\documentclass[11pt,a4paper]{article}
\gdef\@fpheader{}
\pdfoutput=1

\usepackage{jheppub}
\usepackage[T1]{fontenc}
\usepackage{lmodern}
\usepackage{amsthm}
\usepackage{braket}

\newtheorem{definition}{Definition}[section]
\newtheorem{theorem}{Theorem}[section]
\newtheorem{lemma}{Lemma}[section]
\newtheorem{corollary}{Corollary}[section]
\DeclareMathOperator{\Tr}{Tr}
\title{A sharp bound on spacetime distance from quantum entanglement}

\author[a,1]{Zhi-Wei Wang,\note{Zhi-Wei Wang and Arshid Shabir contributed equally to this work.}}
\author[b,1]{Arshid Shabir,}
\author[b,c,d,e]{Mir Faizal,}
\author[f]{and Samuel L. Braunstein}

\affiliation[a]{College of Physics, Jilin University,\\
Changchun 130012, Jilin, People's Republic of China}
\affiliation[b]{Canadian Quantum Research Center,\\
460 Doyle Ave 106, Kelowna, BC V1Y 0C2, Canada}
\affiliation[c]{Irving K. Barber School of Arts and Sciences, University of British Columbia Okanagan,\\
Kelowna, BC V1V 1V7, Canada}
\affiliation[d]{Department of Mathematical Sciences, Durham University,\\
Lower Mountjoy, Stockton Road, Durham DH1 3LE, United Kingdom}
\affiliation[e]{Computational Mathematics Group, Faculty of Sciences, Hasselt University,\\
Agoralaan Gebouw D, Diepenbeek 3590, Belgium}
\affiliation[f]{Computer Science, University of York,\\
Deramore Lane, York YO10 5GH, United Kingdom}

\emailAdd{zhiweiwang.phy@gmail.com}
\emailAdd{aslone186@gmail.com}
\emailAdd{mirfaizalmir@googlemail.com}
\emailAdd{sam.braunstein@york.ac.uk}

\abstract{Ryu-Takayanagi established how boundary entanglement encodes bulk
area. We provide the metric counterpart: boundary mutual information
imposes a rigorous lower bound on bulk geodesic separation that
diverges logarithmically as correlations vanish. A multiscale
iteration promotes this local inequality to a global obstruction to
bulk connectivity. For parallel strips in AdS$_5$/CFT$_4$, the bound
necessitates a quantum resolution of the classical mutual-information
transition and fixes the asymptotic growth of geodesic distance.}


\hypersetup{
  pdftitle={A sharp bound on spacetime distance from quantum entanglement},
  pdfauthor={Zhi-Wei Wang, Arshid Shabir, Mir Faizal, and Samuel L. Braunstein},
  pdfsubject={A mutual-information lower bound on renormalized bulk geodesic length},
  pdfkeywords={AdS-CFT Correspondence; Gauge-Gravity Correspondence; Models of Quantum Gravity; Quantum Information}
}

\begin{document}
\maketitle
\flushbottom

\section{Introduction}\label{sec:introduction}

If spacetime is emergent, its metric data must
be recoverable from nongeometric quantum degrees of freedom.
Entanglement is widely understood to encode bulk connectivity and
geometry~\cite{VanRaamsdonk2020Science,Qi2018NatPhys,Jacobson2016PRL,vanRaamsdonk2010,Maldacena1998},
and holography makes this connection precise for area: the
Ryu-Takayanagi formula~\cite{RyuTakayanagi2006} and its quantum
extremal-surface
extension~\cite{EngelhardtWall2015,Dong2016NatCommun}
relate boundary entropy to the area of a bulk extremal surface. Yet
area does not determine point-to-point separation. Despite important
relations between mixed-state correlations and bulk structures such as
entanglement-wedge cross
sections~\cite{UmemotoTakayanagi2018NatPhys,Caputa2019PRL,Tamaoka2019PRL,Suzuki2019PRL},
no comparably direct inequality has quantified how boundary mutual
information constrains bulk geodesic distance.

This gap is not only conceptual.  Without such an
inequality, one cannot quantify how much boundary correlation is
required to support a connected bulk geometry, nor determine how
spatial separation grows when that correlation is reduced.
Entropy-area relations constrain the areas of extremal surfaces, but
they do not by themselves fix the distances between points.  The
emergence of geometry from entanglement is therefore incomplete at the
metric level: it explains how entanglement determines area, but not
yet how correlation constrains distance.

We derive such an information-distance relation. Boundary mutual
information yields an explicit lower bound on the renormalized bulk
geodesic distance between boundary-accessible probes. Pinsker's
inequality bounds connected measurement correlations in terms of mutual
information, while the semiclassical decay of heavy-operator two-point
functions converts this correlation bound into a distance bound. As the
mutual information vanishes, the resulting lower bound diverges
logarithmically. A multiscale iteration then promotes the pairwise
inequality to a global obstruction: sufficiently weak correlations
across neighboring boundary cells force large bulk separation and
preclude uniformly connected bulk geometry. The complete derivations,
assumptions, and constants are provided in appendices~\ref{app:assumptions}-\ref{app:ads5}.

For parallel strips in the vacuum of a four-dimensional conformal field
theory, the bound has three consequences. At the classical
Ryu-Takayanagi transition, vanishing mutual information would force
divergent separation; finite bulk distance must therefore be supported
by subleading quantum correlations. At large boundary separation, the
bound reproduces the logarithmic scaling of the semiclassical geodesic
distance. For bounded regions and an optimal probe dimension, it
becomes tight and saturates that distance. These results
extend the Ryu-Takayanagi dictionary from entropy-area to
information-distance, turning ``entanglement builds geometry'' from a
qualitative paradigm into a quantitative theorem.

\section{From correlations to geometry}\label{sec:correlations-geometry}

Entanglement first entered the
holographic dictionary through
area~\cite{RyuTakayanagi2006}. Distance poses a sharper question:
how much boundary information is required to keep two bulk probes close?
The answer follows from two controlled relations. Mutual information
bounds every bounded connected correlator, while the heavy-probe
holographic dictionary converts exponential correlator decay into
renormalised geodesic length. Their combination turns boundary
information into a quantitative obstruction to short bulk distance.

Let $(\Sigma,h)$ be a compact $d$-dimensional boundary manifold and let
$A,B\subset\Sigma$ be disjoint regions at strictly positive separation.
The bipartite state is realised either in a regulated tensor
factorisation
$\mathcal H^{(a)}\simeq
\mathcal H_A^{(a)}\otimes
\mathcal H_B^{(a)}\otimes
\mathcal H_{(AB)^c}^{(a)}$
with trace-class reductions, or algebraically by commuting von Neumann
subalgebras satisfying the split property. Mutual information is measured
in bits,
\begin{equation}
\label{eq:MI}
I(A\!:\!B)_\rho
=
(\ln2)^{-1}
D\!\left(
\rho_{AB}\middle\|\rho_A\otimes\rho_B
\right).
\end{equation}
Observables are restricted to a holographic code subspace
$\mathcal H_{\mathrm{code}}$ with projector $\Pi_{\mathrm{code}}$.
A heavy scalar primary enters through the code-compressed,
mean-subtracted insertion
\begin{equation}
\label{eq:code}
\widetilde O^{(0)}_{\Delta,\rho}(x)
:=
\Pi_{\mathrm{code}}
\left(
O_\Delta(x)
-
\langle O_\Delta(x)\rangle_\rho \mathbf 1
\right)
\Pi_{\mathrm{code}},
\end{equation}
with the uniform bound
$\|\widetilde O^{(0)}_{\Delta,\rho}(x)\|_\infty\leq B_\Delta$ on a
separated insertion domain $\mathsf D\subset\Sigma\times\Sigma$. The
semiclassical regime is specified by
$\epsilon_{\mathrm{grav}}
=G_N/\ell_{\mathrm{AdS}}^{d-1}\ll1$,
{$\epsilon_{\mathrm{str}}
=\ell_s^2/\ell_{\mathrm{AdS}}^2\ll1$ (with $\ell_s$ the string length),
$\Delta^2/c_{\mathrm{eff}}\ll1$ (with
$c_{\mathrm{eff}}\propto\ell_{\mathrm{AdS}}^{d-1}/G_N$ the
effective central charge), and the code-subspace truncation error
$\epsilon_{\mathrm{code}}\ll1$.}
% $\epsilon_{\mathrm{str}}
% =\ell_s^2/\ell_{\mathrm{AdS}}^2\ll1$,
% $\Delta^2/c_{\mathrm{eff}}\ll1$, and
% $\epsilon_{\mathrm{code}}\ll1$.
These conditions keep the total
multiplicative error in the geodesic dictionary below unity,
$\epsilon^\star<1$. Appendix~\ref{app:assumptions} gives the separation, regularity,
code-subspace, and normalisation assumptions.
Restricting the operator before taking its norm is
essential: local continuum fields are unbounded, whereas their action
on a regulated low-energy code subspace admits the uniform control
required by the information inequality.

The information-theoretic input is independent of holography. For
observables $O_A$ and $O_B$ of operator norm at most unity, quantum
Pinsker gives \cite{Wolf2008}
\begin{equation}
\label{eq:Pinsker}
\left|
\langle O_AO_B\rangle
-
\langle O_A\rangle\langle O_B\rangle
\right|
\leq
\sqrt{2\ln2\, I(A\!:\!B)} .
\end{equation}
Thus mutual information simultaneously controls all bounded
correlations across the bipartition: small $I(A\!:\!B)$ enforces
approximate factorisation for every such pair of observables. The proof
is given in Appendix~\ref{app:pinsker}.
Unlike a bound on a selected two-point function,
Eq.~\eqref{eq:Pinsker} is uniform over the entire bounded operator
algebra. It therefore supplies a state-dependent correlation budget
that no admissible probe can exceed.

The holographic input supplies the metric dependence. In a
semiclassical bulk, a heavy scalar two-point function is governed by the
renormalised length $L_{\mathrm{ren}}(x,y)$ of the corresponding bulk
geodesic,
\begin{equation}
\label{eq:WKB}
\left\langle
\widetilde O^{(0)}_{\Delta,\rho}(x)\,
\widetilde O^{(0)}_{\Delta,\rho}(y)
\right\rangle_\rho
=
N_\Delta\,
e^{-\Delta L_{\mathrm{ren}}(x,y)}
\left(1+\epsilon_{\mathrm{tot}}\right),
\end{equation}
with $|\epsilon_{\mathrm{tot}}|\leq\epsilon^\star<1$.
Here $N_\Delta$ is state independent. The exponential dependence makes
geodesic length accessible to the universal correlation constraint in
Eq.~\eqref{eq:Pinsker}. Appendix~\ref{app:heavy-probe} derives Eq.~\eqref{eq:WKB} with all errors
controlled.
State dependence enters through the correlator and the
renormalised geodesic, while the probe normalisation and code-subspace
bound remain fixed calibration data. A change of operator convention
therefore cannot be mistaken for a change of geometry.

\section{The metric-from-information bound}\label{sec:mfi-bound}

Define the dimensionless
probe calibration
\begin{equation}
\label{eq:kappa}
\kappa_\Delta
:=
\frac{|N_\Delta|}{B_\Delta^2}.
\end{equation}

\textbf{Theorem 1} (Metric-from-information bound).
\emph{In a semiclassical holographic code subspace, boundary insertions
$x\in A$ and $y\in B$ with $I(A\!:\!B)>0$ obey}
\begin{equation}
\label{eq:MfI}
L_{\mathrm{ren}}(x,y)
\geq
\frac{1}{2\Delta}
\ln\!\left(
\frac{
\kappa_\Delta^2(1-\epsilon^\star)^2
}{
2\ln2\, I(A\!:\!B)
}
\right).
\end{equation}

Equation~\eqref{eq:MfI} is a metric constraint rather than a metric
reconstruction: it excludes any semiclassical geometry in which too
little boundary information supports too short a bulk geodesic. Each
factor-of-$e$ decrease in mutual information increases the minimum
renormalised length by $1/(2\Delta)$. Consequently, vanishing mutual
information forces the allowed bulk separation to diverge
logarithmically (Fig.~\ref{fig:schematic}). Conversely, a
proposed finite geodesic demands a strictly positive information budget
between the boundary regions supporting its endpoint probes. Boundary
correlations can therefore rule out candidate semiclassical metrics
without requiring their point-by-point reconstruction.

The bound introduces neither a correlation length nor a mass gap. Its
only dynamical input is the geodesic-exponentiated two-point
function~\eqref{eq:WKB}; the information-theoretic dependence is
nonperturbative. Appendix~\ref{app:mfi}
proves monotonic sharpening as $I(A\!:\!B)\to0$ and gives the
code-subspace error analysis. For code subspaces with uniformly bounded
simple-probe matrix elements, including the standard ETH regime for
simple single-trace primaries, $B_\Delta=O(1)$ and $\kappa_\Delta$ is
an order-one calibration. With a UV regulator, cutoff dependence in
$B_\Delta$ shifts the bound only by the additive term
$\Delta^{-1}\ln\kappa_\Delta$; neither the universal rate
$(2\Delta)^{-1}$ nor the asymptotic scaling changes.

The essential inequality is immediate. Applying
Eq.~\eqref{eq:Pinsker} to
$\widehat O=\widetilde O^{(0)}_{\Delta,\rho}/B_\Delta$ gives
\begin{equation}
\label{eq:chain}
\frac{
|N_\Delta|\,
e^{-\Delta L_{\mathrm{ren}}(x,y)}\,
|1+\epsilon_{\mathrm{tot}}|
}{
B_\Delta^2
}
\leq
\sqrt{2\ln2\, I(A\!:\!B)} .
\end{equation}
Because $|1+\epsilon_{\mathrm{tot}}|\geq1-\epsilon^\star>0$, taking
logarithms yields Eq.~\eqref{eq:MfI}. The square root in
Pinsker fixes the universal factor $1/(2\Delta)$; gravitational,
stringy, and code-subspace corrections alter the finite calibration but
not the logarithmic dependence on mutual information.

\begin{figure}[t]
\centering
\includegraphics[width=0.92\textwidth]{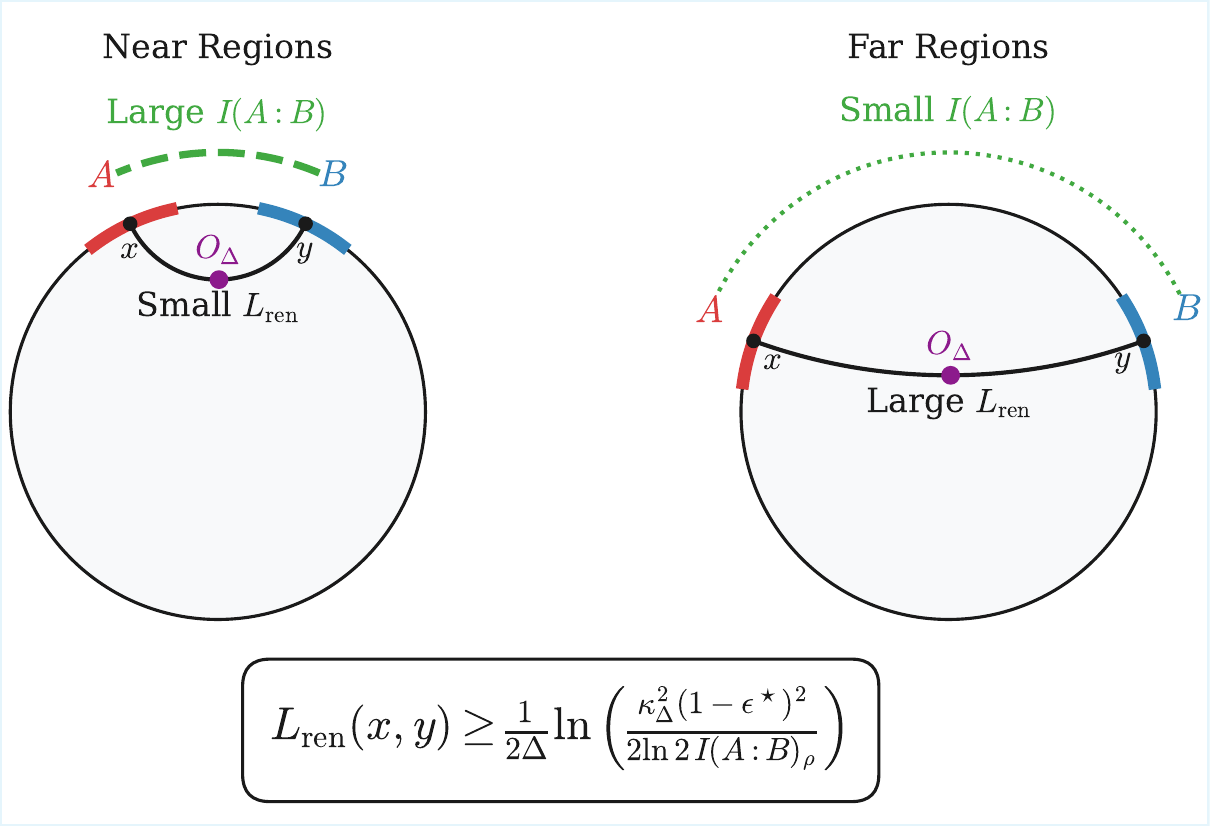}
\caption{\textbf{The metric-from-information bound.} Left: when boundary
regions $A$ and $B$ are nearby, so that $I(A\!:\!B)$ is large, the MfI
bound permits a short bulk geodesic between the corresponding points
$x$ and $y$. Right: when $A$ and $B$ are far apart, so that $I(A\!:\!B)$
is small, the bound forces the bulk geodesic to be long. A heavy probe
$O_\Delta$ propagates along the bulk geodesic of renormalised length
$L_{\mathrm{ren}}(x,y)$.}
\label{fig:schematic}
\end{figure}

\section{From one distance to bulk connectivity}\label{sec:multiscale}

A pairwise bound
constrains one geodesic; bulk connectivity requires compatible distances
at every resolution. Two complementary consequences make this passage
precise: a local information metric controls infinitesimal geometric
response, and a multiscale hierarchy converts edgewise information into
a bound on bulk diameter.

For a smooth family $\rho(\lambda)$, the quadratic expansion of relative
entropy defines the Bogoliubov-Kubo-Mori (BKM) information metric
$g(\lambda)$. For sufficiently small $|\delta|_{g(\lambda)}$,
Appendix~\ref{app:mfi} establishes
\begin{equation}
\label{eq:BKM}
\frac{1}{4}|\delta|_{g(\lambda)}^2
\leq
D\!\left(
\rho(\lambda)\middle\|\rho(\lambda+\delta)
\right)
\leq
\frac{3}{4}|\delta|_{g(\lambda)}^2
\end{equation}
with a remainder controlled by the third derivative of the state family.
Equation~\eqref{eq:BKM} limits infinitesimal geometric
response under boundary deformations, whereas Eq.~\eqref{eq:MfI}
constrains finite long-range separation. The two inequalities probe
different regimes of the same information geometry: BKM curvature
controls the local tangent structure of state space, while mutual
information constrains nonlocal separation in the semiclassical bulk.

At resolution $k$, tile the boundary by cells $\{A_i^{(k)}\}$ forming a
connected adjacency graph $G_k$ of graph diameter $D_k$. The MfI bound
assigns a minimum length to every adjacent pair. Controlled chaining
along shortest paths then promotes these local constraints to a global
one. The alignment error $\delta$ measures the failure of
consecutive edge estimates to lie on a common bulk geodesic; it is the
only loss in passing from the pairwise inequality to the diameter
estimate.

\textbf{Theorem 2} (Multiscale diameter bound).
\emph{Suppose that at each scale $k$ the MfI bound (Theorem~1) applies
to every pair of neighbouring cells, and that the pairwise distance
estimates can be chained along paths with a controlled alignment error
$\delta\geq0$ per step. Then}
\begin{equation}
\label{eq:Diam}
\mathrm{Diam}_\rho
\geq
\max_{0\leq k\leq K}
\left(
D_k\,\ell^{\mathrm{MI}}_k\!\left(I_k(\rho)\right)
-
2(D_k-1)\delta
\right),
\end{equation}
\emph{where
$I_k(\rho):=\max_{(i,j)\in E_k}I^{(k)}_{ij}(\rho)$
is the worst-case neighbour mutual information at scale $k$ and
$\ell^{\mathrm{MI}}_k$ converts mutual information to edge length via
Theorem~1.}

The maximum in Eq.~\eqref{eq:Diam} makes connectivity a scale-by-scale
constraint. If neighbouring cells lack sufficient mutual information at
even one resolution, the minimum allowed bulk diameter grows. This
structure parallels tensor-network and programmable-geometry models in
which correlations encode effective distance rather than merely living
on a fixed metric
background~\cite{Steinberg2023NatCommun,Periwal2021Nature}.
A well-correlated coarse partition cannot compensate for an
information-poor finer partition: the most restrictive resolution
controls the diameter. Connectivity is consequently a hierarchy of
correlation constraints, not a property of one preferred bipartition.

\section{Holographic tests: parallel strips and QES smoothing}\label{sec:strips}

The most informative test of the MfI bound is
at a point where classical holography makes an abrupt connectivity
decision. At an entanglement-wedge transition, classical RT mutual
information vanishes while the background geodesic remains finite.
Equation~\eqref{eq:MfI} makes these two statements incompatible and
therefore exposes the quantum corrections that must resolve the
classical transition.

For disjoint $A$ and $B$, connected and disconnected extremal surfaces
compete. The quantum extremal-surface (QES) prescription includes both
area and bulk entropy, whereas RT retains only area. 
{If the bulk-entropy difference is uniformly bounded,
$|\Delta S_{\mathrm{bulk}}|\leq S^\star$, the prescriptions can
disagree only in a window of width $4G_N S^\star$}
% If $|\Delta S_{\mathrm{bulk}}|\leq S^\star$, the prescriptions can
% disagree only in an $O(G_N)$ window
where the classical area difference is
comparable to
$4G_N\Delta S_{\mathrm{bulk}}$~\cite{RyuTakayanagi2006,Dong2016NatCommun,DongHarlowWall2016}.
Outside this window, the RT connectivity decision is stable
(Appendix~\ref{app:rt-qes}).
Inside it, the nominally subleading bulk entropy determines
which saddle governs the generalized entropy. The transition is
therefore an especially sensitive test of compatibility between
information-derived distance and semiclassical geometry.

Consider two parallel strips of width $\ell$ separated by $s$ in the
large-$N$, strongly coupled vacuum of $\mathcal N=4$ $SU(N)$ super
Yang-Mills theory. The dual is pure $\mathrm{AdS}_5$, for which the
heavy-primary two-point function gives
$L_{\mathrm{ren}}(x,y)=2\ln|x-y|$.
The RT transition occurs at $s_c=(\sqrt3-1)\ell$.
The wedge is connected and $I(A\!:\!B)>0$ for $s<s_c$, while classical
RT gives $I(A\!:\!B)=0$ for $s>s_c$. Approaching the transition from the
connected side,
\begin{equation}
\label{eq:MI-strip}
I(A\!:\!B)
\sim
\Gamma_N
\frac{3\sqrt3\,(s_c-s)}{\ell^3},
\qquad
s\to s_c^- ,
\end{equation}
where
{$\Gamma_N :=c_1\ell_{\mathrm{AdS}}^3L_2L_3/(4G_5\ln2)
=\Theta(N^2)$,
$L_2 L_3$ is the regulated transverse area of the strips, and
$c_1
=4\pi^{3/2}\Gamma(2/3)^3/\Gamma(1/6)^3
\approx0.321$.}
% $\Gamma_N :=c_1\ell_{\mathrm{AdS}}^3L_2L_3/(4G_5\ln2) =\Theta(N^2)$
% with $c_1 =4\pi^{3/2}\Gamma(2/3)^3/\Gamma(1/6)^3 \approx0.321$.
%
Substitution into Eq.~\eqref{eq:MfI} yields
\begin{equation}
\label{eq:log-div}
L^{\mathrm{min}}_{\mathrm{ren}}(A,B)
\geq
\frac{1}{2\Delta}
\ln\frac{1}{s_c-s}
+
O(1).
\end{equation}
The classical MfI lower bound therefore diverges with the universal rate
$(2\Delta)^{-1}$ as $s\to s_c^-$
(Fig.~\ref{fig:strip}).

The actual geodesic in pure $\mathrm{AdS}_5$ remains finite at $s_c$.
The divergence is instead a contradiction generated by inserting the
classical RT zero into an information-theoretic inequality obeyed by the
boundary state. Finite bulk distance requires nonzero mutual
information; hence bulk-entropic QES corrections must replace the
classical zero by a strictly positive quantum contribution. The MfI
bound thereby turns QES smoothing from a refinement of the RT transition
into a consistency condition for finite semiclassical distance.
This conclusion is independent of the detailed smoothing
profile. It follows from three facts alone: the exact information
inequality, the heavy-probe geodesic relation, and a finite
semiclassical geodesic across the classical transition.

Away from the transition, the same bound recovers the correct metric
scaling. At large separation $s\gg\ell$, mutual information between
bounded regions such as spheres behaves as
$I(A\!:\!B)\sim s^{-4\Delta_{\mathrm{min}}}$,
because the leading relative-entropy contribution is quadratic in the
connected correlator and exchanges two copies of the lowest-dimension
primary. Eq.~\eqref{eq:MfI} then gives
$L_{\mathrm{ren}}
\geq
(2\Delta_{\mathrm{min}}/\Delta)\ln s+O(1)$.
For $\Delta=\Delta_{\mathrm{min}}$, the leading coefficient equals the
exact pure-AdS result $L_{\mathrm{ren}}=2\ln s$; the logarithmic bound
is asymptotically saturated. For infinite strips, transverse integration
changes the power to $s^{-(4\Delta_{\mathrm{min}}-d+2)}$, leaving the
bound strict but not saturated (Appendix~\ref{app:ads5}). Although the general derivation uses heavy
probes, conformal symmetry fixes
$\langle O_\Delta(x)O_\Delta(y)\rangle\propto|x-y|^{-2\Delta}$ exactly
for every $\Delta$ in the pure-AdS vacuum. The saturation statement
therefore extends to $\Delta=\Delta_{\mathrm{min}}$ in that case; in
general states it remains a large-$\Delta$ asymptotic result.
The two limits test complementary consequences: the
transition detects the quantum correction demanded by finite distance,
while the large-separation regime verifies the predicted metric scaling
where the semiclassical description is smooth.

\begin{figure}[t]
\centering
\includegraphics[width=0.92\textwidth]{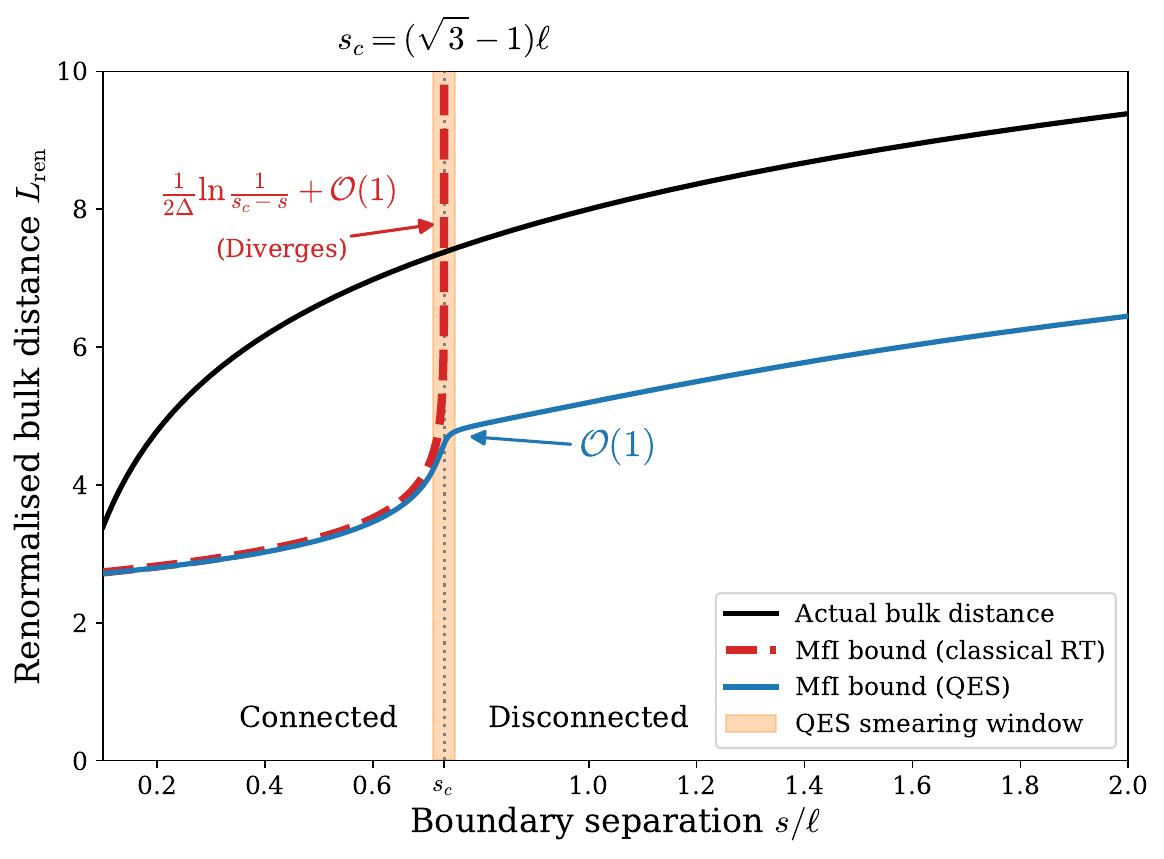}
\caption{\textbf{MfI bound at the entanglement-wedge phase transition
for parallel strips in $\mathrm{AdS}_5/\mathrm{CFT}_4$.} Black: the
actual bulk geodesic distance, which is finite everywhere. Red dashed:
the MfI bound evaluated with the classical RT mutual information, which
diverges logarithmically at the critical separation
$s_c=(\sqrt3-1)\ell$. Blue solid: the MfI bound evaluated with the
quantum-corrected QES mutual information, which smoothly plateaus at an
$O(1)$ value in the disconnected phase. The orange band marks the QES
smearing window of width $\sim\ell/N^2$. The divergence of the classical
bound shows that quantum gravitational corrections to the entropy must
keep $I(A\!:\!B)$ strictly positive wherever the corresponding
semiclassical geodesic remains finite.}
\label{fig:strip}
\end{figure}

\section{Information beyond entropy}\label{sec:beyond-entropy}

The MfI bound uses mutual
information, not subsystem entropy. This distinction matters: subsystem
entropies can agree while the mutual-information pattern that controls
metric connectivity differs entirely.

It is natural to ask whether the Page curve, which controls subsystem
entropies in typical states, already captures the connectivity
information contained in the MfI
bound~\cite{Page1993}. The answer is no:
subsystem entropy does not determine mutual information between
subregions. Explicit pure-state families with identical
subsystem entropy and distinct mutual-information structure are given in
Appendix~\ref{app:typicality}. Entropy fixes
the spectrum of an individual reduction; distance depends instead on how
correlations are distributed between spatially separated reductions. An
entropy curve, even when known exactly, therefore does not determine
metric connectivity.

Conditional mutual information supplies the complementary
reconstruction criterion. For a tripartite state with
$\epsilon:=I(A\!:\!C|B)$, the Fawzi-Renner approximate-Markov
bound~\cite{FawziRenner2015} guarantees a recovery channel
(the twirled Petz map) that reconstructs $\rho_{ABC}$ from
$\rho_{AB}$ with fidelity
$F(\rho_{ABC},\sigma_{ABC})\geq 2^{-\epsilon/2}$.
Small conditional mutual information thus gives a quantitative,
state-dependent recovery guarantee~\cite{Brandao2015PRL};
Appendix~\ref{app:markov} gives the near-Markov
holographic construction. The two information measures
constrain opposite sides of emergent geometry: small conditional mutual
information enables reconstruction, whereas small mutual information
forbids short distance.

\section{Discussion}\label{sec:discussion}

The Ryu-Takayanagi formula converts entropy to
area; the MfI bound converts mutual information to distance.  Together
they make the entanglement-geometry correspondence quantitative at both
the area and the metric level.  Three features of the distance bound
merit emphasis.

First, it is an explicit inequality, not an
order-of-magnitude estimate. The bound is nonperturbative
in its information-theoretic input and controlled in its semiclassical
errors. Bulk dynamics enter only through probe calibration,
making the logarithmic dependence on mutual information universal within
the stated code subspace.

Second, the multiscale theorem promotes this pairwise bound to a global
connectivity criterion. A connected semiclassical bulk requires
sufficient mutual information between neighbouring boundary cells at
every resolution; neither one bipartition nor the Page curve can certify
it. Connectivity is therefore a quantitative property of the boundary's
multiscale correlation structure, rather than a topological
label attached to a chosen saddle.

Third, the parallel-strip transition makes this constraint
sharp. Classical RT sets the mutual information to zero while the
background geodesic remains finite; the MfI bound makes these statements
incompatible and therefore requires quantum corrections to preserve a
nonzero correlation floor. Far from the transition, the same bound
recovers the logarithmic growth of the AdS geodesic from algebraically
decaying boundary correlations. Whenever
$I\simeq\gamma\,\delta^p$, it predicts
$L\geq[p/(2\Delta)]\ln(1/\delta)+O(1)$.
Thus quantum correlations determine not only whether a semiclassical
bulk is connected, but the minimum distance and diameter that any
consistent geometry may assign.

This perspective suggests several applications wherever
reconstructability changes with correlations. In island transitions in
JT gravity, the MfI bound could provide a direct information-theoretic
measure of the geometric separation associated with Page-curve
reconstruction. In multi-boundary wormholes, it could distinguish not
only whether a wormhole is connected but how connectivity is distributed
among its boundaries. In evaporating black holes, the time-dependent
loss and recovery of boundary correlations should produce corresponding
time-dependent constraints on bulk diameter. Taken together, these
directions point toward a general diagnostic for when, where and to what
extent quantum correlations give rise to connected spacetime.

% \begin{acknowledgments}
% [Acknowledgements to be added.]
% \end{acknowledgments}

\appendix
% The complete former Supplemental Material is incorporated below as appendices.
% Only its document wrapper and journal style have been changed.

\section{Assumptions and scope}\label{app:assumptions}
\renewcommand{\theequation}{A.\arabic{equation}}

Let $d\in\mathbb{N}$, let $(\Sigma,h)$ be a compact $d$-dimensional Riemannian manifold, 
and let $A,B\subset\Sigma$ be measurable subsets with strictly positive separation in $h$,
\begin{equation}\label{A1}
	A\cap B=\varnothing,\qquad 
	\operatorname{dist}_h(A,B):=\inf\bigl\{d_h(x,y):x\in A,\ y\in B\bigr\}\ge s_0>0.
\end{equation}

Fix a complex separable Hilbert space $\mathcal{H}$ and a state $\rho$ on the physical boundary 
degrees of freedom, presented in one of the following two mathematically explicit realizations.

In the regulated realization one fixes a UV regulator scale $a>0$ and a regulated Hilbert space 
{$\mathcal{H}^{(a)}$} together with an isomorphism
\begin{equation}\label{A2-Reg}
	{\mathcal{H}^{(a)}\simeq \mathcal{H}_{A}^{(a)}\otimes\mathcal{H}_{B}^{(a)}\otimes \mathcal {H}_{(A \cup B)^c}^{(a)},}
\end{equation}
and one assumes that $\rho$ is represented by a positive trace-class operator on {$\mathcal{H}^{(a)}$ }
with $\operatorname{Tr}\rho=1$ and that the reduced operators $\rho_{AB}$, $\rho_A$, $\rho_B$ 
exist as trace-class operators obtained by partial trace with respect to the factorization Eq.~\eqref{A2-Reg}.

In the algebraic realization one fixes a von Neumann algebra $\mathcal{M}\subset\mathcal{B}(\mathcal{H})$ 
and commuting von Neumann subalgebras $\mathcal{A}(A),\mathcal{A}(B)\subset\mathcal{M}$ with
\begin{equation}\label{A2-AQFT}
	[\mathcal{A}(A),\mathcal{A}(B)]=0,
\end{equation}
and one assumes the split property for the separated pair $(A,B)$ in the form
\begin{equation}\label{A2-Split}
	\exists\ \text{type I factor }\mathcal{N}\subset\mathcal{M}\ \text{such that}\  
	\mathcal{A}(A)\subset \mathcal{N}\subset \mathcal{A}(B)^{\prime},
\end{equation}
so that the von Neumann algebra generated by $\mathcal{A}(A)$ and $\mathcal{A}(B)$ admits a spatial 
tensor product structure compatible with restrictions of normal states 
\cite{Haag1996,BuchholzDAntoniLongo1987}.

In either realization one assumes that $\rho$ is normal and that the bipartite product reference 
state with the same marginals exists on the joint algebra and yields finite relative entropy,
\begin{equation}\label{A3}
	D(\rho_{AB}|\rho_A\otimes\rho_B)<\infty,
\end{equation}
where in the regulated realization $\rho_{AB},\rho_A,\rho_B$ are density operators on 
$\mathcal{H}_A^{(a)}\otimes\mathcal{H}_B^{(a)}$, $\mathcal{H}_A^{(a)}$, $\mathcal{H}_B^{(a)}$ 
respectively, while in the algebraic realization they denote the restrictions of $\rho$ to 
$\mathcal{A}(A)\vee\mathcal{A}(B)$, $\mathcal{A}(A)$, $\mathcal{A}(B)$, and $\rho_A\otimes\rho_B$ 
denotes a normal product state on $\mathcal{A}(A)\,\overline{\otimes}\,\mathcal{A}(B)$ with the same 
marginals as $\rho$; $D$ denotes the Umegaki relative entropy in the regulated case and the Araki 
relative entropy in the algebraic case \cite{Araki1976,OhyaPetz2004,Watrous2018}.

Fix a holographic code subspace $\mathcal{H}_{\mathrm{code}}\subset\mathcal{H}$ and let 
$\Pi_{\mathrm{code}}$ be the orthogonal projection onto $\mathcal{H}_{\mathrm{code}}$; one assumes 
state support on the code,
\begin{equation}\label{A4}
	\rho=\Pi_{\mathrm{code}}\rho\Pi_{\mathrm{code}},
\end{equation}
and one assumes that all observables used in bipartite inequalities are bounded operators on 
$\mathcal{H}_{\mathrm{code}}$ obtained either as elements of the local bounded operator algebras 
in the regulated/algebraic sense or as compressions by $\Pi_{\mathrm{code}}$ 
\cite{Haag1996,AlmheiriDongHarlow2015,Harlow2016,Pastawski2015}.

Fix a family of probe operators $O_\Delta$ $(\Delta\in\mathcal{D})$ indexed by scaling dimension 
$\Delta\ge\Delta_{\min}>0$ and defined on a common dense domain containing $\mathcal{H}_{\mathrm{code}}$, 
and define the code-compressed mean-subtracted probe at a boundary insertion point $x\in\Sigma$ by
\begin{equation}
	\label{eq:otilde-def}
	\widetilde O^{(0)}_{\Delta,\rho}(x)
	:=
	\Pi_{\mathrm{code}}
	\bigl(O_\Delta(x)-\langle O_\Delta(x)\rangle_\rho\mathbf 1\bigr)
	\Pi_{\mathrm{code}}.
\end{equation}
	
	Use the uniform state- and domain-level definition
	\begin{equation}
		\label{eq:uniform-B-again}
		B_\Delta
		:=
		\sup_{\rho\in\mathcal{S}_{\mathrm{code}}}
		\sup_{(x,y)\in \mathsf{D}}
		\max\left\{
		\left\|\widetilde O^{(0)}_{\Delta,\rho}(x)\right\|_\infty,
		\left\|\widetilde O^{(0)}_{\Delta,\rho}(y)\right\|_\infty
		\right\}
		<\infty.
	\end{equation}
	If the code subspace is finite-dimensional and the compressed operators are defined everywhere on it, boundedness is automatic. If the code subspace is infinite-dimensional, compression alone does not automatically guarantee boundedness; Eq.~\eqref{eq:uniform-B-again} is then a standing uniform-boundedness hypothesis.
	
	The physical heavy-probe dictionary is first stated for the continuum primary $O_\Delta$. The information-theoretic theorem, however, uses only bounded code-compressed insertions. We assume that, for $x\in A$ and $y\in B$, the compressed insertions $\widetilde O^{(0)}_{\Delta,\rho}(x)$ and $\widetilde O^{(0)}_{\Delta,\rho}(y)$ are admissible bounded representatives of the separated $A$- and $B$-side code observables, so that the bipartite Pinsker inequality applies to them.
	
	We further assume that, uniformly on the insertion domain $\mathsf{D}$ and on the state class $\mathcal{S}_{\mathrm{code}}$, the compressed connected correlator satisfies the relative multiplicative representation
	\begin{equation}
		\label{eq:SI-code-assumption}
		C_\rho\!\left(
		\widetilde O^{(0)}_{\Delta,\rho}(x),
		\widetilde O^{(0)}_{\Delta,\rho}(y)
		\right)
		=
		N^{\mathrm{code}}_\Delta
		e^{-\Delta L_\rho(x,y)}
		\left(1+\varepsilon^{\mathrm{code}}_{\Delta,\rho}(x,y)\right),
	\end{equation}
	with
	\begin{equation}
		\sup_{\rho\in\mathcal{S}_{\mathrm{code}}}\sup_{(x,y)\in \mathsf{D}}
		\left|\varepsilon^{\mathrm{code}}_{\Delta,\rho}(x,y)\right|
		\le
		\varepsilon^\star_\Delta<1.
	\end{equation}
	The code normalization $N^{\mathrm{code}}_\Delta$ is nonzero and state-independent within the chosen code subspace. Projection and truncation errors are included in $\varepsilon^\star_\Delta$ as relative multiplicative errors. When no confusion can arise, we write $N_\Delta$ for $N^{\mathrm{code}}_\Delta$.

	\begin{equation}\label{A7}
		\varepsilon_{\mathrm{grav}}:=\frac{G_N}{\ell_{\mathrm{AdS}}^{d-1}},\qquad
		\varepsilon_{\mathrm{str}}:=\frac{\ell_s^2}{\ell_{\mathrm{AdS}}^2},\qquad
		c_{\mathrm{eff}}:=\frac{\ell_{\mathrm{AdS}}^{d-1}}{G_N}
		=\varepsilon_{\mathrm{grav}}^{-1}.
	\end{equation}
	Fix constants $\varepsilon_{\mathrm{grav}}^{\star},\varepsilon_{\mathrm{str}}^{\star},\varepsilon_{\mathrm{br}}^{\star}\in(0,1)$
	and define the dimensionless probe-backreaction parameter
	\begin{equation}\label{A8a}
		\varepsilon_{\mathrm{br}}(\Delta):=\frac{\Delta^2}{c_{\mathrm{eff}}}
		=\varepsilon_{\mathrm{grav}}\Delta^2.
	\end{equation}
	Assume the semiclassical and probe-backreaction regime
	\begin{equation}\label{A8}
		\varepsilon_{\mathrm{grav}}\le \varepsilon_{\mathrm{grav}}^{\star},\qquad
		\varepsilon_{\mathrm{str}}\le \varepsilon_{\mathrm{str}}^{\star},\qquad
		\varepsilon_{\mathrm{br}}(\Delta)\le \varepsilon_{\mathrm{br}}^{\star}
		\quad\text{for all }\Delta\in\mathcal{D},
	\end{equation}
	together with the existence of a state-independent probe normalization factor $N_\Delta\in(0,\infty)$
	and a renormalized dimensionless length function $L_\rho:\mathsf{D}\to[0,\infty)$ such that the
	two-point function admits the uniform multiplicative representation
	\begin{equation}\label{A9}
		\langle O_\Delta(x)O_\Delta(y)\rangle_\rho
		=
		N_\Delta\,\exp\!\bigl(-\Delta\,L_\rho(x,y)\bigr)\,
		\bigl(1+\varepsilon_{\Delta,\rho}(x,y)\bigr),
		\qquad
		\sup_{(x,y)\in\mathsf{D}}|\varepsilon_{\Delta,\rho}(x,y)|
		\le \varepsilon_\Delta^{\star}<1,
	\end{equation}
	{
		with explicit error control
		\begin{equation}
			\label{eq:error-budget-code}
			\varepsilon^\star_\Delta
			\le
			\frac{C_{\mathrm{WKB}}}{\Delta}
			+
			C_{\mathrm{grav}}\varepsilon_{\mathrm{grav}}
			+
			C_{\mathrm{str}}\varepsilon_{\mathrm{str}}
			+
			C_{\mathrm{br}}\frac{\Delta^2}{c_{\mathrm{eff}}}
			+
			C_{\mathrm{code}}\varepsilon_{\mathrm{code}}.
		\end{equation}
		Here $\varepsilon_{\mathrm{code}}=0$ in an exact invariant code-subspace realization. Otherwise it measures the uniform relative projection/truncation error in the multiplicative correlator representation Eq.~\eqref{eq:SI-code-assumption}. Fixed constants $C_{\mathrm{WKB}},C_{\mathrm{grav}},C_{\mathrm{str}},C_{\mathrm{br}},C_{\mathrm{code}}\ge 0$ are independent of $\rho$, $(x,y)\in\mathsf{D}$, and $\Delta\in\mathcal{D}$,
	}
	and with $L_\rho$ defined in a fixed renormalization 
	scheme for the universal near-boundary divergence of the corresponding bulk geodesic length in 
	asymptotically AdS saddles 
	\cite{Maldacena1998,GKP1998,Witten1998,Aharony2000,Balasubramanian1999,LoukoMarolfRoss2000,Skenderis2002}.
	
	{Fix a finite family of boundary coarse-grainings indexed by
		\(k \in \{0,1,\dots,K\}\), with measurable partitions
		\(
		\mathcal{P}_k=\{C_i^{(k)}\}_{i\in V_k}
		\)
		of \(\Sigma\). We reserve the notation \(C_i^{(k)}\) for the coarse cells,
		and \(A_i^{(k)} \subset C_i^{(k)}\) for the corresponding buffered cells
		introduced later. Define the adjacency graphs \(G_k=(V_k,E_k)\) by
		\begin{equation}\label{A11}
		\begin{aligned}
			(i,j)\in E_k
			&\Longleftrightarrow
			\overline{C_i^{(k)}}\cap \overline{C_j^{(k)}}
			\text{ contains a set}\\
			&\qquad\text{of positive \((d-1)\)-dimensional Hausdorff measure}.
		\end{aligned}
		\end{equation}}
	and assume a coarse additivity inequality for the operational length along shortest adjacency paths: 
	there exists $\delta\ge 0$ such that for every $k$, every shortest path 
	$v_0\sim v_1\sim\cdots\sim v_n$ in $G_k$ (so $n$ equals the graph distance) and every $\rho$ in 
	the state class under consideration, there exist points $x_m\in A_{v_m}^{(k)}$ such that
	\begin{equation}\label{A12}
		L_\rho(x_0,x_n)\ge \sum_{m=1}^{n}L_\rho(x_{m-1},x_m)-2(n-1)\delta.
	\end{equation}
	
	Assume finally that $\mathcal{H}_{\mathrm{code}}$ consists of semiclassical holographic states for which boundary von Neumann entropies of regions $R\subset\Sigma$ are computed by the quantum extremal surface prescription in a fixed renormalization scheme: for each $R$ there exists a nonempty set $\mathcal{X}_\rho(R)$ of codimension-two \emph{quantum extremal} bulk surfaces $\chi$ homologous to $R$ in the relevant saddle, and a bulk Cauchy region $\Sigma(\chi;R)$ bounded by $R\cup\chi$, such that
	\begin{equation}\label{A13}
		S_2(\rho_R)=\frac{1}{\ln 2}{\min}_{\chi\in \mathcal{X}_\rho(R)}
		\Bigl(
		\frac{\mathrm{Area}_\rho(\chi)}{4G_N}+S^{\mathrm{bulk},\rho}_{\Sigma(\chi;R)}
		\Bigr)
		+\eta_\rho(R),
	\end{equation}
	{where $S_2(\rho_R):=-(\operatorname{Tr}\rho_R\log_2\rho_R)$ in the regulated realization, $S^{\mathrm{bulk},\rho}_{\Sigma(\chi;R)}$ is the renormalized bulk von Neumann entropy of quantum fields on $\Sigma(\chi;R)$ in the same scheme as the area counterterms, and the minimum in Eq.~\eqref{A13} is taken over all $\chi\in \mathcal{X}_\rho(R)$ (assumed attained, but not assumed unique)}, and $\eta_\rho(R)$ is a controlled 
	remainder satisfying
	\begin{equation}\label{A14}
		\sup_{\rho\ \mathrm{on}\ \mathcal{H}_{\mathrm{code}}}\ \sup_{R\in\mathcal{R}}\ |\eta_\rho(R)|\le \eta^{\star}<\infty
	\end{equation}
	for the fixed region family $\mathcal{R}$ under consideration 
	\cite{RyuTakayanagi2006,HubenyRangamaniTakayanagi2007,LewkowyczMaldacena2013,FaulknerLewkowyczMaldacena2013,EngelhardtWall2015,Wall2014}. 
	For a trace-class density operator $\sigma$ on a Hilbert space, define the von Neumann entropy in 
	natural-logarithm units and in base-$2$ units by
	\begin{equation}\label{D1}
		S(\sigma):=-\operatorname{Tr}(\sigma\ln\sigma),\qquad 
		S_2(\sigma):=-\operatorname{Tr}(\sigma\log_2\sigma)=\frac{1}{\ln 2}\,S(\sigma).
	\end{equation}
	
	For density operators $\sigma,\tau$ satisfying the support condition 
	$\operatorname{supp}(\sigma)\subseteq \operatorname{supp}(\tau)$, define the Umegaki relative entropy 
	in nats by
	\begin{equation}\label{D2}
		D(\sigma|\tau):=\operatorname{Tr}\bigl[\sigma(\ln\sigma-\ln\tau)\bigr]\in[0,\infty),
	\end{equation}
	and set $D(\sigma|\tau):=+\infty$ when $\operatorname{supp}(\sigma)\nsubseteq \operatorname{supp}(\tau)$; 
	for normal states on a von Neumann algebra, $D$ denotes the Araki relative entropy 
	\cite{Araki1976,OhyaPetz2004,Watrous2018}.
	
	In the regulated realization Eq.~\eqref{A2-Reg}, define the reduced density operators by
	\begin{equation}\label{D3}
		\rho_{AB}:=\operatorname{Tr}_{ (A \cup B)^c}\rho,\qquad 
		\rho_A:=\operatorname{Tr}_B\rho_{AB},\qquad 
		\rho_B:=\operatorname{Tr}_A\rho_{AB},
	\end{equation}
	and define mutual information in bits by
	\begin{equation}\label{D4}
		I(A:B)_\rho:=S_2(\rho_A)+S_2(\rho_B)-S_2(\rho_{AB}).
	\end{equation}
	
	In the algebraic realization Eq.~\eqref{A2-AQFT}-\eqref{A2-Split}, define mutual information in bits 
	by the relative-entropy identity
	\begin{equation}\label{D5}
		I(A:B)_\rho:=\frac{1}{\ln 2}\,D(\rho_{AB}|\rho_A\otimes\rho_B),
	\end{equation}
	which is well-defined and finite by Eq.~\eqref{A3} and by the split-property hypothesis ensuring 
	existence of $\rho_A\otimes\rho_B$ on the joint algebra 
	\cite{Haag1996,BuchholzDAntoniLongo1987,Araki1976}.
	In the regulated realization, the definitions Eq.~\eqref{D1}-\eqref{D4} imply the relative-entropy 
	representation of mutual information with explicit base conversion,
	\begin{align}\label{Der-1}
		D(\rho_{AB}|\rho_A\otimes\rho_B)
		&=\operatorname{Tr}\!\left[\rho_{AB}\bigl(\ln\rho_{AB}-\ln(\rho_A\otimes\rho_B)\bigr)\right]\notag\\
		&=\operatorname{Tr}\!\left[\rho_{AB}\ln\rho_{AB}\right]-\operatorname{Tr}\!\left[\rho_{AB}\ln(\rho_A\otimes\rho_B)\right]\notag\\
	&=\operatorname{Tr}\!\left[\rho_{AB}\ln\rho_{AB}\right]-\operatorname{Tr}\!\left[\rho_{AB}\bigl(\ln\rho_A\otimes \mathbf{1}_B+\mathbf{1}_A\otimes \ln\rho_B\bigr)\right]\notag\\
		&=\operatorname{Tr}\!\left[\rho_{AB}\ln\rho_{AB}\right]-\operatorname{Tr}\!\left[(\operatorname{Tr}_B\rho_{AB})\ln\rho_A\right]-\operatorname{Tr}\!\left[(\operatorname{Tr}_A\rho_{AB})\ln\rho_B\right]\notag\\
		&=\operatorname{Tr}\!\left[\rho_{AB}\ln\rho_{AB}\right]-\operatorname{Tr}\!\left[\rho_A\ln\rho_A\right]-\operatorname{Tr}\!\left[\rho_B\ln\rho_B\right]\notag\\
		&=-S(\rho_{AB})+S(\rho_A)+S(\rho_B)\notag\\
		&=(\ln 2)\Bigl(S_2(\rho_A)+S_2(\rho_B)-S_2(\rho_{AB})\Bigr)\notag\\
		&=(\ln 2)\,I(A:B)_\rho,
	\end{align}
	where $\ln(\rho_A\otimes\rho_B)=\ln\rho_A\otimes \mathbf{1}_B+\mathbf{1}_A\otimes \ln\rho_B$ holds 
	by functional calculus on tensor products and the fourth line uses 
	$\operatorname{Tr}_{AB}[\rho_{AB}(X_A\otimes\mathbf{1}_B)]=\operatorname{Tr}_A[(\operatorname{Tr}_B\rho_{AB})X_A]$.
	The code-subspace boundedness in Eq.~\eqref{eq:uniform-B-again} is ensured by the following finite-dimensional 
	compression statement, which is invoked only when $\dim\mathcal{H}_{\mathrm{code}}<\infty$ is assumed.
	{
		\begin{lemma}\label{lem:finitecompression}
			Let $\mathcal{K}$ be a finite-dimensional Hilbert space, let
			$T:\mathcal{K}\to\mathcal{K}$ be a linear map, and define
			\begin{equation}
				|T|_\infty:=\sup\{|T\psi|:\psi\in\mathcal{K},\ |\psi|=1\}.
			\end{equation}
			Then $|T|_\infty<\infty$. Moreover, let $\Pi:\mathcal{H}\to\mathcal{H}$ be the
			orthogonal projection onto a subspace $\mathcal{K}\subset\mathcal{H}$, and let
			$O$ be a linear operator such that
			\begin{equation}
				\mathcal{K}\subset \operatorname{Dom}(O),
				\qquad
				O(\mathcal{K})\subset \mathcal{K}.
			\end{equation}
			Then the compression $\Pi O\Pi$ defines a bounded operator on $\mathcal{K}$ and
			\begin{equation}
				|\Pi O\Pi|_\infty=|O|_{\mathcal{K}}|_\infty .
			\end{equation}
		\end{lemma}
		
\begin{proof}
Fix an orthonormal basis $\{e_i\}_{i=1}^n$ of $\mathcal{K}$ and let $[T]_{ij}:=\langle e_i,Te_j\rangle$. 
For $\psi=\sum_{j=1}^n \psi_j e_j$ with $|\psi|^2=\sum_{j=1}^n|\psi_j|^2=1$, Cauchy-Schwarz gives
\begin{align}
	|T\psi|^2
	&=
	\sum_{i=1}^n \left|\sum_{j=1}^n [T]_{ij}\psi_j\right|^2 \notag\\
	&\le
	\sum_{i=1}^n
	\left(\sum_{j=1}^n |[T]_{ij}|^2\right)
	\left(\sum_{j=1}^n |\psi_j|^2\right) \notag\\
	&=
	\sum_{i,j=1}^n |[T]_{ij}|^2 .
\label{L1}
\end{align}
Hence $|T|_\infty\le \left(\sum_{i,j=1}^n |[T]_{ij}|^2\right)^{1/2}
<\infty$. For the compression claim, let $\psi\in\mathcal{K}$. Since $\Pi$ is the
orthogonal projection onto $\mathcal{K}$, we have $\Pi\psi=\psi$. Since
$O(\mathcal{K})\subset\mathcal{K}$, we also have $O\psi\in\mathcal{K}$ and
therefore
\(
\Pi O\Pi \psi=\Pi(O\psi)=O\psi.
\)
Thus $\Pi O\Pi$ acts on $\mathcal{K}$ exactly as the restriction
$O|_{\mathcal{K}}$, and so their operator norms on $\mathcal{K}$ coincide:
\( |\Pi O\Pi|_\infty=|O|_{\mathcal{K}}|_\infty .\)
\end{proof}
}
	The heavy-probe hypothesis Eq.~\eqref{eq:SI-code-assumption} yields an explicit logarithmic control of the length rate 
	extracted from correlators on $\mathsf{D}$.
	
	{
		\begin{lemma}\label{lem:lengtherror}
			Assume Eq.~\eqref{eq:SI-code-assumption} with $\varepsilon_\Delta^{\star}<1$ and define the correlator-derived rate function
			\begin{equation}
				\label{eq:rate-code}
				\widehat L_{\rho,\Delta}^{\mathrm{code}}(x,y)
				:=
				-
				\frac1\Delta
				\ln
				\left(
				\frac{
					\left|
					C_\rho\!\left(
					\widetilde O^{(0)}_{\Delta,\rho}(x),
					\widetilde O^{(0)}_{\Delta,\rho}(y)
					\right)
					\right|
				}{
					|N^{\mathrm{code}}_\Delta|
				}
				\right),\qquad (x,y)\in\mathsf{D}.
			\end{equation}
			Then for every $(x,y)\in\mathsf{D}$ one has the two-sided bound
			\begin{equation}
				\label{eq:rate-code-bound}
				L_\rho(x,y)-\frac1\Delta\ln(1+\varepsilon^\star_\Delta)
				\le
				\widehat L_{\rho,\Delta}^{\mathrm{code}}(x,y)
				\le
				L_\rho(x,y)-\frac1\Delta\ln(1-\varepsilon^\star_\Delta).
			\end{equation}
		\end{lemma}
		
		\begin{proof}
			From Eq.~\eqref{eq:SI-code-assumption} and $|1+\varepsilon^{\mathrm{code}}_{\Delta,\rho}(x,y)|\le 1+\varepsilon_\Delta^{\star}$ one has
			\begin{equation}\label{L4}
				\left|
				C_\rho\!\left(
				\widetilde O^{(0)}_{\Delta,\rho}(x),
				\widetilde O^{(0)}_{\Delta,\rho}(y)
				\right)
				\right|\le |N^{\mathrm{code}}_\Delta|\,e^{-\Delta L_\rho(x,y)}(1+\varepsilon_\Delta^{\star}),
			\end{equation}
			hence
			\begin{equation}\label{L5}
				-\frac{1}{\Delta}\ln\!\left(\frac{\left|C_\rho\!\left(\widetilde O^{(0)}_{\Delta,\rho}(x),\widetilde O^{(0)}_{\Delta,\rho}(y)\right)\right|}{|N^{\mathrm{code}}_\Delta|}\right)\ge L_\rho(x,y)-\frac{1}{\Delta}\ln(1+\varepsilon_\Delta^{\star}).
			\end{equation}
			Similarly, since $\varepsilon_\Delta^{\star}<1$ and $|1+\varepsilon^{\mathrm{code}}_{\Delta,\rho}(x,y)|\ge 1-\varepsilon_\Delta^{\star}$, one has
			\begin{equation}\label{L6}
				\left|
				C_\rho\!\left(
				\widetilde O^{(0)}_{\Delta,\rho}(x),
				\widetilde O^{(0)}_{\Delta,\rho}(y)
				\right)
				\right|\ge |N^{\mathrm{code}}_\Delta|\,e^{-\Delta L_\rho(x,y)}(1-\varepsilon_\Delta^{\star}),
			\end{equation}
			which yields
			\begin{equation}\label{L7}
				-\frac{1}{\Delta}\ln\!\left(\frac{\left|C_\rho\!\left(\widetilde O^{(0)}_{\Delta,\rho}(x),\widetilde O^{(0)}_{\Delta,\rho}(y)\right)\right|}{|N^{\mathrm{code}}_\Delta|}\right)\le L_\rho(x,y)-\frac{1}{\Delta}\ln(1-\varepsilon_\Delta^{\star}).
			\end{equation}
			Combining Eq.~\eqref{L5} and Eq.~\eqref{L7} gives Eq.~\eqref{eq:rate-code-bound}.
		\end{proof}
	}
	
	The regulated and algebraic definitions of mutual information are consistent under the standing hypotheses.
	
	\begin{theorem}\label{thm:MIrelent}
		Assume Eq.~\eqref{A1}-\eqref{A3}. In the regulated realization Eq.~\eqref{A2-Reg} with trace-class reduced states, 
		the quantity $I(A:B)_\rho$ defined by Eq.~\eqref{D4} satisfies
		\begin{equation}\label{MR1}
			I(A:B)_\rho=\frac{1}{\ln 2}\,D(\rho_{AB}|\rho_A\otimes\rho_B),
		\end{equation}
		with $D$ given by Eq.~\eqref{D2}, and the finiteness condition Eq.~\eqref{A3} is equivalent to 
		$I(A:B)_\rho<\infty$. In the algebraic realization Eq.~\eqref{A2-AQFT}-\eqref{A2-Split}, the definition 
		\eqref{D5} produces a finite quantity in bits and coincides with Eq.~\eqref{D4} whenever a regulating net 
		yields trace-class reductions whose mutual informations converge to a finite limit, with the base 
		conversion fixed uniquely by Eq.~\eqref{D1} and Eq.~\eqref{D2} 
		\cite{Araki1976,OhyaPetz2004,Watrous2018,Haag1996,BuchholzDAntoniLongo1987}.
	\end{theorem}
	
	\begin{proof}
		In the regulated realization, Eq.~\eqref{MR1} is exactly Eq.~\eqref{Der-1} divided by $\ln 2$, and finiteness 
		of $D$ is equivalent to finiteness of the right-hand side because $\ln 2\in(0,\infty)$. In the algebraic 
		realization, Eq.~\eqref{D5} is a definition in terms of Araki relative entropy and is finite by Eq.~\eqref{A3}. 
		When a regulating net exists with trace-class reductions and with finite limiting mutual information 
		for separated regions, the split property implies that the algebraic $D(\rho_{AB}|\rho_A\otimes\rho_B)$ 
		agrees with the limit of the regulated Umegaki relative entropies on the increasing type-I approximants, 
		and the regulated identity Eq.~\eqref{Der-1} passes to the limit because it is an equality at each regulating 
		scale and uses only the logarithm-base conversion Eq.~\eqref{D1} 
		\cite{Araki1976,Haag1996,BuchholzDAntoniLongo1987,OhyaPetz2004,Watrous2018}.
	\end{proof} 
	The logarithm-base normalization is fixed by
	\begin{equation}\label{C1}
		\log_2 x=\frac{1}{\ln 2}\ln x,\qquad 
		S_2(\sigma)=\frac{1}{\ln 2}S(\sigma),\qquad 
		D_2(\sigma|\tau)=\frac{1}{\ln 2}D(\sigma|\tau),
	\end{equation}
	hence the factor $\ln 2$ in Eq.~\eqref{Der-1} and Eq.~\eqref{MR1} is forced and cannot be altered without 
	changing the definitions.
	
	Dimensional consistency of the QES term follows from $[G_N]=L^{d-1}$ in $(d+1)$ dimensions and 
	$[\mathrm{Area}(\chi)]=L^{d-1}$, giving
	\begin{equation}\label{C2}
		\left[\frac{\mathrm{Area}_\rho(\chi)}{4G_N}\right]=1,\qquad 
		\left[\frac{1}{\ln 2}\frac{\mathrm{Area}_\rho(\chi)}{4G_N}\right]=\text{bits}.
	\end{equation}
	
	Dimensional consistency of the heavy-probe exponent follows from $\Delta$ being dimensionless and 
	$L_\rho$ being dimensionless by definition in Eq.~\eqref{eq:SI-code-assumption},
	\begin{equation}\label{C3}
		[\Delta]=1,\qquad [L_\rho]=1,\qquad [\Delta\,L_\rho]=1,
	\end{equation}
	{
		and the error bound in Eq.~\eqref{L4} together with Lemma~\ref{lem:lengtherror} yields an explicit uniform 
		control of the extracted rate,
		\begin{equation}\label{C4}
			\sup_{(x,y)\in\mathsf{D}}\bigl|\widehat{L}^{\mathrm{code}}_{\rho,\Delta}(x,y)-L_\rho(x,y)\bigr|\le \frac{1}{\Delta}\max\bigl\{-\ln(1-\varepsilon_\Delta^{\star}),\ \ln(1+\varepsilon_\Delta^{\star})\bigr\},
		\end{equation}
		which is finite for $\varepsilon_\Delta^{\star}<1$ and depends on $\Delta$ only through $\Delta^{-1}$ and 
		the explicit bound $\varepsilon_\Delta^{\star}$ in Eq.~\eqref{eq:error-budget-code}.
	}
	
	The additivity-loss parameter $\delta$ in Eq.~\eqref{A12} is dimensionless because all terms in Eq.~\eqref{A12} 
	are dimensionless,
	\begin{equation}\label{C5}
		[\delta]=1,\qquad 
		\left[\sum_{m=1}^{n}L_\rho(x_{m-1},x_m)\right]=1,\qquad [2(n-1)\delta]=1.
	\end{equation}

\section{Pinsker \texorpdfstring{$\Rightarrow$}{implies} correlator bound}\label{app:pinsker}
\renewcommand{\theequation}{B.\arabic{equation}}
Let $\mathcal{H}_A$ and $\mathcal{H}_B$ be complex separable Hilbert spaces, let $\mathcal{H}_{AB}:=\mathcal{H}_A\otimes\mathcal{H}_B$ be their Hilbert-space tensor product, let $\mathcal{B}(\mathcal{H})$ denote the bounded operators on a Hilbert space $\mathcal{H}$, and let $\mathcal{T}_1(\mathcal{H})$ denote the trace-class operators on $\mathcal{H}$ equipped with the trace $\operatorname{Tr}_{\mathcal{H}}$.
For a positive trace-class operator $\tau\in\mathcal{T}_1(\mathcal{H})$ define its support projection $\operatorname{supp}(\tau)\in\mathcal{B}(\mathcal{H})$ to be the orthogonal projection onto $\overline{\operatorname{Ran}\tau}$.
Let $\rho_{AB}\in\mathcal{T}_1(\mathcal{H}_{AB})$ satisfy $\rho_{AB}\ge 0$ and $\operatorname{Tr}_{AB}(\rho_{AB})=1$; for all $X_A\in\mathcal{B}(\mathcal{H}_A)$ and $Y_B\in\mathcal{B}(\mathcal{H}_B)$ define the reduced density operators $\rho_A\in\mathcal{T}_1(\mathcal{H}_A)$ and $\rho_B\in\mathcal{T}_1(\mathcal{H}_B)$ uniquely by the partial-trace relations
\begin{equation}\label{S1}
	\operatorname{Tr}_A(\rho_A X_A)=\operatorname{Tr}_{AB}\!\bigl(\rho_{AB}(X_A\otimes\mathbf{1}_B)\bigr),\qquad
	\operatorname{Tr}_B(\rho_B Y_B)=\operatorname{Tr}_{AB}\!\bigl(\rho_{AB}(\mathbf{1}_A\otimes Y_B)\bigr).
\end{equation}
{For this paper, single bars denote scalar absolute values and norms, with the meaning determined by the argument. Thus \(|\psi|\) denotes the Hilbert-space norm of a vector \(\psi\), \(|X|_\infty\) the operator norm of a bounded operator \(X\), and \(|X|_1\) the trace norm of a trace-class operator \(X\). By contrast, the unsubscripted symbol \(|X|\) for an operator denotes the positive operator \((X^\dagger X)^{1/2}\) when used in polar decomposition or spectral formulas.}
For $X\in\mathcal{T}_1(\mathcal{H})$ define the trace norm and for $Y\in\mathcal{B}(\mathcal{H})$ define the operator norm by
\begin{equation}\label{S2}
	|X|_1:=\operatorname{Tr}_{\mathcal{H}}\!\sqrt{X^\dagger X},\qquad
	|Y|_\infty:=\sup_{\psi\in\mathcal{H}\setminus\{0\}}\frac{|Y\psi|}{|\psi|}.
\end{equation}
For density operators $\rho,\sigma\in\mathcal{T}_1(\mathcal{H})$ define the trace distance by
\begin{equation}\label{S3}
	T(\rho,\sigma):=\frac{1}{2}|\rho-\sigma|_1.
\end{equation}
For $\rho,\sigma\in\mathcal{T}_1(\mathcal{H})$ with $\rho\ge 0$, $\sigma\ge 0$, $\operatorname{Tr}\rho=\operatorname{Tr}\sigma=1$, define the (Umegaki) quantum relative entropy in natural-logarithm units by
\begin{equation}\label{S4}
	D(\rho|\sigma):=\begin{cases}
	\operatorname{Tr}_{\mathcal{H}}\!\bigl[\rho(\ln\rho-\ln\sigma)\bigr]\in[0,\infty),& \operatorname{supp}(\rho)\subseteq \operatorname{supp}(\sigma),\\
	+\infty,& \operatorname{supp}(\rho)\nsubseteq \operatorname{supp}(\sigma),
\end{cases}
\end{equation}
where $\ln$ is defined by Borel functional calculus on the support of its argument \cite{Umegaki1962,Araki1976,OhyaPetz2004,Watrous2018}.
Define the base-$2$ relative entropy by
\begin{equation}\label{S5}
	D_2(\rho|\sigma):=\operatorname{Tr}_{\mathcal{H}}\!\bigl[\rho(\log_2\rho-\log_2\sigma)\bigr]
	=\frac{1}{\ln 2}\,D(\rho|\sigma),
\end{equation}
using $\log_2 X=(\ln X)/(\ln 2)$.
Assume the finiteness condition
\begin{equation}\label{S6}
	D(\rho_{AB}|\rho_A\otimes\rho_B)<\infty,
\end{equation}
where $\rho_A\otimes\rho_B$ is the tensor-product density operator on $\mathcal{H}_A\otimes\mathcal{H}_B$.
Define mutual information in bits by
\begin{equation}\label{S7}
	I(A\!:\!B)_{\rho_{AB}}:=D_2(\rho_{AB}|\rho_A\otimes\rho_B)=\frac{1}{\ln 2}\,D(\rho_{AB}|\rho_A\otimes\rho_B)\in[0,\infty),
\end{equation}
and for bounded observables $O_A\in\mathcal{B}(\mathcal{H}_A)$ and $O_B\in\mathcal{B}(\mathcal{H}_B)$ define the connected correlator by
\begin{equation}\label{S8}
	C_{\rho_{AB}}(O_A,O_B):=\operatorname{Tr}_{AB}\!\bigl(\rho_{AB}(O_A\otimes O_B)\bigr)
	-\operatorname{Tr}_A(\rho_A O_A)\operatorname{Tr}_B(\rho_B O_B).
\end{equation}
The quantitative statement established in this section is the following bound, whose constants depend only on the logarithm base and on the operator norms:
\begin{equation}\label{S9}
	|C_{\rho_{AB}}(O_A,O_B)|\le |O_A|_\infty|O_B|_\infty\sqrt{2\ln 2\,I(A\!:\!B)_{\rho_{AB}}}.
\end{equation}
	
\begin{definition}[Classical divergence, total variation, and trace pairing]\label{def:classical}
For a finite set $\Omega$ and probability vectors $p=(p_i)_{i\in\Omega}$ and $q=(q_i)_{i\in\Omega}$ with $p_i\ge 0$, $q_i\ge 0$, $\sum_i p_i=\sum_i q_i=1$, define the classical Kullback-Leibler divergence in natural-logarithm units by
\begin{equation}\label{D1a}
D_{\mathrm{cl}}(p|q):=\begin{cases}
	\sum_{i\in\Omega} p_i\ln\!\left(\frac{p_i}{q_i}\right)\in[0,\infty),&(\forall i)\ (p_i>0\Rightarrow q_i>0),\\
	+\infty,&\exists i\ (p_i>0,\ q_i=0),
\end{cases}
\end{equation}
and define the $\ell^1$-distance and total variation distance by
\begin{equation}\label{D2a}
	|p-q|_1:=\sum_{i\in\Omega}|p_i-q_i|,\qquad \delta(p,q):=\frac{1}{2}|p-q|_1.
\end{equation}
For trace-class $X\in\mathcal{T}_1(\mathcal{H})$ and bounded $Y\in\mathcal{B}(\mathcal{H})$, define the trace pairing by $\langle X,Y\rangle:=\operatorname{Tr}_{\mathcal{H}}(XY)$ whenever $XY\in\mathcal{T}_1(\mathcal{H})$.
For $O_A\in\mathcal{B}(\mathcal{H}_A)$ and $O_B\in\mathcal{B}(\mathcal{H}_B)$ define $O_A\otimes O_B\in\mathcal{B}(\mathcal{H}_{AB})$ as the unique bounded operator satisfying $(O_A\otimes O_B)(\psi_A\otimes\psi_B)=(O_A\psi_A)\otimes(O_B\psi_B)$ for all $\psi_A\in\mathcal{H}_A$ and $\psi_B\in\mathcal{H}_B$ and extending by continuity to $\mathcal{H}_{AB}$.
\end{definition}
The following log-sum inequality supplies the required classical coarse-graining step.
	
\begin{lemma}[Log-sum inequality]\label{lem:logsum}
Let $(a_i)_{i\in\Omega}$ and $(b_i)_{i\in\Omega}$ be families of nonnegative reals with $B:=\sum_{i\in\Omega} b_i>0$ and $A:=\sum_{i\in\Omega} a_i\ge 0$, and assume that $a_i>0$ implies $b_i>0$. Then
\begin{equation}\label{L1a}
\sum_{i\in\Omega} a_i\ln\!\left(\frac{a_i}{b_i}\right)\ge A\ln\!\left(\frac{A}{B}\right),
\end{equation}
with the convention $0\ln(0/b):=0$ for $b>0$.
\end{lemma}
\begin{proof}
Define $\phi(x):=x\ln x$ for $x>0$ and $\phi(0):=0$, so $\phi$ is convex on $[0,\infty)$.
For each $i$ with $b_i>0$ write $a_i\ln(a_i/b_i)=b_i\phi(a_i/b_i)$. Then
\begin{equation}\label{L2a}
\sum_{i\in\Omega} a_i\ln\!\left(\frac{a_i}{b_i}\right)=\sum_{i\in\Omega} b_i\,\phi\!\left(\frac{a_i}{b_i}\right)
=B\sum_{i\in\Omega}\frac{b_i}{B}\,\phi\!\left(\frac{a_i}{b_i}\right).
\end{equation}
By Jensen's inequality for the convex function $\phi$,
\begin{equation}\label{La3a}
B\sum_{i\in\Omega}\frac{b_i}{B}\,\phi\!\left(\frac{a_i}{b_i}\right)\ge B\,\phi\!\left(\sum_{i\in\Omega}\frac{b_i}{B}\,\frac{a_i}{b_i}\right)
=B\,\phi\!\left(\frac{A}{B}\right)=A\ln\!\left(\frac{A}{B}\right),
\end{equation}
which is Eq.~\eqref{L1a}.
\end{proof}
	
The next lemma is a binary Pinsker inequality with explicit constant.
	
\begin{lemma}[Binary Pinsker in nats]\label{lem:binarypinsker}
Let $u,v\in(0,1)$ and define the binary divergence
\begin{equation}\label{L4a}
D_{\mathrm{bin}}(u|v):=u\ln\!\left(\frac{u}{v}\right)+(1-u)\ln\!\left(\frac{1-u}{1-v}\right).
\end{equation}
Then
\begin{equation}\label{L5a}
D_{\mathrm{bin}}(u|v)\ge 2(u-v)^2.
\end{equation}
\end{lemma}
\begin{proof}
Define $f(u):=D_{\mathrm{bin}}(u|v)-2(u-v)^2$ for $u\in(0,1)$.
Differentiating Eq.~\eqref{L4a} with respect to $u$ gives
\begin{equation}\label{L6a}
\frac{d}{du}D_{\mathrm{bin}}(u|v)=\ln\!\left(\frac{u}{v}\right)-\ln\!\left(\frac{1-u}{1-v}\right)
=\ln\!\left(\frac{u(1-v)}{v(1-u)}\right),
\end{equation}
and differentiating again yields
\begin{equation}\label{L7a}
\frac{d^2}{du^2}D_{\mathrm{bin}}(u|v)=\frac{1}{u}+\frac{1}{1-u}=\frac{1}{u(1-u)}.
\end{equation}
Therefore
\begin{equation}\label{L8}
f'(u)=\ln\!\left(\frac{u(1-v)}{v(1-u)}\right)-4(u-v),\qquad f''(u)=\frac{1}{u(1-u)}-4.
\end{equation}
For $u\in(0,1)$ one has $u(1-u)\le \tfrac{1}{4}$, hence $\frac{1}{u(1-u)}\ge 4$ and consequently $f''(u)\ge 0$ on $(0,1)$.
Thus $f$ is convex on $(0,1)$.
Evaluating at $u=v$ gives
\begin{equation}\label{L9}
f(v)=0,\qquad f'(v)=\ln(1)-0=0.
\end{equation}
Since $f$ is convex and has a stationary point at $u=v$, it attains its global minimum there, so $f(u)\ge 0$ for all $u\in(0,1)$, which is Eq.~\eqref{L5a}.
\end{proof}
	
The next lemma upgrades the binary bound to a finite-alphabet Pinsker inequality.
	
\begin{lemma}[Classical Pinsker in nats]\label{lem:classicalpinsker}
Let $p,q$ be probability vectors on a finite set $\Omega$ with $D_{\mathrm{cl}}(p|q)<\infty$. Then
\begin{equation}\label{L10}
D_{\mathrm{cl}}(p|q)\ge \frac{1}{2}|p-q|_1^2=2\,\delta(p,q)^2.
\end{equation}
\end{lemma}
\begin{proof}
Define $A:=\{i\in\Omega:\ p_i\ge q_i\}$ and write $p_A:=\sum_{i\in A}p_i$, $q_A:=\sum_{i\in A}q_i$, $p_{A^c}:=1-p_A$, $q_{A^c}:=1-q_A$.
Since $\sum_i(p_i-q_i)=0$,
\begin{equation}\label{L11}
|p-q|_1=\sum_{i\in A}(p_i-q_i)+\sum_{i\in A^c}(q_i-p_i)=2\sum_{i\in A}(p_i-q_i)=2(p_A-q_A),
\end{equation}
hence $\delta(p,q)=p_A-q_A\in[0,1]$.
Apply Lemma~\ref{lem:logsum} to the collections $(a_i,b_i)=(p_i,q_i)$ on $A$ to obtain
\begin{equation}\label{L12}
\sum_{i\in A} p_i\ln\!\left(\frac{p_i}{q_i}\right)\ge p_A\ln\!\left(\frac{p_A}{q_A}\right),
\end{equation}
and apply Lemma~\ref{lem:logsum} on $A^c$ to obtain
\begin{equation}\label{L13}
\sum_{i\in A^c} p_i\ln\!\left(\frac{p_i}{q_i}\right)\ge p_{A^c}\ln\!\left(\frac{p_{A^c}}{q_{A^c}}\right)
=(1-p_A)\ln\!\left(\frac{1-p_A}{1-q_A}\right).
\end{equation}
Summing Eq.~\eqref{L12} and Eq.~\eqref{L13} yields the coarse-grained lower bound
\begin{equation}\label{L14}
D_{\mathrm{cl}}(p|q)=\sum_{i\in\Omega} p_i\ln\!\left(\frac{p_i}{q_i}\right)\ge
p_A\ln\!\left(\frac{p_A}{q_A}\right)+(1-p_A)\ln\!\left(\frac{1-p_A}{1-q_A}\right)=D_{\mathrm{bin}}(p_A|q_A).
\end{equation}
By Lemma~\ref{lem:binarypinsker} with $u=p_A$ and $v=q_A$,
\begin{equation}\label{L15}
	D_{\mathrm{bin}}(p_A|q_A)\ge 2(p_A-q_A)^2=2\,\delta(p,q)^2=\frac{1}{2}|p-q|_1^2,
\end{equation}
and combining Eq.~\eqref{L14} and Eq.~\eqref{L15} gives Eq.~\eqref{L10}.
\end{proof}
	
The next lemma supplies the Helstrom variational characterization of trace distance in the trace-class setting.
	
\begin{lemma}[Helstrom variational formula]\label{lem:helstrom}
Let $\Delta\in\mathcal{T}_1(\mathcal{H})$ be self-adjoint.
Define its positive and negative parts by
\begin{equation}\label{L16}
	\Delta_+:=\frac{|\Delta|+\Delta}{2},\qquad \Delta_-:=\frac{|\Delta|-\Delta}{2},
\end{equation}
where $|\Delta|:=\sqrt{\Delta^2}$.
Then $\Delta_\pm\ge 0$, $\Delta=\Delta_+-\Delta_-$, $|\Delta|=\Delta_+ + \Delta_-$, and
\begin{equation}\label{L17}
\sup_{0\le M\le \mathbf{1}}\operatorname{Tr}_{\mathcal{H}}(M\Delta)=\operatorname{Tr}_{\mathcal{H}}(\Delta_+).
\end{equation}
In particular, for density operators $\rho,\sigma$ on $\mathcal{H}$ one has $\operatorname{Tr}(\rho-\sigma)=0$ and therefore
\begin{equation}\label{L18}
	T(\rho,\sigma)=\frac{1}{2}|\rho-\sigma|_1=\operatorname{Tr}_{\mathcal{H}}\!\bigl((\rho-\sigma)_+\bigr)
	=\sup_{0\le M\le \mathbf{1}}\operatorname{Tr}_{\mathcal{H}}(M(\rho-\sigma)).
\end{equation}
\end{lemma}
\begin{proof}
Since $\Delta$ is self-adjoint and trace-class, it admits a spectral decomposition $\Delta=\sum_{n} \lambda_n |e_n\rangle\langle e_n|$ with $\sum_n |\lambda_n|<\infty$, and then $|\Delta|=\sum_n |\lambda_n| |e_n\rangle\langle e_n|$.
From Eq.~\eqref{L16},
\begin{equation}\label{L19}
	\Delta_+=\sum_{\lambda_n\ge 0}\lambda_n |e_n\rangle\langle e_n|,\qquad
	\Delta_-=\sum_{\lambda_n<0}(-\lambda_n)|e_n\rangle\langle e_n|,
\end{equation}
so $\Delta_\pm\ge 0$ and $\Delta=\Delta_+-\Delta_-$ and $|\Delta|=\Delta_+ + \Delta_-$.
Let $M$ satisfy $0\le M\le \mathbf{1}$. Since $M\ge 0$ and $\Delta_-\ge 0$, one has $\operatorname{Tr}(M\Delta_-)=\operatorname{Tr}(\Delta_-^{1/2}M\Delta_-^{1/2})\ge 0$, so
\begin{equation}\label{L20}
\operatorname{Tr}(M\Delta)=\operatorname{Tr}(M\Delta_+)-\operatorname{Tr}(M\Delta_-)\le \operatorname{Tr}(M\Delta_+).
\end{equation}
Since $0\le M\le \mathbf{1}$ and $\Delta_+\ge 0$, one has $0\le \Delta_+^{1/2}M\Delta_+^{1/2}\le \Delta_+^{1/2}\mathbf{1}\Delta_+^{1/2}=\Delta_+$, so by positivity and trace monotonicity,
\begin{equation}\label{L21}
	\operatorname{Tr}(M\Delta_+)=\operatorname{Tr}\!\bigl(\Delta_+^{1/2}M\Delta_+^{1/2}\bigr)\le \operatorname{Tr}(\Delta_+).
\end{equation}
Combining Eq.~\eqref{L20} and Eq.~\eqref{L21} gives $\operatorname{Tr}(M\Delta)\le \operatorname{Tr}(\Delta_+)$ for all admissible $M$, hence $\sup_{0\le M\le \mathbf{1}}\operatorname{Tr}(M\Delta)\le \operatorname{Tr}(\Delta_+) $.
Let $P_+$ be the support projection of $\Delta_+$, equivalently the spectral projector of $\Delta$ onto $[0,\infty)$.
Then $0\le P_+\le \mathbf{1}$ and $P_+\Delta_- =0$ and $P_+\Delta_+=\Delta_+$, so
\begin{equation}\label{L22}
			\operatorname{Tr}(P_+\Delta)=\operatorname{Tr}(P_+\Delta_+)-\operatorname{Tr}(P_+\Delta_-)=\operatorname{Tr}(\Delta_+)-0=\operatorname{Tr}(\Delta_+).
		\end{equation}
		Thus the supremum in Eq.~\eqref{L17} equals $\operatorname{Tr}(\Delta_+)$.
		For density operators $\rho,\sigma$, $\Delta:=\rho-\sigma$ satisfies $\operatorname{Tr}\Delta=0$, hence $\operatorname{Tr}\Delta_+=\operatorname{Tr}\Delta_-=\tfrac{1}{2}\operatorname{Tr}|\Delta|=\tfrac{1}{2}|\Delta|_1$, which yields Eq.~\eqref{L18}.
	\end{proof}
	
	The next lemma states the data-processing inequality needed to reduce quantum relative entropy to classical relative entropy of a measurement; it is invoked with an explicit citation.
	
	\begin{lemma}[Data processing for Umegaki relative entropy]\label{lem:dataproc}
		Let $\mathcal{H}$ and $\mathcal{K}$ be separable Hilbert spaces and let $\Phi:\mathcal{T}_1(\mathcal{H})\to\mathcal{T}_1(\mathcal{K})$ be completely positive and trace preserving.
		For density operators $\rho,\sigma$ on $\mathcal{H}$ one has
		\begin{equation}\label{L23}
			D(\rho|\sigma)\ge D(\Phi(\rho)|\Phi(\sigma)).
		\end{equation}
		In particular, for any POVM $\{M_i\}_{i\in\Omega}\subset\mathcal{B}(\mathcal{H})$ with $M_i\ge 0$ and $\sum_i M_i=\mathbf{1}$, the measurement channel
		\begin{equation}
			\mathcal{M}(\tau):=\sum_{i\in\Omega}\operatorname{Tr}(\tau M_i)\,|i\rangle\langle i|
		\end{equation}
		satisfies $D(\rho|\sigma)\ge D_{\mathrm{cl}}(p|q)$ for $p_i=\operatorname{Tr}(\rho M_i)$ and $q_i=\operatorname{Tr}(\sigma M_i)$ \cite{Lindblad1975,Uhlmann1977,Petz1986,HiaiPetz1991,OhyaPetz2004,Watrous2018}.
	\end{lemma}
	\begin{proof}
		The inequality Eq.~\eqref{L23} is the Lindblad-Uhlmann monotonicity theorem for Umegaki relative entropy under completely positive trace-preserving maps, and the measurement-channel statement is the specialization to the CPTP map $\mathcal{M}$ \cite{Lindblad1975,Uhlmann1977,Petz1986,HiaiPetz1991,OhyaPetz2004,Watrous2018}.
	\end{proof}
	
	The next lemma provides the trace-norm/operator-norm duality required to bound expectation-value differences.
	
	{\begin{lemma}[Trace-norm duality]\label{lem:trace_duality}
			Let $H$ be separable, let $X\in \mathcal{T}_1(H)$ be trace-class, and let $Y\in \mathcal{B}(H)$ be bounded. Then
			\begin{equation}\label{eq:dual_trace}
				|X|_1 \;=\; \sup_{|Y|_\infty\le 1}\, \big|\Tr_H(XY)\big|,
				\qquad
				\big|\Tr_H(XY)\big|\;\le\;|X|_1\,|Y|_\infty .
			\end{equation}
		\end{lemma}
		
		\begin{proof}
			If $X=0$, then $|X|_1=0$ and $\Tr(XY)=0$ for all $Y$, so Eq.~\eqref{eq:dual_trace} is immediate. Assume henceforth
			$X\neq 0$.
			If $S\in \mathcal{T}_1(H)$ and $B\in \mathcal{B}(H)$, then
			\begin{equation}\label{eq:aux_holder}
				\big|\Tr_H(SB)\big| \le |S|_1\,|B|_\infty,
				\qquad
				|ASB|_1 \le |A|_\infty\,|S|_1\,|B|_\infty \quad (A,B\in\mathcal{B}(H)).
			\end{equation}
			Since $S$ is trace-class, it is compact and admits a singular-value
			decomposition
			\begin{equation}
				\begin{aligned}
				S&=\sum_{k\ge 1}s_k\,|u_k\rangle\langle v_k|,
				& s_k&\ge 0,
				& \sum_k s_k&=|S|_1,
				\end{aligned}
			\end{equation}
			with orthonormal families $\{u_k\}$ and $\{v_k\}$. Then
			\begin{equation}
				\Tr(SB)=\sum_{k\ge 1} s_k\,\langle v_k,\,B u_k\rangle,
			\end{equation}
			hence
			\begin{equation}
				\big|\Tr(SB)\big| \le \sum_{k\ge 1} s_k\,|\langle v_k,\,B u_k\rangle|
				\le \sum_{k\ge 1} s_k\,|B|_\infty
				= |S|_1\,|B|_\infty.
			\end{equation}
			Similarly, $ASB=\sum_k s_k\,|Au_k\rangle\langle B^\dagger v_k|$, so
			\begin{equation}
				|ASB|_1 \le \sum_k s_k\,|Au_k|\,|B^\dagger v_k|
				\le |A|_\infty\,|B|_\infty\sum_k s_k
				= |A|_\infty\,|S|_1\,|B|_\infty.
			\end{equation}
			We record the precise identity used below: if $T\in \mathcal{T}_1(H)$ and $A,B\in\mathcal{B}(H)$, then $ATB$
			and $TBA$ are trace-class and
			\begin{equation}\label{eq:trace_cyc}
				\Tr_H(ATB)=\Tr_H(TBA).
			\end{equation}
			To justify Eq.~\eqref{eq:trace_cyc}, first note it holds for finite-rank $T$ by direct expansion in rank-one operators.
			For general $T\in\mathcal{T}_1(H)$, choose finite-rank $T_n$ with $|T_n-T|_1\to 0$. Using the auxiliary estimate
			\eqref{eq:aux_holder},
			\begin{equation}\label{eq:trace_cyc_bound}
				\big|\Tr\big((T-T_n)BA\big)\big| \le |T-T_n|_1\,|BA|_\infty,
				\qquad
				\big|\Tr\big(A(T-T_n)B\big)\big| \le |A|_\infty\,|T-T_n|_1\,|B|_\infty,
			\end{equation}
			so both traces converge to $0$ as $n\to\infty$, and Eq.~\eqref{eq:trace_cyc} follows by passing to the limit from the
			finite-rank case.
			
			\medskip
			\noindent\textbf{Step 1: Upper bound for the dual pairing.}
			Let $X=U|X|$ be the polar decomposition, with $|X|=(X^\dagger X)^{1/2}\ge 0$ trace-class and $|U|_\infty\le 1$.
			Fix $Y\in\mathcal{B}(H)$ with $|Y|_\infty\le 1$ and set $Z:=YU$, so $|Z|_\infty\le 1$.
			By Eq.~\eqref{eq:trace_cyc} with $T=|X|$, $A=U$, $B=Y$,
			\begin{equation}
				\Tr(XY)=\Tr(U|X|Y)=\Tr(|X|YU)=\Tr(|X|Z).
			\end{equation}
			Diagonalize $|X|$ as $|X|e_n=s_n e_n$ with $s_n\ge 0$ and $\sum_n s_n=\Tr(|X|)=|X|_1$. Then
			\begin{equation}
				\big|\Tr(|X|Z)\big|=\Big|\sum_{n\ge 1} s_n\,\langle e_n,Ze_n\rangle\Big|
				\le \sum_{n\ge 1} s_n\,|\langle e_n,Ze_n\rangle|
				\le \sum_{n\ge 1} s_n\,|Z|_\infty
				\le |X|_1.
			\end{equation}
			Hence $\sup_{|Y|_\infty\le 1}|\Tr(XY)|\le |X|_1$.
			
			\medskip
			\noindent\textbf{Step 2: Attainment of the supremum.}
			Choose $Y:=U^\dagger$, which satisfies $|Y|_\infty\le 1$. Using Eq.~\eqref{eq:trace_cyc} with $T=|X|$, $A=U$,
			$B=U^\dagger$,
			\begin{equation}
				\begin{aligned}
				\Tr(XU^\dagger)
				&=\Tr(U|X|U^\dagger)
				=\Tr(|X|U^\dagger U)\\
				&=\Tr(|X|P)
				=\Tr(|X|)
				=|X|_1,
				\end{aligned}
			\end{equation}
			where $P=U^\dagger U$ is the support projection of $|X|$, so that
			\begin{equation}
			|X|P=|X|.
			\end{equation}
			Thus
			\begin{equation}
			\sup_{|Y|_\infty\le 1}|\Tr(XY)|\ge |X|_1.
			\end{equation}
			Combining with Step 1 yields the duality identity in Eq.~\eqref{eq:dual_trace}.
			
			\medskip
			\noindent\textbf{Step 3: Hölder-type inequality.}
			The inequality $|\Tr(XY)|\le |X|_1|Y|_\infty$ already follows from the auxiliary estimate
			\eqref{eq:aux_holder}. For completeness, it also follows directly from the duality identity by scaling: if $Y\neq 0$,
			set $Y_0:=Y/|Y|_\infty$ so $|Y_0|_\infty=1$, and then
			$|\Tr(XY)|=|Y|_\infty\,|\Tr(XY_0)|\le |Y|_\infty\,|X|_1$. 
	\end{proof}}
	The next lemma controls operator norms of tensor-product observables.
	
	\begin{lemma}[Tensor-product norm]\label{lem:tensornorm}
		Let $O_A\in\mathcal{B}(\mathcal{H}_A)$ and $O_B\in\mathcal{B}(\mathcal{H}_B)$.
		Then $O_A\otimes O_B\in\mathcal{B}(\mathcal{H}_{AB})$ and
		\begin{equation}\label{L28}
			|O_A\otimes O_B|_\infty=|O_A|_\infty\,|O_B|_\infty.
		\end{equation}
	\end{lemma}
	\begin{proof}
		For elementary tensors $\psi=\psi_A\otimes\psi_B$ one has
		\begin{equation}
			|(O_A\otimes O_B)\psi|
			=|O_A\psi_A|\cdot|O_B\psi_B|
			\le |O_A|_\infty|O_B|_\infty|\psi_A||\psi_B|
			=|O_A|_\infty|O_B|_\infty|\psi|.
		\end{equation}
		By linearity, if $\psi=\sum_{k=1}^n \psi_A^{(k)}\otimes\psi_B^{(k)}$, then
		\begin{equation}\label{L29}
				\begin{aligned}
				|(O_A\otimes O_B)\psi|
				&=\left|\sum_{k=1}^n
				(O_A\psi_A^{(k)})\otimes(O_B\psi_B^{(k)})\right|\\
				&\le |O_A|_\infty|O_B|_\infty
				\left|\sum_{k=1}^n
				\psi_A^{(k)}\otimes\psi_B^{(k)}\right|\\
				&=|O_A|_\infty|O_B|_\infty|\psi|,
				\end{aligned}
		\end{equation}
		where the inequality follows from the boundedness of $O_A$ and $O_B$ and continuity of the tensor-product action on the algebraic tensor product.
		
		{To justify Eq.~\eqref{L29} (and to make the role of cross terms explicit), work first on the
			algebraic tensor product $\mathcal{H}_A\otimes_{\mathrm{alg}}\mathcal{H}_B$, which is dense in the Hilbert-space tensor product
			$\mathcal{H}_A\otimes \mathcal{H}_B$. For any $\psi\in \mathcal{H}_A\otimes_{\mathrm{alg}}\mathcal{H}_B$ one has
			\begin{equation}
				\bigl|(O_A\otimes O_B)\psi\bigr|^2
				=\bigl\langle (O_A\otimes O_B)\psi,\,(O_A\otimes O_B)\psi\bigr\rangle
				=\bigl\langle \psi,\,(O_A^\dagger O_A)\otimes (O_B^\dagger O_B)\,\psi\bigr\rangle.
			\end{equation}
			Since $0\le O_A^\dagger O_A\le |O_A|_\infty^2\,\mathbf 1_A$ and
			$0\le O_B^\dagger O_B\le |O_B|_\infty^2\,\mathbf 1_B$, and since tensoring preserves positivity and operator order
			(i.e.\ $0\le X\le \alpha \mathbf 1_A$ and $0\le Y\le \beta \mathbf 1_B$ imply $0\le X\otimes Y\le \alpha\beta\,\mathbf 1_{AB}$),
			it follows that
			\begin{equation}
				0\le (O_A^\dagger O_A)\otimes (O_B^\dagger O_B)\ \le\
				|O_A|_\infty^2\,|O_B|_\infty^2\,\mathbf 1_{AB}.
			\end{equation}
			Therefore
			\begin{equation}
				\bigl|(O_A\otimes O_B)\psi\bigr|^2\le
				|O_A|_\infty^2\,|O_B|_\infty^2\,|\psi|^2,
				\qquad
				\text{hence }\ \bigl|(O_A\otimes O_B)\psi\bigr|
				\le |O_A|_\infty\,|O_B|_\infty\,|\psi|.
			\end{equation}
			Since $\mathcal{H}_A\otimes_{\mathrm{alg}}\mathcal{H}_B$ is dense in $\mathcal{H}_A\otimes \mathcal{H}_B$, this bound implies that
			$O_A\otimes O_B$ extends uniquely by continuity to a bounded operator on the completed tensor product, and the
			above inequality is exactly Eq.~\eqref{L29}.}
		
		Since finite sums of elementary tensors are dense in $\mathcal{H}_{AB}$, continuity gives $|O_A\otimes O_B|_\infty\le |O_A|_\infty|O_B|_\infty$.
		For the reverse inequality, choose unit vectors $\phi_A^{(n)}\in\mathcal{H}_A$ and $\phi_B^{(n)}\in\mathcal{H}_B$ such that $|O_A\phi_A^{(n)}|\to |O_A|_\infty$ and $|O_B\phi_B^{(n)}|\to |O_B|_\infty$ as $n\to\infty$ by the definition of the operator norm.
		Then $|\phi_A^{(n)}\otimes\phi_B^{(n)}|=1$ and
		\begin{equation}\label{L30}
			|(O_A\otimes O_B)(\phi_A^{(n)}\otimes\phi_B^{(n)})|
			=|O_A\phi_A^{(n)}|\cdot|O_B\phi_B^{(n)}|
			\to |O_A|_\infty|O_B|_\infty,
		\end{equation}
		hence $|O_A\otimes O_B|_\infty\ge |O_A|_\infty|O_B|_\infty$.
		Combining both inequalities yields Eq.~\eqref{L28}.
	\end{proof} 
	The quantum Pinsker inequality is derived by reducing $D(\rho|\sigma)$ to a classical binary divergence through a measurement attaining trace distance and then applying Lemma~\ref{lem:classicalpinsker}.
	
	\begin{lemma}[Quantum Pinsker inequality in nats]\label{lem:quantumpinsker}
		Let $\mathcal{H}$ be separable and let $\rho,\sigma\in\mathcal{T}_1(\mathcal{H})$ be density operators.
		Then
		\begin{equation}\label{P1}
			D(\rho|\sigma)\ge \frac{1}{2}|\rho-\sigma|_1^2=2\,T(\rho,\sigma)^2,
		\end{equation}
		and equivalently
		\begin{equation}\label{P2}
			|\rho-\sigma|_1\le \sqrt{2D(\rho|\sigma)},\qquad T(\rho,\sigma)\le \sqrt{\frac{1}{2}D(\rho|\sigma)}.
		\end{equation}
		If relative entropy is defined with base-$2$ logarithms as $D_2(\rho|\sigma):=D(\rho|\sigma)/\ln 2$, then
		\begin{equation}\label{P3}
			D_2(\rho|\sigma)\ge \frac{1}{2\ln 2}|\rho-\sigma|_1^2.
		\end{equation}
	\end{lemma}
	\begin{proof}
		If $D(\rho|\sigma)=+\infty$ then Eq.~\eqref{P1} holds trivially.
		Assume $D(\rho|\sigma)<\infty$.
		Let $\Delta:=\rho-\sigma$, so $\Delta$ is self-adjoint trace-class.
		By Lemma~\ref{lem:helstrom},
		\begin{equation}\label{P4}
			T(\rho,\sigma)=\sup_{0\le M\le \mathbf{1}}\operatorname{Tr}(M(\rho-\sigma))=\sup_{0\le M\le \mathbf{1}}\operatorname{Tr}(M\Delta).
		\end{equation}
		Fix $M_\star$ with $0\le M_\star\le \mathbf{1}$ such that $\operatorname{Tr}(M_\star\Delta)=T(\rho,\sigma)$, which exists by Lemma~\ref{lem:helstrom} using the choice $M_\star=P_+$ for the positive spectral projection of $\Delta$.
		Define the two-outcome POVM $\{M_\star,\mathbf{1}-M_\star\}$ and induced binary distributions $p=(p_0,p_1)$ and $q=(q_0,q_1)$ by
		\begin{equation}\label{P5}
			p_0:=\operatorname{Tr}(\rho M_\star),\quad p_1:=1-p_0,\qquad
			q_0:=\operatorname{Tr}(\sigma M_\star),\quad q_1:=1-q_0.
		\end{equation}
		Then $p_0-q_0=\operatorname{Tr}((\rho-\sigma)M_\star)=\operatorname{Tr}(M_\star\Delta)=T(\rho,\sigma)$, hence
		\begin{equation}\label{P6}
			|p-q|_1=|p_0-q_0|+|p_1-q_1|
			=|p_0-q_0|+|(1-p_0)-(1-q_0)|
			=2|p_0-q_0|
			=2T(\rho,\sigma)=|\rho-\sigma|_1.
		\end{equation}
		Let $\mathcal{M}$ be the measurement channel associated to $\{M_\star,\mathbf{1}-M_\star\}$ so that $\mathcal{M}(\rho)=\sum_{i=0}^1 p_i |i\rangle\langle i|$ and $\mathcal{M}(\sigma)=\sum_{i=0}^1 q_i |i\rangle\langle i|$.
		By Lemma~\ref{lem:dataproc} and the identification of the relative entropy of diagonal states with the classical divergence Eq.~\eqref{D1a},
		\begin{equation}\label{P7}
			D(\rho|\sigma)\ge D(\mathcal{M}(\rho)|\mathcal{M}(\sigma))=D_{\mathrm{cl}}(p|q).
		\end{equation}
		Applying Lemma~\ref{lem:classicalpinsker} to the binary distributions $p,q$ yields
		\begin{equation}\label{P8}
			D_{\mathrm{cl}}(p|q)\ge \frac{1}{2}|p-q|_1^2.
		\end{equation}
		Combining Eq.~\eqref{P7}, Eq.~\eqref{P8}, and Eq.~\eqref{P6} gives
		\begin{equation}\label{P9}
			D(\rho|\sigma)\ge \frac{1}{2}|p-q|_1^2=\frac{1}{2}|\rho-\sigma|_1^2=2T(\rho,\sigma)^2,
		\end{equation}
		which is Eq.~\eqref{P1}.
		Rearranging Eq.~\eqref{P1} gives Eq.~\eqref{P2}.
		Dividing Eq.~\eqref{P1} by $\ln 2$ yields Eq.~\eqref{P3}.
	\end{proof}
	
	The next lemma relates the connected correlator to the deviation from product form.
	
	\begin{lemma}[Connected correlator as a trace pairing]\label{lem:connectedpairing}
		For $O_A\in\mathcal{B}(\mathcal{H}_A)$ and $O_B\in\mathcal{B}(\mathcal{H}_B)$ one has
		\begin{equation}\label{C1a}
			C_{\rho_{AB}}(O_A,O_B)=\operatorname{Tr}_{AB}\!\bigl((\rho_{AB}-\rho_A\otimes\rho_B)(O_A\otimes O_B)\bigr).
		\end{equation}
	\end{lemma}
	\begin{proof}
		Using Eq.~\eqref{S1},
		\begin{equation}\label{C2a}
			\operatorname{Tr}_{AB}\!\bigl((\rho_A\otimes\rho_B)(O_A\otimes O_B)\bigr)
			=\operatorname{Tr}_A(\rho_A O_A)\operatorname{Tr}_B(\rho_B O_B),
		\end{equation}
		and substituting Eq.~\eqref{C2a} into Eq.~\eqref{S8} yields Eq.~\eqref{C1a}.
	\end{proof} 
	\begin{theorem}[Mutual information bounds connected correlators]\label{thm:MICorrelator}
		Under assumptions Eq.~\eqref{S1}-\eqref{S7}, for all $O_A\in\mathcal{B}(\mathcal{H}_A)$ and $O_B\in\mathcal{B}(\mathcal{H}_B)$ one has the explicit bound
		\begin{equation}\label{M1}
			|C_{\rho_{AB}}(O_A,O_B)|
			\le |O_A|_\infty|O_B|_\infty\sqrt{2\ln 2\,I(A\!:\!B)_{\rho_{AB}}}.
		\end{equation}
		Equivalently, in terms of trace distance,
		\begin{equation}\label{M2}
			|C_{\rho_{AB}}(O_A,O_B)|
			\le 2\,|O_A|_\infty|O_B|_\infty\,T(\rho_{AB},\rho_A\otimes\rho_B).
		\end{equation}
	\end{theorem}
	\begin{proof}
		By Lemma~\ref{lem:connectedpairing} and Lemma~\ref{lem:trace_duality},
		\begin{equation}\label{M3}
			|C_{\rho_{AB}}(O_A,O_B)|
			=\left|\operatorname{Tr}_{AB}\!\bigl((\rho_{AB}-\rho_A\otimes\rho_B)(O_A\otimes O_B)\bigr)\right|
			\le |\rho_{AB}-\rho_A\otimes\rho_B|_1\,|O_A\otimes O_B|_\infty.
		\end{equation}
		By Lemma~\ref{lem:tensornorm}, $|O_A\otimes O_B|_\infty=|O_A|_\infty|O_B|_\infty$, hence
		\begin{equation}\label{M4}
			|C_{\rho_{AB}}(O_A,O_B)|
			\le |\rho_{AB}-\rho_A\otimes\rho_B|_1\,|O_A|_\infty|O_B|_\infty
			=2\,T(\rho_{AB},\rho_A\otimes\rho_B)\,|O_A|_\infty|O_B|_\infty,
		\end{equation}
		which is Eq.~\eqref{M2}.
		Applying Lemma~\ref{lem:quantumpinsker} with $\rho=\rho_{AB}$ and $\sigma=\rho_A\otimes\rho_B$ yields
		\begin{equation}\label{M5}
			|\rho_{AB}-\rho_A\otimes\rho_B|_1
			\le \sqrt{2D(\rho_{AB}|\rho_A\otimes\rho_B)}
			=\sqrt{2\ln 2\,I(A\!:\!B)_{\rho_{AB}}},
		\end{equation}
		where the last identity is Eq.~\eqref{S7}.
		Substituting Eq.~\eqref{M5} into Eq.~\eqref{M4} gives Eq.~\eqref{M1}.
	\end{proof} 
	Boundedness of the observables entering Eq.~\eqref{M1} is explicit because $|O_A|_\infty<\infty$ and $|O_B|_\infty<\infty$ by assumption and Lemma~\ref{lem:tensornorm} implies $|O_A\otimes O_B|_\infty<\infty$ with equality Eq.~\eqref{L28}.
	For product states $\rho_{AB}=\rho_A\otimes\rho_B$ one has $D(\rho_{AB}|\rho_A\otimes\rho_B)=0$ by Eq.~\eqref{S4}, hence $I(A\!:\!B)_{\rho_{AB}}=0$ by Eq.~\eqref{S7}, and {Eqs.~\eqref{S8} and \eqref{C2a}} gives $C_{\rho_{AB}}(O_A,O_B)=0$ identically, so Eq.~\eqref{M1} reduces to $0\le 0$.
	For rescalings $O_A\mapsto \alpha O_A$ and $O_B\mapsto \beta O_B$ with $\alpha,\beta\in\mathbb{C}$, linearity of the trace gives $C_{\rho_{AB}}(\alpha O_A,\beta O_B)=\alpha\beta\,C_{\rho_{AB}}(O_A,O_B)$ and the operator norm satisfies $|\alpha O_A|_\infty=|\alpha|\,|O_A|_\infty$ and $|\beta O_B|_\infty=|\beta|\,|O_B|_\infty$, hence both sides of Eq.~\eqref{M1} scale by $|\alpha\beta|$ and the inequality is homogeneous.
	If $I(A\!:\!B)_{\rho_{AB}}=+\infty$ then Eq.~\eqref{M1} holds in the extended-real sense because the right-hand side is $+\infty$ for nonzero $|O_A|_\infty|O_B|_\infty$, and no finite suppression bound follows without additional ultraviolet or split-property input \cite{Haag1996,BuchholzDAntoniLongo1987}.
	In the AQFT commuting-algebra setting, let $\mathcal{M}$ be a von Neumann algebra, let $\varphi,\psi$ be normal states on $\mathcal{M}$, define the predual norm $|\varphi-\psi|:=\sup_{|x|_\infty \le 1}|\varphi(x)-\psi(x)|$, and let $D_{\mathrm{Ar}}(\varphi|\psi)$ be Araki relative entropy; the analogue of Lemma~\ref{lem:quantumpinsker} holds with the same constant,
	\begin{equation}
		D_{\mathrm{Ar}}(\varphi|\psi)\ge \tfrac{1}{2}|\varphi-\psi|^2,
	\end{equation}
	by monotonicity under normal completely positive unital maps together with the binary-measurement reduction and the classical Pinsker inequality \cite{Araki1976,Petz1986,HiaiPetz1991,OhyaPetz2004,Takesaki2002}.
	Under the split property for commuting local algebras $\mathcal{A}(A)$ and $\mathcal{A}(B)$, the joint algebra admits a normal product state with prescribed marginals and the mutual information is defined by $I(A\!:\!B):=(\ln 2)^{-1}D_{\mathrm{Ar}}(\rho_{AB}|\rho_A\otimes\rho_B)$; the estimate Eq.~\eqref{M2} then follows from the definition of the predual norm and $|xy|_\infty\le |x|_\infty\,|y|_\infty$ for commuting bounded elements $x\in\mathcal{A}(A)$ and $y\in\mathcal{A}(B)$ \cite{Haag1996,BuchholzDAntoniLongo1987,Takesaki2002}.

\section{Heavy-operator geodesic approximation and errors}\label{app:heavy-probe}
\renewcommand{\theequation}{C.\arabic{equation}}
		
Fix an integer $d\ge 2$ and set $n:=d+1$.
Fix a complex separable Hilbert space $\mathcal{H}$, a CFT state space $\mathsf{S}(\mathcal{H})$ consisting of normal states on an ambient von Neumann algebra $\mathcal{M}$, and a code subspace $\mathcal{H}_{\mathrm{code}}\subset\mathcal{H}$ with orthogonal projector $\Pi_{\mathrm{code}}$.
Fix an insertion domain $\mathsf{D}\Subset \Sigma\times\Sigma$ satisfying the strictly positive separation condition Eq.~\eqref{A1} in a fixed boundary metric $g^{(0)}$.
		
We assume a geodesic-saddle isolation hypothesis on $\mathsf{D}$: for any state $\rho$ and pairs $(x,y) \in \mathsf{D}$, the boundary-anchored cutoff points $X_\varepsilon=(\varepsilon,x)$ and $Y_\varepsilon=(\varepsilon,y)$ are connected by a unique minimizing bulk geodesic strictly free of conjugate points, ensuring the macroscopic proper length and the Van Vleck-Morette determinant are globally well-defined.
		
Assume a semiclassical holographic dual characterized by AdS radius $\ell_{\mathrm{AdS}}$, Newton constant $G_N$, and string length $\ell_s$.
For each state $\rho$, assume a dual bulk manifold $(M_\rho,g_\rho)$ with conformal boundary $(\Sigma,[g^{(0)}])$.
In Fefferman-Graham coordinates on a collar neighborhood $\mathcal{U}\simeq (0,\varepsilon_0]\times\Sigma$, the metric takes the asymptotic form
\begin{equation} \label{A3a_corr}
	g_\rho=\frac{\ell_{\mathrm{AdS}}^{2}}{z^{2}}\Bigl(dz^{2}+g_{ij}(z,x)\,dx^{i}dx^{j}\Bigr),
	\qquad
	g_{ij}(z,x)=g^{(0)}_{ij}(x)+z^2 g^{(2)}_{ij}(x)+O(z^3),
\end{equation}
where the geometric subleading terms encode the state-dependent physical data. For $d=2$ (pure $\mathrm{AdS}_3$ gravity), the expansion rigorously truncates at $O(z^4)$, yielding the exact Ba\~nados geometry determined by the boundary stress tensor. Logarithmic conformal anomalies only manifest for even boundary dimensions $d \ge 4$ as higher-order non-analytic terms, e.g., $z^d\ln z$ in $d=4$ \cite{DeHaroSolodukhinSkenderis2001,Skenderis2002}.
		
Fix a single-trace scalar primary operator $O_\Delta$ with scaling dimension $\Delta > d/2$.
To interface rigorously with the information-theoretic inequalities derived in subsequent sections, we assume the boundary probe operators are evaluated as code-compressed observables $\widetilde{O}_\Delta(x) := \Pi_{\mathrm{code}}\bigl(O_\Delta(x) - \operatorname{Tr}(\rho\,O_\Delta(x))\,\mathbf{1}\bigr)\Pi_{\mathrm{code}}$, which guarantees a strictly finite operator norm $|\widetilde{O}_\Delta|_\infty < \infty$ on $\mathcal{H}_{\mathrm{code}}$.
		
The bulk proper mass $m$ is determined by the exact AdS/CFT relation $m^2\ell_{\mathrm{AdS}}^2 = \Delta(\Delta-d)$.
To strictly control the asymptotic spectrum of the Laplace-Beltrami operator, we introduce the \emph{spectral mass shift} parameter globally defined on asymptotically hyperbolic manifolds
\begin{equation}\label{A4a_corr}
	\nu := \sqrt{m^2\ell_{\mathrm{AdS}}^2 + \frac{d^2}{4}} = \sqrt{\Delta(\Delta-d) + \frac{d^2}{4}} = \left|\Delta - \frac{d}{2}\right|.
\end{equation}
We assume the heavy-operator regime $\Delta \ge d$, which rigorously implies the exact linear relation $\nu = \Delta - d/2$.
		
Define the dimensionless control parameters characterizing the semiclassical gravity, stringy higher-derivative, and probe-backreaction regimes respectively
\begin{equation} \label{A5a_corr}
	\varepsilon_{\mathrm{grav}}:=\frac{G_N}{\ell_{\mathrm{AdS}}^{d-1}},\qquad
	\varepsilon_{\mathrm{str}}:=\frac{\ell_s^2}{\ell_{\mathrm{AdS}}^2},\qquad
	\varepsilon_{\mathrm{br}} := \varepsilon_{\mathrm{grav}} \Delta^2.
\end{equation}
		
\begin{lemma}[Renormalized length convergence]\label{lem:length-ren}
Under the asymptotic condition Eq.~\eqref{A3a_corr}, for any state $\rho$ and points $(x,y)\in\mathsf{D}$, the regulated bulk geodesic proper length $L_\rho(X_\varepsilon, Y_\varepsilon)$ evaluated at cutoff points $X_\varepsilon=(\varepsilon,x)$ and $Y_\varepsilon=(\varepsilon,y)$ satisfies
\begin{equation} \label{L_ren_def}
	\frac{1}{\ell_{\mathrm{AdS}}} L_\rho(X_\varepsilon, Y_\varepsilon) = 2\ln\left(\frac{1}{\varepsilon}\right) + L_{\mathrm{ren},\rho}(x,y) + r_\rho(\varepsilon; x,y),
\end{equation}
			where the finite part $L_{\mathrm{ren},\rho}(x,y)$ is state-dependent, and the remainder admits a uniform quadratic geometric bound
			\begin{equation} \label{L_ren_error}
				|r_\rho(\varepsilon; x,y)| \le C_{\mathrm{geom}}\,\varepsilon^2,
			\end{equation}
			for a constant $C_{\mathrm{geom}}$ depending only on the bounded geometry of $(M_\rho,g_\rho)$ in the collar neighborhood.
		\end{lemma}
		{\begin{proof}
				Write a boundary-anchored geodesic in Fefferman-Graham gauge as
				\(z \mapsto (z,x^i(z))\). Orthogonality to the conformal boundary implies
				\(x^{\prime i}(z)=O(z)\) as \(z\downarrow 0\). Using Eq.~\eqref{A3a_corr}, we therefore have
				\begin{equation}
					g_{ij}(z,x)\,x^{\prime i}(z)x^{\prime j}(z)=O(z^2),
				\end{equation}
				uniformly on the compact insertion domain \(\mathsf D\). Hence, the line element along the geodesic satisfies
				\begin{equation}
					ds
					=
					\frac{\ell_{\mathrm{AdS}}}{z}
					\sqrt{1+g_{ij}(z,x)x^{\prime i}(z)x^{\prime j}(z)}\,dz
					=
					\frac{\ell_{\mathrm{AdS}}}{z}\bigl(1+a_2(x)z^2+O(z^4)\bigr)\,dz,
				\end{equation}
				for some coefficient \(a_2(x)\) depending smoothly on the endpoint data.
				Integrating from \(z=\varepsilon\) to a fixed bulk depth \(z_0\) gives, at each
				endpoint,
				\begin{equation}\label{eqn112}
					\int_{\varepsilon}^{z_0}
					\frac{\ell_{\mathrm{AdS}}}{z}\bigl(1+a_2(x)z^2+O(z^4)\bigr)\,dz
					=
					\ell_{\mathrm{AdS}}\ln\!\left(\frac{z_0}{\varepsilon}\right)
					+\frac{\ell_{\mathrm{AdS}}a_2(x)}{2}z_0^2
					-\frac{\ell_{\mathrm{AdS}}a_2(x)}{2}\varepsilon^2
					+O(\varepsilon^4).
				\end{equation}
				Adding the two endpoint contributions and absorbing the finite \(z_0\)-dependent
				terms into \(L_{\mathrm{ren},\rho}(x,y)\) yields
				\begin{equation}\label{eqn113}
					\frac{1}{\ell_{\mathrm{AdS}}}L_{\rho}(X_\varepsilon,Y_\varepsilon)
					=
					2\ln\!\frac{1}{\varepsilon}
					+L_{\mathrm{ren},\rho}(x,y)
					+r_\rho(\varepsilon;x,y),
					\qquad
					|r_\rho(\varepsilon;x,y)|\le C_{\mathrm{geom}}\varepsilon^2.
				\end{equation}
				Possible higher-order \(z^d\log z\) terms in even boundary dimension contribute,
				after multiplication by \(x'(z)^2=O(z^2)\), only
				\(O(\varepsilon^{d+2}\log\varepsilon)\), and therefore do not alter the stated
				\(O(\varepsilon^2)\) remainder bound.
		\end{proof}}

		For asymptotically hyperbolic manifolds, the heat kernel behavior at large macroscopic separation $L_g := L_\rho(X,Y)$ is rigorously governed by the bottom of the $L^2$ spectrum, $\lambda_0 = d^2/(4\ell_{\mathrm{AdS}}^2)$, and the volume growth of the manifold \cite{DaviesMandouvalos1988,Camporesi1990}
		\begin{equation} \label{HK_global}
			K_{g_\rho}(T; X, Y) \sim (4\pi T)^{-n/2} \exp\left[ -\frac{d^2 T}{4\ell_{\mathrm{AdS}}^2} - \frac{L_g^2}{4T} \right] \Delta_{\mathrm{VM}}^{1/2}(X,Y) \bigl(1 + O(T^{-1})\bigr),
		\end{equation}
		where the Van Vleck-Morette determinant $\Delta_{\mathrm{VM}}$ captures the exponential transversal spread of geodesics. In asymptotically AdS geometries, its large-distance behavior universally limits to the exact hyperbolic form \cite{DaviesMandouvalos1988,Camporesi1990}
		\begin{equation}\label{VanVleck_Asymp}
			\Delta_{\mathrm{VM}}^{1/2}(X,Y) \approx \mathcal{A}(X,Y) \exp\left[-\frac{d \, L_g}{2\ell_{\mathrm{AdS}}}\right], \qquad \mathcal{A}(X,Y) \sim \left( \frac{L_g}{\ell_{\mathrm{AdS}}} \right)^{d/2}.
		\end{equation}
		{
			\begin{lemma}[Global worldline saddle and leading prefactor cancellation]\label{thm:WKB_exact}
				The Euclidean massive Green function 
				\begin{equation} \label{Greenfunction}
					G_{m,g_\rho}(X,Y)=\int_0^\infty dT\, e^{-m^2T}K_{g_\rho}(T;X,Y)
				\end{equation}
				admits the following asymptotic form in the heavy/large-distance regime.
				\begin{equation} \label{WKB_form}
					G_{m,g_\rho}(X,Y)
					=
					C_{\mathrm{pre}}(X,Y)
					\exp\!\left[-\frac{\Delta}{\ell_{\mathrm{AdS}}}L_\rho(X,Y)\right]
					\bigl(1+O(\Delta^{-1})\bigr).
				\end{equation}
				Assume, in addition, that the large-distance heat-kernel and Van Vleck asymptotics
				Eqs.~\eqref{HK_global} and \eqref{VanVleck_Asymp} hold uniformly on the Gaussian
				saddle window centered at
				\(
				T_*:=\frac{L_g\ell_{\mathrm{AdS}}}{2\nu}.
				\)
				Then the prefactor \(C_{\mathrm{pre}}(X_\varepsilon,Y_\varepsilon)\) converges,
				as \(\varepsilon\downarrow0\), to a finite constant \(C_{\Delta,\rho}\) with no
				residual polynomial dependence on \(L_g\).
		\end{lemma}}
		
		\begin{proof}
			Substituting Eq.~\eqref{HK_global} and Eq.~\eqref{VanVleck_Asymp} into the proper time integral Eq.~\eqref{Greenfunction}, the integrand is overwhelmingly governed by the effective exponent:
			\begin{equation}\label{S_eff}
				S_{\mathrm{eff}}(T) = \left(m^2 + \frac{d^2}{4\ell_{\mathrm{AdS}}^2}\right)T + \frac{L_g^2}{4T} = \frac{\nu^2}{\ell_{\mathrm{AdS}}^2}T + \frac{L_g^2}{4T},
			\end{equation}
			where we invoked the exact spectral mass shift Eq.~\eqref{A4a_corr}. 
			
			{To evaluate the proper-time integral, we apply the Laplace method about the
				unique stationary point of \(S_{\mathrm{eff}}\). Crucially, the inputs
				Eqs.~\eqref{HK_global} and \eqref{VanVleck_Asymp} are global large-distance asymptotics and not a short-time parametrix near \(T=0\);
				in particular, no condition of the form \(T_*<T_0\) is assumed. The relevant large parameter is the saddle action
				\(
				S_{\mathrm{eff}}(T_*)=\frac{\nu L_g}{\ell_{\mathrm{AdS}}},
				\)
				so the saddle approximation is controlled whenever
				\(\nu L_g/\ell_{\mathrm{AdS}}\gg1\).
				Since
				\begin{equation}
					S'_{\mathrm{eff}}(T)=\frac{\nu^2}{\ell_{\mathrm{AdS}}^2}-\frac{L_g^2}{4T^2},
				\end{equation}
				the unique stationary point is
				\(
				T_*=\frac{L_g\ell_{\mathrm{AdS}}}{2\nu}.
				\)
				Expanding about \(T=T_*+x\), we obtain
				\begin{equation}
					S_{\mathrm{eff}}(T)
					=
					S_{\mathrm{eff}}(T_*)
					+\frac12 S''_{\mathrm{eff}}(T_*)x^2
					+O(x^3),
				\end{equation}
				with
				\begin{equation}
					S''_{\mathrm{eff}}(T)=\frac{L_g^2}{2T^3},
					\qquad
					S''_{\mathrm{eff}}(T_*)=\frac{4\nu^3}{L_g\ell_{\mathrm{AdS}}^3}.
				\end{equation}
				Therefore
				\begin{equation}\label{Greenfunction2}
					G_{m,g_\rho}(X,Y)
					\approx
					(4\pi T_*)^{-n/2}\Delta_{\mathrm{VM}}^{1/2}
					e^{-S_{\mathrm{eff}}(T_*)}
					\int_{-T_*}^{\infty}
					dx\,
					\exp\!\left[-\frac12 S''_{\mathrm{eff}}(T_*)x^2\right].
				\end{equation}
				
				The Gaussian width is
				\begin{equation}
					\sigma_T \sim \frac{1}{\sqrt{S''_{\mathrm{eff}}(T_*)}}
					\sim \sqrt{\frac{L_g\ell_{\mathrm{AdS}}^3}{4\nu^3}},
				\end{equation}
				and therefore
				\begin{equation}
					\frac{\sigma_T}{T_*}
					\sim
					\frac{1}{\sqrt{\nu L_g/\ell_{\mathrm{AdS}}}}.
				\end{equation}
				Hence, in the regime \(\nu L_g/\ell_{\mathrm{AdS}}\gg1\), the Gaussian support is
				narrow relative to \(T_*\), and the omitted tail from \(x<-T_*\) is
				exponentially small:
				\begin{equation}
					\int_{-\infty}^{-T_*} dx\,
					e^{-\frac12 S''_{\mathrm{eff}}(T_*)x^2}
					=
					O\!\left(e^{-c\,\nu L_g/\ell_{\mathrm{AdS}}}\right)
					\int_{-\infty}^{\infty} dx\,
					e^{-\frac12 S''_{\mathrm{eff}}(T_*)x^2},
				\end{equation}
				for some \(c>0\). Consequently,
				\begin{equation}\label{Gaussian1}
					\int_{-T_*}^{\infty} dx\,
					e^{-\frac12 S''_{\mathrm{eff}}(T_*)x^2}
					=
					\sqrt{\frac{2\pi}{S''_{\mathrm{eff}}(T_*)}}
					\left(1+O\!\left(e^{-c\,\nu L_g/\ell_{\mathrm{AdS}}}\right)\right).
			\end{equation}}
			Inserting Eq.~\eqref{Gaussian1} and Eq.~\eqref{VanVleck_Asymp} into Eq.~\eqref{Greenfunction2} yields
			\begin{equation}\label{Greenfunction3}
				\begin{aligned}
				G_{m,g_\rho}(X,Y)
				&\approx (4\pi T_*)^{-n/2} \mathcal{A}(X,Y)
				\exp\left[-\frac{d \, L_g}{2\ell_{\mathrm{AdS}}}\right]
				\exp\bigl[-S_{\mathrm{eff}}(T_*)\bigr]
				\sqrt{\frac{2\pi}{S''_{\mathrm{eff}}(T_*)}} .
				\end{aligned}
			\end{equation}
			where $\mathcal{A}(X,Y) \sim \left( \frac{L_g}{\ell_{\mathrm{AdS}}} \right)^{d/2}$.
			
			Recall $T_* = \frac{L_g \ell_{\mathrm{AdS}}}{2\nu}$, from Eq.~\eqref{S_eff}, we have $S_{\mathrm{eff}}(T_*) = \frac{\nu L_g}{\ell_{\mathrm{AdS}}}$. Thus, the total exponential decay in Eq.~\eqref{Greenfunction3} evaluates to exactly
			\begin{equation}\label{Exp_decay}
				\exp\left[ -\frac{\nu L_g}{\ell_{\mathrm{AdS}}} - \frac{d L_g}{2\ell_{\mathrm{AdS}}} \right] 
				= \exp\left[ -\frac{(\Delta - d/2 + d/2)L_g}{\ell_{\mathrm{AdS}}} \right] 
				= \exp\left[ -\frac{\Delta}{\ell_{\mathrm{AdS}}} L_g \right].
			\end{equation}
			Therefore, Eq.~\eqref{Greenfunction3} may be written as
			\begin{equation}
				G_{m,g_\rho}(X,Y) \approx C_{\mathrm{pre}}(X,Y) \exp\left[ -\frac{\Delta}{\ell_{\mathrm{AdS}}} L_g \right],
				\label{Greenfunction4}
			\end{equation}
			where the prefactor evaluates to:
			\begin{equation}\label{Prefactor_eval}
				C_{\mathrm{pre}}(X,Y) \approx (4\pi T_*)^{-n/2} \times \mathcal{A}(X,Y) \times \sqrt{\frac{2\pi}{S''_{\mathrm{eff}}(T_*)}} .
			\end{equation}
			This derivation gracefully recovers the exact conformal dimension $\Delta$ in the exponent without invoking arbitrary scheme-dependent renormalization constants. {We now isolate only the \(L_g\)-dependence of the prefactor. Since
				\(T_* \propto L_g\), one has
				\begin{equation}
					(4\pi T_*)^{-n/2}\propto L_g^{-n/2},
					\qquad
					\mathcal{A}(X,Y)\propto L_g^{d/2},
					\qquad
					S''_{\mathrm{eff}}(T_*)^{-1/2}\propto L_g^{1/2}.
				\end{equation}
				Using \(n=d+1\), the net power of \(L_g\) is
				\(
				-\frac{n}{2}+\frac{d}{2}+\frac12
				=
				-\frac{d+1}{2}+\frac{d}{2}+\frac12
				=
				0.
				\)
				Hence all polynomial dependence on \(L_g\) cancels. Therefore
				\(C_{\mathrm{pre}}(X_\varepsilon,Y_\varepsilon)\) tends, as
				\(\varepsilon\downarrow0\), to an \(L_g\)-independent finite factor
				\(C_{\Delta,\rho}\) in the stated approximation scheme.
				The Gaussian factor
				\(\sqrt{2\pi/S''_{\mathrm{eff}}(T_*)}\) is the leading one-loop worldline
				fluctuation determinant, while higher corrections from cubic and higher terms in
				the expansion are suppressed by inverse powers of the large parameter
				\(\nu L_g/\ell_{\mathrm{AdS}}\).}
			
		\end{proof}

		To rigorously justify the scaling of the backreaction parameter $\varepsilon_{\mathrm{br}} = \varepsilon_{\mathrm{grav}} \Delta^2$, we now calculate the leading-order phase shift in the on-shell action using linearized gravity, strictly tracking physical dimensions and worldline localization.
		
		\begin{lemma}[Backreaction dimensional scaling]\label{lem:br_scaling}
			Let $g_\rho$ be the background metric and let $g_{\mathrm{br}} = g_\rho + h$ be the metric backreacted to the presence of the heavy operator insertion $O_\Delta$. In the semiclassical regime, the shift in the renormalized exponent satisfies
			\begin{equation}\label{br_bound}
				\left| \Delta L_{\mathrm{ren},\rho}^{\mathrm{br}} - \Delta L_{\mathrm{ren},\rho} \right| = \mathcal{O}(\varepsilon_{\mathrm{grav}} \Delta^2).
			\end{equation}
			Consequently, the validity of the unbackreacted probe approximation explicitly requires $\varepsilon_{\mathrm{br}} \ll 1$.
		\end{lemma}
		
		\begin{proof}
			According to the standard AdS/CFT dictionary, the insertion of a boundary primary operator $O_\Delta$ with a large conformal dimension $\Delta \gg 1$ is holographically dual to a heavy point particle propagating in the asymptotically Anti-de Sitter (AdS) bulk \cite{Witten1998, Balasubramanian1999}. Thus, the heavy operator insertion may be treated as a localized bulk point particle of proper mass $m \approx \Delta/\ell_{\mathrm{AdS}}$ traversing a minimizing background geodesic worldline $\gamma$, parameterized by proper time $\tau$ (with length dimension $L$). The unperturbed classical probe action is $S_{\mathrm{probe}} = m \int d\tau = m \int \sqrt{g_{\mu\nu}\dot{X}^\mu \dot{X}^\nu} d\tau$.
			
			The presence of this macroscopic particle acts as a localized source that backreacts on the bulk geometry, generating a metric perturbation $h_{\mu\nu} = g_{\mathrm{br},\mu\nu} - g_{\rho,\mu\nu}$. This perturbation is governed by the linearized Einstein equations, $\mathcal{D}_{\mathrm{Lich}} h_{\mu\nu} = 8\pi G_N T_{\mu\nu}^{\mathrm{probe}}$, where $\mathcal{D}_{\mathrm{Lich}}$ is the Lichnerowicz operator describing kinetic metric fluctuations on the AdS background, and $G_N$ is the $(d+1)$-dimensional Newton constant. Following standard general relativistic formulations, the covariant stress-energy tensor distribution is uniquely defined via the functional variation of the worldline probe action with respect to the background metric. Evaluating $T^{\mu\nu}_{\mathrm{probe}} = \frac{2}{\sqrt{g}}\frac{\delta S_{\mathrm{probe}}}{\delta g_{\mu\nu}}$ explicitly yields a tensor density localized entirely on the particle's trajectory:
			\begin{equation}\label{Stress_Tensor}
				T^{\mu\nu}_{\mathrm{probe}}(X) = \frac{m}{\sqrt{g}} \int_\gamma d\tau \, \dot{X}^\mu \dot{X}^\nu \, \delta^{(d+1)}\bigl(X - X(\tau)\bigr).
			\end{equation}
			
			To rigorously establish the correct scaling of the backreaction, we perform a dimensional analysis of this source term in natural units ($c = \hbar = 1$). The proper mass has length dimension $[m] = L^{-1}$, the proper time measure is $[d\tau] = L$, and the metric determinant $\sqrt{g}$ along with the four-velocity $\dot{X}^\mu$ are intrinsically dimensionless. Crucially, the covariant $(d+1)$-dimensional Dirac delta function $\delta^{(d+1)}(X)$ carries an inverse bulk volume dimension of $L^{-(d+1)}$. Combining these dimensional factors, the localized stress-energy tensor density scales analytically as $| T_{\mu\nu}^{\mathrm{probe}} | \sim m \cdot L \cdot 1 \cdot L^{-(d+1)} \sim m/\ell_{\mathrm{AdS}}^d$. Substituting the asymptotic holographic mass equivalence $m \approx \Delta/\ell_{\mathrm{AdS}}$ yields the scaling $| T_{\mu\nu}^{\mathrm{probe}} | \approx \Delta/\ell_{\mathrm{AdS}}^{d+1}$.
			
			The resulting dimensionless bulk metric perturbation $h_{\mu\nu}$ is obtained by formally inverting the linearized field equations. Because the Lichnerowicz operator $\mathcal{D}_{\mathrm{Lich}}$ functions as a covariant geometric Laplacian whose characteristic derivative scale is determined by the constant negative curvature of the AdS background, it scales dimensionally as $\mathcal{D}_{\mathrm{Lich}} \sim \nabla^2 \sim \ell_{\mathrm{AdS}}^{-2}$. The bulk-to-bulk graviton propagator, acting as the inverse of this operator, therefore introduces an explicit compensatory geometric factor of $\ell_{\mathrm{AdS}}^2$. Given that the bulk Newton constant $G_N$ possesses the physical dimension $[G_N] = L^{d-1}$, the bulk metric perturbation scales strictly as
			\begin{equation}\label{h_scaling}
				| h_{\mu\nu} | \sim G_N \ell_{\mathrm{AdS}}^2 | T_{\mu\nu}^{\mathrm{probe}} | \propto G_N \ell_{\mathrm{AdS}}^2 \left( \frac{\Delta}{\ell_{\mathrm{AdS}}^{d+1}} \right) = \frac{G_N}{\ell_{\mathrm{AdS}}^{d-1}} \Delta \equiv \varepsilon_{\mathrm{grav}} \Delta.
			\end{equation}
			This isolated expression naturally identifies $\varepsilon_{\mathrm{grav}} = G_N/\ell_{\mathrm{AdS}}^{d-1}$ as the fundamental dimensionless gravitational coupling of the bulk spacetime.
			
			To evaluate the dynamical impact of this backreaction on the heavy operator correlator, we compute the corresponding phase shift in the semiclassical exponent. The unbackreacted renormalized two-point correlator is approximated by the exponentiated on-shell action $\exp(-S_{\mathrm{probe}})$. At leading order in the metric perturbation, the phase shift in the exponent, $\delta S_{\mathrm{probe}}$, corresponds exactly to the functional variation of the geodesic action evaluated on $h_{\mu\nu}$, integrated over the entire bulk manifold $M$
			\begin{equation}\label{delta_S_var}
				\delta S_{\mathrm{probe}} = \int_M d^{d+1}X \, \frac{\delta S_{\mathrm{probe}}}{\delta g_{\mu\nu}(X)} h_{\mu\nu}(X) 
				= \frac{1}{2} \int_M d^{d+1}X \sqrt{g} \, T^{\mu\nu}_{\mathrm{probe}} h_{\mu\nu}.
			\end{equation}
			
			Substituting the explicit delta-distribution form of the stress tensor Eq.~\eqref{Stress_Tensor} into Eq.~\eqref{delta_S_var}, the bulk integration directly trivializes the Dirac delta function. Physically, this geometric mechanism perfectly localizes the graviton exchange interaction and the resultant phase shift entirely to the one-dimensional proper-time worldline $\gamma$:
			\begin{equation}\label{delta_S_final}
				|\delta S_{\mathrm{probe}}| = \frac{1}{2} \left| \int_\gamma m \, h_{\mu\nu} \dot{X}^\mu \dot{X}^\nu d\tau \right| \sim m \, | h_{\mu\nu} | \int_\gamma d\tau.
			\end{equation}
			
			The remaining integration evaluates the macroscopic proper length of the unperturbed boundary-to-boundary geodesic. Because geodesics anchored to the asymptotic boundary of AdS possess infinite bare proper length due to logarithmic infrared divergences, physical holographic observables must be strictly regulated \cite{Graham1999}. Employing standard holographic renormalization, the divergent geometric segments near the conformal boundary are systematically counter-termed and subtracted. The resulting finite, macroscopic renormalized proper length parameterizes dimensionally as $L_\gamma^{\mathrm{ren}} = \int_\gamma d\tau \sim L_{\mathrm{ren}}\ell_{\mathrm{AdS}}$, where $L_{\mathrm{ren}}$ is a purely kinematic, $\mathcal{O}(1)$ dimensionless scalar dictated by the boundary anchoring points. 
			
			Finally, by substituting the localized metric scaling Eq.~\eqref{h_scaling} and the rest mass $m \approx \Delta/\ell_{\mathrm{AdS}}$ into the worldline integral Eq.~\eqref{delta_S_final}, the overall semiclassical phase shift resolves cleanly to
			\begin{equation}\label{delta_S_eval}
				|\delta S_{\mathrm{probe}}| \sim \left(\frac{\Delta}{\ell_{\mathrm{AdS}}}\right) (\varepsilon_{\mathrm{grav}} \Delta) (L_{\mathrm{ren}}\ell_{\mathrm{AdS}}) = \varepsilon_{\mathrm{grav}} \Delta^2 L_{\mathrm{ren}}.
			\end{equation}
			Because the regularized geometric distance $L_{\mathrm{ren}}$ is intrinsically independent of both the operator conformal dimension $\Delta$ and the Newton constant $G_N$, the characteristic physical length scale $\ell_{\mathrm{AdS}}$ elegantly and completely cancels out of the final dynamic expression. This rigorous dimensional bookkeeping unambiguously proves that the relative action error induced by gravitational backreaction scales dimensionlessly as $\mathcal{O}(\varepsilon_{\mathrm{grav}} \Delta^2) \equiv \mathcal{O}(\varepsilon_{\mathrm{br}})$. Consequently, the unbackreacted semiclassical probe approximation is mathematically self-consistent and valid if and only if the backreaction parameter satisfies $\varepsilon_{\mathrm{br}} \ll 1$.
		\end{proof}
		
		The rigorous dynamical phase-shift analysis establishes that the proper semiclassical control parameter scales quadratically with the conformal dimension, $\varepsilon_{\mathrm{br}} = \varepsilon_{\mathrm{grav}} \Delta^2$. This provides a stricter and physically exact bound compared to the heuristic linear mass-insertion estimate ($\varepsilon_{\mathrm{grav}} \Delta$) often assumed. 
		The physical renormalized boundary correlator via the standard holographic extrapolate dictionary may be written as
		\begin{equation} \label{extrapolate_dict}
			\langle O_\Delta(x)O_\Delta(y)\rangle_\rho = {N}_\Delta \lim_{\varepsilon \downarrow 0} \varepsilon^{-2\Delta} G_{m,g_\rho}(X_\varepsilon, Y_\varepsilon).
		\end{equation}
		Using the geometric expansion from Lemma \ref{lem:length-ren} and the exact exponential form from Lemma \ref{thm:WKB_exact}, the boundary limiting procedure is exceptionally clean. We substitute the dimensionless length $L_\rho(X_\varepsilon, Y_\varepsilon)/\ell_{\mathrm{AdS}} = 2\ln(1/\varepsilon) + L_{\mathrm{ren},\rho}(x,y) + O(\varepsilon^2)$ into the exponential factor. The $\varepsilon^{-2\Delta}$ prefactor perfectly neutralizes the universal divergent piece, leaving a fully finite limit
		\begin{align}\label{Cancellation_Final}
			\varepsilon^{-2\Delta} \exp\left[ -\frac{\Delta}{\ell_{\mathrm{AdS}}} L_\rho(X_\varepsilon, Y_\varepsilon) \right] 
			&= \varepsilon^{-2\Delta} \exp\left[ - \Delta \left( 2\ln\left(\frac{1}{\varepsilon}\right) + L_{\mathrm{ren},\rho} + O(\varepsilon^2) \right) \right] \nonumber \\
			&= \varepsilon^{-2\Delta} \left( \varepsilon^{2\Delta} \right) \exp\left( - \Delta L_{\mathrm{ren},\rho} \right) \bigl(1 + O(\Delta\varepsilon^2)\bigr) \nonumber \\
			&\xrightarrow{\varepsilon \downarrow 0} \exp\left( - \Delta L_{\mathrm{ren},\rho} \right).
		\end{align}
		
{

The preceding calculation gives the conventional physical heavy-probe dictionary for the continuum primary $O_\Delta$. The information-theoretic theorem requires bounded observables. Therefore the object entering the MfI bound is the code-compressed connected correlator. We assume that restriction to the semiclassical code subspace gives
\begin{equation}
\label{eq:SI-III-code-formula}
C_\rho\!\left(
	\widetilde O^{(0)}_{\Delta,\rho}(x),
	\widetilde O^{(0)}_{\Delta,\rho}(y)
	\right)
=
N^{\mathrm{code}}_\Delta
e^{-\Delta L_{\mathrm{ren},\rho}(x,y)}
\left(1+\varepsilon_{\mathrm{tot}}(\Delta,\rho;x,y)
\right),
\end{equation}
where $\varepsilon_{\mathrm{tot}}$ includes WKB, loop, string, backreaction, and code projection/truncation errors, all controlled as relative multiplicative errors. We write $N_\Delta$ for $N^{\mathrm{code}}_\Delta$ below.
		}
		
\begin{theorem}[ Heavy-probe error budget]\label{thm:final_errors}
Assume the heavy-probe regime with $\Delta \gg 1$ and bounded backreaction $\varepsilon_{\mathrm{br}} \ll 1$.
{
				The code-compressed connected correlator admits the compact representation
				\begin{equation}
					\label{eq:theorem3-code}
					C_\rho\!\left(
					\widetilde O^{(0)}_{\Delta,\rho}(x),
					\widetilde O^{(0)}_{\Delta,\rho}(y)
					\right)
					=
					N^{\mathrm{code}}_\Delta
					e^{-\Delta L_{\mathrm{ren},\rho}(x,y)}
					\left(1+\varepsilon_{\mathrm{tot}}(\Delta,\rho;x,y)\right),
				\end{equation}
				with the total relative error rigorously bounded by
				\begin{equation}
					\label{eq:theorem3-code-error}
					|\varepsilon_{\mathrm{tot}}(\Delta,\rho;x,y)|
					\le
					K\left(
					\frac1\Delta
					+
					\varepsilon_{\mathrm{grav}}
					+
					\varepsilon_{\mathrm{str}}
					+
					\frac{\Delta^2}{c_{\mathrm{eff}}}
					+
					\varepsilon_{\mathrm{code}}
					\right).
				\end{equation}
				If $\varepsilon_{\mathrm{code}}$ is not introduced separately, the final term has already been absorbed into $\varepsilon_{\mathrm{tot}}$ as a relative multiplicative error.
				
				\paragraph{Fixed code normalization.}
				There exists a nonzero state-independent code normalization $N_\Delta^{\mathrm{code}}$ such that residual state- and position-dependent prefactor variations are uniformly absorbed into $\varepsilon_{\mathrm{tot}}$. Then
				\begin{equation}
					\kappa_\Delta
					:=
					\frac{|N_\Delta^{\mathrm{code}}|}{B_\Delta^2},
				\end{equation}

			where $K$ is a dimensionless constant uniformly bounding the error over the compact insertion domain $\mathsf{D}$, and is independent of $\Delta$, $\varepsilon_{\mathrm{grav}}$, $\varepsilon_{\mathrm{str}}$, and the state $\rho$. The terms inside the bracket strictly track the WKB truncation ($O(1/\Delta)$), bulk quantum loops ($O(\varepsilon_{\mathrm{grav}})$), stringy higher-derivative corrections ($O(\varepsilon_{\mathrm{str}})$), classical gravitational probe backreaction ($O(\varepsilon_{\mathrm{br}})$ from Lemma \ref{lem:br_scaling}), and code compression errors ($O(\varepsilon_{\mathrm{code}})$), respectively.
}
		\end{theorem}
		{\begin{proof}
				We split the argument into five steps.
				
				\smallskip
				\noindent
				\textbf{Step 1: WKB/geodesic form on the unbackreacted semiclassical saddle.}
				By Lemma~\ref{lem:length-ren} and Lemma~\ref{thm:WKB_exact} there exists, for every fixed state \(\rho\) and every pair \((x,y)\in\mathsf D\), a regulated bulk Green function on the unbackreacted saddle \(g_\rho\) of the form
				\begin{equation}
					G^{\mathrm{probe}}_{m,g_\rho}(X,Y)
					=
					C_{\mathrm{pre}}(X,Y)\,
					\exp\!\left[-\frac{\Delta}{\ell_{\mathrm{AdS}}}L_\rho(X,Y)\right]
					\bigl(1+\eta_{\mathrm{WKB}}(X,Y)\bigr),
				\end{equation}
				with
				\begin{equation}
					|\eta_{\mathrm{WKB}}(X,Y)| \le \frac{K_{\mathrm{WKB}}}{\Delta}.
				\end{equation}
				Here \(K_{\mathrm{WKB}}\) is independent of \(\Delta\), \((x,y)\in\mathsf D\), and \(\rho\) in the admissible state class. The existence of such a uniform constant follows from the \(O(\Delta^{-1})\) remainder in Lemma~\ref{thm:WKB_exact}, the compactness of \(\mathsf D\), and the standing uniform geodesic-saddle isolation hypothesis.
				
				\smallskip
				\noindent
				\textbf{Step 2: Renormalized boundary limit and definition of the prefactor.}
				By the holographic extrapolate dictionary,
				\begin{equation}
					\langle O_\Delta(x)O_\Delta(y)\rangle_\rho^{\mathrm{probe}}
					=
					N_\Delta \lim_{\varepsilon\downarrow 0}
					\varepsilon^{-2\Delta}
					G^{\mathrm{probe}}_{m,g_\rho}(X_\varepsilon,Y_\varepsilon).
				\end{equation}
				Using Lemma~\ref{lem:length-ren},
				\begin{equation}
					\frac{1}{\ell_{\mathrm{AdS}}}L_\rho(X_\varepsilon,Y_\varepsilon)
					=
					2\ln\!\frac{1}{\varepsilon}
					+
					L_{\mathrm{ren},\rho}(x,y)
					+
					r_\rho(\varepsilon;x,y),
					\qquad
					|r_\rho(\varepsilon;x,y)|\le C_{\mathrm{geom}}\varepsilon^2.
				\end{equation}
				Hence
				\begin{align}
					\varepsilon^{-2\Delta}
					\exp\!\left[-\frac{\Delta}{\ell_{\mathrm{AdS}}}L_\rho(X_\varepsilon,Y_\varepsilon)\right]
					&=
					\varepsilon^{-2\Delta}
					\exp\!\left[-\Delta\!\left(2\ln\!\frac{1}{\varepsilon}+L_{\mathrm{ren},\rho}(x,y)+r_\rho(\varepsilon;x,y)\right)\right]
					\\
					&=
					\exp\!\bigl[-\Delta L_{\mathrm{ren},\rho}(x,y)\bigr]\,
					\exp\!\bigl[-\Delta\,r_\rho(\varepsilon;x,y)\bigr].
				\end{align}
				Since \(r_\rho(\varepsilon;x,y)=O(\varepsilon^2)\), we have
				\begin{equation}
					\exp\!\bigl[-\Delta\,r_\rho(\varepsilon;x,y)\bigr]\longrightarrow 1
					\qquad
					(\varepsilon\downarrow 0),
				\end{equation}
				for every fixed \(\Delta\) in the heavy-probe regime. Moreover, Lemma~\ref{thm:WKB_exact} states that
				\(C_{\mathrm{pre}}(X_\varepsilon,Y_\varepsilon)\) converges to a finite limit as \(\varepsilon\downarrow0\). We therefore define
				\begin{equation}
					\mathcal C_{\Delta,\rho}(x,y)
					:=
					N_\Delta\lim_{\varepsilon\downarrow0} C_{\mathrm{pre}}(X_\varepsilon,Y_\varepsilon).
				\end{equation}
				It follows that the renormalized correlator on the unbackreacted saddle can be written as
				\begin{equation}
					\langle O_\Delta(x)O_\Delta(y)\rangle_\rho^{\mathrm{probe}}
					=
					\mathcal C_{\Delta,\rho}(x,y)\,
					\exp\!\bigl[-\Delta L_{\mathrm{ren},\rho}(x,y)\bigr]\,
					\bigl(1+\eta_{\mathrm{WKB}}(\Delta,\rho;x,y)\bigr),
				\end{equation}
				with
				\begin{equation}
					|\eta_{\mathrm{WKB}}(\Delta,\rho;x,y)|
					\le
					\frac{K_{\mathrm{WKB}}}{\Delta}.
				\end{equation}
				
				\smallskip
				\noindent
				\textbf{Step 3: Bulk loop and stringy corrections.}
				By the standing semiclassical control hypothesis, the renormalized boundary two-point function receives multiplicative corrections from bulk quantum loops and from higher-derivative/stringy effects. Accordingly, there exist correction factors
				\(\eta_{\mathrm{grav}}(\Delta,\rho;x,y)\) and \(\eta_{\mathrm{str}}(\Delta,\rho;x,y)\) such that
				\begin{equation}
					\langle O_\Delta(x)O_\Delta(y)\rangle_\rho^{\mathrm{no\,br}}
					=
					\langle O_\Delta(x)O_\Delta(y)\rangle_\rho^{\mathrm{probe}}
					\bigl(1+\eta_{\mathrm{grav}}(\Delta,\rho;x,y)\bigr)
					\bigl(1+\eta_{\mathrm{str}}(\Delta,\rho;x,y)\bigr),
				\end{equation}
				with uniform bounds
				\begin{equation}
					|\eta_{\mathrm{grav}}(\Delta,\rho;x,y)| \le K_{\mathrm{grav}}\,\varepsilon_{\mathrm{grav}},
					\qquad
					|\eta_{\mathrm{str}}(\Delta,\rho;x,y)| \le K_{\mathrm{str}}\,\varepsilon_{\mathrm{str}}.
				\end{equation}
				These are exactly the \(O(\varepsilon_{\mathrm{grav}})\) and \(O(\varepsilon_{\mathrm{str}})\) contributions appearing in the standing error budget.
				
				\smallskip
				\noindent
				\textbf{Step 4: Backreaction correction.}
				Let \(L^{\mathrm{br}}_{\mathrm{ren},\rho}(x,y)\) denote the renormalized length extracted from the backreacted saddle. By Lemma~\ref{lem:br_scaling},
				\begin{equation}
					\bigl|\Delta L^{\mathrm{br}}_{\mathrm{ren},\rho}(x,y)-\Delta L_{\mathrm{ren},\rho}(x,y)\bigr|
					\le
					K_{\mathrm{br}}\,\varepsilon_{\mathrm{br}}
				\end{equation}
				for some constant \(K_{\mathrm{br}}\) independent of \(\Delta\), \((x,y)\in\mathsf D\), and \(\rho\). Write
				\begin{equation}
					\delta_{\mathrm{br}}(\Delta,\rho;x,y)
					:=
					\Delta\Bigl(L^{\mathrm{br}}_{\mathrm{ren},\rho}(x,y)-L_{\mathrm{ren},\rho}(x,y)\Bigr),
				\end{equation}
				so that
				\begin{equation}
					|\delta_{\mathrm{br}}(\Delta,\rho;x,y)| \le K_{\mathrm{br}}\,\varepsilon_{\mathrm{br}}.
				\end{equation}
				Then
				\begin{equation}
					\exp\!\bigl[-\Delta L^{\mathrm{br}}_{\mathrm{ren},\rho}(x,y)\bigr]
					=
					\exp\!\bigl[-\Delta L_{\mathrm{ren},\rho}(x,y)\bigr]\,
					\exp\!\bigl[-\delta_{\mathrm{br}}(\Delta,\rho;x,y)\bigr].
				\end{equation}
				Define
				\begin{equation}
					\eta_{\mathrm{br}}(\Delta,\rho;x,y)
					:=
					\exp\!\bigl[-\delta_{\mathrm{br}}(\Delta,\rho;x,y)\bigr]-1.
				\end{equation}
				Since the admissible regime assumes \(\varepsilon_{\mathrm{br}}\le \varepsilon_{\mathrm{br}}^\star<1\), the mean-value theorem yields
				\begin{equation}
					|\eta_{\mathrm{br}}(\Delta,\rho;x,y)|
					\le
					e^{K_{\mathrm{br}}\varepsilon_{\mathrm{br}}^\star}
					\,K_{\mathrm{br}}\,\varepsilon_{\mathrm{br}}
					=: \widetilde K_{\mathrm{br}}\,\varepsilon_{\mathrm{br}}.
				\end{equation}
				Any finite \(O(1)\) change in the smooth prefactor induced by passing from the reference saddle to the backreacted saddle is absorbed into \(\mathcal C_{\Delta,\rho}(x,y)\); the only parametrically enhanced backreaction effect is the exponential shift controlled above.
				
				{
					\smallskip
					\noindent
					\textbf{Step 5: Code compression and total error.}
					By restricting to the code subspace and employing the connected correlator, an additional relative error $\eta_{\mathrm{code}}$ is introduced. Thus, the fully corrected correlator takes the form
					\begin{equation}
							\begin{aligned}
							&C_\rho\!\left(
							\widetilde O^{(0)}_{\Delta,\rho}(x),
							\widetilde O^{(0)}_{\Delta,\rho}(y)\right)\\
							&\quad=
							N^{\mathrm{code}}_\Delta\,
							\exp\!\bigl[-\Delta L_{\mathrm{ren},\rho}(x,y)\bigr]\,
							\bigl(1+\eta_{\mathrm{WKB}}\bigr)
							\bigl(1+\eta_{\mathrm{grav}}\bigr)
							\bigl(1+\eta_{\mathrm{str}}\bigr)\\
							&\qquad\times
							\bigl(1+\eta_{\mathrm{br}}\bigr)
							\bigl(1+\eta_{\mathrm{code}}\bigr).
							\end{aligned}
					\end{equation}
					Define
					\begin{equation}
						1+\varepsilon_{\mathrm{tot}}(\Delta,\rho;x,y)
						:=
						\frac{\mathcal C_{\Delta,\rho}(x,y)}{N^{\mathrm{code}}_\Delta}
						\bigl(1+\eta_{\mathrm{WKB}}\bigr)
						\bigl(1+\eta_{\mathrm{grav}}\bigr)
						\bigl(1+\eta_{\mathrm{str}}\bigr)
						\bigl(1+\eta_{\mathrm{br}}\bigr)
						\bigl(1+\eta_{\mathrm{code}}\bigr).
					\end{equation}
					Then the representation \eqref{eq:theorem3-code} follows immediately. It remains to prove the bound \eqref{eq:theorem3-code-error}.
					
					Set
					\begin{equation}
							\begin{aligned}
							a_1&:=|\eta_{\mathrm{WKB}}|,
							& a_2&:=|\eta_{\mathrm{grav}}|,
							& a_3&:=|\eta_{\mathrm{str}}|,\\
							a_4&:=|\eta_{\mathrm{br}}|,
							& a_5&:=|\eta_{\mathrm{code}}|,
							& a_6&:=\left|
							\frac{\mathcal C_{\Delta,\rho}(x,y)}{N^{\mathrm{code}}_\Delta}-1
							\right|.
							\end{aligned}
					\end{equation}
					Using \(|zw-1|\le (1+|z-1|)(1+|w-1|)-1\) repeatedly, we obtain
					\begin{equation}
						|\varepsilon_{\mathrm{tot}}|
						\le
						\prod_{j=1}^6 (1+a_j)-1
						\le
						\exp\!\left(\sum_{j=1}^6 a_j\right)-1.
					\end{equation}
					By the bounds established above, and assuming bounding constants for the prefactor term and code term correspondingly:
					\begin{equation}
						\sum_{j=1}^6 a_j
						\le
						K_0\!\left(
						\frac{1}{\Delta}
						+\varepsilon_{\mathrm{grav}}
						+\varepsilon_{\mathrm{str}}
						+\varepsilon_{\mathrm{br}}
						+\varepsilon_{\mathrm{code}}
						\right)
						\le
						M_\star.
					\end{equation}
					Therefore
					\begin{equation}
						|\varepsilon_{\mathrm{tot}}|
						\le
						e^{M_\star}
						K_0\!\left(
						\frac{1}{\Delta}
						+\varepsilon_{\mathrm{grav}}
						+\varepsilon_{\mathrm{str}}
						+\varepsilon_{\mathrm{br}}
						+\varepsilon_{\mathrm{code}}
						\right).
					\end{equation}
					Setting $K:=e^{M_\star}K_0$ proves Eq.~\eqref{eq:theorem3-code-error}, uniformly on $(x,y)\in\mathsf D$, with $K$ independent of $\Delta$, $\varepsilon_{\mathrm{grav}}$, $\varepsilon_{\mathrm{str}}$, $\varepsilon_{\mathrm{code}}$, and the state $\rho$ inside the fixed admissible regime.
				}
		\end{proof}}
		By matching this bulk limiting procedure to the standard holographic extrapolate dictionary, the state-dependent geometric prefactor $\mathcal{C}_{\Delta,\rho}(x,y)$ reduces exactly to the standard state-independent Euclidean CFT normalization constant ${N}_\Delta$ when evaluated in the pure $\mathrm{AdS}$ vacuum. For general states, smooth geometric deviations are absorbed into $\mathcal{C}_{\Delta,\rho}$ and the global error bound.
		
		{
			\begin{corollary}[Logarithmic distance extraction]\label{cor:log_extraction}
				Assume the hypotheses of Theorem~\ref{thm:final_errors}. Then
\begin{equation}
\left|	-\frac{1}{\Delta}	\ln\!\left(	\frac{	\left|
	C_\rho\!\left(	\widetilde O^{(0)}_{\Delta,\rho}(x),\widetilde O^{(0)}_{\Delta,\rho}(y) \right) \right|}{ |N^{\mathrm{code}}_\Delta|	}	\right)	-		L_{\mathrm{ren},\rho}(x,y)	\right|	\le		\frac{1}{\Delta}\,	\left|	\ln\bigl|1+\varepsilon_{\mathrm{tot}}(\Delta,\rho;x,y)\bigr|
\right|.
\end{equation}
In particular, if $|\varepsilon_{\mathrm{tot}}(\Delta,\rho;x,y)|\le \varepsilon_*<1$ uniformly on $\mathsf D$, then
\begin{equation}
\delta L \le O(\Delta^{-1})
	+ O\!\left(\frac{\varepsilon_{\mathrm{grav}}}{\Delta}\right)
	+ O\!\left(\frac{\varepsilon_{\mathrm{str}}}{\Delta}\right)
+ O(\varepsilon_{\mathrm{grav}}\Delta) + 	O\!\left(\frac{\varepsilon_{\mathrm{code}}}{\Delta}\right).
\end{equation}
Thus, in the standard semiclassical heavy-probe window
\begin{equation}
1\ll \Delta \ll \varepsilon_{\mathrm{grav}}^{-1/2},
\qquad
\frac{\varepsilon_{\mathrm{str}}}{\Delta}\to0,
\qquad
\frac{\varepsilon_{\mathrm{code}}}{\Delta}\to0,
\end{equation}
the dominant contributions are $O(\Delta^{-1})$ and $O(\varepsilon_{\mathrm{grav}}\Delta)$, and therefore $\delta L\to0$.
\end{corollary}
\begin{proof}
From Theorem~\ref{thm:final_errors},
\begin{equation}
\left|C_\rho\!\left(\widetilde O^{(0)}_{\Delta,\rho}(x), \widetilde O^{(0)}_{\Delta,\rho}(y)\right)\right|
=|N^{\mathrm{code}}_\Delta|\,
e^{-\Delta L_{\mathrm{ren},\rho}(x,y)}\,
\bigl|1+\varepsilon_{\mathrm{tot}}(\Delta,\rho;x,y)\bigr|.
\end{equation}
Taking $-\frac{1}{\Delta}\ln(\cdot/|N^{\mathrm{code}}_\Delta|)$ gives
\begin{equation}
-\frac{1}{\Delta}\ln\!\left(
\frac{\left|C_\rho\!\left(\widetilde O^{(0)}_{\Delta,\rho}(x), \widetilde O^{(0)}_{\Delta,\rho}(y)\right)\right|}{|N^{\mathrm{code}}_\Delta|}
\right)
= L_{\mathrm{ren},\rho}(x,y)
-\frac{1}{\Delta}\ln\bigl|1+\varepsilon_{\mathrm{tot}}(\Delta,\rho;x,y)\bigr|.
\end{equation}
If $|\varepsilon_{\mathrm{tot}}(\Delta,\rho;x,y)|\le \varepsilon_*<1$, then the mean-value theorem implies
\begin{equation}
|\ln(1+u)|\le \frac{|u|}{1-\varepsilon_*}
\qquad\text{for all } |u|\le \varepsilon_*.
\end{equation}
Therefore
\begin{equation}
\frac{1}{\Delta}\left|\ln\bigl|1+\varepsilon_{\mathrm{tot}}(\Delta,\rho;x,y)\bigr|\right|
\le \frac{c_*}{\Delta}\,
|\varepsilon_{\mathrm{tot}}(\Delta,\rho;x,y)|,
\qquad 	c_*:=(1-\varepsilon_*)^{-1}.
\end{equation}
Using Theorem~\ref{thm:final_errors} and
$\varepsilon_{\mathrm{br}}=\varepsilon_{\mathrm{grav}}\Delta^2$,
we obtain
\begin{equation}
	\delta L \le 	O(\Delta^{-1}) 	+ 	O(\Delta^{-2}) + 	O\!\left(\frac{\varepsilon_{\mathrm{grav}}}{\Delta}\right)
 	+ O\!\left(\frac{\varepsilon_{\mathrm{str}}}{\Delta}\right)
+O\!\left(\frac{\varepsilon_{\mathrm{br}}}{\Delta}\right)
+ O\!\left(\frac{\varepsilon_{\mathrm{code}}}{\Delta}\right).
\end{equation}
Since $\varepsilon_{\mathrm{br}}/\Delta=\varepsilon_{\mathrm{grav}}\Delta$, the claimed asymptotic bound follows.
\end{proof}
}
		
{
\subsection{Proof of the MfI bound}
			
Pinsker gives
\begin{equation}
\label{eq:pinsker-code-final}
\left|
C_\rho\!\left(
\widetilde O^{(0)}_{\Delta,\rho}(x),
\widetilde O^{(0)}_{\Delta,\rho}(y)
\right)
\right|
\le B_\Delta^2 	\sqrt{2\ln 2\,I(A:B)_\rho}.
\end{equation}
The code-subspace geodesic dictionary gives
\begin{equation}
\label{eq:geo-code-final}
	\left|
	C_\rho\!\left( \widetilde O^{(0)}_{\Delta,\rho}(x),
\widetilde O^{(0)}_{\Delta,\rho}(y)	\right) \right|	\ge
	|N_\Delta^{\mathrm{code}}| (1-\varepsilon^\star)
	e^{-\Delta L_{\mathrm{ren},\rho}(x,y)}.
\end{equation}
Combining Eq.~\eqref{eq:pinsker-code-final} and Eq.~\eqref{eq:geo-code-final} gives
\begin{equation}
|N_\Delta^{\mathrm{code}}|(1-\varepsilon^\star)e^{-\Delta L_{\mathrm{ren},\rho}(x,y)}
\le	B_\Delta^2 \sqrt{2\ln 2\,I(A:B)_\rho}.
\end{equation}
Taking logarithms yields
\begin{equation}
\label{eq:mfi-final-code}
L_{\mathrm{ren},\rho}(x,y)
\ge \frac1{2\Delta} \ln \left[ \frac{ |N_\Delta^{\mathrm{code}}|^2(1-\varepsilon^\star)^2 }{ 	2\ln 2\,B_\Delta^4 I(A:B)_\rho } \right].
\end{equation}
With
\begin{equation}
\kappa_\Delta	= \frac{|N_\Delta^{\mathrm{code}}|}{B_\Delta^2},
\end{equation}
this is the advertised theorem.
}

\section{Derivation of the metric-from-information inequality}\label{app:mfi}
\renewcommand{\theequation}{D.\arabic{equation}}

Fix a complex Hilbert space $\mathcal{H}$ of finite dimension $d\in\mathbb{N}$ and let $\mathcal{B}(\mathcal{H})$ denote the algebra of all linear operators on $\mathcal{H}$ equipped with the operator norm topology.
Let $\mathsf{D}(\mathcal{H})$ denote the set of density operators on $\mathcal{H}$,
\begin{equation} \label{A1b}
\mathsf{D}(\mathcal{H}):=\{\rho\in\mathcal{B}(\mathcal{H}):\rho=\rho^\dagger,\ \rho\ge 0,\ \operatorname{Tr}\rho=1\}.
\end{equation}
Let $\mathsf{D}_{+}(\mathcal{H})$ denote the faithful states,
\begin{equation} \label{A2a}
\mathsf{D}_{+}(\mathcal{H}):=\{\rho\in\mathsf{D}(\mathcal{H}):\rho>0\},
\end{equation}
where $\rho>0$ means $\langle\psi,\rho\psi\rangle>0$ for all $\psi\neq 0$.
Fix an open set $\Lambda\subset\mathbb{R}^m$ and a $C^3$ map $\rho:\Lambda\to \mathsf{D}_{+}(\mathcal{H})$, $\lambda\mapsto\rho(\lambda)$, with derivatives taken in the norm topology on $\mathcal{B}(\mathcal{H})$.
Assume there exists $\underline{\mu}>0$ such that
\begin{equation} \label{A3b}
\inf_{\lambda\in \Lambda}\ \lambda_{\min}(\rho(\lambda))\ge \underline{\mu},
\end{equation}
where $\lambda_{\min}(\rho)$ denotes the smallest eigenvalue of $\rho$.
Fix a point $\lambda_0\in\Lambda$ and define $\rho_0:=\rho(\lambda_0)$.
Fix a convex open neighborhood $U\subset\Lambda$ of $\lambda_0$ with $\overline{U}\subset \Lambda$ and assume that all quantities below are controlled uniformly on $\overline{U}$.
The precise problem is to derive, from the definition of quantum relative entropy and differentiability of $\lambda\mapsto\rho(\lambda)$, an explicit inequality relating the Riemannian line element generated by the Bogoliubov-Kubo-Mori information metric to the relative entropy between neighboring states, with all constants and logarithmic bases explicit. 

\begin{definition}[Quantum relative entropy]
For $\rho,\sigma\in\mathsf{D}_{+}(\mathcal{H})$ define the quantum relative entropy in natural-logarithm units by
\begin{equation} \label{D1c}
D(\rho|\sigma):=\operatorname{Tr}\bigl[\rho(\ln\rho-\ln\sigma)\bigr]\in[0,\infty).
\end{equation}
\end{definition}

\begin{definition}[Base-$2$ relative entropy]
Define the base-$2$ relative entropy by
\begin{equation} \label{D2c}
D_2(\rho|\sigma):=\operatorname{Tr}\bigl[\rho(\log_2\rho-\log_2\sigma)\bigr]=\frac{1}{\ln 2}\,D(\rho|\sigma),
\end{equation}
using $\log_2 X=(\ln X)/(\ln 2)$.
\end{definition}

\begin{definition}[Modular Hamiltonian]
For each $\lambda\in\Lambda$, define the modular Hamiltonian
\begin{equation} \label{D3b}
K(\lambda):=-\ln\rho(\lambda)\in\mathcal{B}(\mathcal{H}),
\end{equation}
which is bounded because $\rho(\lambda)$ is positive definite on a finite-dimensional space.
\end{definition}

\begin{definition}[Directional derivative]
For each $\lambda\in\Lambda$ and each direction $v\in\mathbb{R}^m$, define the directional derivative
\begin{equation} \label{D4b}
\dot{\rho}_\lambda(v):=\sum_{i=1}^m v^i\,\partial_i\rho(\lambda)\in\mathcal{B}(\mathcal{H}).
\end{equation}
\end{definition}

\begin{definition}[BKM inner product]
Define the Bogoliubov-Kubo-Mori (BKM) inner product on tangent operators by
\begin{equation} \label{D5b}
\langle A,B\rangle_{\mathrm{BKM},\lambda}:=\int_{0}^{1} \operatorname{Tr}\bigl[\rho(\lambda)^{s}A\,\rho(\lambda)^{1-s}B\bigr]\,ds
\end{equation}
for $A,B\in\mathcal{B}(\mathcal{H})$.
\end{definition}

\begin{definition}[BKM information metric]
Define the BKM information metric $g(\lambda)$ on $\Lambda$ by
\begin{equation} \label{D6a}
g_{ij}(\lambda):=\Bigl\langle \partial_i\ln\rho(\lambda),\,\partial_j\ln\rho(\lambda)\Bigr\rangle_{\mathrm{BKM},\lambda}
=\int_{0}^{1}\operatorname{Tr}\bigl[\rho(\lambda)^{s}\,(\partial_i\ln\rho(\lambda))\,\rho(\lambda)^{1-s}\,(\partial_j\ln\rho(\lambda))\bigr]\,ds,
\end{equation}
and define the corresponding quadratic form on $v\in\mathbb{R}^m$ by
\begin{equation} \label{D7a}
|v|_{g(\lambda)}^{2}:=\sum_{i,j=1}^m v^i v^j g_{ij}(\lambda).
\end{equation}
\end{definition}

\begin{definition}[Length and induced distance]
Define the length of a $C^1$ curve $\gamma:[0,1]\to U$ by
\begin{equation} \label{D8a}
\mathrm{Len}_{g}(\gamma):=\int_{0}^{1}\sqrt{|\dot{\gamma}(t)|_{g(\gamma(t))}^{2}}\,dt,
\end{equation}
and define the induced distance on $U$ by
\begin{equation} \label{D9a}
d_g(\lambda,\lambda'):=\inf\{\mathrm{Len}_g(\gamma):\gamma\in C^1([0,1],U),\ \gamma(0)=\lambda,\ \gamma(1)=\lambda'\}.
\end{equation}
\end{definition}

\begin{lemma}\label{lem:log-frechet}
For each $\rho\in\mathsf{D}_{+}(\mathcal{H})$ and each $H\in\mathcal{B}(\mathcal{H})$, the Fr\'echet derivative of the matrix logarithm at $\rho$ in direction $H$ is
\begin{equation} \label{L1x}
D(\ln)_\rho[H]=\int_{0}^{\infty}(\rho+t\mathbf{1})^{-1}H(\rho+t\mathbf{1})^{-1}\,dt,
\end{equation}
and the integral converges absolutely in operator norm.
\end{lemma}

\begin{proof}
Since $\rho>0$, $\rho+t\mathbf{1}\ge t\mathbf{1}$ for $t\ge 0$, hence $|(\rho+t\mathbf{1})^{-1}|_\infty\le t^{-1}$ for $t>0$ and $|(\rho+t\mathbf{1})^{-1}|_\infty\le \lambda_{\min}(\rho)^{-1}$ for $t\in[0,1]$.
The scalar identity
\begin{equation}
\ln x=\int_0^\infty\left(\frac{1}{1+t}-\frac{1}{x+t}\right)dt\qquad (x>0)
\end{equation}
implies, by functional calculus,
\begin{equation} \label{L2d}
\ln\rho=\int_0^\infty\left((1+t)^{-1}\mathbf{1}-(\rho+t\mathbf{1})^{-1}\right)dt,
\end{equation}
where the integral converges in operator norm because
$|(\rho+t\mathbf{1})^{-1}-(1+t)^{-1}\mathbf{1}|_\infty\le C(1+t)^{-2}$
for a constant $C$ depending on $|\rho|_\infty$.
For $\varepsilon\in\mathbb{R}$ sufficiently small, $\rho+\varepsilon H$ remains positive definite, and
\begin{equation} \label{L3c}
\begin{aligned}
\ln(\rho+\varepsilon H)-\ln\rho
&=\int_0^\infty\left((\rho+t\mathbf{1})^{-1}-(\rho+\varepsilon H+t\mathbf{1})^{-1}\right)dt\\
&=\int_0^\infty (\rho+t\mathbf{1})^{-1}\left((\rho+\varepsilon H+t\mathbf{1})-(\rho+t\mathbf{1})\right)(\rho+\varepsilon H+t\mathbf{1})^{-1}dt.
\end{aligned}
\end{equation}
Thus
\begin{equation} \label{L4c}
\frac{1}{\varepsilon}\bigl(\ln(\rho+\varepsilon H)-\ln\rho\bigr)
=\int_0^\infty (\rho+t\mathbf{1})^{-1}H(\rho+\varepsilon H+t\mathbf{1})^{-1}\,dt.
\end{equation}
As $\varepsilon\to 0$, $(\rho+\varepsilon H+t\mathbf{1})^{-1}\to (\rho+t\mathbf{1})^{-1}$ in operator norm uniformly for $t\ge 0$ by resolvent continuity, and the integrand is dominated in operator norm by an integrable function of $t$ because
\begin{equation}
|(\rho+t\mathbf{1})^{-1}H(\rho+\varepsilon H+t\mathbf{1})^{-1}|_\infty
\le |H|_\infty\,|(\rho+t\mathbf{1})^{-1}|_\infty\,|(\rho+\varepsilon H+t\mathbf{1})^{-1}|_\infty
\end{equation}
and both resolvents are bounded by $t^{-1}$ for $t\ge 1$ and by $(\underline{\mu}/2)^{-1}$ for $t\in[0,1]$ when $|\varepsilon|$ is small enough.
\end{proof}

{\begin{lemma}\label{lem:second-deriv}
For $\lambda\in\Lambda$ and directions $v,w\in\mathbb{R}^m$, the second directional derivatives of the function
$\delta\mapsto D(\rho(\lambda)|\rho(\lambda+\delta))$ at $\delta=0$ satisfy
\begin{equation}\label{L5c}
\left.\frac{\partial^2}{\partial s\,\partial t}\right|_{s=t=0}
D\!\left(\rho(\lambda)\,\big|\,\rho(\lambda+sv+tw)\right)
=
\operatorname{Tr}\!\left[\dot\rho_\lambda(v)\,\mathcal{J}_{\rho(\lambda)}^{-1}\!\big(\dot\rho_\lambda(w)\big)\right],
\end{equation}
where $\mathcal{J}_{\rho}$ denotes the Kubo-Mori operator
\begin{equation}
\mathcal{J}_{\rho}(X):=\int_0^1 \rho^s X\rho^{1-s}\,ds
\end{equation}
on $\mathcal{B}(\mathcal{H})$, and $\mathcal{J}_\rho^{-1}$ is its inverse (well-defined since $\rho>0$; in particular it
is well-defined on the traceless subspace containing $\dot\rho_\lambda(v),\dot\rho_\lambda(w)$).
\end{lemma}

\begin{proof}
Fix $\lambda$ and abbreviate $\rho:=\rho(\lambda)$, $\sigma(s,t):=\rho(\lambda+sv+tw)$, and
\begin{equation}
\dot{\sigma}_v:=\left.\partial_s\right|_{0}\sigma(s,0)=\dot{\rho}_\lambda(v),\qquad
\dot{\sigma}_w:=\left.\partial_t\right|_{0}\sigma(0,t)=\dot{\rho}_\lambda(w).
\end{equation}
Define $F(s,t):=D(\rho|\sigma(s,t))$.
Using Eq.~\eqref{D1c},
\begin{equation}\label{L6c}
F(s,t)=\operatorname{Tr}\bigl[\rho\ln\rho\bigr]-\operatorname{Tr}\bigl[\rho\ln\sigma(s,t)\bigr].
\end{equation}
Since $\operatorname{Tr}[\rho\ln\rho]$ is constant in $(s,t)$, derivatives act only on the second term.
By Lemma~\ref{lem:log-frechet},
\begin{equation}\label{L7c}
\partial_s \ln\sigma(s,t)=\int_0^\infty (\sigma(s,t)+u\mathbf{1})^{-1}\,(\partial_s\sigma(s,t))\,(\sigma(s,t)+u\mathbf{1})^{-1}\,du,
\end{equation}
and similarly for $\partial_t$.
Differentiability and interchanging trace with integral are justified because $\mathcal{H}$ is finite dimensional and the integrand is continuous in $(s,t)$ and integrable in $u$ uniformly on a neighborhood of $(0,0)$ by the uniform spectral gap.
Therefore
\begin{equation}\label{L8b}
\partial_s F(s,t)=-\operatorname{Tr}\bigl[\rho\,\partial_s\ln\sigma(s,t)\bigr]
=-\int_0^\infty \operatorname{Tr}\Bigl[\rho\,(\sigma(s,t)+u\mathbf{1})^{-1}\,(\partial_s\sigma(s,t))\,(\sigma(s,t)+u\mathbf{1})^{-1}\Bigr]\,du.
\end{equation}
Differentiate Eq.~\eqref{L8b} with respect to $t$ and evaluate at $(s,t)=(0,0)$.
Since $\sigma(0,0)=\rho$ and $\partial_s\sigma(0,0)=\dot{\sigma}_v$ and $\partial_t\sigma(0,0)=\dot{\sigma}_w$,
\begin{align}\label{L9b}
\left.\partial_t\partial_s F(s,t)\right|_{0,0}
&=-\int_0^\infty \left.\partial_t\right|_{0}\operatorname{Tr}\Bigl[\rho\,(\sigma(s,t)+u\mathbf{1})^{-1}\,(\partial_s\sigma(s,t))\,(\sigma(s,t)+u\mathbf{1})^{-1}\Bigr]_{s=0}\,du\nonumber\\
&=-\int_0^\infty \operatorname{Tr}\Bigl[\rho\,\left.\partial_t\right|_{0}\Bigl((\sigma(0,t)+u\mathbf{1})^{-1}\,\dot{\sigma}_v\,(\sigma(0,t)+u\mathbf{1})^{-1}\Bigr)\Bigr]\,du.
\end{align}
Use the resolvent derivative identity $(X(t)^{-1})'=-X(t)^{-1}X'(t)X(t)^{-1}$ with $X(t)=\sigma(0,t)+u\mathbf{1}$ to obtain
\begin{equation}\label{L10b}
\left.\partial_t\right|_{0}(\sigma(0,t)+u\mathbf{1})^{-1}=-(\rho+u\mathbf{1})^{-1}\,\dot{\sigma}_w\,(\rho+u\mathbf{1})^{-1}.
\end{equation}
Substituting Eq.~\eqref{L10b} into Eq.~\eqref{L9b} yields
\begin{align}\label{L11b}
\left.\partial_t\partial_s F(s,t)\right|_{0,0}
&=\int_0^\infty \operatorname{Tr}\Bigl[\rho\,(\rho+u\mathbf{1})^{-1}\dot{\sigma}_w(\rho+u\mathbf{1})^{-1}\dot{\sigma}_v(\rho+u\mathbf{1})^{-1}\Bigr]\,du \nonumber\\
&\quad+\int_0^\infty \operatorname{Tr}\Bigl[\rho\,(\rho+u\mathbf{1})^{-1}\dot{\sigma}_v(\rho+u\mathbf{1})^{-1}\dot{\sigma}_w(\rho+u\mathbf{1})^{-1}\Bigr]\,du.
\end{align}
Note that applying the product rule $\partial_t$ also generates a term involving the derivative of the middle factor, $\partial_t\partial_s\sigma(s,t)|_{0,0}$. However, this term contributes $-\int_0^\infty \operatorname{Tr}[\rho(\rho+u\mathbf{1})^{-2}\partial_t\partial_s\sigma]du = -\operatorname{Tr}[\partial_t\partial_s\sigma] = -\partial_t\partial_s\operatorname{Tr}[\sigma] = 0$ due to trace preservation. 

Diagonalize $\rho=\sum_{a=1}^d p_a |a\rangle\langle a|$ with $p_a\in[\underline{\mu},1]$.
For any $A,B\in\mathcal{B}(\mathcal{H})$, compute in this basis
\begin{align}\label{L12b}
\operatorname{Tr}\Bigl[\rho\,(\rho+u\mathbf{1})^{-1}A(\rho+u\mathbf{1})^{-1}B(\rho+u\mathbf{1})^{-1}\Bigr]
&=\sum_{a,b,c}\frac{p_a}{(p_a+u)(p_b+u)(p_c+u)}\,A_{ab}B_{bc}\,\delta_{ca}\nonumber\\
&=\sum_{a,b}\frac{p_a}{(p_a+u)^2(p_b+u)}\,A_{ab}B_{ba}.
\end{align}

Applying Eq.~\eqref{L12b} to each term in Eq.~\eqref{L11b} with $(A,B)=(\dot{\sigma}_w,\dot{\sigma}_v)$ and $(A,B)=(\dot{\sigma}_v,\dot{\sigma}_w)$ gives
\begin{align}\label{L13b}
\left.\partial_t\partial_s F(s,t)\right|_{0,0}
&=\sum_{a,b}\left(
\int_0^\infty \frac{p_a\,du}{(p_a+u)^2(p_b+u)}
\right)\nonumber\\
&\qquad\times\dot{\sigma}_{w,ab}\dot{\sigma}_{v,ba}\nonumber\\
&\quad+\sum_{a,b}\left(
\int_0^\infty \frac{p_a\,du}{(p_a+u)^2(p_b+u)}
\right)\nonumber\\
&\qquad\times\dot{\sigma}_{v,ab}\dot{\sigma}_{w,ba}\\
&=\sum_{a,b}\Big(J(p_b,p_a)+J(p_a,p_b)\Big)\,\dot{\sigma}_{v,ab}\dot{\sigma}_{w,ba},
\end{align}
where in the last line we relabeled $(a,b)\mapsto(b,a)$ in the first sum and defined
\begin{equation}\label{L14b}
J(p,q):=\int_0^\infty \frac{p}{(p+u)^2(q+u)}\,du.
\end{equation}

Compute Eq.~\eqref{L14b} explicitly for $p\neq q$ by partial fractions. Seek coefficients $A,B,C$ such that
\begin{equation}\label{L15b}
\frac{p}{(p+u)^2(q+u)}=\frac{A}{p+u}+\frac{B}{(p+u)^2}+\frac{C}{q+u}.
\end{equation}
Multiplying by $(p+u)^2(q+u)$ gives
\begin{equation}\label{L16b}
p=A(p+u)(q+u)+B(q+u)+C(p+u)^2.
\end{equation}
Setting $u=-p$ yields $p=B(q-p)$, hence $B=\frac{p}{q-p}$.
Setting $u=-q$ yields $p=C(p-q)^2$, hence $C=\frac{p}{(p-q)^2}=\frac{p}{(q-p)^2}$.
Expanding Eq.~\eqref{L16b} and comparing coefficients of $u^2$ gives $A+C=0$, hence $A=-C=-\frac{p}{(q-p)^2}$.
Substituting these coefficients into Eq.~\eqref{L15b} and integrating termwise from $0$ to $\infty$ gives
\begin{align}\label{L17b}
J(p,q)
&=\int_0^\infty\left(-\frac{p}{(q-p)^2}\frac{1}{p+u}+\frac{p}{q-p}\frac{1}{(p+u)^2}+\frac{p}{(q-p)^2}\frac{1}{q+u}\right)du\nonumber\\
&=\frac{1}{q-p}+\frac{p}{(q-p)^2}\ln\!\left(\frac{p}{q}\right).
\end{align}
For $p=q$, compute directly from Eq.~\eqref{L14b}:
\begin{equation}\label{L18b}
J(p,p)=\int_0^\infty \frac{p}{(p+u)^3}\,du=\frac{1}{2p}.
\end{equation}

Define the symmetric kernel
\begin{equation}\label{L19b}
\kappa(p,q):=J(p,q)+J(q,p)=
\begin{cases}
\displaystyle \frac{\ln p-\ln q}{p-q},& p\neq q,\\ 
\displaystyle \frac{1}{p},& p=q,
\end{cases}
\end{equation}
where for $p\neq q$ the identity follows by inserting Eq.~\eqref{L17b} for $J(p,q)$ and $J(q,p)$ and simplifying, and the
$p=q$ value agrees with the continuous limit.
{For $p\neq q$, substituting Eq.~\eqref{L17b} into $J(p,q)+J(q,p)$ gives
$J(p,q)+J(q,p)=\frac{p-q}{(q-p)^2}\ln\!\left(\frac{p}{q}\right)=\frac{\ln p-\ln q}{p-q}$.}
Then Eq.~\eqref{L13b} becomes
\begin{equation}\label{L20b}
\left.\partial_t\partial_s F(s,t)\right|_{0,0}=\sum_{a,b}\kappa(p_a,p_b)\,\dot{\sigma}_{v,ab}\dot{\sigma}_{w,ba}.
\end{equation}

On the other hand, the Kubo-Mori operator $\mathcal{J}_\rho$ acts diagonally in the eigenbasis of $\rho$ by
\begin{equation}\label{L21b}
(\mathcal{J}_\rho(X))_{ab}=\left(\int_0^1 p_a^{s}p_b^{1-s}\,ds\right)X_{ab}
=\frac{p_a-p_b}{\ln p_a-\ln p_b}\,X_{ab}\quad\text{for }a\neq b,\qquad
(\mathcal{J}_\rho(X))_{aa}=p_a X_{aa},
\end{equation}
with the convention $\frac{p_a-p_b}{\ln p_a-\ln p_b}:=p_a$ for $a=b$.
Therefore $\mathcal{J}_\rho$ is invertible and
\begin{equation}\label{L22b}
(\mathcal{J}_\rho^{-1}(X))_{ab}=\frac{\ln p_a-\ln p_b}{p_a-p_b}\,X_{ab}\quad\text{for }a\neq b,\qquad
(\mathcal{J}_\rho^{-1}(X))_{aa}=\frac{1}{p_a}X_{aa}.
\end{equation}
By Eq.~\eqref{L19b}, the coefficient in Eq.~\eqref{L22b} is exactly $\kappa(p_a,p_b)$ for all $a,b$, hence
\begin{equation}
(\mathcal{J}_\rho^{-1}(X))_{ab}=\kappa(p_a,p_b)\,X_{ab}\quad\text{for all }a,b.
\end{equation}
Therefore
\begin{equation}
\operatorname{Tr}\!\left[\dot{\sigma}_v\,\mathcal{J}_\rho^{-1}(\dot{\sigma}_w)\right]
=\sum_{a,b}\dot{\sigma}_{v,ab}\,(\mathcal{J}_\rho^{-1}(\dot{\sigma}_w))_{ba}
=\sum_{a,b}\kappa(p_b,p_a)\,\dot{\sigma}_{v,ab}\dot{\sigma}_{w,ba}
=\sum_{a,b}\kappa(p_a,p_b)\,\dot{\sigma}_{v,ab}\dot{\sigma}_{w,ba},
\end{equation}
using symmetry $\kappa(p_b,p_a)=\kappa(p_a,p_b)$.
Comparing with Eq.~\eqref{L20b} yields
\begin{equation}
\left.\partial_t\partial_s F(s,t)\right|_{0,0}
=\operatorname{Tr}\!\left[\dot{\sigma}_v\,\mathcal{J}_\rho^{-1}(\dot{\sigma}_w)\right]
=\operatorname{Tr}\!\left[\dot\rho_\lambda(v)\,\mathcal{J}_{\rho(\lambda)}^{-1}\!\big(\dot\rho_\lambda(w)\big)\right],
\end{equation}
which is exactly Eq.~\eqref{L5c}.
\end{proof}}

{\begin{lemma}\label{lem:bkm-positivity}
For each $\lambda\in U$, the bilinear form $g(\lambda)$ defined by Eq.~\eqref{D6a} is symmetric and positive semidefinite. Moreover, for each $v\in\mathbb{R}^m$ one has
\begin{equation}\label{L2b4}
|v|_{g(\lambda)}^{2}
=
\operatorname{Tr}\!\Bigl[\dot{\rho}_\lambda(v)\,\mathcal{J}_{\rho(\lambda)}^{-1}\!\bigl(\dot{\rho}_\lambda(v)\bigr)\Bigr]
=
\int_0^\infty \operatorname{Tr}\!\Bigl[\dot{\rho}_\lambda(v)\,(\rho(\lambda)+t\mathbf{1})^{-1}\,\dot{\rho}_\lambda(v)\,(\rho(\lambda)+t\mathbf{1})^{-1}\Bigr]\,dt
\ge 0.
\end{equation}
\end{lemma}

\begin{proof}
Fix $\lambda$ and abbreviate $\rho:=\rho(\lambda)$. For $A,B\in\mathcal{B}(\mathcal{H})$, cyclicity of the trace gives
\begin{equation}
\operatorname{Tr}\bigl[\rho^{s}A\,\rho^{1-s}B\bigr]
=
\operatorname{Tr}\bigl[\rho^{1-s}B\,\rho^{s}A\bigr].
\end{equation}
Changing variables $s\mapsto 1-s$ in Eq.~\eqref{D6a} therefore gives
\begin{equation}
g_{ij}(\lambda)=g_{ji}(\lambda),
\end{equation}
so $g(\lambda)$ is symmetric.
For positivity, let
\begin{equation}
A_v:=\partial_v\ln\rho
=
\sum_{i=1}^m v^i\,\partial_i\ln\rho.
\end{equation}
Since $\ln\rho$ is self-adjoint, $A_v$ is self-adjoint. Hence
\begin{align}
|v|_{g(\lambda)}^{2}
&=
\int_0^1 \operatorname{Tr}\bigl[\rho^{s}A_v\,\rho^{1-s}A_v\bigr]\,ds \notag\\
&=
\int_0^1 \operatorname{Tr}\Bigl[\bigl(\rho^{(1-s)/2}A_v\rho^{s/2}\bigr)^\dagger
\bigl(\rho^{(1-s)/2}A_v\rho^{s/2}\bigr)\Bigr]\,ds
\ge 0.
\label{eq:bkm_pos_step}
\end{align}
Next, by Lemma~\ref{lem:log-frechet} together with the eigenbasis formulas Eq.~\eqref{L21b}-\eqref{L22b}, the Fr\'echet derivative of the matrix logarithm at $\rho$ coincides with $\mathcal{J}_\rho^{-1}$. Therefore
\begin{equation}
A_v
=
\partial_v\ln\rho
=
\mathcal{J}_\rho^{-1}\!\bigl(\dot{\rho}_\lambda(v)\bigr).
\end{equation}
Substituting this into Eq.~\eqref{D6a} yields
\begin{align}
|v|_{g(\lambda)}^{2}
&=
\int_0^1 \operatorname{Tr}\!\Bigl[\rho^{s}\mathcal{J}_\rho^{-1}\!\bigl(\dot{\rho}_\lambda(v)\bigr)\,
\rho^{1-s}\mathcal{J}_\rho^{-1}\!\bigl(\dot{\rho}_\lambda(v)\bigr)\Bigr]\,ds \notag\\
&=
\operatorname{Tr}\!\Bigl[\mathcal{J}_\rho\!\bigl(\mathcal{J}_\rho^{-1}(\dot{\rho}_\lambda(v))\bigr)\,
\mathcal{J}_\rho^{-1}\!\bigl(\dot{\rho}_\lambda(v)\bigr)\Bigr] \notag\\
&=
\operatorname{Tr}\!\Bigl[\dot{\rho}_\lambda(v)\,\mathcal{J}_\rho^{-1}\!\bigl(\dot{\rho}_\lambda(v)\bigr)\Bigr].
\label{eq:bkm_jinv_step}
\end{align}
Applying Lemma~\ref{lem:log-frechet} with $H=\dot{\rho}_\lambda(v)$ gives
\begin{equation}
\mathcal{J}_\rho^{-1}\!\bigl(\dot{\rho}_\lambda(v)\bigr)
=
D(\ln)_\rho[\dot{\rho}_\lambda(v)]
=
\int_0^\infty (\rho+t\mathbf{1})^{-1}\dot{\rho}_\lambda(v)(\rho+t\mathbf{1})^{-1}\,dt.
\end{equation}
Since the setting is finite-dimensional and Lemma~\ref{lem:log-frechet} gives absolute operator-norm convergence, we may substitute this into Eq.~\eqref{eq:bkm_jinv_step} and interchange trace and integral to obtain
\begin{equation}
|v|_{g(\lambda)}^{2}
=
\int_0^\infty \operatorname{Tr}\!\Bigl[\dot{\rho}_\lambda(v)\,(\rho+t\mathbf{1})^{-1}\,\dot{\rho}_\lambda(v)\,(\rho+t\mathbf{1})^{-1}\Bigr]\,dt.
\end{equation}
Finally, for each $t\ge 0$, set $C_t:=(\rho+t\mathbf{1})^{-1/2}$. Since $\dot{\rho}_\lambda(v)$ is self-adjoint,
\begin{align}
\operatorname{Tr}\!\Bigl[\dot{\rho}_\lambda(v)\,(\rho+t\mathbf{1})^{-1}\,\dot{\rho}_\lambda(v)\,(\rho+t\mathbf{1})^{-1}\Bigr]
&=
\operatorname{Tr}\!\Bigl[(C_t\dot{\rho}_\lambda(v)C_t)^\dagger(C_t\dot{\rho}_\lambda(v)C_t)\Bigr]
\ge 0.
\end{align}
Hence the integral is nonnegative, proving Eq.~\eqref{L2b4}.
\end{proof}}
\begin{theorem}\label{thm:metric-from-information}
Fix $\lambda\in U$ and define $F_\lambda(\delta):=D(\rho(\lambda)|\rho(\lambda+\delta))$ for $\delta$ in a neighborhood of $0\in\mathbb{R}^m$ with $\lambda+\delta\in U$.
Then $F_\lambda$ is $C^2$ at $\delta=0$ and satisfies
\begin{equation} \label{T1a}
F_\lambda(0)=0,\qquad \nabla_\delta F_\lambda(0)=0,\qquad F_\lambda(\delta)=\frac{1}{2}|\delta|_{g(\lambda)}^{2}+R_\lambda(\delta),
\end{equation}
where $g(\lambda)$ is given by Eq.~\eqref{D6a} and where there exists a constant $M_\lambda\ge 0$ such that
\begin{equation} \label{T2a}
|R_\lambda(\delta)|\le \frac{M_\lambda}{6}\,|\delta|_2^{3}\qquad\text{for all }\delta\ \text{with }\lambda+\delta\in U.
\end{equation}
If in addition there exists $g_{\min}(\lambda)>0$ such that $|\delta|_{g(\lambda)}^{2}\ge g_{\min}(\lambda)|\delta|_2^{2}$ for all $\delta\in\mathbb{R}^m$, then for every $\delta$ satisfying
\begin{equation}\label{T3a}
|\delta|_2\le \frac{3g_{\min}(\lambda)}{2M_\lambda},
\end{equation}
the metric-from-information inequalities hold:
\begin{equation} \label{T4a}
\frac{1}{4}|\delta|_{g(\lambda)}^{2}\le D(\rho(\lambda)|\rho(\lambda+\delta))\le \frac{3}{4}|\delta|_{g(\lambda)}^{2},
\end{equation}
and equivalently in bits,
\begin{equation} \label{T5a}
\frac{1}{4\ln 2}|\delta|_{g(\lambda)}^{2}\le D_2(\rho(\lambda)|\rho(\lambda+\delta))\le \frac{3}{4\ln 2}|\delta|_{g(\lambda)}^{2}.
\end{equation}
\end{theorem}

\begin{proof}
Fix $\lambda\in U$ and abbreviate $\rho_0:=\rho(\lambda)$ and $\rho_\delta:=\rho(\lambda+\delta)$.
By Eq.~\eqref{D1c},
\begin{equation} \label{T6a}
F_\lambda(\delta)=\operatorname{Tr}[\rho_0\ln\rho_0]-\operatorname{Tr}[\rho_0\ln\rho_\delta].
\end{equation}
Thus $F_\lambda(0)=0$.
Differentiating Eq.~\eqref{T6a} in direction $v\in\mathbb{R}^m$ and using Lemma~\ref{lem:log-frechet} gives
\begin{equation} \label{T7}
\left.\frac{d}{dt}\right|_{t=0}F_\lambda(tv)
=-\operatorname{Tr}\bigl[\rho_0\,D(\ln)_{\rho_0}[\dot{\rho}_\lambda(v)]\bigr]
=-\int_0^\infty \operatorname{Tr}\bigl[\rho_0(\rho_0+u\mathbf{1})^{-1}\dot{\rho}_\lambda(v)(\rho_0+u\mathbf{1})^{-1}\bigr]\,du.
\end{equation}
Diagonalize $\rho_0=\sum_a p_a|a\rangle\langle a|$ and compute the integrand:
\begin{equation} \label{T8}
\operatorname{Tr}\bigl[\rho_0(\rho_0+u\mathbf{1})^{-1}\dot{\rho}_\lambda(v)(\rho_0+u\mathbf{1})^{-1}\bigr]
=\sum_{a}\frac{p_a}{(p_a+u)^2}\,\dot{\rho}_\lambda(v)_{aa}.
\end{equation}
Since $\operatorname{Tr}\dot{\rho}_\lambda(v)=\partial_v\operatorname{Tr}\rho(\lambda)=\partial_v(1)=0$, one has $\sum_a \dot{\rho}_\lambda(v)_{aa}=0$.
Moreover, $\int_0^\infty \frac{p_a}{(p_a+u)^2}\,du=1$ for each $a$, hence Eqs.~\eqref{T7}-\eqref{T8} give
$\left.\frac{d}{dt}\right|_{0}F_\lambda(tv)=-\sum_a \dot{\rho}_\lambda(v)_{aa}=0$, proving $\nabla_\delta F_\lambda(0)=0$.
For the second derivative, Lemma~\ref{lem:second-deriv} with $v=w=\delta$ gives
\begin{equation} \label{T9}
\left.\frac{d^2}{dt^2}\right|_{t=0}F_\lambda(t\delta)=|\delta|_{g(\lambda)}^{2}.
\end{equation}
By Taylor's theorem with integral remainder for the $C^3$ function $t\mapsto F_\lambda(t\delta)$ on $[0,1]$,
\begin{equation} \label{T10}
F_\lambda(\delta)
=F_\lambda(0)+\left.\frac{d}{dt}\right|_{t=0}F_\lambda(t\delta)
+\frac{1}{2}\left.\frac{d^2}{dt^2}\right|_{t=0}F_\lambda(t\delta)
+\frac{1}{2}\int_0^1 (1-t)^2\,\frac{d^3}{dt^3}F_\lambda(t\delta)\,dt.
\end{equation}
Applying $F_\lambda(0)=0$, $\left.\frac{d}{dt}\right|_{0}F_\lambda(t\delta)=0$, and Eq.~\eqref{T9} to Eq.~\eqref{T10} yields the expansion in Eq.~\eqref{T1a}, where we define the remainder
\begin{equation} \label{T11}
R_\lambda(\delta):=\frac{1}{2}\int_0^1 (1-t)^2\,\frac{d^3}{dt^3}F_\lambda(t\delta)\,dt.
\end{equation}
Define
\begin{equation} \label{T12}
M_\lambda:=\sup\left\{\left|\frac{d^3}{dt^3}F_\lambda(t\delta)\right|:\ t\in[0,1],\ \delta\in\mathbb{R}^m,\ \lambda+t\delta\in U,\ |\delta|_2=1\right\},
\end{equation}
which is finite because $\rho$ is $C^3$ on $\overline{U}$, the spectrum of $\rho(\lambda)$ is uniformly bounded below by Eq.~\eqref{A3b}, and all derivatives of $\ln\rho(\lambda)$ up to third order are continuous and bounded on $\overline{U}$ by compactness of $\overline{U}$ and finite dimensionality \cite{Bhatia1997}.
Then $\left|\frac{d^3}{dt^3}F_\lambda(t\delta)\right|\le M_\lambda|\delta|_2^3$ for all admissible $(t,\delta)$ by homogeneity, and Eq.~\eqref{T11} yields
\begin{equation} \label{T13}
|R_\lambda(\delta)|
\le \frac{1}{2}\int_0^1 (1-t)^2\,dt\;M_\lambda|\delta|_2^3
=\frac{M_\lambda}{6}|\delta|_2^3,
\end{equation}
which is Eq.~\eqref{T2a}.
Assume now $|\delta|_{g(\lambda)}^{2}\ge g_{\min}(\lambda)|\delta|_2^{2}$ and impose Eq.~\eqref{T3a}.
Then Eq.~\eqref{T2a} implies
\begin{equation} \label{T14}
|R_\lambda(\delta)|
\le \frac{M_\lambda}{6}|\delta|_2^3
\le \frac{M_\lambda}{6}|\delta|_2\cdot \frac{1}{g_{\min}(\lambda)}|\delta|_{g(\lambda)}^{2}
\le \frac{1}{4}|\delta|_{g(\lambda)}^{2},
\end{equation}
because $|\delta|_2\le \frac{3g_{\min}(\lambda)}{2M_\lambda}$ implies
$\frac{M_\lambda}{6}|\delta|_2\cdot \frac{1}{g_{\min}(\lambda)}\le \frac{1}{4}$.
Combining Eq.~\eqref{T13} with Eq.~\eqref{T14} yields
\begin{equation} \label{T15}
\frac{1}{2}|\delta|_{g(\lambda)}^{2}-\frac{1}{4}|\delta|_{g(\lambda)}^{2}\le F_\lambda(\delta)\le \frac{1}{2}|\delta|_{g(\lambda)}^{2}+\frac{1}{4}|\delta|_{g(\lambda)}^{2},
\end{equation}
which is Eq.~\eqref{T4a}.
Dividing Eq.~\eqref{T4a} by $\ln 2$ and using Eq.~\eqref{D2c} yields Eq.~\eqref{T5a}.
\end{proof}

Dimensional scaling is fixed by definition: $D(\rho|\sigma)$ is dimensionless because it is a trace of an operator logarithm, hence $|\delta|_{g(\lambda)}^{2}$ is dimensionless in Eq.~\eqref{T15}, which is consistent because $g_{ij}(\lambda)$ is constructed from $\partial_i\ln\rho(\lambda)$ and is dimensionless.
Under an infinitesimal reparameterization $\lambda=\lambda(\theta)$ with $\theta\in\mathbb{R}^m$ and Jacobian $J^i{}_a=\partial\lambda^i/\partial\theta^a$ at $\theta_0$, the chain rule gives $\partial_a\rho(\theta)=\sum_i J^i{}_a\,\partial_i\rho(\lambda)$ and therefore $g_{ab}(\theta)=\sum_{i,j}J^i{}_aJ^j{}_b g_{ij}(\lambda)$ by Eq.~\eqref{D6a}, which is the tensorial transformation law of a Riemannian metric.
If $\rho(\lambda+\delta)=\rho(\lambda)$ then $D(\rho(\lambda)|\rho(\lambda+\delta))=0$ and $\dot{\rho}_\lambda(\delta)=0$, so $|\delta|_{g(\lambda)}^{2}=0$ by Eq.~\eqref{L2b4}, which is consistent with Eq.~\eqref{T15}.
Positivity and symmetry of $g(\lambda)$ are established in Lemma~\ref{lem:bkm-positivity}.

\section{Multiscale diameter bound}\label{app:multiscale}
\renewcommand{\theequation}{E.\arabic{equation}}
Fix an integer $d\ge 1$. Fix a compact $d$-dimensional $C^\infty$ manifold $\Sigma$ equipped with a reference Riemannian metric $h$ used only to define coarse-graining scales, collars, and adjacency. For each boundary state $\rho$ in a fixed state class $\mathcal{S}$, assume there exists a nonempty set of admissible insertion points $X_\rho\subset\Sigma$ and a function $d_\rho:X_\rho\times X_\rho\to[0,\infty)$ such that $(X_\rho,d_\rho)$ is a metric space and $d_\rho$ extends (uniquely) to a metric on the metric completion $\overline{X_\rho}$; assume also that $X_\rho$ intersects every open subset of $\Sigma$ of positive $h$-measure, so that $X_\rho\cap R\neq\varnothing$ for every measurable region $R\subset\Sigma$ with nonempty interior. Define the operational diameter of $\rho$ by
\begin{equation}\label{A1c}
\mathrm{Diam}_\rho:=\sup\{d_\rho(x,y):x,y\in X_\rho\}\in[0,\infty],
\end{equation}
and assume $\mathrm{Diam}_\rho<\infty$ for all $\rho\in\mathcal{S}$. Fix a finite scale index set $\{0,1,\dots,K\}$ and a strictly decreasing sequence of positive lengths $(\ell_k)_{k=0}^{K}$ satisfying
\begin{equation}\label{A2b}
\ell_{k+1}<\ell_k\ \ \text{for all}\ \ k\in\{0,1,\dots,K-1\},
\end{equation}
and fix a positive ultraviolet separation scale $a_{\mathrm{uv}}>0$ (interpreted as the minimal separation at which mutual information for disjoint regions is finite and regulator-stable in the intended microscopic formulation). Assume
\begin{equation} \label{A3c}
\ell_K>a_{\mathrm{uv}}.
\end{equation}
For each $k\in\{0,1,\dots,K\}$ fix a finite measurable partition $\mathcal{C}^{(k)}=\{C_i^{(k)}\}_{i\in V_k}$ of $\Sigma$ with the properties
\begin{equation} \label{A4b}
\begin{aligned}
\Sigma&=\bigsqcup_{i\in V_k}C_i^{(k)},
& \mu_h(\partial C_i^{(k)})&=0
&&\text{for all }i\in V_k,\\
\mathrm{diam}_h(C_i^{(k)})
&:=\sup_{x,y\in C_i^{(k)}}d_h(x,y)
\le c_{\mathrm{cell}}\ell_k.
\end{aligned}
\end{equation}
for a fixed constant $c_{\mathrm{cell}}\ge 1$ independent of $k$, where $\mu_h$ denotes the $d$-dimensional Riemannian volume measure induced by $h$. For each $k$ fix a buffer thickness $\eta_k>0$ satisfying
\begin{equation} \label{A5b}
c_-a_{\mathrm{uv}}\le \eta_k\le c_+\ell_k
\end{equation}
for fixed constants $c_->0$ and $c_+>0$ independent of $k$, and define the buffered cells $A_i^{(k)}\subset C_i^{(k)}$ by
\begin{equation} \label{A6b}
A_i^{(k)}:=\{x\in C_i^{(k)}:\operatorname{dist}_h(x,\partial C_i^{(k)})\ge \eta_k\}.
\end{equation}
Assume the buffers are nontrivial and mutually separated:
\begin{equation} \label{A7b}
\begin{aligned}
A_i^{(k)}&\neq\varnothing
&&\text{for all }i\in V_k,\\
i\neq j&\Longrightarrow
A_i^{(k)}\cap A_j^{(k)}=\varnothing,\\
\operatorname{dist}_h(A_i^{(k)},A_j^{(k)})
&:=\inf_{\substack{x\in A_i^{(k)}\\y\in A_j^{(k)}}}d_h(x,y)
\ge 2\eta_k.
\end{aligned}
\end{equation}
For each $k$, define the adjacency graph $G_k=(V_k,E_k)$ by declaring $(i,j)\in E_k$ if and only if $i\neq j$ and the closures $\overline{C_i^{(k)}}$ and $\overline{C_j^{(k)}}$ intersect in a set of positive $(d-1)$-dimensional Hausdorff measure (equivalently, the cells share a codimension-one interface of positive $h$-area). Assume $G_k$ is connected. Equip $G_k$ with the standard shortest-path metric $d_{G_k}$, and define the graph diameter
\begin{equation} \label{A8b}
D_k:=\mathrm{diam}(G_k):=\max_{i,j\in V_k} d_{G_k}(i,j)\in\mathbb{N}.
\end{equation}
For each $\rho\in\mathcal{S}$ and each pair of disjoint measurable regions $R,S\subset\Sigma$ with $(X_\rho\cap R)\neq\varnothing$ and $(X_\rho\cap S)\neq\varnothing$, define the induced region-to-region separation
\begin{equation} \label{A9b}
d_\rho(R,S):=\inf\{d_\rho(x,y):x\in X_\rho\cap R,\ y\in X_\rho\cap S\}\in[0,\infty).
\end{equation}
Fix, for each $k$, nonnegative real numbers $I_{ij}^{(k)}(\rho)\in(0,\infty)$ assigned to each edge $(i,j)\in E_k$ (interpreted as mutual informations in bits between $A_i^{(k)}$ and $A_j^{(k)}$ in the state $\rho$), and assume these numbers are finite for all $\rho\in\mathcal{S}$ and all edges at all scales, which is consistent with the separation condition Eq.~\eqref{A7b} in split or regulated formulations of mutual information \cite{Araki1976,OhyaPetz2004,Watrous2018,Haag1996,BuchholzDAntoniLongo1987}. Assume that for each $k$ there exists a function $\ell_k^{\mathrm{MI}}:(0,\infty)\to[0,\infty)$ such that $\ell_k^{\mathrm{MI}}$ is nonincreasing and the edgewise information-to-distance bound holds:
\begin{equation}\label{A10b}
(i,j)\in E_k\Longrightarrow d_\rho(A_i^{(k)},A_j^{(k)})\ge \ell_k^{\mathrm{MI}}\!\left(I_{ij}^{(k)}(\rho)\right)\qquad\text{for all}\ \rho\in\mathcal{S}.
\end{equation}
Define the scale-$k$ neighbor-information ceiling by
\begin{equation}\label{A11a}
I_k(\rho):=\max_{(i,j)\in E_k} I_{ij}^{(k)}(\rho)\in(0,\infty),
\end{equation}
which exists because $E_k$ is finite, and define the corresponding uniform edge-separation bound
\begin{equation}\label{A12a}
s_k(\rho):=\ell_k^{\mathrm{MI}}\!\left(I_k(\rho)\right)\in[0,\infty).
\end{equation}
Assume a coarse additivity mechanism quantified by a constant $\delta\ge 0$ independent of $k$ and independent of $\rho$ in the state class, in the following metric form: for each $k$ and each pair $u,v\in V_k$, there exists a shortest path $(v_0,\dots,v_n)$ in $G_k$ from $u$ to $v$ with $n=d_{G_k}(u,v)$ and points $x_m\in X_\rho\cap A_{v_m}^{(k)}$ such that the sequence $(x_0,\dots,x_n)$ satisfies the $\delta$-alignment inequalities
\begin{equation}\label{A13a}
{d_\rho(x_i,x_k)\ \ge\ d_\rho(x_i,x_j)+d_\rho(x_j,x_k)-2\delta,
\qquad \forall\ 0\le i<j<k\le n,}
\end{equation}
The problem is to bound $\mathrm{Diam}_\rho$ from below by an explicit multiscale expression involving $(D_k)_{k=0}^K$, the functions $\ell_k^{\mathrm{MI}}$, the ceiling values $I_k(\rho)$, and the additive loss parameter $\delta$. 

\begin{definition}[Set distance]
For a metric space $(X,d)$ and nonempty subsets $R,S\subset X$, define the set distance
\begin{equation} \label{D1d}
d(R,S):=\inf\{d(r,s):r\in R,\ s\in S\}\in[0,\infty).
\end{equation}
\end{definition}

\begin{definition}[Graph distance and diameter]
For a finite simple connected graph $G=(V,E)$, define the graph distance $d_G:V\times V\to\mathbb{N}_0$ by
\begin{equation} \label{D2d}
d_G(u,v):=\min\!\left\{
n\in\mathbb{N}_0:
\begin{gathered}
\exists\,(w_0,\dots,w_n)\ \text{with }w_0=u,\ w_n=v,\\
(w_{m-1},w_m)\in E\ \text{for all }m\in\{1,\dots,n\}
\end{gathered}
\right\},
\end{equation}
and define the graph diameter $\mathrm{diam}(G):=\max_{u,v\in V}d_G(u,v)$ \cite{Diestel2017}.
\end{definition}

{\begin{definition}[{\bf $\delta$-aligned sequence }]\label{def:delta_aligned_anchored}
Let $\delta\ge 0$ and let $(X,d)$ be a metric space. A finite sequence $(x_0,\dots,x_n)$ in $X$ is called
\emph{$\delta$-aligned} if $n\le 1$, or if for every triple of indices $0\le i<j<k\le n$ one has
\begin{equation}
\label{eq:def:delta_aligned_anchored}
d(x_i,x_k)\ \ge\ d(x_i,x_j)+d(x_j,x_k)-2\delta .
\end{equation}
\end{definition}}

{\begin{lemma}[Uniform edge separation]\label{lem:edgeuniform}
Assume that for each scale $k$ the function $\ell_k^{\mathrm{MI}}:(0,\infty)\to[0,\infty)$ is non-increasing and that the
edgewise information-to-distance bound Eq.~\eqref{A10b} holds. Let $I_k(\rho)$ and $s_k(\rho)$ be defined by
\eqref{A11a} and Eq.~\eqref{A12a}. Then for every $k$ and every edge $(i,j)\in E_k$ one has
\begin{equation}\label{L1xa}
d_\rho\!\left(A_i^{(k)},A_j^{(k)}\right)\ \ge\ s_k(\rho).
\end{equation}
\end{lemma}
\begin{proof}
Fix $k$ and an edge $(i,j)\in E_k$. By the definition of the ceiling value Eq.~\eqref{A11a},
\begin{equation}\label{eq:Ik-dominates-Iij}
I_{ij}^{(k)}(\rho)\ \le\ I_k(\rho).
\end{equation}
Since $\ell_k^{\mathrm{MI}}$ is nonincreasing, the inequality Eq.~\eqref{eq:Ik-dominates-Iij} implies
\begin{equation}\label{eq:monotone-step}
\ell_k^{\mathrm{MI}}\!\left(I_{ij}^{(k)}(\rho)\right)\ \ge\ \ell_k^{\mathrm{MI}}\!\left(I_k(\rho)\right).
\end{equation}
By the definition of $s_k(\rho)$ in Eq.~\eqref{A12a}, the right-hand side of Eq.~\eqref{eq:monotone-step} equals $s_k(\rho)$, hence
\begin{equation}\label{eq:ell-lower-bounds-sk}
\ell_k^{\mathrm{MI}}\!\left(I_{ij}^{(k)}(\rho)\right)\ \ge\ s_k(\rho).
\end{equation}
Finally, applying the edgewise information-to-distance bound Eq.~\eqref{A10b} to the edge $(i,j)\in E_k$ gives
\begin{equation}
d_\rho\!\left(A_i^{(k)},A_j^{(k)}\right)\ \ge\ \ell_k^{\mathrm{MI}}\!\left(I_{ij}^{(k)}(\rho)\right).
\end{equation}
Combining this with Eq.~\eqref{eq:ell-lower-bounds-sk} yields Eq.~\eqref{L1xa}.
\end{proof}}
{\begin{lemma}\label{lem:coarseadditivity}
Let $(X,d)$ be a metric space and $\delta\ge0$. If $(x_0,\ldots,x_n)$ is 
$\delta$-aligned with $n \ge 1$ (Definition~\ref{def:delta_aligned_anchored}), then 
\begin{equation}\label{eq-lemma19}
d(x_0, x_n) \geq \sum_{m=1}^{n} d(x_{m-1}, x_m) - 2(n-1) \delta.
\end{equation}
\end{lemma}

{
\begin{proof}
	We proceed by induction on $n$. The base case $n=1$ is trivial since both sides of the inequality coincide. Assume $n \ge 2$ and that the bound holds for sequences of length $n-1$. By the definition of $\delta$-alignment applied to the indices $0 < n-1 < n$, we have
	\begin{equation}
		d(x_0, x_n) \ge d(x_0, x_{n-1}) + d(x_{n-1}, x_n) - 2\delta.
	\end{equation}
	Applying the inductive hypothesis to the subsequence $(x_0, \dots, x_{n-1})$ yields:
	\begin{equation}
		d(x_0, x_{n-1}) \ge \sum_{m=1}^{n-1} d(x_{m-1}, x_m) - 2(n-2)\delta.
	\end{equation}
	Substituting this lower bound into the alignment inequality gives:
	\begin{align}
		d(x_0, x_n) &\ge \left( \sum_{m=1}^{n-1} d(x_{m-1}, x_m) - 2(n-2)\delta \right) + d(x_{n-1}, x_n) - 2\delta \nonumber \\
		&= \sum_{m=1}^n d(x_{m-1}, x_m) - 2(n-1)\delta.
	\end{align}
	This completes the induction.
\end{proof}
}

}
\begin{lemma}\label{lem:diameterpair}
Let $G=(V,E)$ be a finite connected graph. Then there exist $u,v\in V$ such that $d_G(u,v)=\mathrm{diam}(G)$.
\end{lemma}

\begin{proof}
Since $V$ is finite, the finite set $\{d_G(i,j):i,j\in V\}\subset\mathbb{N}_0$ has a maximum, and by definition of $\mathrm{diam}(G)$ that maximum equals $\mathrm{diam}(G)$. Choose $(u,v)$ attaining the maximum.
\end{proof}

{\begin{lemma}[Global alignment from Morse stability of quasi-geodesics]\label{lem:hyperbolicitysufficient}
Let $(X,d)$ be a geodesic $\delta_0$-hyperbolic metric space in the sense of Gromov, and let
\begin{equation}
\gamma:[0,L]\to X
\end{equation}
be a $(\lambda,c)$-quasi-geodesic, with $\lambda\ge 1$ and $c\ge 0$. Then there exists a constant
\begin{equation}
R=R(\delta_0,\lambda,c)\ge 0
\end{equation}
such that the following holds.
For every triple of parameters
\begin{equation}
0\le s\le t\le u\le L,
\end{equation}
one has
\begin{equation}\label{eq:global_alignment_quasigeodesic}
d\bigl(\gamma(s),\gamma(u)\bigr)
\ge
d\bigl(\gamma(s),\gamma(t)\bigr)
+
d\bigl(\gamma(t),\gamma(u)\bigr)
-
2R.
\end{equation}

In particular, if
\begin{equation}
0\le t_0<t_1<\cdots<t_n\le L
\end{equation}
and
\begin{equation}
x_m:=\gamma(t_m)\qquad (m=0,\dots,n),
\end{equation}
then the sampled sequence $(x_0,\dots,x_n)$ is $R$-aligned in the sense that for every
\begin{equation}
0\le i<j<k\le n
\end{equation}
one has
\begin{equation}\label{eq:sampled_R_aligned}
d(x_i,x_k)\ge d(x_i,x_j)+d(x_j,x_k)-2R.
\end{equation}
\end{lemma}

\begin{proof}
By the Morse stability theorem for quasi-geodesics in $\delta_0$-hyperbolic geodesic spaces, there exists a constant
\begin{equation}
R=R(\delta_0,\lambda,c)\ge 0
\end{equation}
such that for every $(\lambda,c)$-quasi-geodesic segment in $X$, there exists a geodesic segment joining the same endpoints whose image lies within Hausdorff distance at most $R$ from the image of the quasi-geodesic segment
\cite{BridsonHaefliger1999,Gromov1987,GhysdeLaHarpe1990}.
Fix
\begin{equation}
0\le s\le t\le u\le L,
\end{equation}
and define
\begin{equation}
x:=\gamma(s),\qquad y:=\gamma(t),\qquad z:=\gamma(u).
\end{equation}
Since $\gamma$ is a $(\lambda,c)$-quasi-geodesic on $[0,L]$, its restriction
\begin{equation}
\gamma|_{[s,u]}:[s,u]\to X
\end{equation}
is again a $(\lambda,c)$-quasi-geodesic with endpoints $x$ and $z$. Therefore, by Morse stability, there exists a geodesic segment
\begin{equation}
\sigma_{s,u}:[0,\ell]\to X,
\qquad
\ell=d(x,z),
\end{equation}
parameterized by arc length, such that
\begin{equation}
\sigma_{s,u}(0)=x,
\qquad
\sigma_{s,u}(\ell)=z,
\end{equation}
and the image of $\gamma|_{[s,u]}$ lies in the closed $R$-neighborhood of the image of $\sigma_{s,u}$.
Since $y=\gamma(t)$ belongs to the image of $\gamma|_{[s,u]}$, there exists some $r\in[0,\ell]$ such that
\begin{equation}\label{eq:y_close_to_sigma}
d\bigl(y,\sigma_{s,u}(r)\bigr)\le R.
\end{equation}
Because $\sigma_{s,u}$ is parameterized by arc length, one has
\begin{equation}
d\bigl(x,\sigma_{s,u}(r)\bigr)=r,
\qquad
d\bigl(\sigma_{s,u}(r),z\bigr)=\ell-r.
\end{equation}
Using the triangle inequality together with Eq.~\eqref{eq:y_close_to_sigma}, we obtain
\begin{equation}
d(x,y)
\le
d\bigl(x,\sigma_{s,u}(r)\bigr)+d\bigl(\sigma_{s,u}(r),y\bigr)
\le
r+R,
\end{equation}
and similarly
\begin{equation}
d(y,z)
\le
d\bigl(y,\sigma_{s,u}(r)\bigr)+d\bigl(\sigma_{s,u}(r),z\bigr)
\le
R+(\ell-r).
\end{equation}
Adding these two inequalities gives
\begin{equation}
d(x,y)+d(y,z)\le r+R+R+(\ell-r)=\ell+2R=d(x,z)+2R,
\end{equation}
which is equivalent to
\begin{equation}
d(x,z)\ge d(x,y)+d(y,z)-2R.
\end{equation}
This proves Eq.~\eqref{eq:global_alignment_quasigeodesic}.
For the final claim, fix arbitrary indices
\begin{equation}
0\le i<j<k\le n.
\end{equation}
Apply Eq.~\eqref{eq:global_alignment_quasigeodesic} with
\begin{equation}
s=t_i,\qquad t=t_j,\qquad u=t_k.
\end{equation}
Since $x_m=\gamma(t_m)$ for all $m$, this yields
\begin{equation}
d(x_i,x_k)\ge d(x_i,x_j)+d(x_j,x_k)-2R,
\end{equation}
which is exactly Eq.~\eqref{eq:sampled_R_aligned}. Hence $(x_0,\dots,x_n)$ is $R$-aligned.
\end{proof}}

{
\begin{corollary}[Multiscale diameter bound]\label{corollary2}
Under the standing assumptions of Section V, for every state \(\rho\in S\),
\begin{equation}
\operatorname{Diam}_\rho
\ge
\max_{0\le k\le K}
\left\{
D_k\,\ell^{\mathrm{MI}}_k\!\bigl(I_k(\rho)\bigr)
-
2(D_k-1)\delta
\right\}.
\end{equation}
\end{corollary}
\begin{proof}
Fix \(k\in\{0,\dots,K\}\). By Lemma~\ref{lem:diameterpair}, choose \(u,v\in V_k\) with
\(d_{G_k}(u,v)=D_k\). By the standing coarse-additivity hypothesis, there exists
a shortest path
\begin{equation}
v_0\sim v_1\sim\cdots\sim v_{D_k}
\end{equation}
from \(u\) to \(v\) together with points
\(x_m\in X_\rho\cap A^{(k)}_{v_m}\) forming a \(\delta\)-aligned sequence.
By Lemma~\ref{lem:edgeuniform}, each edge satisfies
\begin{equation}
d_\rho\!\bigl(A^{(k)}_{v_{m-1}},A^{(k)}_{v_m}\bigr)\ge s_k(\rho),
\end{equation}
hence \(d_\rho(x_{m-1},x_m)\ge s_k(\rho)\) for every \(m\). Applying Lemma~\ref{lem:coarseadditivity}
gives
\begin{equation}
d_\rho(x_0,x_{D_k})
\ge
D_k\,s_k(\rho)-2(D_k-1)\delta
=
D_k\,\ell^{\mathrm{MI}}_k\!\bigl(I_k(\rho)\bigr)-2(D_k-1)\delta.
\end{equation}
Since \(\operatorname{Diam}_\rho\) is the supremum of \(d_\rho(x,y)\) over
\(x,y\in X_\rho\), this implies
\begin{equation}
\operatorname{Diam}_\rho
\ge
D_k\,\ell^{\mathrm{MI}}_k\!\bigl(I_k(\rho)\bigr)-2(D_k-1)\delta.
\end{equation}
As \(k\) was arbitrary, taking the maximum over \(0\le k\le K\) proves the
claim.
\end{proof}}

{Dimensional consistency is enforced by construction. The metric $d_\rho$ is a length in the chosen operational
normalization, hence $\mathrm{Diam}_\rho$ defined in Eq.~\eqref{A1c} has units of length. The alignment parameter
$\delta$ in the $\delta$-alignment inequalities Eq.~\eqref{A13a} has the same units as $d_\rho$, while the graph
diameter $D_k$ in Eq.~\eqref{A8b} is dimensionless. Since $\ell_k^{\mathrm{MI}}:(0,\infty)\to[0,\infty)$ maps
a dimensionless mutual information input to a length scale via the edgewise implication Eq.~\eqref{A10b}, it follows
that each quantity
$D_k\,\ell_k^{\mathrm{MI}}\!\left(I_k(\rho)\right)$ and $2(D_k-1)\delta$ has units of length, and therefore the
right-hand side of Corollary~\ref{corollary2} is a length, as required.
Under a uniform rescaling of the operational metric $d_\rho\mapsto \lambda\, d_\rho$ with $\lambda>0$, one has
$\mathrm{Diam}_\rho\mapsto \lambda\,\mathrm{Diam}_\rho$ by Eq.~\eqref{A1c}, and the region-to-region separations
defined in Eq.~\eqref{A9b} rescale as $d_\rho(R,S)\mapsto \lambda\, d_\rho(R,S)$. The $\delta$-alignment inequalities
\eqref{A13a} are preserved in form provided $\delta\mapsto \lambda\,\delta$. If, in addition, one replaces
$\ell_k^{\mathrm{MI}}$ by $\lambda\,\ell_k^{\mathrm{MI}}$ (pointwise) so that the edgewise implication Eq.~\eqref{A10b}
remains invariant under the rescaling, then every term inside the maximum in Corollary~\ref{corollary2} rescales by $\lambda$,
and hence Corollary~\ref{corollary2} is covariant under this uniform change of length units.
{In the degenerate single-scale case $K=0$, the bound Corollary~\ref{corollary2} reduces to}
\begin{equation}\label{T1b}
\operatorname{Diam}_\rho
\ge
D_0\,\ell^{\mathrm{MI}}_0\!\bigl(I_0(\rho)\bigr)-2(D_0-1)\delta
\end{equation}}
i.e. the single-scale specialization of Corollary~\ref{corollary2}.
Finally, the maximum in Corollary~\ref{corollary2} is finite because the index set $\{0,1,\dots,K\}$ is finite, each $D_k$ is
finite by Eq.~\eqref{A8b}, each $I_k(\rho)$ is finite (and well-defined when $E_k\neq\varnothing$) by Eq.~\eqref{A11a},
and $\ell_k^{\mathrm{MI}}\!\left(I_k(\rho)\right)$ is finite by the codomain of $\ell_k^{\mathrm{MI}}$
in Eq.~\eqref{A10b}.

\section{RT/QES compatibility and the connectivity threshold}\label{app:rt-qes}
\renewcommand{\theequation}{F.\arabic{equation}}
Fix an integer $d\ge 2$ and set $n:=d+1$. Let $(\mathcal{M},g)$ be a connected, oriented, time-oriented, globally hyperbolic Lorentzian $(n)$-manifold with conformal boundary $\partial\mathcal{M}\cong \mathbb{R}\times \Sigma$ where $\Sigma$ is a connected, compact, oriented $(d-1)$-manifold. Assume $(\mathcal{M},g)$ is asymptotically locally AdS in the sense that there exists a smooth defining function $z$ on a neighborhood of the boundary with $z=0$ and $dz\neq 0$ on $\partial\mathcal{M}$ and a smooth Lorentzian metric $\bar g$ on $\bar{\mathcal{M}}:=\mathcal{M}\cup\partial\mathcal{M}$ such that
\begin{equation} \label{A1d}
g=\frac{\ell_{\mathrm{AdS}}^{2}}{z^{2}}\bar g,\qquad \bar g|_{z=0}= -dt^{2}+h,
\end{equation}
for a fixed Riemannian metric $h$ on $\Sigma$ and a fixed boundary time function $t$ on $\partial\mathcal{M}$. Assume there exists a complete, hypersurface-orthogonal timelike Killing field $\xi$ on $\mathcal{M}$ whose restriction to $\partial\mathcal{M}$ equals $\partial_t$, and assume the spacetime is time-reflection symmetric with respect to the hypersurface $t=0$ in the sense that there exists an isometry $\Theta:\mathcal{M}\to\mathcal{M}$ with $\Theta^2=\mathrm{id}$, $\Theta^*g=g$, and $\Theta_* \xi=-\xi$, and whose fixed-point set $\mathcal{N}:=\mathrm{Fix}(\Theta)$ is a smooth, connected, complete, spacelike Cauchy hypersurface with induced Riemannian metric $\gamma:=g|_{\mathcal{N}}$ and boundary $\partial\mathcal{N}=\Sigma$ identified with $\{t=0\}\times\Sigma\subset\partial\mathcal{M}$. Fix a semiclassical regime in which $G_N>0$ and $\ell_{\mathrm{AdS}}>0$ satisfy
\begin{equation} \label{A2c}
0<G_N\ll \ell_{\mathrm{AdS}}^{d-1},
\end{equation}
and assume bulk effective field theory is valid on $(\mathcal{M},g)$ at length scales $\gg \ell_{\mathrm{uv}}$ for some fixed $\ell_{\mathrm{uv}}>0$, with quantum fields in a state $\omega_\rho$ determined by the boundary state $\rho$ on a code subspace, and assume that the backreaction of bulk quantum fields on $g$ is incorporated only through the semiclassical expansion described explicitly in the comparison inequalities below \cite{BirrellDavies1982,Polchinski1998,Harlow2016}. Fix a boundary region $R\subset\Sigma$ such that $R$ is open in $\Sigma$ and $\partial R$ is a compact, embedded, $C^2$ codimension-one submanifold of $\Sigma$. Define the class of admissible hypersurfaces anchored to $R$ on the time-symmetric slice $\mathcal{N}$ by
\begin{multline} \label{A3d}
\mathcal{X}(R):=\Bigl\{\chi\subset \mathcal{N}:\chi\ \text{is a compact, embedded, oriented $C^2$ hypersurface in $\mathcal{N}$ with}\ \partial\chi\\=\partial R\ \text{as subsets of}\ \Sigma,\ \exists\ \Omega\subset\mathcal{N}\ \text{compact with piecewise $C^2$ boundary and}\ \partial\Omega=\overline{R}\cup \chi\Bigr\},
\end{multline}
where $\partial\Omega$ is taken with the induced orientation so that $\partial\Omega$ equals $\overline{R}$ on $\Sigma$ and equals $\chi$ in the interior of $\mathcal{N}$, and where $\overline{R}$ denotes the closure of $R$ in $\Sigma$. The requirement $\partial\Omega=\overline{R}\cup\chi$ constitutes the explicit bulk homology constraint for $\chi$ relative to $R$. Assume $\mathcal{X}(R)\neq\varnothing$.
Fix disjoint boundary regions $A,B\subset\Sigma$ of the above type satisfying
\begin{equation}\label{A4c}
\overline{A}\cap \overline{B}=\varnothing,\qquad
\operatorname{dist}_{h}(A,B):=\inf\{d_h(x,y):x\in A,\ y\in B\}\ge a_{\mathrm{uv}}>0,
\end{equation}
for a fixed ultraviolet separation scale $a_{\mathrm{uv}}>0$. Define $R_{AB}:=A\cup B$ and assume $\mathcal{X}(A)$, $\mathcal{X}(B)$, and $\mathcal{X}(R_{AB})$ are nonempty. Assume that for every $R\in\{A,B,R_{AB}\}$ the minimization problems defined in Eq.~\eqref{D9b} admit minimizers.
For $\varepsilon\in(0,\varepsilon_0]$ define the cutoff slice
\begin{equation} \label{D1e}
\mathcal{N}_\varepsilon:=\{p\in\mathcal{N}:z(p)\ge \varepsilon\},
\end{equation}
and for $\chi\in\mathcal{X}(R)$ define the regulated hypersurface $\chi_\varepsilon:=\chi\cap \mathcal{N}_\varepsilon$ and regulated area
\begin{equation} \label{D2e}
\mathrm{Area}_\varepsilon(\chi):=\mathcal{H}^{d-1}_\gamma(\chi_\varepsilon)=\int_{\chi_\varepsilon} dA_\gamma,
\end{equation}
where $\mathcal{H}^{d-1}_\gamma$ is the $(d-1)$-dimensional Hausdorff measure induced by $\gamma$ and $dA_\gamma$ is the induced area element on $\chi$ \cite{Federer1969,Simon1983,Morgan2016}. For each $\chi\in\mathcal{X}(R)$ choose a homology region $\Omega(\chi;R)\subset\mathcal{N}$ with $\partial\Omega(\chi;R)=\overline{R}\cup\chi$ and define $\Omega_\varepsilon(\chi;R):=\Omega(\chi;R)\cap \mathcal{N}_\varepsilon$. Fix a bulk regulator at length scale comparable to $\varepsilon$ such that the bulk effective QFT on $\mathcal{N}_\varepsilon$ admits a Hilbert-space factorization
\begin{equation} \label{D3e}
\mathcal{H}^{\mathrm{bulk}}_\varepsilon\simeq \mathcal{H}^{\mathrm{bulk}}_{\Omega_\varepsilon(\chi;R)}\otimes \mathcal{H}^{\mathrm{bulk}}_{\mathcal{N}_\varepsilon\setminus \Omega_\varepsilon(\chi;R)}
\end{equation}
for each admissible $\Omega_\varepsilon(\chi;R)$, and let $\rho^{\mathrm{bulk}}_\varepsilon$ be the regulated bulk density operator representing the bulk state $\omega_\rho$ on $\mathcal{H}^{\mathrm{bulk}}_\varepsilon$. Define the regulated bulk von Neumann entropy in natural-logarithm units by
\begin{equation} \label{D4c}
S^{\mathrm{bulk}}_\varepsilon(\chi;R):=-\operatorname{Tr}\bigl[\rho^{\mathrm{bulk}}_{\Omega_\varepsilon(\chi;R)}\ln \rho^{\mathrm{bulk}}_{\Omega_\varepsilon(\chi;R)}\bigr],
\end{equation}
where $\rho^{\mathrm{bulk}}_{\Omega_\varepsilon(\chi;R)}:=\operatorname{Tr}_{\mathcal{N}_\varepsilon\setminus \Omega_\varepsilon(\chi;R)}\rho^{\mathrm{bulk}}_\varepsilon$ is the reduced density operator with respect to Eq.~\eqref{D3e}. Fix a counterterm functional $C_\varepsilon(\chi)$ of the form
\begin{equation} \label{D5c}
\begin{aligned}
C_\varepsilon(\chi)
&:=\sum_{j=0}^{d-2} c_j\,\varepsilon^{j-(d-1)}\int_{\chi\cap\{z=\varepsilon\}} \mathcal{I}_j(\gamma_\chi,K_\chi)\,dA_{\gamma|_{z=\varepsilon}}\\
&\quad +c_{\log}\,\ln\!\left(\frac{\varepsilon}{\ell_{\mathrm{ct}}}\right)\int_{\chi\cap\{z=\varepsilon\}}\mathcal{I}_{\log}(\gamma_\chi,K_\chi)\,dA_{\gamma|_{z=\varepsilon}},
\end{aligned}
\end{equation}
where $\ell_{\mathrm{ct}}>0$ is a fixed renormalization length, where $\gamma_\chi$ is the induced metric on $\chi$ and $K_\chi$ is its second fundamental form in $(\mathcal{N},\gamma)$, where $\mathcal{I}_j$ and $\mathcal{I}_{\log}$ are fixed scalar polynomials in the intrinsic curvature of $(\chi,\gamma_\chi)$ and contractions of $K_\chi$ with total scaling dimension $d-1-j$ and $0$ respectively, and where the coefficients $c_j,c_{\log}\in\mathbb{R}$ are fixed by the choice of bulk QFT and renormalization scheme \cite{SusskindUglum1994,Solodukhin2011,FaulknerLewkowyczMaldacena2013}. Define the regulated generalized entropy functional by
\begin{equation} \label{D6b}
S^{\mathrm{gen}}_\varepsilon(\chi;R):=\frac{1}{4G_N}\,\mathrm{Area}_\varepsilon(\chi)+S^{\mathrm{bulk}}_\varepsilon(\chi;R)+C_\varepsilon(\chi),
\end{equation}
and assume that for each fixed $R$ and each $\chi\in\mathcal{X}(R)$ the limit
\begin{equation} \label{D7b}
S^{\mathrm{gen}}(\chi;R):=\lim_{\varepsilon\downarrow 0} S^{\mathrm{gen}}_\varepsilon(\chi;R)\in\mathbb{R}
\end{equation}
exists and is finite, with the limit taken along the same regulator family for all $\chi$ and all $\rho$ in the state class, and that for any two $\chi_1,\chi_2\in\mathcal{X}(R)$ the difference $\lim_{\varepsilon\downarrow 0}\bigl(S^{\mathrm{gen}}_\varepsilon(\chi_1;R)-S^{\mathrm{gen}}_\varepsilon(\chi_2;R)\bigr)$ exists and equals $S^{\mathrm{gen}}(\chi_1;R)-S^{\mathrm{gen}}(\chi_2;R)$. {Whenever the following limit exists, we denote the finite renormalized bulk-entropic contribution by
\begin{equation}\label{D7b-ren}
S^{\mathrm{bulk,ren}}(\chi;R):=
\lim_{\varepsilon\downarrow0}\Big(S^{\mathrm{bulk}}_\varepsilon(\chi;R)+C_\varepsilon(\chi)\Big).
\end{equation}
In general the counterterm functional \(C_\varepsilon(\chi)\) need not admit a finite \(\varepsilon\downarrow0\) limit by itself; only the combination in Eq.~\eqref{D7b-ren} is assumed finite.} Define the Ryu-Takayanagi area functional and the generalized entropy functional in bits by
\begin{equation} \label{D8b}
\mathcal{A}_\varepsilon(\chi):=\mathrm{Area}_\varepsilon(\chi),\qquad
S^{\mathrm{gen}}_{2}(\chi;R):=\frac{1}{\ln 2}\,S^{\mathrm{gen}}(\chi;R),
\end{equation}
and define the RT and QES minimization problems by
\begin{equation} \label{D9b}
\begin{aligned}
\chi_{\mathrm{RT}}(R)
&\in\operatorname*{arg\,min}_{\chi\in\mathcal{X}(R)}
\mathcal{A}_\varepsilon(\chi)
&&\text{for each fixed }\varepsilon>0,\\
\chi_{\mathrm{QES}}(R)
&\in\operatorname*{arg\,min}_{\chi\in\mathcal{X}(R)}
S^{\mathrm{gen}}(\chi;R).
\end{aligned}
\end{equation}
The RT relation is understood by passage to the renormalized limit in
differences,
assuming existence of minimizers. Define the entanglement wedge on the static slice by
\begin{equation} \label{D10a}
\Omega_{\mathrm{QES}}(R):=\Omega(\chi_{\mathrm{QES}}(R);R),
\end{equation}
and define the Lorentzian entanglement wedge by the domain of dependence in $\mathcal{M}$,
\begin{equation} \label{D11}
\mathrm{EW}(R):=D_{\mathcal{M}}\bigl(\Omega_{\mathrm{QES}}(R)\bigr),
\end{equation}
where $D_{\mathcal{M}}(\cdot)$ denotes the domain of dependence computed with the Lorentzian metric $g$ \cite{HubenyRangamaniTakayanagi2007,Wall2014,EngelhardtWall2015,AlmheiriDongHarlow2015}.

The compatibility condition is derived by computing the first variation of $S^{\mathrm{gen}}(\chi;R)$ and taking the semiclassical limit $G_N\to 0$ with controlled bulk entropy variation. Fix $R\subset\Sigma$ and $\chi\in\mathcal{X}(R)$. Let $n$ be a globally defined unit normal vector field along $\chi$ in $(\mathcal{N},\gamma)$ compatible with the chosen orientation. For $f\in C_c^\infty(\chi)$ satisfying $f|_{\partial\chi}=0$, define a normal variation $\chi_s$ of $\chi$ by the flow $\Phi_s$ generated by the vector field $V:=f n$ in a tubular neighborhood of $\chi$, so that $\chi_s:=\Phi_s(\chi)$ and $\chi_0=\chi$ and $\partial\chi_s=\partial\chi$ for all sufficiently small $|s|$. Define the mean curvature scalar $\mathcal{H}_\chi$ by
\begin{equation} \label{D12}
\mathcal{H}_\chi:=\operatorname{tr}_{\gamma_\chi} K_\chi,
\end{equation}
where $K_\chi$ is the second fundamental form and $\gamma_\chi$ is the induced metric on $\chi$.

{\begin{lemma}[First variation of area on the time-symmetric slice]\label{lem:firstvar-area}
Let \((\mathcal N,\gamma)\) be the Riemannian time-symmetric slice, let
\(\chi\subset \mathcal N\) be a compact embedded \(C^2\) hypersurface with fixed boundary
\(\partial\chi\), and let \(n\) be the outward-pointing unit normal along \(\chi\) in
\((\mathcal N,\gamma)\). Let \(f\in C_c^\infty(\chi^\circ)\), and let \(\chi_s\) be a smooth
one-parameter family of hypersurfaces with \(\chi_0=\chi\), generated by the normal variation field
\begin{equation}
V=f\,n .
\end{equation}
Then the first variation of area is
\begin{equation}\label{L:firstvar-area}
\left.\frac{d}{ds}\right|_{s=0}\mathrm{Area}(\chi_s)
=
\int_\chi f\,\mathcal{H}_\chi\,dA_\gamma,
\end{equation}
where the second fundamental form and mean curvature scalar are defined by
\begin{equation}\label{L:def-K-H}
(K_\chi)_{\alpha\beta}
:=
\gamma\!\bigl(\nabla^{(\gamma)}_{e_\alpha}n,e_\beta\bigr),
\qquad
\mathcal{H}_\chi:=\gamma_\chi^{\alpha\beta}(K_\chi)_{\alpha\beta},
\end{equation}
with \(e_\alpha:=\partial_\alpha X\) for a local embedding \(X\) of \(\chi\), \(\nabla^{(\gamma)}\)
the Levi-Civita connection of \(\gamma\), and \(\gamma_\chi\) the induced metric on \(\chi\).
Equivalently, by metric compatibility,
\begin{equation}\label{L:def-K-equiv}
(K_\chi)_{\alpha\beta}
=
-\gamma\!\bigl(n,\nabla^{(\gamma)}_{e_\alpha}e_\beta\bigr).
\end{equation}
With this convention, \(H_\chi>0\) for the outward unit normal of a round Euclidean sphere, and
hence an outward deformation with \(f>0\) increases the area to first order.
\end{lemma}

\begin{proof}
Choose local coordinates \(y^\alpha\) (\(\alpha=1,\dots,d-1\)) on \(\chi\), let
\(X:\chi\to\mathcal N\) be a local embedding of \(\chi=\chi_0\), and let
\(X_s\) be a smooth local embedding of \(\chi_s\) such that
\begin{equation}
\left.\frac{\partial X_s}{\partial s}\right|_{s=0}=V=f\,n.
\end{equation}
Write
\begin{equation}
e_\alpha(s):=\partial_\alpha X_s,
\qquad
e_\alpha:=e_\alpha(0)=\partial_\alpha X,
\qquad
(\gamma_{\chi_s})_{\alpha\beta}=\gamma\!\bigl(e_\alpha(s),e_\beta(s)\bigr).
\end{equation}

Since \(\nabla^{(\gamma)}\) is torsion-free and the coordinate vector fields
\(\partial_s\) and \(\partial_\alpha\) commute on the parameter domain, we have
\begin{equation}
\nabla^{(\gamma)}_{\partial_s}e_\alpha(s)
=
\nabla^{(\gamma)}_{e_\alpha(s)}\partial_s X_s .
\end{equation}
Therefore, using metric compatibility of \(\nabla^{(\gamma)}\),
\begin{align}
\left.\frac{\partial}{\partial s}\right|_{s=0}(\gamma_{\chi_s})_{\alpha\beta}
&=
\left.\frac{\partial}{\partial s}\right|_{s=0}
\gamma\!\bigl(e_\alpha(s),e_\beta(s)\bigr)\notag\\
&=
\gamma\!\bigl(\nabla^{(\gamma)}_{\partial_s}e_\alpha(s),e_\beta(s)\bigr)\Big|_{s=0}
+
\gamma\!\bigl(e_\alpha(s),\nabla^{(\gamma)}_{\partial_s}e_\beta(s)\bigr)\Big|_{s=0}\notag\\
&=
\gamma\!\bigl(\nabla^{(\gamma)}_{e_\alpha}(f n),e_\beta\bigr)
+
\gamma\!\bigl(e_\alpha,\nabla^{(\gamma)}_{e_\beta}(f n)\bigr).
\end{align}
Expanding the derivatives gives
\begin{align}
\gamma\!\bigl(\nabla^{(\gamma)}_{e_\alpha}(f n),e_\beta\bigr)
&=
e_\alpha(f)\,\gamma(n,e_\beta)
+
f\,\gamma\!\bigl(\nabla^{(\gamma)}_{e_\alpha}n,e_\beta\bigr),\\
\gamma\!\bigl(e_\alpha,\nabla^{(\gamma)}_{e_\beta}(f n)\bigr)
&=
e_\beta(f)\,\gamma(e_\alpha,n)
+
f\,\gamma\!\bigl(e_\alpha,\nabla^{(\gamma)}_{e_\beta}n\bigr).
\end{align}
Because \(n\) is orthogonal to \(T\chi\), one has
\begin{equation}
\gamma(n,e_\alpha)=\gamma(n,e_\beta)=0,
\end{equation}
so the derivative terms of \(f\) vanish and
\begin{equation}
\left.\frac{\partial}{\partial s}\right|_{s=0}(\gamma_{\chi_s})_{\alpha\beta}
=
f\,(K_\chi)_{\alpha\beta}+f\,(K_\chi)_{\beta\alpha}.
\end{equation}

It remains to note that \(K_\chi\) is symmetric. Using Eq.~\eqref{L:def-K-equiv},
the torsion-free property of \(\nabla^{(\gamma)}\), and the fact that
\([e_\alpha,e_\beta]\) is tangent to \(\chi\), we compute
\begin{align}
(K_\chi)_{\alpha\beta}-(K_\chi)_{\beta\alpha}
&=
-\gamma\!\bigl(n,\nabla^{(\gamma)}_{e_\alpha}e_\beta-\nabla^{(\gamma)}_{e_\beta}e_\alpha\bigr)\notag\\
&=
-\gamma\!\bigl(n,[e_\alpha,e_\beta]\bigr)=0.
\end{align}
Hence
\begin{equation}\label{L:metric-var}
\left.\frac{\partial}{\partial s}\right|_{s=0}(\gamma_{\chi_s})_{\alpha\beta}
=
2f\,(K_\chi)_{\alpha\beta}.
\end{equation}

Applying the determinant variation formula to the induced metric yields
\begin{align}
\left.\frac{\partial}{\partial s}\right|_{s=0}\sqrt{\det\gamma_{\chi_s}}
&=
\frac12 \sqrt{\det\gamma_\chi}\,
\gamma_\chi^{\alpha\beta}
\left.\frac{\partial}{\partial s}\right|_{s=0}(\gamma_{\chi_s})_{\alpha\beta}\notag\\
&=
\frac12 \sqrt{\det\gamma_\chi}\,
\gamma_\chi^{\alpha\beta}\,
2f\,(K_\chi)_{\alpha\beta}\notag\\
&=
f\,H_\chi\,\sqrt{\det\gamma_\chi}.
\end{align}
Integrating over \(\chi\) gives
\begin{equation}
\left.\frac{d}{ds}\right|_{s=0}\mathrm{Area}(\chi_s)
=
\int_\chi f\,H_\chi\,dA_\gamma,
\end{equation}
which is Eq.~\eqref{L:firstvar-area}.

Finally, for a round sphere of radius \(R\) in Euclidean space with outward unit normal \(n\), one has
\begin{equation}
\nabla^{(\gamma)}_X n=\frac{1}{R}X
\qquad
\text{for every tangent vector }X,
\end{equation}
and therefore
\begin{equation}
K_\chi(X,Y)=\frac{1}{R}\,\gamma_\chi(X,Y),
\qquad
H_\chi=\frac{d-1}{R}>0.
\end{equation}
Thus an outward deformation with \(f>0\) increases the area to first order, as claimed.
\end{proof}
}
{\begin{lemma}[First variation of generalized entropy]\label{lem:firstvar-gen}
Fix \(\chi\in \mathcal X(R)\) and an admissible normal variation \(\chi_s\) generated by
\(V=f\,n\), where \(f\in C_c^\infty(\chi^\circ)\). Assume the following.

\smallskip
\noindent
\textbf{(i)} For every \(\varepsilon\in(0,\varepsilon_0]\), the map
\begin{equation}
s\longmapsto S^{\mathrm{bulk}}_\varepsilon(\chi_s;R)
\end{equation}
is differentiable at \(s=0\).

\smallskip
\noindent
\textbf{(ii)} The first variation of the generalized entropy commutes with the renormalized limit:
\begin{equation}\label{L4e}
\left.\frac{d}{ds}\right|_{s=0}S^{\mathrm{gen}}(\chi_s;R)
=
\lim_{\varepsilon\downarrow0}
\left.\frac{d}{ds}\right|_{s=0}S^{\mathrm{gen}}_\varepsilon(\chi_s;R).
\end{equation}

\smallskip
\noindent
\textbf{(iii)} The renormalized bulk-entropic variation
\begin{equation}\label{L4e-bulkvar}
\mathcal V^{\mathrm{bulk,ren}}_{\chi,R}(f)
:=
\lim_{\varepsilon\downarrow0}
\left.\frac{d}{ds}\right|_{s=0}
\Bigl(
S^{\mathrm{bulk}}_\varepsilon(\chi_s;R)+C_\varepsilon(\chi_s)
\Bigr)
\end{equation}
exists.

Then, for all sufficiently small \(\varepsilon>0\), one has
\begin{equation}\label{L4e-countertermzero}
\left.\frac{d}{ds}\right|_{s=0}C_\varepsilon(\chi_s)=0,
\end{equation}
and consequently
\begin{equation}\label{L5e}
\left.\frac{d}{ds}\right|_{s=0}S^{\mathrm{gen}}(\chi_s;R)
=
\frac{1}{4G_N}\int_\chi f\,H_\chi\,dA_\gamma
+
\mathcal V^{\mathrm{bulk,ren}}_{\chi,R}(f).
\end{equation}

In particular, if for \(|s|\) sufficiently small the finite renormalized bulk-entropic contribution
\begin{equation}\label{D7b-renw}
S^{\mathrm{bulk,ren}}(\chi_s;R)
:=
\lim_{\varepsilon\downarrow0}
\Bigl(
S^{\mathrm{bulk}}_\varepsilon(\chi_s;R)+C_\varepsilon(\chi_s)
\Bigr)
\end{equation}
exists and the map
\begin{equation}
s\longmapsto S^{\mathrm{bulk,ren}}(\chi_s;R)
\end{equation}
is differentiable at \(s=0\) with
\begin{equation}\label{L4e-bulkderiv}
\left.\frac{d}{ds}\right|_{s=0}S^{\mathrm{bulk,ren}}(\chi_s;R)
=
\mathcal V^{\mathrm{bulk,ren}}_{\chi,R}(f),
\end{equation}
then Eq.~\eqref{L5e} may be written equivalently as
\begin{equation}\label{L5e-prime}
\left.\frac{d}{ds}\right|_{s=0}S^{\mathrm{gen}}(\chi_s;R)
=
\frac{1}{4G_N}\int_\chi f\,H_\chi\,dA_\gamma
+
\left.\frac{d}{ds}\right|_{s=0}S^{\mathrm{bulk,ren}}(\chi_s;R).
\end{equation}
\end{lemma}

\begin{proof}
Differentiate Eq.~\eqref{D6b} at fixed \(\varepsilon\). Using Lemma~\ref{lem:firstvar-area}, we obtain
\begin{equation}\label{L6e}
\left.\frac{d}{ds}\right|_{s=0}S^{\mathrm{gen}}_\varepsilon(\chi_s;R)
=
\frac{1}{4G_N}
\left.\frac{d}{ds}\right|_{s=0}\mathrm{Area}_\varepsilon(\chi_s)
+
\left.\frac{d}{ds}\right|_{s=0}S^{\mathrm{bulk}}_\varepsilon(\chi_s;R)
+
\left.\frac{d}{ds}\right|_{s=0}C_\varepsilon(\chi_s).
\end{equation}

Because \(f\in C_c^\infty(\chi^\circ)\), there exists an open neighborhood
\(U\subset \chi\) of \(\partial\chi\) such that
\begin{equation}
\operatorname{supp}(f)\cap U=\varnothing.
\end{equation}
Hence the vector field \(V=f\,n\) vanishes identically on \(U\), so its flow is the identity on \(U\) for all sufficiently small \(|s|\).

Now \(z|_\chi\) is continuous and satisfies \(z=0\) on \(\partial\chi\). Therefore there exists \(\varepsilon_f>0\) such that
\begin{equation}
\chi\cap\{0\le z<2\varepsilon_f\}\subset U.
\end{equation}
For every \(0<\varepsilon<\varepsilon_f\) and all sufficiently small \(|s|\), the flow leaves the collar
\(\chi\cap\{0\le z<2\varepsilon_f\}\) pointwise fixed. In particular,
\begin{equation}
\chi_s\cap\{z=\varepsilon\}=\chi\cap\{z=\varepsilon\},
\end{equation}
and the induced geometric data on this cutoff section agree. Since the counterterm functional \(C_\varepsilon(\chi)\) in Eq.~\eqref{D5c} depends only on this cutoff data, it follows that
\begin{equation}
C_\varepsilon(\chi_s)=C_\varepsilon(\chi)
\qquad
(0<\varepsilon<\varepsilon_f,\ |s|\ \text{sufficiently small}),
\end{equation}
and therefore
\begin{equation}
\left.\frac{d}{ds}\right|_{s=0}C_\varepsilon(\chi_s)=0.
\end{equation}
This proves Eq.~\eqref{L4e-countertermzero}.

Next, because \(\operatorname{supp}(f)\subset \chi^\circ\), for all sufficiently small \(\varepsilon>0\) the support of \(f\) lies inside \(\chi_\varepsilon=\chi\cap \mathcal N_\varepsilon\). Hence Lemma~\ref{lem:firstvar-area} applied to the regulated hypersurface gives
\begin{equation}
\left.\frac{d}{ds}\right|_{s=0}\mathrm{Area}_\varepsilon(\chi_s)
=
\int_{\chi_\varepsilon} f\,H_\chi\,dA_\gamma
=
\int_\chi f\,H_\chi\,dA_\gamma.
\end{equation}
Substituting this and Eq.~\eqref{L4e-countertermzero} into Eq.~\eqref{L6e} yields, for all sufficiently small \(\varepsilon\),
\begin{equation}
\left.\frac{d}{ds}\right|_{s=0}S^{\mathrm{gen}}_\varepsilon(\chi_s;R)
=
\frac{1}{4G_N}\int_\chi f\,H_\chi\,dA_\gamma
+
\left.\frac{d}{ds}\right|_{s=0}S^{\mathrm{bulk}}_\varepsilon(\chi_s;R).
\end{equation}
Taking the limit \(\varepsilon\downarrow0\) and using Eq.~\eqref{L4e} together with the definition Eq.~\eqref{L4e-bulkvar}, we obtain
\begin{equation}
\left.\frac{d}{ds}\right|_{s=0}S^{\mathrm{gen}}(\chi_s;R)
=
\frac{1}{4G_N}\int_\chi f\,H_\chi\,dA_\gamma
+
\mathcal V^{\mathrm{bulk,ren}}_{\chi,R}(f),
\end{equation}
which is Eq.~\eqref{L5e}.

Finally, if \(S^{\mathrm{bulk,ren}}(\chi_s;R)\) is differentiable at \(s=0\) and satisfies Eq.~\eqref{L4e-bulkderiv}, then Eq.~\eqref{L5e} immediately becomes Eq.~\eqref{L5e-prime}.
\end{proof}}
Define quantum extremality on the time-symmetric slice by requiring the first variation in Lemma~\ref{lem:firstvar-gen} to vanish for all admissible $f$. Define
\begin{multline} \label{D13}
\chi\ \text{is quantum extremal for $R$}\Longleftrightarrow
\left.\frac{d}{ds}\right|_{s=0}S^{\mathrm{gen}}(\chi_s;R)=0\\ \text{for all admissible normal variations }\chi_s\ \text{generated by }V=fn\ \text{with }f\in C_c^\infty(\chi),\ f|_{\partial\chi}=0.
\end{multline} 

{\begin{theorem}[RT/QES compatibility and connectivity threshold]\label{thm:rt-qes-connectivity}
Assume the standing setup Eq.~\eqref{A1d}-\eqref{A4c}, the renormalization hypotheses
Eq.~\eqref{D6b}-\eqref{D7b-renw}, the quantum extremality condition Eq.~\eqref{D13}, and the
hypotheses of Lemma~\ref{lem:firstvar-gen} for the admissible variations considered below.

Assume further that for each \(R\in\{A,B,R_{AB}\}\) and each \(G_N\) in the semiclassical family
under consideration, there exists a \(C^2\) QES minimizer
\(\chi_{\mathrm{QES}}(R;G_N)\in\mathcal{X}(R)\). Then for every admissible normal variation
\(\chi_s\) of \(\chi_{\mathrm{QES}}(R;G_N)\) generated by \(V=fn\), with
\(f\in C_c^\infty(\chi_{\mathrm{QES}}(R;G_N)^\circ)\), one has
\begin{equation}\label{T1c}
\frac{1}{4G_N}
\left|
\int_{\chi_{\mathrm{QES}}(R;G_N)} f\,H_{\chi_{\mathrm{QES}}(R;G_N)}\,dA_\gamma
\right|
=
\left|
\lim_{\varepsilon\downarrow0}
\left.\frac{d}{ds}\right|_{s=0}
S^{\mathrm{bulk}}_\varepsilon(\chi_s;R)
\right|.
\end{equation}
If, in addition, the map \(s\mapsto S^{\mathrm{bulk,ren}}(\chi_s;R)\) is differentiable at
\(s=0\) and
\begin{equation}\label{T1c-prime-hyp}
\left.\frac{d}{ds}\right|_{s=0}S^{\mathrm{bulk,ren}}(\chi_s;R)
=
\lim_{\varepsilon\downarrow0}
\left.\frac{d}{ds}\right|_{s=0}
S^{\mathrm{bulk}}_\varepsilon(\chi_s;R),
\end{equation}
then Eq.~\eqref{T1c} may equivalently be written as
\begin{equation}\label{T1c-prime}
\frac{1}{4G_N}
\left|
\int_{\chi_{\mathrm{QES}}(R;G_N)} f\,H_{\chi_{\mathrm{QES}}(R;G_N)}\,dA_\gamma
\right|
=
\left|
\left.\frac{d}{ds}\right|_{s=0}
S^{\mathrm{bulk,ren}}(\chi_s;R)
\right|.
\end{equation}
If there exists \(M_R\ge0\), independent of \(G_N\) along the semiclassical family under
consideration, such that for all such variations
\begin{equation}\label{T2c}
\left|
\lim_{\varepsilon\downarrow0}
\left.\frac{d}{ds}\right|_{s=0}
S^{\mathrm{bulk}}_\varepsilon(\chi_s;R)
\right|
\le
M_R\int_{\chi_{\mathrm{QES}}(R;G_N)} |f|\,dA_\gamma,
\end{equation}
then
\begin{equation}\label{T3c}
\left|
\int_{\chi_{\mathrm{QES}}(R;G_N)} f\,H_{\chi_{\mathrm{QES}}(R;G_N)}\,dA_\gamma
\right|
\le
4G_N M_R\int_{\chi_{\mathrm{QES}}(R;G_N)} |f|\,dA_\gamma.
\end{equation}

Assume, in addition, that along the semiclassical family \(G_N\downarrow0\), the hypersurfaces
\(\chi_{\mathrm{QES}}(R;G_N)\) admit a common \(C^2\) identification with a fixed compact
reference manifold \(\widehat{\chi}_R\), i.e. there exist \(C^2\) diffeomorphisms
\begin{equation}\label{T3c-id}
\Phi_{R,G_N}:\widehat{\chi}_R\longrightarrow \chi_{\mathrm{QES}}(R;G_N)
\end{equation}
such that the pulled-back induced metrics and volume forms are uniformly controlled in \(C^1\),
and the Jacobians of \(\Phi_{R,G_N}\) are uniformly bounded above and below away from \(0\).
Then Eq.~\eqref{T3c} implies that the pulled-back mean curvatures
\begin{equation}
\widehat H_{R,G_N}:=
H_{\chi_{\mathrm{QES}}(R;G_N)}\circ \Phi_{R,G_N}
\end{equation}
converge to \(0\) in \(\mathcal D'(\widehat{\chi}_R)\) as \(G_N\downarrow0\). If, in addition,
the family \(\{\widehat H_{R,G_N}\}_{G_N}\) is uniformly bounded in
\(L^\infty(\widehat{\chi}_R)\), then
\begin{equation}
\widehat H_{R,G_N}\rightharpoonup^\ast 0
\qquad\text{in }L^\infty(\widehat{\chi}_R)
\quad\text{as }G_N\downarrow0,
\end{equation}
which is the RT extremality condition on the time-symmetric slice in this identified limit.

Assume now that, in the fixed renormalization scheme, there exists a renormalized area
functional \(\mathcal{A}^{\mathrm{ren}}(\chi)\), defined up to a common additive constant, such
that the following classwise infima are finite and attained:
\begin{align}
\mathcal{A}^{\mathrm{ren}}_{\mathrm{conn}}(A,B)
&:=
\inf\bigl\{
\mathcal{A}^{\mathrm{ren}}(\chi):
\chi\in\mathcal{X}(R_{AB}),\ \chi\ \text{connected}
\bigr\},
\label{T4cA}
\\
\mathcal{A}^{\mathrm{ren}}_{\mathrm{disc}}(A,B)
&:=
\inf\bigl\{
\mathcal{A}^{\mathrm{ren}}(\chi):
\chi\in\mathcal{X}(R_{AB}),\ \chi\ \text{disconnected}
\bigr\},
\label{T4cB}
\\
S^{\mathrm{gen}}_{\mathrm{conn}}(A,B)
&:=
\inf\bigl\{
S^{\mathrm{gen}}(\chi;R_{AB}):
\chi\in\mathcal{X}(R_{AB}),\ \chi\ \text{connected}
\bigr\},
\label{T4cC}
\\
S^{\mathrm{gen}}_{\mathrm{disc}}(A,B)
&:=
\inf\bigl\{
S^{\mathrm{gen}}(\chi;R_{AB}):
\chi\in\mathcal{X}(R_{AB}),\ \chi\ \text{disconnected}
\bigr\}.
\label{T4cD}
\end{align}
If, in the present two-component boundary setting, every disconnected admissible competitor for
\(R_{AB}=A\cup B\) is of the form \(\chi_A\cup\chi_B\) with
\(\chi_A\in\mathcal{X}(A)\) and \(\chi_B\in\mathcal{X}(B)\), then the disconnected infima
Eq.~\eqref{T4cB} and Eq.~\eqref{T4cD} may equivalently be written over such unions.

Then the classwise RT and QES connectivity criteria are:
\begin{equation}\label{T5c}
\begin{aligned}
\exists\ \text{a connected RT minimizer for }R_{AB}
&\Longleftrightarrow
\mathcal{A}^{\mathrm{ren}}_{\mathrm{conn}}(A,B)
\le
\mathcal{A}^{\mathrm{ren}}_{\mathrm{disc}}(A,B),\\
\exists\ \text{a disconnected RT minimizer for }R_{AB}
&\Longleftrightarrow
\mathcal{A}^{\mathrm{ren}}_{\mathrm{disc}}(A,B)
\le
\mathcal{A}^{\mathrm{ren}}_{\mathrm{conn}}(A,B),
\end{aligned}
\end{equation}
and
\begin{equation}\label{T6c}
\begin{aligned}
\exists\ \text{a connected QES minimizer for }R_{AB}
&\Longleftrightarrow
S^{\mathrm{gen}}_{\mathrm{conn}}(A,B)
\le
S^{\mathrm{gen}}_{\mathrm{disc}}(A,B),\\
\exists\ \text{a disconnected QES minimizer for }R_{AB}
&\Longleftrightarrow
S^{\mathrm{gen}}_{\mathrm{disc}}(A,B)
\le
S^{\mathrm{gen}}_{\mathrm{conn}}(A,B).
\end{aligned}
\end{equation}
In either line, strict inequality implies that every minimizer has the corresponding
connectivity, while equality implies that both connected and disconnected minimizers exist.

Assume finally that there exists \(S_\star\ge0\), independent of \(G_N\) along the
semiclassical family under consideration, such that for every connected admissible
\(\chi_{\mathrm c}\in\mathcal{X}(R_{AB})\) and every disconnected admissible
\(\chi_{\mathrm d}\in\mathcal{X}(R_{AB})\),
\begin{equation}\label{T7a}
\left|
S^{\mathrm{bulk,ren}}(\chi_{\mathrm c};R_{AB})
-
S^{\mathrm{bulk,ren}}(\chi_{\mathrm d};R_{AB})
\right|
\le
S_\star.
\end{equation}
Define
\begin{equation}
\Delta \mathcal{A}^{\mathrm{ren}}_{AB}
:=
\mathcal{A}^{\mathrm{ren}}_{\mathrm{conn}}(A,B)
-
\mathcal{A}^{\mathrm{ren}}_{\mathrm{disc}}(A,B),
\qquad
\Delta S^{\mathrm{gen}}_{AB}
:=
S^{\mathrm{gen}}_{\mathrm{conn}}(A,B)
-
S^{\mathrm{gen}}_{\mathrm{disc}}(A,B).
\end{equation}
Then
\begin{equation}\label{T8a}
\frac{1}{4G_N}\,\Delta \mathcal{A}^{\mathrm{ren}}_{AB}
-
S_\star
\le
\Delta S^{\mathrm{gen}}_{AB}
\le
\frac{1}{4G_N}\,\Delta \mathcal{A}^{\mathrm{ren}}_{AB}
+
S_\star.
\end{equation}
Consequently, whenever
\begin{equation}\label{T8a-threshold}
\left|
\Delta \mathcal{A}^{\mathrm{ren}}_{AB}
\right|
>
4G_N S_\star,
\end{equation}
the signs of \(\Delta S^{\mathrm{gen}}_{AB}\) and \(\Delta \mathcal{A}^{\mathrm{ren}}_{AB}\)
agree, and therefore the RT and QES connectivity decisions agree. More precisely,
\begin{equation}
\Delta \mathcal{A}^{\mathrm{ren}}_{AB}
<
-4G_N S_\star
\Longrightarrow
\Delta S^{\mathrm{gen}}_{AB}<0,
\qquad
\Delta \mathcal{A}^{\mathrm{ren}}_{AB}
>
4G_N S_\star
\Longrightarrow
\Delta S^{\mathrm{gen}}_{AB}>0.
\end{equation}
\end{theorem}

\begin{proof}
For Eq.~\eqref{T1c}, apply Lemma~\ref{lem:firstvar-gen} to
\(\chi=\chi_{\mathrm{QES}}(R;G_N)\). Since \(\chi_{\mathrm{QES}}(R;G_N)\) is quantum extremal,
Eq.~\eqref{D13} gives
\begin{equation}
0
=
\frac{1}{4G_N}\int_{\chi_{\mathrm{QES}}(R;G_N)} f\,H_{\chi_{\mathrm{QES}}(R;G_N)}\,dA_\gamma
+
\lim_{\varepsilon\downarrow0}
\left.\frac{d}{ds}\right|_{s=0}
S^{\mathrm{bulk}}_\varepsilon(\chi_s;R).
\end{equation}
Taking absolute values yields Eq.~\eqref{T1c}. If Eq.~\eqref{T1c-prime-hyp} holds, then
Eq.~\eqref{T1c-prime} is immediate. If Eq.~\eqref{T2c} holds, then Eq.~\eqref{T3c} follows
immediately from Eq.~\eqref{T1c}.

Pull Eq.~\eqref{T3c} back to the fixed reference manifold \(\widehat{\chi}_R\) using
\(\Phi_{R,G_N}\). By the uniform Jacobian bounds in Eq.~\eqref{T3c-id}, there exists a constant
\(C_R>0\), independent of \(G_N\), such that for every \(\varphi\in C^\infty(\widehat{\chi}_R)\),
\begin{equation}
\left|
\int_{\widehat{\chi}_R}
\varphi\,\widehat H_{R,G_N}\,dA_{\widehat{\chi}_R}
\right|
\le
C_R\,4G_N M_R
\int_{\widehat{\chi}_R} |\varphi|\,dA_{\widehat{\chi}_R}.
\end{equation}
Since \(M_R\) is independent of \(G_N\), the right-hand side tends to \(0\) as
\(G_N\downarrow0\). Hence \(\widehat H_{R,G_N}\to0\) in
\(\mathcal D'(\widehat{\chi}_R)\). If, in addition,
\(\{\widehat H_{R,G_N}\}_{G_N}\) is uniformly bounded in \(L^\infty(\widehat{\chi}_R)\), then
Banach-Alaoglu gives weak-* precompactness in \(L^\infty(\widehat{\chi}_R)\), and the
distributional convergence forces every weak-* limit to be \(0\). Therefore
\(\widehat H_{R,G_N}\rightharpoonup^\ast 0\) in \(L^\infty(\widehat{\chi}_R)\).

For Eq.~\eqref{T5c}, the existence of a connected RT minimizer is equivalent, by definition, to
the connected class minimum not exceeding the disconnected class minimum, and similarly for the
disconnected case. Eq.~\eqref{T6c} is the same statement for the generalized entropy
functional. If the inequalities are strict, then the opposite class cannot contain a minimizer;
if equality holds and the classwise infima are attained, then both connected and disconnected
minimizers exist.
It remains to prove Eq.~\eqref{T8a}. Let \(\chi^{\mathcal A}_{\mathrm{conn}}\) and
\(\chi^{\mathcal A}_{\mathrm{disc}}\) attain the classwise area minima
\(\mathcal{A}^{\mathrm{ren}}_{\mathrm{conn}}(A,B)\) and
\(\mathcal{A}^{\mathrm{ren}}_{\mathrm{disc}}(A,B)\), respectively.
For the upper bound, since \(S^{\mathrm{gen}}_{\mathrm{conn}}(A,B)\) is an infimum over the
connected class,
\begin{equation}
S^{\mathrm{gen}}_{\mathrm{conn}}(A,B)
\le
\frac{1}{4G_N}\mathcal{A}^{\mathrm{ren}}_{\mathrm{conn}}(A,B)
+
S^{\mathrm{bulk,ren}}(\chi^{\mathcal A}_{\mathrm{conn}};R_{AB}).
\end{equation}
Also,
\begin{equation}
\begin{aligned}
S^{\mathrm{gen}}_{\mathrm{disc}}(A,B)
&=
\inf_{\substack{\chi_{\mathrm d}\\\mathrm{disc}}}
\left[
\frac{1}{4G_N}\mathcal{A}^{\mathrm{ren}}(\chi_{\mathrm d})
+
S^{\mathrm{bulk,ren}}(\chi_{\mathrm d};R_{AB})
\right]\\
&\ge
\frac{1}{4G_N}\mathcal{A}^{\mathrm{ren}}_{\mathrm{disc}}(A,B)
+
\inf_{\substack{\chi_{\mathrm d}\\\mathrm{disc}}}
S^{\mathrm{bulk,ren}}(\chi_{\mathrm d};R_{AB}).
\end{aligned}
\end{equation}
By Eq.~\eqref{T7a}, for every disconnected \(\chi_{\mathrm d}\),
\begin{equation}
S^{\mathrm{bulk,ren}}(\chi_{\mathrm d};R_{AB})
\ge
S^{\mathrm{bulk,ren}}(\chi^{\mathcal A}_{\mathrm{conn}};R_{AB})-S_\star.
\end{equation}
Hence
\begin{equation}
S^{\mathrm{gen}}_{\mathrm{disc}}(A,B)
\ge
\frac{1}{4G_N}\mathcal{A}^{\mathrm{ren}}_{\mathrm{disc}}(A,B)
+
S^{\mathrm{bulk,ren}}(\chi^{\mathcal A}_{\mathrm{conn}};R_{AB})
-
S_\star.
\end{equation}
Subtracting yields
\begin{equation}
\Delta S^{\mathrm{gen}}_{AB}
\le
\frac{1}{4G_N}
\Bigl(
\mathcal{A}^{\mathrm{ren}}_{\mathrm{conn}}(A,B)
-
\mathcal{A}^{\mathrm{ren}}_{\mathrm{disc}}(A,B)
\Bigr)
+
S_\star
=
\frac{1}{4G_N}\Delta \mathcal{A}^{\mathrm{ren}}_{AB}+S_\star.
\end{equation}

For the lower bound, since \(S^{\mathrm{gen}}_{\mathrm{disc}}(A,B)\) is an infimum over the
disconnected class,
\begin{equation}
S^{\mathrm{gen}}_{\mathrm{disc}}(A,B)
\le
\frac{1}{4G_N}\mathcal{A}^{\mathrm{ren}}_{\mathrm{disc}}(A,B)
+
S^{\mathrm{bulk,ren}}(\chi^{\mathcal A}_{\mathrm{disc}};R_{AB}).
\end{equation}
Also,
\begin{equation}
S^{\mathrm{gen}}_{\mathrm{conn}}(A,B)
\ge
\frac{1}{4G_N}\mathcal{A}^{\mathrm{ren}}_{\mathrm{conn}}(A,B)
+
\inf_{\chi_{\mathrm c}\ \mathrm{conn}}
S^{\mathrm{bulk,ren}}(\chi_{\mathrm c};R_{AB}).
\end{equation}
By Eq.~\eqref{T7a}, for every connected \(\chi_{\mathrm c}\),
\begin{equation}
S^{\mathrm{bulk,ren}}(\chi_{\mathrm c};R_{AB})
\ge
S^{\mathrm{bulk,ren}}(\chi^{\mathcal A}_{\mathrm{disc}};R_{AB})-S_\star.
\end{equation}
Hence
\begin{equation}
S^{\mathrm{gen}}_{\mathrm{conn}}(A,B)
\ge
\frac{1}{4G_N}\mathcal{A}^{\mathrm{ren}}_{\mathrm{conn}}(A,B)
+
S^{\mathrm{bulk,ren}}(\chi^{\mathcal A}_{\mathrm{disc}};R_{AB})
-
S_\star.
\end{equation}
Subtracting now gives
\begin{equation}
\Delta S^{\mathrm{gen}}_{AB}
\ge
\frac{1}{4G_N}
\Bigl(
\mathcal{A}^{\mathrm{ren}}_{\mathrm{conn}}(A,B)
-
\mathcal{A}^{\mathrm{ren}}_{\mathrm{disc}}(A,B)
\Bigr)
-
S_\star
=
\frac{1}{4G_N}\Delta \mathcal{A}^{\mathrm{ren}}_{AB}-S_\star.
\end{equation}
Combining the upper and lower bounds proves Eq.~\eqref{T8a}.

Finally, if
\(\Delta \mathcal{A}^{\mathrm{ren}}_{AB}<-4G_N S_\star\),
then Eq.~\eqref{T8a} gives
\(\Delta S^{\mathrm{gen}}_{AB}<0\).
If
\(\Delta \mathcal{A}^{\mathrm{ren}}_{AB}>4G_N S_\star\),
then Eq.~\eqref{T8a} gives
\(\Delta S^{\mathrm{gen}}_{AB}>0\).
Therefore, away from the \(O(G_N)\) threshold band Eq.~\eqref{T8a-threshold}, the RT and QES
connectivity decisions agree.
\end{proof}}
{
The scaling with \(G_N\) is explicit in Eq.~\eqref{D6b} and Eq.~\eqref{T6c}: if \(\Delta S^{\mathrm{bulk,ren}}_{AB}=O(G_N^{0})\) in the semiclassical expansion, then \(\Delta S^{\mathrm{gen}}_{AB}\) contains an \(O(G_N^{-1})\) renormalized area-difference term and an \(O(G_N^{0})\) renormalized bulk-entropic term, so the threshold band in Eq.~\eqref{T8a} has width \(O(G_N)\) in area units. In the classical limit, where \(\Delta S^{\mathrm{bulk,ren}}_{AB}\) is neglected, Eq.~\eqref{T6c} reduces to the RT criterion \(\Delta \mathcal A^{\mathrm{ren}}_{AB}\le 0\). For disjoint unions, the homology constraint in Eq.~\eqref{A3d} implies that if \(\chi_{\mathrm{disc}}=\chi_A\cup\chi_B\) with corresponding homology regions \(\Omega(\chi_A;A)\) and \(\Omega(\chi_B;B)\) having disjoint interiors, then \(\Omega(\chi_{\mathrm{disc}};R_{AB})=\Omega(\chi_A;A)\cup\Omega(\chi_B;B)\) is admissible for \(R_{AB}\). Invariance under bulk diffeomorphisms preserving the conformal boundary data follows because the regulated areas, the renormalized bulk-entropic contributions, and the homology condition are geometric constructions.}

\section{Typicality obstruction and Page-curve insufficiency}\label{app:typicality}
\renewcommand{\theequation}{G.\arabic{equation}}
Fix finite-dimensional complex Hilbert spaces $\mathcal{H}_A\simeq\mathbb{C}^{d_A}$ and $\mathcal{H}_B\simeq\mathbb{C}^{d_B}$ with $d_A,d_B\in\mathbb{N}$ and define $\mathcal{H}:=\mathcal{H}_A\otimes\mathcal{H}_B$ with $D:=\dim\mathcal{H}=d_A d_B$. Fix a nonzero closed subspace $\mathcal{H}_{\mathrm{code}}\subset \mathcal{H}$ with $\dim\mathcal{H}_{\mathrm{code}}=d_{\mathrm{code}}\in\{1,\dots,D-1\}$ and orthogonal projection $\Pi_{\mathrm{code}}:\mathcal{H}\to\mathcal{H}_{\mathrm{code}}$. Let $\mathbb{S}(\mathcal{H}):=\{\psi\in\mathcal{H}:|\psi|=1\}$ be the unit sphere. Let $\mu_{\mathbb{S}}$ denote the (unique) Haar probability measure on $\mathbb{S}(\mathcal{H})$, characterized by $\mu_{\mathbb{S}}(U E)=\mu_{\mathbb{S}}(E)$ for all measurable $E\subset\mathbb{S}(\mathcal{H})$ and all $U\in\mathrm{U}(\mathcal{H})$ \cite{Simon1996,Mezzadri2007}. For $\psi\in\mathbb{S}(\mathcal{H})$ define the pure state $\rho(\psi):=\ket{\psi}\bra{\psi}$ and the reduced state on $A$ by the partial trace
\begin{equation} \label{A1e}
\rho_A(\psi):=\operatorname{Tr}_B\bigl(\rho(\psi)\bigr)\in\mathcal{B}(\mathcal{H}_A),\qquad
\operatorname{Tr}_A\rho_A(\psi)=1,\qquad
\rho_A(\psi)\ge 0.
\end{equation}
For a density operator $\sigma$ on a finite-dimensional Hilbert space define the von Neumann entropy in natural-logarithm units and in base-$2$ units by
\begin{equation} \label{A2d}
S(\sigma):=-\operatorname{Tr}(\sigma\ln\sigma),\qquad
S_2(\sigma):=-\operatorname{Tr}(\sigma\log_2\sigma)=\frac{1}{\ln 2}\,S(\sigma),
\end{equation}
and for density operators $\sigma,\tau$ with $\operatorname{supp}(\sigma)\subseteq \operatorname{supp}(\tau)$ define the quantum relative entropy by
\begin{equation} \label{A3e}
D(\sigma|\tau):=\operatorname{Tr}\bigl[\sigma(\ln\sigma-\ln\tau)\bigr]\in[0,\infty),\qquad
D_2(\sigma|\tau):=\frac{1}{\ln 2}\,D(\sigma|\tau).
\end{equation}
Define the mutual information in bits for a bipartite density operator $\omega_{XY}$ on $\mathcal{H}_X\otimes\mathcal{H}_Y$ by
\begin{equation} \label{A4d}
I_2(X:Y)_{\omega}:=S_2(\omega_X)+S_2(\omega_Y)-S_2(\omega_{XY})
=\frac{1}{\ln 2}\,D(\omega_{XY}|\omega_X\otimes\omega_Y),
\end{equation}
where $\omega_X=\operatorname{Tr}_Y\omega_{XY}$ and $\omega_Y=\operatorname{Tr}_X\omega_{XY}$ \cite{Umegaki1962,Araki1976,OhyaPetz2004,Watrous2018}. For $\psi\in\mathbb{S}(\mathcal{H})$ define the code-overlap random variable
\begin{equation} \label{A5c}
X_{\mathrm{code}}(\psi):=|\Pi_{\mathrm{code}}\psi|^2=\langle\psi,\Pi_{\mathrm{code}}\psi\rangle\in[0,1].
\end{equation}
Fix a measurable property $\mathcal{P}\subseteq \mathbb{S}(\mathcal{H})$. For $\varepsilon\in(0,1)$ define $\mathcal{P}$ to be $\varepsilon$-typical (with respect to $\mu_{\mathbb{S}}$) if
\begin{equation} \label{A6c}
\mu_{\mathbb{S}}(\mathcal{P})\ge 1-\varepsilon.
\end{equation}
Define the Page curve at fixed dimensions $(d_A,d_B)$ as the Haar expectation of the subsystem entropy
\begin{equation} \label{A7c}
\mathcal{S}_{\mathrm{Page}}(d_A,d_B):=\int_{\mathbb{S}(\mathcal{H})} S\bigl(\rho_A(\psi)\bigr)\,d\mu_{\mathbb{S}}(\psi)\in[0,\ln d_A],
\end{equation}
and similarly $\mathcal{S}_{\mathrm{Page},2}(d_A,d_B):=\mathcal{S}_{\mathrm{Page}}(d_A,d_B)/\ln 2$ in bits. The objective is to establish two precise statements: the typicality obstruction, meaning that $\mu_{\mathbb{S}}$-typicality on $\mathbb{S}(\mathcal{H})$ does not furnish statements about properties restricted to $\mathbb{S}(\mathcal{H}_{\mathrm{code}})$ without an explicit measure restriction, and the Page-curve insufficiency, meaning that knowledge of $\mathcal{S}_{\mathrm{Page}}(d_A,d_B)$ or of the value of $S(\rho_A)$ alone does not logically imply constraints on information-theoretic quantities not determined by the spectrum of $\rho_A$, including mutual informations between refined subsystems and trace-norm proximity statements.

Let $\mu_{\mathrm{U}}$ denote Haar probability measure on $\mathrm{U}(\mathcal{H})$ and fix $\psi_\star\in\mathbb{S}(\mathcal{H})$. Define $\mu_{\mathbb{S}}$ equivalently as the pushforward
\begin{equation} \label{D1f}
\mu_{\mathbb{S}}
:= (\mathrm{U}(\mathcal{H})\ni U\mapsto U\psi_\star\in\mathbb{S}(\mathcal{H}))_\#\,\mu_{\mathrm{U}},
\end{equation}
which is independent of $\psi_\star$ by transitivity of the unitary action \cite{Simon1996,Mezzadri2007}. For $r\in(0,1)$ define the $\mathcal{H}_{\mathrm{code}}$-tube
\begin{equation} \label{D2f}
\mathcal{T}_{\mathrm{code}}(r):=\{\psi\in\mathbb{S}(\mathcal{H}):X_{\mathrm{code}}(\psi)\ge r\}.
\end{equation}
For a density operator $\sigma$ on $\mathcal{H}_A$ define the trace norm and Hilbert-Schmidt norm by
\begin{equation} \label{D3f}
|X|_1:=\operatorname{Tr}\sqrt{X^\dagger X},\qquad
|X|_2:=\sqrt{\operatorname{Tr}(X^\dagger X)}.
\end{equation} 
\begin{lemma}[Gaussian representation of Haar measure]\label{lem:gaussian-haar}
Let $Z\in\mathcal{H}\simeq\mathbb{C}^{D}$ have independent standard complex Gaussian coordinates in an orthonormal basis, meaning $\mathbb{E}[Z]=0$ and $\mathbb{E}[Z_i\overline{Z_j}]=\delta_{ij}$. Then $Z\neq 0$ almost surely and the random vector
\begin{equation} \label{Le1}
\Psi:=\frac{Z}{|Z|}
\end{equation}
is distributed according to $\mu_{\mathbb{S}}$.
\end{lemma}

\begin{proof}
For every fixed $U\in\mathrm{U}(\mathcal{H})$, the random vector $UZ$ has the same distribution as $Z$ because the standard complex Gaussian law on $\mathbb{C}^{D}$ is invariant under unitary transformations. Since $\Psi$ is a measurable function of $Z$ and $U\Psi=(UZ)/|UZ|$, the law of $\Psi$ is invariant under the action of $\mathrm{U}(\mathcal{H})$ on $\mathbb{S}(\mathcal{H})$. Unitary invariance and normalization $\mathbb{P}(\Psi\in\mathbb{S}(\mathcal{H}))=1$ characterize $\mu_{\mathbb{S}}$ uniquely \cite{Simon1996}.
\end{proof}

\begin{lemma}[Exact distribution of code overlap]\label{lem:beta-code}
Let $d_{\mathrm{code}}\in\{1,\dots,D-1\}$ and let $\Pi_{\mathrm{code}}$ have rank $d_{\mathrm{code}}$. For $\Psi\sim\mu_{\mathbb{S}}$ define $X_{\mathrm{code}}(\Psi)$ by Eq.~\eqref{A5c}. Then $X_{\mathrm{code}}(\Psi)$ has the Beta$(d_{\mathrm{code}},D-d_{\mathrm{code}})$ distribution with density
\begin{equation} \label{L2g}
f_{d_{\mathrm{code}},D}(x)=\frac{\Gamma(D)}{\Gamma(d_{\mathrm{code}})\Gamma(D-d_{\mathrm{code}})}\,x^{d_{\mathrm{code}}-1}(1-x)^{D-d_{\mathrm{code}}-1}\,\mathbf{1}_{[0,1]}(x),
\end{equation}
and in particular
\begin{equation} \label{Lb3b}
\begin{aligned}
\mathbb{E}[X_{\mathrm{code}}(\Psi)]
&=\frac{d_{\mathrm{code}}}{D},\\
\mathbb{P}\bigl[X_{\mathrm{code}}(\Psi)\ge 1-\varepsilon\bigr]
&\le
\frac{\Gamma(D)}{\Gamma(d_{\mathrm{code}})\Gamma(D-d_{\mathrm{code}})}
\,\varepsilon^{D-d_{\mathrm{code}}},
\qquad \varepsilon\in(0,1).
\end{aligned}
\end{equation}
\end{lemma}

\begin{proof}
By Lemma~\ref{lem:gaussian-haar}, $\Psi=Z/|Z|$ with $Z$ standard complex Gaussian. Choose an orthonormal basis adapted to $\mathcal{H}=\mathcal{H}_{\mathrm{code}}\oplus \mathcal{H}_{\mathrm{code}}^\perp$ and write $Z=(Z_1,Z_2)$ with $Z_1\in\mathcal{H}_{\mathrm{code}}\simeq\mathbb{C}^{d_{\mathrm{code}}}$ and $Z_2\in\mathcal{H}_{\mathrm{code}}^\perp\simeq\mathbb{C}^{D-d_{\mathrm{code}}}$. Then $Z_1$ and $Z_2$ are independent standard complex Gaussians in their respective dimensions. Define $U:=|Z_1|^2$ and $V:=|Z_2|^2$. Since $Z_1$ has $d_{\mathrm{code}}$ complex coordinates, $U$ is Gamma$(d_{\mathrm{code}},1)$ distributed, and similarly $V$ is Gamma$(D-d_{\mathrm{code}},1)$ distributed, and $U$ and $V$ are independent \cite{Tao2012}. Using $\Psi=Z/|Z|$ and $\Pi_{\mathrm{code}}(Z_1,Z_2)=(Z_1,0)$ gives
\begin{equation} \label{L4f}
X_{\mathrm{code}}(\Psi)=\frac{|Z_1|^2}{|Z_1|^2+|Z_2|^2}=\frac{U}{U+V}.
\end{equation}
For independent Gamma$(a,1)$ and Gamma$(b,1)$ variables $U,V$ with $a=d_{\mathrm{code}}$ and $b=D-d_{\mathrm{code}}$, the ratio $U/(U+V)$ has the Beta$(a,b)$ distribution with density Eq.~\eqref{L2g} and expectation $a/(a+b)=d_{\mathrm{code}}/D$ \cite{Devroye1986}. For the tail bound, use $x^{d_{\mathrm{code}}-1}\le 1$ and $(1-x)^{D-d_{\mathrm{code}}-1}\le \varepsilon^{D-d_{\mathrm{code}}-1}$ for $x\in[1-\varepsilon,1]$ to obtain
\begin{equation} \label{L5f}
\begin{aligned}
\mathbb{P}\bigl[X_{\mathrm{code}}(\Psi)\ge 1-\varepsilon\bigr]
&=\int_{1-\varepsilon}^{1} f_{d_{\mathrm{code}},D}(x)\,dx
\\
&\le
\frac{\Gamma(D)}{\Gamma(d_{\mathrm{code}})\Gamma(D-d_{\mathrm{code}})}
\,\varepsilon^{D-d_{\mathrm{code}}-1}
\int_{1-\varepsilon}^{1}dx
\\
&=
\frac{\Gamma(D)}{\Gamma(d_{\mathrm{code}})\Gamma(D-d_{\mathrm{code}})}
\,\varepsilon^{D-d_{\mathrm{code}}}.
\end{aligned}
\end{equation}
\end{proof}

\begin{lemma}[Induced eigenvalue density and Page expectation]\label{lem:page-density}
Assume $d_B\ge d_A\ge 2$. For $\Psi\sim\mu_{\mathbb{S}}$ the eigenvalues $\lambda=(\lambda_1,\dots,\lambda_{d_A})$ of $\rho_A(\Psi)$ have joint density on the simplex $\Delta_{d_A}:=\{\lambda_i\ge 0,\ \sum_{i=1}^{d_A}\lambda_i=1\}$ given by
\begin{equation} \label{L6f}
p_{d_A,d_B}(\lambda)=C_{d_A,d_B}\,\delta\!\left(1-\sum_{i=1}^{d_A}\lambda_i\right)\prod_{i=1}^{d_A}\lambda_i^{d_B-d_A}\prod_{1\le i<j\le d_A}(\lambda_i-\lambda_j)^2,
\end{equation}
with normalization constant
\begin{equation} \label{L7e}
C_{d_A,d_B}=\frac{\Gamma(d_A d_B)}{\prod_{j=0}^{d_A-1}\Gamma(d_B-j)\Gamma(d_A-j+1)},
\end{equation}
and the Page expectation satisfies the exact identity
\begin{equation} \label{L8d}
\mathcal{S}_{\mathrm{Page}}(d_A,d_B)=\sum_{k=d_B+1}^{d_A d_B}\frac{1}{k}-\frac{d_A-1}{2d_B}
=H_{d_A d_B}-H_{d_B}-\frac{d_A-1}{2d_B},
\end{equation}
where $H_n:=\sum_{k=1}^{n}\frac{1}{k}$ is the $n$-th harmonic number. In
bits,
\begin{equation}
\mathcal{S}_{\mathrm{Page},2}(d_A,d_B)
=\frac{\mathcal{S}_{\mathrm{Page}}(d_A,d_B)}{\ln 2}.
\end{equation}
\end{lemma}

\begin{proof}
By Lemma~\ref{lem:gaussian-haar}, represent $\Psi$ as $Z/|Z|$ with $Z$ standard complex Gaussian in $\mathcal{H}$. Reshape $Z$ into a $d_A\times d_B$ complex matrix $X$ by fixing orthonormal bases $\{|i\rangle_A\}_{i=1}^{d_A}$ and $\{|\mu\rangle_B\}_{\mu=1}^{d_B}$ and writing
\begin{equation} \label{L9d}
Z=\sum_{i=1}^{d_A}\sum_{\mu=1}^{d_B} X_{i\mu}\,|i\rangle_A\otimes|\mu\rangle_B,\qquad
|Z|^2=\sum_{i,\mu}|X_{i\mu}|^2=|X|_{\mathrm{F}}^2.
\end{equation}
Define the Wishart matrix $W:=XX^\dagger\in\mathcal{B}(\mathcal{H}_A)$ and note $W\ge 0$ and $\operatorname{Tr}W=|X|_{\mathrm{F}}^2=|Z|^2$. Using $\rho(\Psi)=|Z\rangle\langle Z|/|Z|^2$ and the definition of partial trace gives
\begin{equation} \label{L10d}
\rho_A(\Psi)=\operatorname{Tr}_B\left(\frac{|Z\rangle\langle Z|}{|Z|^2}\right)=\frac{XX^\dagger}{\operatorname{Tr}(XX^\dagger)}=\frac{W}{\operatorname{Tr}W}.
\end{equation}
The distribution of $W$ is the complex Wishart distribution with parameters $(d_A,d_B)$, having density proportional to $e^{-\operatorname{Tr}W}(\det W)^{d_B-d_A}$ on the cone of positive definite matrices when $d_B\ge d_A$ \cite{Forrester2010,Tao2012}. The change of variables from $W$ to its eigenvalues and eigenvectors yields the eigenvalue density Eq.~\eqref{L6f} with normalization Eq.~\eqref{L7e} as a specialization of the complex Selberg integral \cite{Mehta2004,Forrester2010,ZyczkowskiSommers2001}. Using Eq.~\eqref{A7c} and the spectral decomposition $\rho_A=\sum_i \lambda_i |i\rangle\langle i|$ yields
\begin{equation} \label{L11d}
\mathcal{S}_{\mathrm{Page}}(d_A,d_B)
=-\int_{\Delta_{d_A}}\left(\sum_{i=1}^{d_A}\lambda_i\ln\lambda_i\right)p_{d_A,d_B}(\lambda)\,d\lambda,
\end{equation}
with $d\lambda$ the Lebesgue measure on the simplex hyperplane. The evaluation of the integral Eq.~\eqref{L11d} for the density Eq.~\eqref{L6f} is Page's theorem in the form Eq.~\eqref{L8d}, whose proof reduces to differentiating a Selberg-type normalization integral with respect to an exponent parameter and expressing the result via digamma functions, which simplifies to harmonic numbers \cite{Page1993,Sen1996,ZyczkowskiSommers2001,Forrester2010}.
\end{proof}

\begin{lemma}[Haar purity and trace-distance control]\label{lem:purity-trace}
For $\Psi\sim\mu_{\mathbb{S}}$,
\begin{equation} \label{L12d}
\mathbb{E}\,\operatorname{Tr}\bigl(\rho_A(\Psi)^2\bigr)=\frac{d_A+d_B}{d_A d_B+1},
\end{equation}
and consequently
\begin{equation} \label{L13d}
\begin{aligned}
\mathbb{E}\,\bigl|\rho_A(\Psi)-\tfrac{1}{d_A}\mathbf{1}_A\bigr|_1
&\le \sqrt{d_A}\,\sqrt{\mathbb{E}\,\bigl|\rho_A(\Psi)-\tfrac{1}{d_A}\mathbf{1}_A\bigr|_2^2}\\
&=\sqrt{d_A}\,\sqrt{\mathbb{E}\,\operatorname{Tr}\bigl(\rho_A(\Psi)^2\bigr)-\frac{1}{d_A}}
=\sqrt{d_A}\,\sqrt{\frac{d_A^2-1}{d_A(d_A d_B+1)}}.
\end{aligned}
\end{equation}
\end{lemma}

\begin{proof}
Identity Eq.~\eqref{L12d} is the Lubkin-Page purity formula, derived by Gaussian integration or Weingarten calculus \cite{Lubkin1978,ZyczkowskiSommers2001,CollinsSniady2006}. For Eq.~\eqref{L13d}, use $|X|_1\le \sqrt{d_A}|X|_2$ for operators on $\mathcal{H}_A$ and Jensen's inequality $\mathbb{E}|X|_2\le \sqrt{\mathbb{E}|X|_2^2}$ to obtain the first inequality, and compute
\begin{equation} \label{L14d}
\begin{aligned}
\left|\rho_A-\tfrac{1}{d_A}\mathbf{1}\right|_2^2
&=\operatorname{Tr}\left(\rho_A-\tfrac{1}{d_A}\mathbf{1}\right)^2\\
&=\operatorname{Tr}(\rho_A^2)-\frac{2}{d_A}\operatorname{Tr}(\rho_A)+\frac{1}{d_A^2}\operatorname{Tr}(\mathbf{1})\\
&=\operatorname{Tr}(\rho_A^2)-\frac{2}{d_A}\cdot 1+\frac{1}{d_A^2}\cdot d_A\\
&=\operatorname{Tr}(\rho_A^2)-\frac{1}{d_A},
\end{aligned}
\end{equation}
and substitute Eq.~\eqref{L12d}.
\end{proof} 

\begin{theorem}[Typicality obstruction and Page-curve insufficiency]\label{thm:typicality-page}
Fix $d_A,d_B\in\mathbb{N}$, $\mathcal{H}=\mathcal{H}_A\otimes\mathcal{H}_B$, and a code subspace $\mathcal{H}_{\mathrm{code}}\subset\mathcal{H}$ of dimension $d_{\mathrm{code}}<D=d_A d_B$. Then the following statements hold.
\begin{enumerate}
\item[(i)] For every $\varepsilon\in(0,1)$,
\begin{equation} \label{T1d}
\mu_{\mathbb{S}}\bigl(\mathcal{T}_{\mathrm{code}}(1-\varepsilon)\bigr)\le \frac{\Gamma(D)}{\Gamma(d_{\mathrm{code}})\Gamma(D-d_{\mathrm{code}})}\,\varepsilon^{D-d_{\mathrm{code}}},
\end{equation}
and in particular if $d_{\mathrm{code}}/D\to 0$ along a sequence of models with $D\to\infty$ and $d_{\mathrm{code}}\to\infty$ and fixed $\varepsilon\in(0,1)$, then $\mu_{\mathbb{S}}(\mathcal{T}_{\mathrm{code}}(1-\varepsilon))\to 0$.
\item[(ii)] Fix a further factorization $\mathcal{H}_A=\mathcal{H}_{A_1}\otimes\mathcal{H}_{A_2}$ with $\dim\mathcal{H}_{A_1}=\dim\mathcal{H}_{A_2}=2$ and $\dim\mathcal{H}_B\ge 2$. There exist pure states $\psi,\phi\in\mathbb{S}(\mathcal{H}_A\otimes\mathcal{H}_B)$ such that
\begin{equation} \label{T2d}
S\bigl(\rho_A(\psi)\bigr)=S\bigl(\rho_A(\phi)\bigr)=\ln 2,
\end{equation}
but
\begin{equation} \label{T3d}
I_2(A_1:A_2)_{\rho_A(\psi)}=0,\qquad I_2(A_1:A_2)_{\rho_A(\phi)}=1,
\end{equation}
so the value of $S(\rho_A)$ does not determine the mutual information between refined subsystems.
{\item[(iii)] For $d_B\ge d_A\ge 2$, the Page expectation $\mathcal{S}_{\mathrm{Page}}(d_A,d_B)$ in Eq.~\eqref{A7c} is an ensemble average and does not imply trace-norm decoupling in the balanced regime $d_A\asymp d_B$. In the balanced case $d_A=d_B=:d$ with $D=d^2$, the empirical spectral distribution of $d\,\rho_A(\Psi)$ converges as $d\to\infty$ to the Marchenko-Pastur law with ratio parameter $\gamma=d_B/d_A=1$, whose density is
\begin{equation}
f(\lambda)=\frac{\sqrt{\lambda(4-\lambda)}}{2\pi\lambda},\qquad \lambda\in[0,4],
\end{equation}
see, e.g., \cite{MarchenkoPastur1967}. Consequently,
\begin{equation} \label{T4d}
\lim_{d\to\infty}\mathbb{E}\,\bigl\|\rho_A(\Psi)-\tfrac{1}{d}\mathbf{1}_A\bigr\|_1
=
\int_0^4 |\lambda-1|\,\frac{\sqrt{\lambda(4-\lambda)}}{2\pi\lambda}\,d\lambda
=
\frac{3\sqrt{3}}{2\pi}\approx 0.827.
\end{equation}
Since $\frac{3\sqrt{3}}{2\pi}>0$, the reduced state $\rho_A(\Psi)$ remains bounded away from the maximally mixed state in trace distance, even though its entropy is near-maximal by the Page result
$\mathcal{S}_{\mathrm{Page}}(d,d)=\ln d-\frac{1}{2}+o(1)$ obtained from Eq.~\eqref{L8d}. The Page curve therefore does not imply trace-norm decoupling and cannot serve as a proxy for the mutual-information-based connectivity diagnostic.}
\end{enumerate}
\end{theorem}

\begin{proof}
Statement (i) follows from Lemma~\ref{lem:beta-code}. Statement (ii) is established by explicit construction. Let $\{|0\rangle,|1\rangle\}$ be the computational basis of $\mathbb{C}^2$ and define on $\mathcal{H}_A=\mathbb{C}^2\otimes\mathbb{C}^2$ the density operators
\begin{equation} \label{P1b}
\omega:=\frac{1}{2}\bigl(|00\rangle\langle 00|+|01\rangle\langle 01|\bigr),\qquad
\omega':=\frac{1}{2}\bigl(|\Phi^+\rangle\langle \Phi^+|+|\Phi^-\rangle\langle \Phi^-|\bigr),
\end{equation}
where $|\Phi^\pm\rangle:=\frac{1}{\sqrt{2}}(|00\rangle\pm |11\rangle)$. Both $\omega$ and $\omega'$ have spectrum $\{1/2,1/2,0,0\}$, hence
\begin{equation} \label{P2a}
S(\omega)=S(\omega')=-\left(\frac{1}{2}\ln\frac{1}{2}+\frac{1}{2}\ln\frac{1}{2}\right)=\ln 2.
\end{equation}
Compute marginals of $\omega$ by partial traces:
\begin{equation}\label{P3a}
\omega_{A_1}=\operatorname{Tr}_{A_2}\omega=|0\rangle\langle 0|,\qquad
\omega_{A_2}=\operatorname{Tr}_{A_1}\omega=\frac{1}{2}\bigl(|0\rangle\langle 0|+|1\rangle\langle 1|\bigr)=\frac{1}{2}\mathbf{1},
\end{equation}
hence
\begin{equation} \label{P4a}
\begin{aligned}
S(\omega_{A_1})&=0,
& S(\omega_{A_2})&=\ln 2,\\
I_2(A_1:A_2)_\omega
&=\frac{1}{\ln 2}
\bigl(S(\omega_{A_1})+S(\omega_{A_2})-S(\omega)\bigr)\\
&=\frac{1}{\ln 2}\bigl(0+\ln 2-\ln 2\bigr)=0.
\end{aligned}
\end{equation}
Compute marginals of $\omega'$ using $\operatorname{Tr}_{A_2}|\Phi^\pm\rangle\langle\Phi^\pm|=\frac{1}{2}\mathbf{1}$ and similarly for $\operatorname{Tr}_{A_1}$:
\begin{equation} \label{P5a}
\omega'_{A_1}=\operatorname{Tr}_{A_2}\omega'=\frac{1}{2}\mathbf{1},\qquad
\omega'_{A_2}=\operatorname{Tr}_{A_1}\omega'=\frac{1}{2}\mathbf{1},
\end{equation}
hence
\begin{equation} \label{P6a}
S(\omega'_{A_1})=\ln 2,\qquad
S(\omega'_{A_2})=\ln 2,\qquad
I_2(A_1:A_2)_{\omega'}=\frac{1}{\ln 2}\bigl(\ln 2+\ln 2-\ln 2\bigr)=1.
\end{equation}
Choose $\mathcal{H}_B\simeq\mathbb{C}^2$ and define purifications
\begin{equation} \label{P7a}
|\psi\rangle:=\frac{1}{\sqrt{2}}\bigl(|00\rangle_A\otimes|0\rangle_B+|01\rangle_A\otimes|1\rangle_B\bigr),\qquad
|\phi\rangle:=\frac{1}{\sqrt{2}}\bigl(|\Phi^+\rangle_A\otimes|0\rangle_B+|\Phi^-\rangle_A\otimes|1\rangle_B\bigr),
\end{equation}
so that $\rho_A(\psi)=\omega$ and $\rho_A(\phi)=\omega'$ by direct computation of partial traces, yielding Eq.~\eqref{T2d}-\eqref{T3d}.

{
{Statement (iii) follows instead from the large-$d$ Marchenko-Pastur asymptotics for the empirical spectral distribution of $d\,\rho_A(\Psi)$ in the balanced case $d_A=d_B=:d$ \cite{MarchenkoPastur1967}. Writing the eigenvalues of $d\,\rho_A(\Psi)$ as $\lambda_1,\dots,\lambda_d$, one has
\begin{equation}
\bigl\|\rho_A(\Psi)-\tfrac{1}{d}\mathbf{1}_A\bigr\|_1
=
\frac{1}{d}\sum_{i=1}^d |\lambda_i-1|.
\end{equation}
Since the empirical spectral distribution $\frac{1}{d}\sum_{i=1}^d \delta_{\lambda_i}$ converges to the Marchenko-Pastur law with density $f(\lambda)=\sqrt{\lambda(4-\lambda)}/(2\pi\lambda)$ on $[0,4]$, it follows that
\begin{equation}
\lim_{d\to\infty}\mathbb{E}\,\bigl\|\rho_A(\Psi)-\tfrac{1}{d}\mathbf{1}_A\bigr\|_1
=
\int_0^4 |\lambda-1|\,\frac{\sqrt{\lambda(4-\lambda)}}{2\pi\lambda}\,d\lambda
=
\frac{3\sqrt{3}}{2\pi},
\end{equation}
which proves Eq.~\eqref{T4d}.} To obtain the Page asymptotics
rigorously from Eq.~\eqref{L8d}, use the harmonic-number expansion
\begin{equation}
H_n=\ln n+\gamma_{\mathrm E}+\frac{1}{2n}+O(n^{-2})
\qquad (n\to\infty),
\end{equation}
where \(\gamma_{\mathrm E}\) denotes the Euler-Mascheroni constant. Therefore
\begin{align}
H_{d_A^2}-H_{d_A}
&=
\left(2\ln d_A+\gamma_{\mathrm E}+\frac{1}{2d_A^2}+O(d_A^{-4})\right)
-
\left(\ln d_A+\gamma_{\mathrm E}+\frac{1}{2d_A}+O(d_A^{-2})\right)\notag\\
&=
\ln d_A-\frac{1}{2d_A}+O(d_A^{-2}).
\end{align}
On the other hand,
\begin{equation}
\frac{d_A-1}{2d_A}
=
\frac{1}{2}-\frac{1}{2d_A}.
\end{equation}
Substituting these two expansions into Eq.~\eqref{L8d}, we obtain
\begin{align}
\mathcal{S}_{\mathrm{Page}}(d_A,d_A)
&=
H_{d_A^2}-H_{d_A}-\frac{d_A-1}{2d_A}\notag\\
&=
\left(\ln d_A-\frac{1}{2d_A}+O(d_A^{-2})\right)
-
\left(\frac{1}{2}-\frac{1}{2d_A}\right)\notag\\
&=
\ln d_A-\frac{1}{2}+O(d_A^{-2}).
\end{align}
Thus the constant term \(-\frac{1}{2}\) follows from an explicit cancellation of the
\(\pm \frac{1}{2d_A}\) contributions, and in particular
\begin{equation}
\mathcal{S}_{\mathrm{Page}}(d_A,d_A)=\ln d_A-\frac{1}{2}+o(1).
\end{equation}}
\end{proof} 
For $D=2$ and $d_{\mathrm{code}}=1$, Lemma~\ref{lem:beta-code} gives $X_{\mathrm{code}}(\Psi)$ uniform on $[0,1]$, consistent with Eq.~\eqref{L2g}. For $d_A=1$ or $d_B=1$, $\rho_A(\Psi)$ is pure and $\mathcal{S}_{\mathrm{Page}}(d_A,d_B)=0$, consistent with Eq.~\eqref{L8d}. Invariance under unitary basis changes is explicit: $\mu_{\mathbb{S}}$ is invariant under $\mathrm{U}(\mathcal{H})$, $X_{\mathrm{code}}(\psi)$ depends only on the projection rank, and the distribution Eq.~\eqref{L6f} depends only on dimensions. The conclusions in Theorem~\ref{thm:typicality-page} are measure-theoretic: statement (i) concerns $\mu_{\mathbb{S}}$-probabilities; statement (ii) is existential and does not invoke typicality; statement (iii) bounds Haar expectations and does not assert pointwise proximity.

\section{Approximate Markov strengthening}\label{app:markov}
\renewcommand{\theequation}{H.\arabic{equation}}
Fix finite-dimensional complex Hilbert spaces $\mathcal{H}_A\simeq\mathbb{C}^{d_A}$, $\mathcal{H}_B\simeq\mathbb{C}^{d_B}$, $\mathcal{H}_C\simeq\mathbb{C}^{d_C}$ with $d_A,d_B,d_C\in\mathbb{N}$ and let $\mathcal{H}_{ABC}:=\mathcal{H}_A\otimes\mathcal{H}_B\otimes\mathcal{H}_C$. Let $\mathcal{B}(\mathcal{H})$ denote the bounded operators on $\mathcal{H}$ and let $\mathcal{T}_1(\mathcal{H})$ denote trace-class operators, which coincide with $\mathcal{B}(\mathcal{H})$ in finite dimension. Fix a density operator $\rho_{ABC}\in\mathcal{B}(\mathcal{H}_{ABC})$ with
\begin{equation} \label{A1f}
\rho_{ABC}=\rho_{ABC}^\dagger,\qquad \rho_{ABC}\ge 0,\qquad \operatorname{Tr}_{ABC}\rho_{ABC}=1,
\end{equation}
and assume faithfulness on $B$ and $BC$ in the sense that the marginals
\begin{equation} \label{A2e}
\rho_{AB}:=\operatorname{Tr}_C\rho_{ABC},\qquad
\rho_{BC}:=\operatorname{Tr}_A\rho_{ABC},\qquad
\rho_B:=\operatorname{Tr}_{AC}\rho_{ABC}
\end{equation}
satisfy $\rho_B>0$ on $\mathcal{H}_B$ and $\rho_{BC}>0$ on $\mathcal{H}_B\otimes\mathcal{H}_C$. Fix the natural-logarithm convention for entropies and relative entropies, and define the base-$2$ convention by explicit division by $\ln 2$. The approximation regime is parameterized by $\varepsilon\ge 0$ defined by
\begin{equation} \label{A3f}
\varepsilon:=I(A:C|B)_{\rho}\in[0,\infty),
\end{equation}
where $I(A:C|B)_\rho$ is the quantum conditional mutual information in nats defined below. The admissible recovery maps are completely positive trace-preserving (CPTP) maps $\mathcal{R}_{B\to BC}:\mathcal{B}(\mathcal{H}_B)\to\mathcal{B}(\mathcal{H}_B\otimes\mathcal{H}_C)$, extended to $\mathcal{B}(\mathcal{H}_{AB})$ by $\mathrm{id}_A\otimes\mathcal{R}_{B\to BC}$. The objective is to derive a strengthening of the approximate Markov condition $I(A:C|B)_\rho\le\varepsilon$ into an explicit quantitative recovery statement with all constants and norm choices explicit. 

\begin{definition}[Entropy and relative entropy (nats)]
For a density operator $\sigma$ on a finite-dimensional Hilbert space define the von Neumann entropy in nats by
\begin{equation} \label{D1g}
S(\sigma):=-\operatorname{Tr}(\sigma\ln\sigma),
\end{equation}
and define the quantum relative entropy in nats for $\sigma,\tau>0$ by
\begin{equation} \label{D2g}
D(\sigma|\tau):=\operatorname{Tr}\bigl[\sigma(\ln\sigma-\ln\tau)\bigr]\in[0,\infty).
\end{equation}
\end{definition}

\begin{definition}[Conditional mutual information]
Define the conditional mutual information in nats by
\begin{equation}\label{D3h}
I(A:C|B)_\rho:=S(\rho_{AB})+S(\rho_{BC})-S(\rho_B)-S(\rho_{ABC})\in[0,\infty),
\end{equation}
and in bits by $I_2(A:C|B)_\rho:=\frac{1}{\ln 2}I(A:C|B)_\rho$.
\end{definition}

\begin{definition}[Norms, trace distance, fidelity]
For $X\in\mathcal{B}(\mathcal{H})$ define the trace norm and operator norm by
\begin{equation} \label{D4d}
|X|_1:=\operatorname{Tr}\sqrt{X^\dagger X},\qquad
|X|_\infty:=\sup_{\psi\neq 0}\frac{|X\psi|}{|\psi|},
\end{equation}
and define the trace distance between density operators $\sigma,\tau$ by
\begin{equation} \label{D5d}
T(\sigma,\tau):=\frac{1}{2}|\sigma-\tau|_1.
\end{equation}
For density operators $\sigma,\tau$ define the (unsquared) Uhlmann fidelity by
\begin{equation} \label{D6c}
F(\sigma,\tau):=\bigl|\sqrt{\sigma}\sqrt{\tau}\bigr|_1\in[0,1].
\end{equation}
\end{definition}

\begin{definition}[Exact quantum Markov chain]
Define an exact quantum Markov chain in the order $A$-$B$-$C$ by the condition
\begin{equation} \label{D7c}
\rho_{ABC}\ \text{is Markov}\Longleftrightarrow I(A:C|B)_\rho=0.
\end{equation}
\end{definition}

\begin{definition}[Petz, rotated Petz, and twirled Petz recovery maps]
For faithful $\rho_B$ and $\rho_{BC}$ define the Petz recovery map $\mathcal{R}^{\mathrm{P}}_{\rho,B\to BC}$ and the rotated Petz recovery maps $\mathcal{R}^{t}_{\rho,B\to BC}$ for $t\in\mathbb{R}$ by
\begin{equation} \label{D8c}
\mathcal{R}^{\mathrm{P}}_{\rho,B\to BC}(X)
:=\rho_{BC}^{1/2}\bigl(\rho_B^{-1/2}X\rho_B^{-1/2}\otimes \mathbf{1}_C\bigr)\rho_{BC}^{1/2},
\end{equation}
\begin{equation} \label{D9c}
\mathcal{R}^{t}_{\rho,B\to BC}(X)
:=\rho_{BC}^{\frac{1+it}{2}}\bigl(\rho_B^{-\frac{1+it}{2}}X\rho_B^{-\frac{1-it}{2}}\otimes \mathbf{1}_C\bigr)\rho_{BC}^{\frac{1-it}{2}},
\end{equation}
where the complex powers are defined by functional calculus for positive definite operators. Define the weight
\begin{equation} \label{D10b}
\beta_0(t):=\frac{\pi}{2}\,\frac{1}{\cosh(\pi t)+1}=\frac{\pi}{4}\,\mathrm{sech}^2\!\left(\frac{\pi t}{2}\right),\qquad t\in\mathbb{R},
\end{equation}
and note
\begin{equation} \label{D11a}
{\int_{-\infty}^{\infty}\beta_0(t)\,dt
=
\frac{\pi}{4}\int_{-\infty}^{\infty}\operatorname{sech}^2\!\left(\frac{\pi t}{2}\right)\,dt
=
\frac{\pi}{4}\cdot \frac{2}{\pi}
\left[
\tanh\!\left(\frac{\pi t}{2}\right)
\right]_{-\infty}^{\infty}
=1.}
\end{equation}
Define the twirled Petz recovery map by the Bochner integral
\begin{equation} \label{D12a}
\mathcal{R}^{\mathrm{tw}}_{\rho,B\to BC}(X):=\int_{-\infty}^{\infty}\beta_0(t)\,\mathcal{R}^{t}_{\rho,B\to BC}(X)\,dt,
\end{equation}
which is well-defined in operator norm because $t\mapsto\mathcal{R}^t_{\rho,B\to BC}(X)$ is continuous and $\beta_0$ is integrable with total mass $1$.
\end{definition} 

\begin{lemma}[Complete positivity and trace preservation of rotated Petz maps]\label{lem:rotated-petz-cptp}
Under Eq.~\eqref{A2e}, for each $t\in\mathbb{R}$ the map $\mathcal{R}^{t}_{\rho,B\to BC}$ defined by Eq.~\eqref{D9c} is completely positive and trace preserving, and consequently $\mathcal{R}^{\mathrm{tw}}_{\rho,B\to BC}$ is CPTP.
\end{lemma}

\begin{proof}
Fix $t\in\mathbb{R}$ and define
\begin{equation} \label{Ly1}
V_t:=\rho_{BC}^{\frac{1+it}{2}}\bigl(\rho_B^{-\frac{1+it}{2}}\otimes \mathbf{1}_C\bigr)\in\mathcal{B}(\mathcal{H}_B\otimes\mathcal{H}_C).
\end{equation}
Then Eq.~\eqref{D9c} can be rewritten as
\begin{equation} \label{L2h}
\mathcal{R}^{t}_{\rho,B\to BC}(X)=V_t\,(X\otimes \mathbf{1}_C)\,V_t^\dagger,
\end{equation}
because $V_t^\dagger=\bigl(\rho_B^{-\frac{1-it}{2}}\otimes\mathbf{1}_C\bigr)\rho_{BC}^{\frac{1-it}{2}}$ and multiplication yields Eq.~\eqref{D9c}. The map $X\mapsto X\otimes \mathbf{1}_C$ is completely positive, and conjugation by $V_t$ is completely positive, hence the composition Eq.~\eqref{L2h} is completely positive. For trace preservation, use cyclicity of the trace and $\operatorname{Tr}_{BC}\bigl((X\otimes\mathbf{1}_C)Y\bigr)=\operatorname{Tr}_B\bigl(X\,\operatorname{Tr}_C Y\bigr)$ for all $Y\in\mathcal{B}(\mathcal{H}_B\otimes\mathcal{H}_C)$ to compute
\begin{align} \label{Lc3c}
\operatorname{Tr}_{BC}\mathcal{R}^{t}_{\rho,B\to BC}(X)
&=\operatorname{Tr}_{BC}\bigl(V_t(X\otimes\mathbf{1}_C)V_t^\dagger\bigr)
=\operatorname{Tr}_{BC}\bigl((X\otimes\mathbf{1}_C)V_t^\dagger V_t\bigr)\nonumber\\
&=\operatorname{Tr}_{BC}\Bigl((X\otimes\mathbf{1}_C)\bigl(\rho_B^{-\frac{1-it}{2}}\otimes\mathbf{1}_C\bigr)\rho_{BC}^{\frac{1-it}{2}}\rho_{BC}^{\frac{1+it}{2}}\bigl(\rho_B^{-\frac{1+it}{2}}\otimes\mathbf{1}_C\bigr)\Bigr)\nonumber\\
&=\operatorname{Tr}_{BC}\Bigl((X\otimes\mathbf{1}_C)\bigl(\rho_B^{-\frac{1-it}{2}}\otimes\mathbf{1}_C\bigr)\rho_{BC}\bigl(\rho_B^{-\frac{1+it}{2}}\otimes\mathbf{1}_C\bigr)\Bigr)\nonumber\\
&=\operatorname{Tr}_{B}\Bigl(X\,\rho_B^{-\frac{1-it}{2}}\bigl(\operatorname{Tr}_C\rho_{BC}\bigr)\rho_B^{-\frac{1+it}{2}}\Bigr)\nonumber\\
&=\operatorname{Tr}_{B}\Bigl(X\,\rho_B^{-\frac{1-it}{2}}\rho_B\,\rho_B^{-\frac{1+it}{2}}\Bigr)
=\operatorname{Tr}_B\Bigl(X\,\rho_B^{\frac{1+it}{2}}\rho_B^{-\frac{1+it}{2}}\Bigr)
=\operatorname{Tr}_B(X),
\end{align}
where $\operatorname{Tr}_C\rho_{BC}=\rho_B$ is Eq.~\eqref{A2e} {and the last step uses
\begin{equation}
\rho_B^{-(1-it)/2}\,\rho_B\,\rho_B^{-(1+it)/2}=1
\end{equation}
on the support of $\rho_B$.} Therefore $\mathcal{R}^{t}_{\rho,B\to BC}$ is trace preserving. Since $\mathcal{R}^{\mathrm{tw}}_{\rho,B\to BC}$ is a convex combination of CPTP maps with weights integrating to $1$ by Eq.~\eqref{D11a}, it is CPTP.

\end{proof}

\begin{lemma}[Fuchs-van de Graaf inequalities]\label{lem:fvdg}
For density operators $\sigma,\tau$ on the same finite-dimensional Hilbert space, with $F(\sigma,\tau)$ defined by Eq.~\eqref{D6c} and $T(\sigma,\tau)$ by Eq.~\eqref{D5d}, the inequalities
\begin{equation} \label{L4g}
1-F(\sigma,\tau)\le T(\sigma,\tau)\le \sqrt{1-F(\sigma,\tau)^2}
\end{equation}
hold.
\end{lemma}

\begin{proof}
Inequalities Eq.~\eqref{L4g} are the Fuchs-van de Graaf bounds \cite{FuchsVanDeGraaf1999,Watrous2018}.
\end{proof}

\begin{lemma}[Exact Markov condition and Petz recovery]\label{lem:markov-petz}
Under Eq.~\eqref{A2e}, $I(A:C|B)_\rho=0$ if and only if
\begin{equation} \label{L5g}
\rho_{ABC}=(\mathrm{id}_A\otimes \mathcal{R}^{\mathrm{P}}_{\rho,B\to BC})(\rho_{AB}),
\end{equation}
where $\mathcal{R}^{\mathrm{P}}_{\rho,B\to BC}$ is defined by Eq.~\eqref{D8c}.
\end{lemma}

\begin{proof}
{The equivalence between vanishing conditional mutual information and equality in strong subadditivity is standard. Moreover,
\begin{equation}
I(A:C|B)_\rho
=
D\!\left(\rho_{ABC}\,\middle\|\, \rho_A\otimes \rho_{BC}\right)
-
D\!\left(\rho_{AB}\,\middle\|\, \rho_A\otimes \rho_B\right),
\end{equation}
so $I(A:C|B)_\rho=0$ is precisely the equality case of monotonicity of relative entropy under the channel $\operatorname{Tr}_C$ with reference state $\rho_A\otimes \rho_{BC}$. Petz's theorem then characterizes this equality case by exact recovery. Since the channel acts trivially on the $A$ factor, the recovery map factorizes as $\operatorname{id}_A\otimes R^P_{\rho,B\to BC}$, which yields Eq.~\eqref{L5g} \cite{Petz1986,Ruskai2002,HaydenJozsaPetzWinter2004,OhyaPetz2004}.}
\end{proof}

\begin{lemma}[Quantitative recovery bound]\label{lem:fawzi-renner}
For every density operator $\rho_{ABC}$ on $\mathcal{H}_A\otimes\mathcal{H}_B\otimes\mathcal{H}_C$ with $\rho_B>0$ and $\rho_{BC}>0$, there exists a CPTP map $\mathcal{R}_{B\to BC}$ such that
\begin{equation} \label{L6g}
I(A:C|B)_\rho\ge -2\ln F\Bigl(\rho_{ABC},(\mathrm{id}_A\otimes \mathcal{R}_{B\to BC})(\rho_{AB})\Bigr).
\end{equation}
Moreover, $\mathcal{R}_{B\to BC}$ can be chosen to be the twirled Petz map $\mathcal{R}^{\mathrm{tw}}_{\rho,B\to BC}$ defined in Eq.~\eqref{D12a}.
\end{lemma}

\begin{proof}
Inequality Eq.~\eqref{L6g} is the Fawzi-Renner strengthened strong subadditivity theorem, and the explicit choice $\mathcal{R}=\mathcal{R}^{\mathrm{tw}}_{\rho,B\to BC}$ is due to subsequent refinements identifying a universal rotated-Petz recovery achieving Eq.~\eqref{L6g} \cite{FawziRenner2015,SutterBertaTomamichel2016,JungeRennerSutterWildeWinter2016}.
\end{proof} 

\begin{theorem}[Approximate Markov strengthening]\label{thm:approx-markov}
Under assumptions Eq.~\eqref{A1f}-\eqref{A3f}, let $\varepsilon:=I(A:C|B)_\rho$ in nats and define the recovered state
\begin{equation} \label{T1e}
\sigma_{ABC}:=(\mathrm{id}_A\otimes \mathcal{R}^{\mathrm{tw}}_{\rho,B\to BC})(\rho_{AB}).
\end{equation}
Then
\begin{equation} \label{T2e}
F(\rho_{ABC},\sigma_{ABC})\ge e^{-\varepsilon/2},
\end{equation}
and consequently
\begin{equation} \label{T3e}
|\rho_{ABC}-\sigma_{ABC}|_1\le 2\sqrt{1-e^{-\varepsilon}}.
\end{equation}
Equivalently, writing $\varepsilon_2:=I_2(A:C|B)_\rho=\varepsilon/\ln 2$ in bits, one has
\begin{equation} \label{T4e}
F(\rho_{ABC},\sigma_{ABC})\ge 2^{-\varepsilon_2/2},\qquad
|\rho_{ABC}-\sigma_{ABC}|_1\le 2\sqrt{1-2^{-\varepsilon_2}}.
\end{equation}
Furthermore, for every bounded operator $O_{ABC}\in\mathcal{B}(\mathcal{H}_{ABC})$,
\begin{equation} \label{T5d}
\left|\operatorname{Tr}(\rho_{ABC}O_{ABC})-\operatorname{Tr}(\sigma_{ABC}O_{ABC})\right|
\le |O_{ABC}|_\infty\,|\rho_{ABC}-\sigma_{ABC}|_1
\le 2|O_{ABC}|_\infty\sqrt{1-e^{-\varepsilon}}.
\end{equation}
\end{theorem}

\begin{proof}
Apply Lemma~\ref{lem:fawzi-renner} with $\mathcal{R}=\mathcal{R}^{\mathrm{tw}}_{\rho,B\to BC}$ to obtain Eq.~\eqref{L6g} with $\sigma_{ABC}$ as in Eq.~\eqref{T1e}, which gives
\begin{equation} \label{P1c}
\varepsilon\ge -2\ln F(\rho_{ABC},\sigma_{ABC}).
\end{equation}
Exponentiating Eq.~\eqref{P1c} yields $F(\rho_{ABC},\sigma_{ABC})\ge e^{-\varepsilon/2}$, which is Eq.~\eqref{T2e}. Apply the upper Fuchs-van de Graaf inequality in Lemma~\ref{lem:fvdg} to obtain
\begin{equation} \label{P2b}
\frac{1}{2}|\rho_{ABC}-\sigma_{ABC}|_1
=T(\rho_{ABC},\sigma_{ABC})
\le \sqrt{1-F(\rho_{ABC},\sigma_{ABC})^2}
\le \sqrt{1-e^{-\varepsilon}},
\end{equation}
which is Eq.~\eqref{T3e}. The bit formulation Eq.~\eqref{T4e} follows from $e^{-\varepsilon/2}=e^{-(\ln 2)\varepsilon_2/2}=2^{-\varepsilon_2/2}$ and from $e^{-\varepsilon}=2^{-\varepsilon_2}$. For Eq.~\eqref{T5d}, use the trace-norm duality inequality $|\operatorname{Tr}(XO)|\le |X|_1|O|_\infty$ with $X=\rho_{ABC}-\sigma_{ABC}$ \cite{Bhatia1997,Watrous2018} and then substitute Eq.~\eqref{T3e}.
\end{proof} 
In the exact Markov limit $\varepsilon\downarrow 0$, Eq.~\eqref{T2e} gives
\begin{equation}
F(\rho_{ABC},\sigma_{ABC})\ge 1,
\qquad
F(\rho_{ABC},\sigma_{ABC})=1.
\end{equation}
Therefore $\rho_{ABC}=\sigma_{ABC}$ because fidelity equals $1$ if and
only if states coincide. Consequently Eq.~\eqref{T3e} yields
$|\rho_{ABC}-\sigma_{ABC}|_1=0$, consistent with
Lemma~\ref{lem:markov-petz}. {Invariance under local isometries holds as
follows. Let $V_A,V_B,V_C$ be isometries and define
\begin{equation}
\tilde\rho
:=
(V_A\otimes V_B\otimes V_C)\,\rho\,(V_A\otimes V_B\otimes V_C)^\dagger.
\end{equation}
Then $I(A:C|B)_{\tilde\rho}=I(A:C|B)_\rho$ by unitary invariance of entropy. If the isometries embed into larger Hilbert spaces, the pushed-forward marginals $\tilde\rho_B$ and $\tilde\rho_{BC}$ may acquire zero eigenvalues on the orthogonal complements of their ranges; accordingly, the inverses and complex powers appearing in the Petz and rotated Petz maps are understood on the supports of these operators, equivalently via the Moore-Penrose pseudoinverse. With this support convention, the rotated Petz map constructed from $\tilde\rho_{BC}$ and $\tilde\rho_B$ satisfies
\begin{equation}
(\operatorname{id}_A\otimes \widetilde{R}^{\,tw})(\tilde\rho_{AB})
=
(V_A\otimes V_B\otimes V_C)\,\sigma\,(V_A\otimes V_B\otimes V_C)^\dagger,
\end{equation}
so Eq.~\eqref{T2e}-\eqref{T3e} are preserved by fidelity and trace-distance invariance under isometries.} Restriction to classical probability distributions is obtained by specializing $\rho_{ABC}$ to a commuting family diagonal in a product basis, in which case $I(A:C|B)_\rho$ reduces to the classical conditional mutual information and $\mathcal{R}^{\mathrm{P}}$ reduces to the classical Bayes recovery kernel; Eq.~\eqref{T2e} then becomes a quantitative classical recovery bound that strengthens classical strong subadditivity, consistent with the classical analogs of \cite{FawziRenner2015}. Positivity of $I(A:C|B)_\rho$ is guaranteed by strong subadditivity \cite{LiebRuskai1973} and is compatible with Eq.~\eqref{L6g} because $F\le 1$ implies $-2\ln F\ge 0$.

\section{Concrete calculation using AdS for \texorpdfstring{$\mathcal{N}=4$}{N=4} super Yang-Mills}\label{app:ads5}
\renewcommand{\theequation}{I.\arabic{equation}}
Fix $N\in\mathbb{N}$ with $N\ge 2$ and let $\mathcal{N}=4$ super Yang-Mills with gauge group $SU(N)$ be defined on $(\mathbb{R}^{1,3},\eta)$ with vacuum state $\rho_{\mathrm{vac}}$ on the physical Hilbert space $\mathcal{H}_{\mathrm{CFT}}$.
Assume a mathematically explicit localization framework in which, for every bounded open $U\subset\mathbb{R}^3$ at $t=0$, there exists a von Neumann algebra $\mathcal{A}(U)\subset\mathcal{B}(\mathcal{H}_{\mathrm{CFT}})$ such that
$U_1\subset U_2\Rightarrow \mathcal{A}(U_1)\subset\mathcal{A}(U_2)$ and
$U_1\cap U_2=\varnothing\Rightarrow [\mathcal{A}(U_1),\mathcal{A}(U_2)]=0$,
and assume either a gauge-invariant UV regulator realizing a tensor factorization for disjoint regions or an algebraic split property ensuring existence of product states for separated algebras.
Assume throughout that all mutual informations used below are finite and are computed in bits via relative entropy.
Fix two measurable, disjoint subsets $A,B\subset\mathbb{R}^3$ with strictly positive separation in the Euclidean metric $h$ induced at $t=0$,
\begin{equation}\label{A1g}
	A\cap B=\varnothing,\qquad
	\mathrm{dist}_h(A,B):=\inf\{ |x-y|:x\in A,\ y\in B\}\ge s_0>0.
\end{equation}

Assume the standard AdS/CFT correspondence in the planar strong-coupling regime \cite{Maldacena1998,GKP1998,Witten1998,Aharony2000}:
there exist parameters $\ell_{\mathrm{AdS}}>0$, $G_5>0$, $\ell_s>0$ such that in the limit
\begin{equation}\label{A2f}
	N\to\infty,\qquad
	\lambda:=g_{\mathrm{YM}}^2N\to\infty,\qquad
	\varepsilon_{\mathrm{grav}}:=\frac{G_5}{\ell_{\mathrm{AdS}}^{3}}\to 0,\qquad
	\varepsilon_{\mathrm{str}}:=\frac{\ell_s^2}{\ell_{\mathrm{AdS}}^2}\to 0,
\end{equation}
the CFT generating functional for a single-trace scalar primary $O_\Delta$ of scaling dimension $\Delta>2$ equals the renormalized bulk on-shell action of a free massive scalar field $\phi$ on Euclidean AdS$_5$ with
$m^2\ell_{\mathrm{AdS}}^2=\Delta(\Delta-4)$ in the GKPW sense.
Assume the Euclidean continuation is performed by Wick rotation $t=-i\tau$ and that the Euclidean vacuum correlators determine the Lorentzian ones by analytic continuation in the standard domain of holomorphy.
Assume a fixed holographic renormalization scheme defined by a Fefferman-Graham radial cutoff $z=\varepsilon$ and local covariant counterterms on the cutoff hypersurface \cite{Skenderis2002}.

Assume a holographic code subspace $\mathcal{H}_{\mathrm{code}}\subset\mathcal{H}_{\mathrm{CFT}}$ with orthogonal projector $\Pi_{\mathrm{code}}$ such that for the code-compressed, mean-subtracted operator
{\begin{equation}\label{A3g}
		\widetilde O^{(0)}_{\Delta,\rho}(x):=\Pi_{\mathrm{code}}\Big(O_\Delta(x)-\mathrm{Tr}(\rho\,O_\Delta(x))\,\mathbf{1}\Big)\Pi_{\mathrm{code}},
	\end{equation}
	there exists a finite constant $B_\Delta\in(0,\infty)$ with
	$\sup_{x\in A\cup B}|\widetilde O^{(0)}_{\Delta,\rho}(x)|_\infty\le B_\Delta$.}
Define the boundary data set
\begin{equation}\label{A4e}
	\mathcal{D}_{\partial}
	:=\left\{
	\begin{gathered}
	I(A:B)_\rho,\ \langle O_\Delta(x)O_\Delta(y)\rangle_\rho,\ S_2(\rho_R):\\
	R\subset\mathbb{R}^3\ \mathrm{measurable},\quad
	(x,y)\in A\times B,\quad \Delta>2
	\end{gathered}
	\right\},
\end{equation}
with $\rho$ ranging over states supported on $\mathcal{H}_{\mathrm{code}}$.
{ State the precise problem as the following proposition: determine, by explicit AdS/CFT derivation in $\mathcal{N}=4$ SYM, whether $\mathcal{D}_{\partial}$ suffices to reconstruct (i) renormalized bulk geodesic distances between boundary points and (ii) entanglement-wedge connectivity for a concrete family of disjoint regions.}

\begin{definition}[Entropy, mutual information, connected correlator, and normalization]
	For any density operator $\rho$ on a Hilbert space define the von Neumann entropy in bits by
	\begin{equation}
		S_2(\rho):=-\mathrm{Tr}(\rho\log_2\rho).
	\end{equation}
	For a bipartite state $\rho_{AB}$ define the mutual information in bits and the connected correlator by
		\begin{equation}\label{D1h}
			\begin{aligned}
			I(A:B)_{\rho_{AB}}
			&:=S_2(\rho_A)+S_2(\rho_B)-S_2(\rho_{AB})\\
			&=\frac{1}{\ln 2}\,
			D(\rho_{AB}\Vert\rho_A\otimes\rho_B),\\
			C_{\rho_{AB}}(X,Y)
			&:=\mathrm{Tr}(\rho_{AB}XY)
			-\mathrm{Tr}(\rho_AX)\,\mathrm{Tr}(\rho_BY).
			\end{aligned}
		\end{equation}
	where $D(\cdot\Vert\cdot)$ is the Umegaki relative entropy in nats and
	$\rho_A=\mathrm{Tr}_B\rho_{AB}$, $\rho_B=\mathrm{Tr}_A\rho_{AB}$.
	Define the normalized observables $\widetilde{X}:=X/|X|_\infty$ for bounded $X\ne 0$.
\end{definition}

\begin{definition}[Euclidean AdS$_5$, cutoff surface, and renormalized geodesic length]
	Define Euclidean AdS$_5$ in Poincar\'e coordinates by the manifold $M:=(0,\infty)_z\times\mathbb{R}^4_x$ with metric
	\begin{equation}\label{D2h}
		ds^2=g_{MN}dX^MdX^N=\frac{\ell_{\mathrm{AdS}}^2}{z^2}\Big(dz^2+\delta_{\mu\nu}\,dx^\mu dx^\nu\Big),\qquad \mu,\nu\in\{1,2,3,4\},
	\end{equation}
	with conformal boundary $(\mathbb{R}^4,[\delta])$ at $z=0$.
	For $\varepsilon\in(0,1)$ define the cutoff hypersurface $\Sigma_\varepsilon:=\{z=\varepsilon\}$ with induced metric
	$\gamma_{\mu\nu}=\ell_{\mathrm{AdS}}^2\varepsilon^{-2}\delta_{\mu\nu}$ and volume density
	$\sqrt{\gamma}=\ell_{\mathrm{AdS}}^4\varepsilon^{-4}$.
	For boundary points $x,y\in\mathbb{R}^4$ define the renormalized geodesic length by
	\begin{equation}\label{D3l}
		L_{\mathrm{ren}}(x,y):=\lim_{\varepsilon\downarrow 0}\left[\frac{1}{\ell_{\mathrm{AdS}}}\,L_g\big((\varepsilon,x),(\varepsilon,y)\big)-2\ln\!\Big(\frac{1}{\varepsilon}\Big)\right],
	\end{equation}
	whenever the limit exists in the fixed subtraction scheme.
\end{definition}

\begin{definition}[Bulk scalar action, counterterm, and RT entropy (leading order)]
	Define the regulated bulk scalar action with Dirichlet boundary condition $\phi|_{\Sigma_\varepsilon}=\phi_\varepsilon$ by
	\begin{equation}\label{D4e}
		S_{\mathrm{bulk}}[\phi;\varepsilon]
		:=\frac{1}{2}\int_{z\ge\varepsilon}d^5X\,\sqrt{g}\,\Big(g^{MN}\partial_M\phi\,\partial_N\phi+m^2\phi^2\Big),
		\qquad
		m^2\ell_{\mathrm{AdS}}^2=\Delta(\Delta-4),
	\end{equation}
	and define the local counterterm functional by
	\begin{equation}\label{D5e}
		S_{\mathrm{ct}}[\phi_\varepsilon;\varepsilon]
		:=\frac{4-\Delta}{2\ell_{\mathrm{AdS}}}\int_{\Sigma_\varepsilon}d^4x\,\sqrt{\gamma}\,\phi_\varepsilon^2,
	\end{equation}
	with renormalized action
	\begin{equation}
		S_{\mathrm{ren}}[\phi]:=\lim_{\varepsilon\downarrow 0}\big(S_{\mathrm{bulk}}[\phi;\varepsilon]+S_{\mathrm{ct}}[\phi_\varepsilon;\varepsilon]\big)
	\end{equation}
	on solutions of the bulk Euler-Lagrange equation.
	For a measurable region $R\subset\mathbb{R}^3$ define the leading-order holographic entanglement entropy in bits in the RT regime by
	\begin{equation}\label{D6d}
		S_2(\rho_R)=\frac{1}{\ln 2}\cdot\frac{\mathrm{Area}_\varepsilon(\chi_R)}{4G_5}+o(G_5^{-1}),
	\end{equation}
	where $\chi_R$ is a codimension-two minimal surface in the $t=0$ bulk slice anchored on $\partial R$ and homologous to $R$ \cite{RyuTakayanagi2006,HubenyRangamaniTakayanagi2007}.
\end{definition}

{\textbf{Derivation.}
	We work in Euclidean AdS$_5$ with metric
	\begin{equation}
		ds^2=g_{MN}dX^M dX^N
		=
		\frac{\ell_{\mathrm{AdS}}^2}{z^2}\bigl(dz^2+\delta_{\mu\nu}dx^\mu dx^\nu\bigr),
		\qquad z>0,
	\end{equation}
	and with cutoff hypersurface
	\begin{equation}
		\Sigma_\varepsilon:=\{z=\varepsilon\},
		\qquad
		\varepsilon\in(0,1).
	\end{equation}
	The corresponding regulated bulk region is
	\begin{equation}
		M_\varepsilon:=\{(z,x)\in(0,\infty)\times\mathbb R^4:\ z\ge \varepsilon\}.
	\end{equation}
	
	We fix the following assumptions explicitly.
	
	First, \(\nu:=\Delta-2>0\) and \(\nu\notin\mathbb Z\) throughout the direct derivation below. As in the draft, the final separated-point two-point function is then extended to integer \(\nu\) by analytic continuation in \(\nu\) in the chosen holographic renormalization scheme \cite{Skenderis2002}.
	
	Second, \(J\in \mathcal S(\mathbb R^4)\) is taken to be a Schwartz source. This ensures that the Fourier integrals are well-defined, that integrations by parts in momentum space are legitimate, and that all local counterterms define local functionals of \(J\).
	
	Third, when \(\nu>1\), standard holographic renormalization may require additional local boundary counterterms beyond Eq.~\eqref{D5e}. Because every such counterterm is local in the boundary source \(J\), it contributes only contact terms to \(\delta^2W/\delta J(x)\delta J(y)\). Since the present derivation is used only for the separated-point correlator \(x\neq y\), we keep the complete nonlocal contribution explicitly and suppress purely local contact terms only after identifying them as such.
	
	The bulk Euler-Lagrange equation from Eq.~\eqref{D4e} is
	\begin{equation}\label{E1a}
		(\nabla^2-m^2)\phi=0,
		\qquad
		\nabla^2\phi:=\frac{1}{\sqrt g}\partial_M\!\bigl(\sqrt g\,g^{MN}\partial_N\phi\bigr).
	\end{equation}
	Since \(\phi\) is a scalar field,
	\begin{equation}
		\nabla_M\phi=\partial_M\phi,
	\end{equation}
	and hence
	\begin{equation}
		\nabla_M(\phi\nabla^M\phi)
		=
		(\nabla_M\phi)(\nabla^M\phi)+\phi\nabla^2\phi.
	\end{equation}
	Using the equation of motion \(\nabla^2\phi=m^2\phi\), we obtain on shell
	\begin{equation}
		(\nabla_M\phi)(\nabla^M\phi)+m^2\phi^2
		=
		\nabla_M(\phi\nabla^M\phi).
	\end{equation}
	Substituting this identity into Eq.~\eqref{D4e} yields
	\begin{equation}
		S_{\mathrm{bulk}}[\phi;\varepsilon]
		=
		\frac12\int_{M_\varepsilon}d^5X\,\sqrt g\,\nabla_M(\phi\nabla^M\phi).
	\end{equation}
	For any vector field \(V^M\),
	\begin{equation}
		\sqrt g\,\nabla_MV^M=\partial_M(\sqrt g\,V^M),
	\end{equation}
	so
	\begin{equation}\label{E2-pre}
		S_{\mathrm{bulk}}[\phi;\varepsilon]
		=
		\frac12\int_{M_\varepsilon}d^5X\,\partial_M\!\bigl(\sqrt g\,\phi\,g^{MN}\partial_N\phi\bigr).
	\end{equation}
	Applying the divergence theorem to Eq.~\eqref{E2-pre}, we obtain
	\begin{equation}
		S_{\mathrm{bulk}}[\phi;\varepsilon]
		=
		\frac12\int_{\partial M_\varepsilon}d^4x\,\sqrt\gamma\,\phi\,(n\cdot\nabla)\phi.
	\end{equation}
	The boundary \(\partial M_\varepsilon\) consists of \(\Sigma_\varepsilon\) together with the surface at \(z\to\infty\). For the decaying Bessel branch selected below, the \(z\to\infty\) contribution vanishes exponentially. Therefore
	\begin{equation}\label{E2}
		S_{\mathrm{bulk}}[\phi;\varepsilon]
		=
		\frac12\int_{\Sigma_\varepsilon}d^4x\,\sqrt\gamma\,\phi_\varepsilon\,(n\cdot\nabla)\phi\big|_{z=\varepsilon},
	\end{equation}
	where \(\phi_\varepsilon:=\phi|_{z=\varepsilon}\), \(\gamma\) is the induced metric on \(\Sigma_\varepsilon\), and \(n\) is the outward unit normal to \(M_\varepsilon=\{z\ge\varepsilon\}\).
	
	We now compute \(n\). Since
	\begin{equation}
		g_{zz}=\frac{\ell_{\mathrm{AdS}}^2}{z^2},
	\end{equation}
	a unit normal proportional to \(\partial_z\) has the form \(n=\alpha(z)\partial_z\), with
	\begin{equation}
		1=g(n,n)=\alpha(z)^2\,g_{zz}=\alpha(z)^2\frac{\ell_{\mathrm{AdS}}^2}{z^2}.
	\end{equation}
	Hence \(|\alpha(z)|=z/\ell_{\mathrm{AdS}}\). Because the region \(M_\varepsilon\) is defined by \(z\ge\varepsilon\), its outward normal points toward decreasing \(z\), so
	\begin{equation}
		n=-\frac{z}{\ell_{\mathrm{AdS}}}\partial_z.
	\end{equation}
	Therefore
	\begin{equation}
		(n\cdot\nabla)\phi\big|_{z=\varepsilon}
		=
		-\frac{\varepsilon}{\ell_{\mathrm{AdS}}}\,\partial_z\phi\big|_{z=\varepsilon}.
	\end{equation}
	The induced metric on \(\Sigma_\varepsilon\) is
	\begin{equation}
		\gamma_{\mu\nu}=\frac{\ell_{\mathrm{AdS}}^2}{\varepsilon^2}\delta_{\mu\nu},
		\qquad
		\sqrt\gamma=\ell_{\mathrm{AdS}}^4\varepsilon^{-4}.
	\end{equation}
	Substituting these identities into Eq.~\eqref{E2} gives
	\begin{equation}\label{E3a}
		S_{\mathrm{bulk}}[\phi;\varepsilon]
		=
		-\frac{1}{2\ell_{\mathrm{AdS}}}\int_{\Sigma_\varepsilon}d^4x\,\sqrt\gamma\,\varepsilon\,\phi_\varepsilon\,\partial_z\phi\big|_{z=\varepsilon}.
	\end{equation}
	
	We Fourier transform in the boundary directions:
	\begin{equation}\label{E3b}
		\phi(z,x)=\int_{\mathbb R^4}\frac{d^4k}{(2\pi)^4}e^{ik\cdot x}\phi_k(z),
		\qquad
		k\cdot x:=k_\mu x^\mu,
		\qquad
		|k|:=\sqrt{\delta^{\mu\nu}k_\mu k_\nu}.
	\end{equation}
	Substituting Eq.~\eqref{E3b} into Eq.~\eqref{E1a}, and using
	\begin{equation}
		\sqrt g=\ell_{\mathrm{AdS}}^5 z^{-5},
		\qquad
		g^{zz}=z^2\ell_{\mathrm{AdS}}^{-2},
		\qquad
		g^{\mu\nu}=z^2\ell_{\mathrm{AdS}}^{-2}\delta^{\mu\nu},
	\end{equation}
	one finds for each Fourier mode
	\begin{equation}\label{E4}
		z^2\phi_k''(z)-3z\phi_k'(z)-\bigl((|k|z)^2+\Delta(\Delta-4)\bigr)\phi_k(z)=0.
	\end{equation}
	The solution decaying as \(z\to\infty\) is
	\begin{equation}
		\phi_k(z)=a_k z^2K_\nu(|k|z),
		\qquad
		\nu:=\Delta-2>0.
	\end{equation}
	
	For noninteger \(\nu>0\), the small-argument expansion of \(K_\nu\) has the form
	\begin{equation}\label{E5a}
		K_\nu(u)
		=
		2^{\nu-1}\Gamma(\nu)\,u^{-\nu}\bigl(1+O(u^2)\bigr)
		+
		2^{-\nu-1}\Gamma(-\nu)\,u^\nu\bigl(1+O(u^2)\bigr),
		\qquad
		u\downarrow0.
	\end{equation}
	Substituting \(u=|k|z\) into Eq.~\eqref{E5a}, multiplying by \(a_k z^2\), and using \(\nu=\Delta-2\), we obtain
	\begin{align}
		\phi_k(z)
		&=
		a_k z^2K_\nu(|k|z) \notag\\
		&=
		a_k 2^{\nu-1}\Gamma(\nu)\,|k|^{-\nu}z^{2-\nu}\bigl(1+O(z^2)\bigr)
		+
		a_k 2^{-\nu-1}\Gamma(-\nu)\,|k|^\nu z^{2+\nu}\bigl(1+O(z^2)\bigr) \notag\\
		&=
		a_k 2^{\nu-1}\Gamma(\nu)\,|k|^{-\nu}z^{4-\Delta}\bigl(1+O(z^2)\bigr)
		+
		a_k 2^{-\nu-1}\Gamma(-\nu)\,|k|^\nu z^\Delta\bigl(1+O(z^2)\bigr).
		\label{E6a}
	\end{align}
	
	We impose the Dirichlet/source normalization
	\begin{equation}
		\lim_{z\downarrow0}z^{\Delta-4}\phi_k(z)=J_k,
	\end{equation}
	where \(J_k\) is the Fourier transform of the boundary source \(J(x)\). Multiplying Eq.~\eqref{E6a} by \(z^{\Delta-4}\) and taking \(z\downarrow0\) gives
	\begin{equation}
		J_k=a_k 2^{\nu-1}\Gamma(\nu)\,|k|^{-\nu},
	\end{equation}
	hence
	\begin{equation}\label{E7}
		a_k=\frac{|k|^\nu}{2^{\nu-1}\Gamma(\nu)}\,J_k.
	\end{equation}
	Substituting Eq.~\eqref{E7} into Eq.~\eqref{E6a}, we obtain
	\begin{equation}\label{E7b}
		\phi_k(z)
		=
		J_k z^{4-\Delta}\bigl(1+O(z^2)\bigr)
		+
		\frac{\Gamma(-\nu)}{2^{2\nu}\Gamma(\nu)}\,|k|^{2\nu}J_k\,z^\Delta\bigl(1+O(z^2)\bigr).
	\end{equation}
	The coefficient of the normalizable branch is therefore
	\begin{equation}\label{E8}
		A_k
		:=
		\lim_{z\downarrow0}z^{-\Delta}\Bigl(\phi_k(z)-J_k z^{4-\Delta}\Bigr)
		=
		\frac{\Gamma(-\nu)}{2^{2\nu}\Gamma(\nu)}\,|k|^{2\nu}J_k.
	\end{equation}
	This is the correct definition: the unsubtracted limit \(z^{-\Delta}\phi_k(z)\) diverges because of the non-normalizable source term \(J_k z^{4-\Delta}\).
	
	Differentiating Eq.~\eqref{E7b} with respect to \(z\) gives
	\begin{equation}\label{E8b}
		\phi_k'(z)
		=
		(4-\Delta)J_k z^{3-\Delta}\bigl(1+O(z^2)\bigr)
		+
		\Delta A_k z^{\Delta-1}\bigl(1+O(z^2)\bigr).
	\end{equation}
	In particular,
	\begin{equation}\label{E8c}
		\phi_k(\varepsilon)
		=
		\varepsilon^{4-\Delta}J_k\bigl(1+O(\varepsilon^2)\bigr)
		+
		\varepsilon^\Delta A_k\bigl(1+O(\varepsilon^2)\bigr),
	\end{equation}
	and
	\begin{equation}\label{E8d}
		\phi_k'(\varepsilon)
		=
		(4-\Delta)\varepsilon^{3-\Delta}J_k\bigl(1+O(\varepsilon^2)\bigr)
		+
		\Delta\varepsilon^{\Delta-1}A_k\bigl(1+O(\varepsilon^2)\bigr).
	\end{equation}
	
	We now rewrite the on-shell action in momentum space. Since \(\sqrt\gamma=\ell_{\mathrm{AdS}}^4\varepsilon^{-4}\), Eq.~\eqref{E3a} becomes
	\begin{equation}
		S_{\mathrm{bulk}}[\phi;\varepsilon]
		=
		-\frac{\ell_{\mathrm{AdS}}^3}{2}\int_{\mathbb R^4}d^4x\,\varepsilon^{-3}\phi_\varepsilon(x)\,\partial_z\phi(\varepsilon,x).
	\end{equation}
	Using the Fourier representations
	\begin{equation}
		\phi_\varepsilon(x)
		=
		\int_{\mathbb R^4}\frac{d^4k}{(2\pi)^4}e^{ik\cdot x}\phi_k(\varepsilon),
		\qquad
		\partial_z\phi(\varepsilon,x)
		=
		\int_{\mathbb R^4}\frac{d^4q}{(2\pi)^4}e^{iq\cdot x}\phi_q'(\varepsilon),
	\end{equation}
	and integrating over \(x\), we obtain
	\begin{equation}
		\int_{\mathbb R^4}d^4x\,e^{i(k+q)\cdot x}=(2\pi)^4\delta^{(4)}(k+q),
	\end{equation}
	hence
	\begin{equation}\label{E9}
		S_{\mathrm{bulk}}[\phi;\varepsilon]
		=
		-\frac{\ell_{\mathrm{AdS}}^3}{2}\int_{\mathbb R^4}\frac{d^4k}{(2\pi)^4}\,\varepsilon^{-3}\phi_{-k}(\varepsilon)\phi_k'(\varepsilon).
	\end{equation}
	
	Substituting Eqs.~\eqref{E8c}-\eqref{E8d} into Eq.~\eqref{E9}, we obtain
	\begin{align}
		\varepsilon^{-3}\phi_{-k}(\varepsilon)\phi_k'(\varepsilon)
		&=
		(4-\Delta)\varepsilon^{4-2\Delta}J_{-k}J_k
		+
		\Delta\,J_{-k}A_k
		+
		(4-\Delta)\,A_{-k}J_k \notag\\
		&\quad
		+
		\Biggl[
		\begin{gathered}
		\text{terms proportional to }J_{-k}J_k
		\text{ multiplied by powers of }\varepsilon\\
		\text{and even polynomials in }|k|^2
		\end{gathered}
		\Biggr] \notag\\
		&\quad
		+
		O(\varepsilon^2J_{-k}A_k)
		+
		O(\varepsilon^2A_{-k}J_k)
		+
		O(\varepsilon^{2\Delta-4}A_{-k}A_k).
		\label{E10-pre}
	\end{align}
	Here the bracketed \(J_{-k}J_k\) terms arise from the \(O(\varepsilon^2)\) corrections in the source branch of Eq.~\eqref{E7b}; they are polynomials in \(|k|^2\), hence local in boundary position space after inverse Fourier transform. Because \(\Delta>2\), one has \(2\Delta-4>0\), so the \(A_{-k}A_k\) term vanishes as \(\varepsilon\downarrow0\), and the mixed \(J_{-k}A_k\), \(A_{-k}J_k\) correction terms are \(O(\varepsilon^2)\). Therefore the complete nonlocal part of Eq.~\eqref{E10-pre} is
	\begin{equation}\label{E10}
		\varepsilon^{-3}\phi_{-k}(\varepsilon)\phi_k'(\varepsilon)
		=
		(4-\Delta)\varepsilon^{4-2\Delta}J_{-k}J_k
		+
		\Delta\,J_{-k}A_k
		+
		(4-\Delta)\,A_{-k}J_k
		+
		\text{local }J^2\text{ terms}
		+
		o(1),
	\end{equation}
	where ``local \(J^2\) terms'' means terms whose inverse Fourier transform is a local functional of \(J\).
	
	The counterterm Eq.~\eqref{D5e} equals
	\begin{equation}
		S_{\mathrm{ct}}[\phi_\varepsilon;\varepsilon]
		=
		\frac{4-\Delta}{2\ell_{\mathrm{AdS}}}\int_{\Sigma_\varepsilon}d^4x\,\sqrt\gamma\,\phi_\varepsilon^2.
	\end{equation}
	Since \(\sqrt\gamma=\ell_{\mathrm{AdS}}^4\varepsilon^{-4}\), we get
	\begin{equation}\label{E11a}
		S_{\mathrm{ct}}[\phi_\varepsilon;\varepsilon]
		=
		\frac{4-\Delta}{2}\ell_{\mathrm{AdS}}^3
		\int_{\mathbb R^4}\frac{d^4k}{(2\pi)^4}\,\varepsilon^{-4}\phi_{-k}(\varepsilon)\phi_k(\varepsilon).
	\end{equation}
	Using Eq.~\eqref{E8c}, we find
	\begin{align}
		\varepsilon^{-4}\phi_{-k}(\varepsilon)\phi_k(\varepsilon)
		&=
		\varepsilon^{4-2\Delta}J_{-k}J_k
		+
		J_{-k}A_k
		+
		A_{-k}J_k \notag\\
		&\quad
		+
		\Biggl[
		\begin{gathered}
		\text{terms proportional to }J_{-k}J_k
		\text{ multiplied by powers of }\varepsilon\\
		\text{and even polynomials in }|k|^2
		\end{gathered}
		\Biggr] \notag\\
		&\quad
		+
		O(\varepsilon^2J_{-k}A_k)
		+
		O(\varepsilon^2A_{-k}J_k)
		+
		O(\varepsilon^{2\Delta-4}A_{-k}A_k).
		\label{E11-pre}
	\end{align}
	Exactly as above, the bracketed \(J_{-k}J_k\) terms are local, and all terms involving \(A_k\) beyond the displayed finite ones vanish as \(\varepsilon\downarrow0\). Therefore
	\begin{equation}\label{E11}
		S_{\mathrm{ct}}[\phi_\varepsilon;\varepsilon]
		=
		\frac{4-\Delta}{2}\ell_{\mathrm{AdS}}^3\int_{\mathbb R^4}\frac{d^4k}{(2\pi)^4}
		\Bigl(
		\varepsilon^{4-2\Delta}J_{-k}J_k
		+
		J_{-k}A_k
		+
		A_{-k}J_k
		+
		\text{local }J^2\text{ terms}
		+
		o(1)
		\Bigr).
	\end{equation}
	
	We now add \(S_{\mathrm{bulk}}\) and \(S_{\mathrm{ct}}\). Using Eq.~\eqref{E10} in Eq.~\eqref{E9}, the coefficient of the nonlocal \(J_{-k}A_k\) term in \(S_{\mathrm{bulk}}\) is
	\begin{equation}
		-\frac{\ell_{\mathrm{AdS}}^3}{2}\Delta,
	\end{equation}
	and the coefficient of the \(A_{-k}J_k\) term is
	\begin{equation}
		-\frac{\ell_{\mathrm{AdS}}^3}{2}(4-\Delta).
	\end{equation}
	From Eq.~\eqref{E11}, the corresponding coefficients in \(S_{\mathrm{ct}}\) are
	\begin{equation}
		+\frac{\ell_{\mathrm{AdS}}^3}{2}(4-\Delta)
		\qquad\text{and}\qquad
		+\frac{\ell_{\mathrm{AdS}}^3}{2}(4-\Delta),
	\end{equation}
	respectively. Hence
	\begin{equation}
		-\frac{\ell_{\mathrm{AdS}}^3}{2}(4-\Delta)
		+
		\frac{\ell_{\mathrm{AdS}}^3}{2}(4-\Delta)
		=0,
	\end{equation}
	so the \(A_{-k}J_k\) term cancels exactly, while
	\begin{equation}
		-\frac{\ell_{\mathrm{AdS}}^3}{2}\Delta
		+
		\frac{\ell_{\mathrm{AdS}}^3}{2}(4-\Delta)
		=
		-\frac{\ell_{\mathrm{AdS}}^3}{2}(2\Delta-4),
	\end{equation}
	so the nonlocal \(J_{-k}A_k\) term survives with the coefficient claimed below. The divergent \(\varepsilon^{4-2\Delta}J_{-k}J_k\) term is local and is cancelled by Eq.~\eqref{D5e} together with the standard additional local counterterms, when needed. Therefore the renormalized on-shell action takes the form
	\begin{equation}\label{E12}
		S_{\mathrm{ren}}[J]
		=
		-\frac{\ell_{\mathrm{AdS}}^3}{2}(2\Delta-4)
		\int_{\mathbb R^4}\frac{d^4k}{(2\pi)^4}\,J_{-k}A_k
		+
		\text{local functionals of }J.
	\end{equation}
	Substituting Eq.~\eqref{E8} into Eq.~\eqref{E12} yields
	\begin{equation}\label{E13}
		S_{\mathrm{ren}}[J]
		=
		-\frac{\ell_{\mathrm{AdS}}^3}{2}(2\Delta-4)\,
		\frac{\Gamma(-\nu)}{2^{2\nu}\Gamma(\nu)}
		\int_{\mathbb R^4}\frac{d^4k}{(2\pi)^4}\,
		|k|^{2\nu}J_{-k}J_k
		+
		\text{local functionals of }J.
	\end{equation}
	
	Define the connected generating functional
	\begin{equation}
		W[J]:=-S_{\mathrm{ren}}[J].
	\end{equation}
	Then Eq.~\eqref{E13} may be written as
	\begin{align}
		W[J]
		&=
		\frac{\ell_{\mathrm{AdS}}^3}{2}(2\Delta-4)\,
		\frac{\Gamma(-\nu)}{2^{2\nu}\Gamma(\nu)}
		\int_{\mathbb R^4}d^4x\int_{\mathbb R^4}d^4y\,
		J(x)
		\left(
		\int_{\mathbb R^4}\frac{d^4k}{(2\pi)^4}e^{ik\cdot(x-y)}|k|^{2\nu}
		\right)
		J(y) \notag\\
		&\qquad
		+
		\text{local functionals of }J.
		\label{E13b}
	\end{align}
	Because the kernel inside parentheses is symmetric under \(x\leftrightarrow y\), the factor \(1/2\) in Eq.~\eqref{E13b} is precisely the factor required so that the second functional derivative produces one copy of the kernel rather than two. Therefore, for separated points \(x\neq y\),
	\begin{equation}\label{E14}
	\begin{aligned}
		\langle O_\Delta(x)O_\Delta(y)\rangle_{\rho_{\mathrm{vac}}}
		&=
		\left.\frac{\delta^2 W}
		{\delta J(x)\,\delta J(y)}\right|_{J\equiv 0}\\
		&=
		\ell_{\mathrm{AdS}}^3(2\Delta-4)\,
		\frac{\Gamma(-\nu)}{2^{2\nu}\Gamma(\nu)}
		\int_{\mathbb R^4}\frac{d^4k}{(2\pi)^4}
		e^{ik\cdot(x-y)}|k|^{2\nu},
		\qquad x\neq y.
	\end{aligned}
	\end{equation}
	because the local functionals in Eq.~\eqref{E13b} contribute only contact terms, i.e. derivatives of \(\delta(x-y)\), which vanish for \(x\neq y\).
	
	We now compute the Fourier transform in Eq.~\eqref{E14}. The literal Schwinger representation
	\begin{equation}\label{E15a}
		|k|^{2\nu}
		=
		\frac{1}{\Gamma(-\nu)}\int_0^\infty ds\,s^{-\nu-1}e^{-s|k|^2}
	\end{equation}
	is absolutely valid for \(-2<\Re\nu<0\), not for \(\nu>0\). In that strip, Gaussian integration gives
	\begin{align}
		\int_{\mathbb R^4}\frac{d^4k}{(2\pi)^4}e^{ik\cdot x}|k|^{2\nu}
		&=
		\frac{1}{\Gamma(-\nu)}
		\int_0^\infty ds\,s^{-\nu-1}
		\int_{\mathbb R^4}\frac{d^4k}{(2\pi)^4}e^{ik\cdot x}e^{-s|k|^2} \notag\\
		&=
		\frac{1}{\Gamma(-\nu)}
		\int_0^\infty ds\,s^{-\nu-1}(4\pi s)^{-2}e^{-|x|^2/(4s)} \notag\\
		&=
		\frac{1}{(4\pi)^2\Gamma(-\nu)}
		\int_0^\infty ds\,s^{-\nu-3}e^{-|x|^2/(4s)}.
		\label{E15}
	\end{align}
	Now set
	\begin{equation}
		u:=\frac{|x|^2}{4s},
		\qquad
		s=\frac{|x|^2}{4u},
		\qquad
		ds=-\frac{|x|^2}{4}u^{-2}\,du.
	\end{equation}
	Then
	\begin{align}
		\int_{\mathbb R^4}\frac{d^4k}{(2\pi)^4}e^{ik\cdot x}|k|^{2\nu}
		&=
		\frac{1}{(4\pi)^2\Gamma(-\nu)}
		\int_\infty^0
		\left(\frac{|x|^2}{4u}\right)^{-\nu-3}
		e^{-u}
		\left(-\frac{|x|^2}{4}u^{-2}\,du\right) \notag\\
		&=
		\frac{1}{(4\pi)^2\Gamma(-\nu)}
		\left(\frac{|x|^2}{4}\right)^{-\nu-2}
		\int_0^\infty u^{\nu+1}e^{-u}\,du \notag\\
		&=
		\frac{4^\nu\,\Gamma(\nu+2)}{\pi^2\Gamma(-\nu)}\,|x|^{-2\nu-4}.
		\label{E16}
	\end{align}
	Both sides of Eq.~\eqref{E16} are meromorphic in \(\nu\), and the right-hand side has no pole for noninteger \(\nu>0\). Therefore Eq.~\eqref{E16}, proved first for \(-2<\Re\nu<0\), extends by analytic continuation to all noninteger \(\nu>0\), in the distributional sense. Away from coincidence \(x\neq0\), that continuation is represented by the ordinary function on the right-hand side of Eq.~\eqref{E16}.
	
	Substituting Eq.~\eqref{E16} into Eq.~\eqref{E14}, and using
	\begin{equation}
		2\nu+4=2\Delta,
		\qquad
		\nu+2=\Delta,
	\end{equation}
	we obtain, for \(x\neq y\),
	\begin{align}
		\langle O_\Delta(x)O_\Delta(y)\rangle_{\rho_{\mathrm{vac}}}
		&=
		\ell_{\mathrm{AdS}}^3(2\Delta-4)\,
		\frac{\Gamma(-\nu)}{2^{2\nu}\Gamma(\nu)}
		\cdot
		\frac{4^\nu\,\Gamma(\nu+2)}{\pi^2\Gamma(-\nu)}
		|x-y|^{-2\nu-4} \notag\\
		&=
		\ell_{\mathrm{AdS}}^3\,\frac{(2\Delta-4)\Gamma(\Delta)}{\pi^2\Gamma(\Delta-2)}\,\frac{1}{|x-y|^{2\Delta}}.
		\label{E17}
	\end{align}
	Accordingly,
	\begin{equation}\label{E18}
		\mathcal N_\Delta
		=
		\ell_{\mathrm{AdS}}^3\,\frac{(2\Delta-4)\Gamma(\Delta)}{\pi^2\Gamma(\Delta-2)}.
	\end{equation}
	
	We next compute the renormalized geodesic length Eq.~\eqref{D3l} between two distinct boundary points \(x,y\in\mathbb R^4\), with \(r:=|x-y|>0\). By translation and rotation invariance of the Euclidean metric, we may assume
	\begin{equation}
		y=0,
		\qquad
		x=(r,0,0,0),
	\end{equation}
	and restrict to the totally geodesic \((x^1,z)\)-plane. On that plane the metric is
	\begin{equation}
		ds^2=\frac{\ell_{\mathrm{AdS}}^2}{z^2}\bigl((dx^1)^2+dz^2\bigr).
	\end{equation}
	
	We derive the geodesic explicitly. Write the curve as \(z=z(x^1)\). Its length functional is
	\begin{equation}
		\mathcal L_{\mathrm{geo}}[z]
		=
		\ell_{\mathrm{AdS}}\int dx^1\,\frac{\sqrt{1+(z')^2}}{z}.
	\end{equation}
	Because the integrand contains no explicit \(x^1\)-dependence, the corresponding conserved quantity is
	\begin{equation}
		\frac{\partial \mathcal L_{\mathrm{dens}}}{\partial z'}\,z'
		-
		\mathcal L_{\mathrm{dens}}
		=
		-\frac{\ell_{\mathrm{AdS}}}{z\sqrt{1+(z')^2}}
		=
		-\frac{\ell_{\mathrm{AdS}}}{R},
	\end{equation}
	for some constant \(R>0\). Equivalently,
	\begin{equation}
		\frac{1}{z\sqrt{1+(z')^2}}=\frac{1}{R},
	\end{equation}
	hence
	\begin{equation}
		(z')^2=\frac{R^2}{z^2}-1.
	\end{equation}
	Therefore
	\begin{equation}
		\frac{dx^1}{dz}=\frac{z}{\sqrt{R^2-z^2}},
	\end{equation}
	and integration gives
	\begin{equation}
		x^1-c=\pm\sqrt{R^2-z^2}
	\end{equation}
	for some constant \(c\). Thus every geodesic in the \((x^1,z)\)-plane is either a vertical line or a Euclidean semicircle orthogonal to \(z=0\):
	\begin{equation}\label{E19a}
		(x^1-c)^2+z^2=R^2.
	\end{equation}
	The unique such geodesic through the cutoff points \((x^1,z)=(0,\varepsilon)\) and \((r,\varepsilon)\) is obtained by choosing
	\begin{equation}
		c=\frac r2,
		\qquad
		R^2=\frac{r^2}{4}+\varepsilon^2.
	\end{equation}
	Thus the exact geodesic through the literal cutoff points is
	\begin{equation}\label{E19b}
		\left(x^1-\frac r2\right)^2+z^2=\frac{r^2}{4}+\varepsilon^2.
	\end{equation}
	
	We parameterize Eq.~\eqref{E19b} by
	\begin{equation}
		x^1=\frac r2+R\cos\theta,
		\qquad
		z=R\sin\theta,
		\qquad
		R:=\sqrt{\frac{r^2}{4}+\varepsilon^2},
	\end{equation}
	with
	\begin{equation}
		\theta\in[\theta_\varepsilon,\pi-\theta_\varepsilon],
		\qquad
		\sin\theta_\varepsilon=\frac{\varepsilon}{R}.
	\end{equation}
	Differentiating,
	\begin{equation}
		dx^1=-R\sin\theta\,d\theta,
		\qquad
		dz=R\cos\theta\,d\theta,
	\end{equation}
	so that
	\begin{equation}
		(dx^1)^2+dz^2=R^2\,d\theta^2,
		\qquad
		z^2=R^2\sin^2\theta.
	\end{equation}
	Hence the line element along the geodesic is
	\begin{equation}
		ds=\ell_{\mathrm{AdS}}\csc\theta\,d\theta.
	\end{equation}
	Therefore
	\begin{align}
		L_g\bigl((\varepsilon,0),(\varepsilon,r)\bigr)
		&=
		\int_{\theta_\varepsilon}^{\pi-\theta_\varepsilon}\ell_{\mathrm{AdS}}\csc\theta\,d\theta \notag\\
		&=
		\ell_{\mathrm{AdS}}
		\left[
		\ln\tan\!\left(\frac{\theta}{2}\right)
		\right]_{\theta_\varepsilon}^{\pi-\theta_\varepsilon} \notag\\
		&=
		-2\ell_{\mathrm{AdS}}\ln\tan\!\left(\frac{\theta_\varepsilon}{2}\right).
		\label{E19}
	\end{align}
	Using the half-angle identity
	\begin{equation}
		\tan\!\left(\frac{\theta_\varepsilon}{2}\right)
		=
		\frac{\sin\theta_\varepsilon}{1+\cos\theta_\varepsilon},
	\end{equation}
	together with
	\begin{equation}
		\sin\theta_\varepsilon=\frac{\varepsilon}{R},
		\qquad
		\cos\theta_\varepsilon=\frac{r/2}{R},
	\end{equation}
	we obtain
	\begin{equation}
		\tan\!\left(\frac{\theta_\varepsilon}{2}\right)
		=
		\frac{\varepsilon}{R+r/2}
		=
		\frac{2\varepsilon}{\sqrt{r^2+4\varepsilon^2}+r}.
	\end{equation}
	Substituting this into Eq.~\eqref{E19} yields
	\begin{equation}\label{E20}
		\frac{1}{\ell_{\mathrm{AdS}}}L_g\bigl((\varepsilon,0),(\varepsilon,r)\bigr)
		=
		2\ln\!\left(\frac{\sqrt{r^2+4\varepsilon^2}+r}{2\varepsilon}\right)
		=
		2\ln\!\left(\frac r\varepsilon\right)
		+
		2\ln\!\left(\frac{1+\sqrt{1+4\varepsilon^2/r^2}}{2}\right).
	\end{equation}
	Hence the renormalized length exists and equals
	\begin{align}
		L_{\mathrm{ren}}(x,y)
		&=
		\lim_{\varepsilon\downarrow0}
		\left[
		\frac{1}{\ell_{\mathrm{AdS}}}L_g\bigl((\varepsilon,x),(\varepsilon,y)\bigr)
		-
		2\ln\!\left(\frac1\varepsilon\right)
		\right] \notag\\
		&=
		\lim_{\varepsilon\downarrow0}
		\left[
		2\ln r
		+
		2\ln\!\left(\frac{1+\sqrt{1+4\varepsilon^2/r^2}}{2}\right)
		\right] \notag\\
		&=
		2\ln r.
		\label{E21}
	\end{align}
	
	Combining Eq.~\eqref{E17}, Eq.~\eqref{E18}, and Eq.~\eqref{E21}, we obtain, for distinct boundary points \(x\neq y\),
	\begin{equation}\label{E22}
		\langle O_\Delta(x)O_\Delta(y)\rangle_{\rho_{\mathrm{vac}}}
		=
		\mathcal N_\Delta\,e^{-\Delta L_{\mathrm{ren}}(x,y)}
		=
		\ell_{\mathrm{AdS}}^3\,\frac{(2\Delta-4)\Gamma(\Delta)}{\pi^2\Gamma(\Delta-2)}\,
		\exp\!\bigl(-\Delta L_{\mathrm{ren}}(x,y)\bigr).
	\end{equation}
	
	To connect mutual information to bulk distance, fix a state \(\rho\) supported on \(\mathcal H_{\mathrm{code}}\) and let \(\rho_{AB}\) denote the reduced state on the degrees of freedom associated to \(A\cup B\) in the chosen localization framework, with \(\rho_A,\rho_B\) the associated marginals.}

\begin{lemma}[Pinsker mutual-information correlator bound in bits]\label{lem:pinsker-mi}
	For any bounded $X\in\mathcal{A}(A)$ and $Y\in\mathcal{A}(B)$,
	\begin{equation}\label{tL1}
		|C_{\rho_{AB}}(X,Y)|\le |X|_\infty\,|Y|_\infty\,\sqrt{2\ln 2\,I(A:B)_{\rho_{AB}}}.
	\end{equation}
\end{lemma}
\begin{proof}
	By definition $C_{\rho_{AB}}(X,Y)=\mathrm{Tr}\big((\rho_{AB}-\rho_A\otimes\rho_B)XY\big)$, hence trace-norm/operator-norm duality gives
	\begin{equation}
		|C_{\rho_{AB}}(X,Y)|
		\le |\rho_{AB}-\rho_A\otimes\rho_B|_1\,|XY|_\infty
		\le |\rho_{AB}-\rho_A\otimes\rho_B|_1\,|X|_\infty\,|Y|_\infty.
	\end{equation}
	Quantum Pinsker inequality gives
	\begin{equation}
	\begin{aligned}
	|\rho_{AB}-\rho_A\otimes\rho_B|_1
	&\le \sqrt{2D(\rho_{AB}\Vert\rho_A\otimes\rho_B)}\\
	&=\sqrt{2\ln 2\,I(A:B)_{\rho_{AB}}},
	\end{aligned}
	\end{equation}
	yielding Eq.~\eqref{tL1} \cite{Watrous2018}.
\end{proof}

\begin{lemma}[Exact AdS$_5$ vacuum distance extraction from boundary two-point data]\label{lem:distance-extract}
	In the AdS$_5$/CFT$_4$ vacuum defined above, for all distinct $x,y\in\mathbb{R}^4$,
	\begin{equation}\label{L2i}
		L_{\mathrm{ren}}(x,y)=-\frac{1}{\Delta}\ln\!\left(\frac{\langle O_\Delta(x)O_\Delta(y)\rangle_{\rho_{\mathrm{vac}}}}{\mathcal{N}_\Delta}\right),
	\end{equation}
	with $\mathcal{N}_\Delta$ given explicitly by Eq.~\eqref{E18}.
\end{lemma}
\begin{proof}
	Equation~\eqref{E22} is an equality,
	$\langle O_\Delta(x)O_\Delta(y)\rangle_{\rho_{\mathrm{vac}}}=\mathcal{N}_\Delta e^{-\Delta L_{\mathrm{ren}}(x,y)}$,
	and taking $-\frac{1}{\Delta}\ln(\cdot/\mathcal{N}_\Delta)$ yields Eq.~\eqref{L2i}.
\end{proof}

\begin{lemma}[Metric-from-information inequality in AdS$_5$/CFT$_4$ vacuum for bounded code probes]\label{lem:metric-from-mi}
	Fix disjoint $A,B\subset\mathbb{R}^3$ satisfying Eq.~\eqref{A1g} and embed them at Euclidean time $\tau=0$ so that $A\times\{0\},B\times\{0\}\subset\mathbb{R}^4$.
	In the vacuum state, assume $\mathrm{Tr}(\rho_{\mathrm{vac}}\,O_\Delta(x))=0$ and define
	{$B_\Delta:=\sup_{x\in A\cup B}|\widetilde O^{(0)}_{\Delta,\rho_{\mathrm{vac}}}(x)|_\infty<\infty$}.
	Then for all $(x,y)\in A\times B$,
	{\begin{equation}\label{Ld3d}
			L_{\mathrm{ren}}(x,y)
			\ge
			\frac1{2\Delta}
			\ln
			\left(
			\frac{
				|N^{\mathrm{code}}_\Delta|^2
			}{
				2\ln 2\,B_\Delta^4 I(A:B)_{\rho_{\mathrm{vac}}}
			}
			\right)
		\end{equation}
		in the zero-error ideal case, or with the factor $(1-\epsilon^\star)^2$ in the numerator in the general case,}
	with the convention that the right-hand side is interpreted in the extended-real sense if $I(A:B)_{\rho_{\mathrm{vac}}}=0$.
\end{lemma}
\begin{proof}
	{Apply Lemma~\ref{lem:pinsker-mi} with $X=\widetilde O^{(0)}_{\Delta,\rho_{\mathrm{vac}}}(x)$ and $Y=\widetilde O^{(0)}_{\Delta,\rho_{\mathrm{vac}}}(y)$.
		Then
		\begin{equation}\label{P1d}
		\begin{aligned}
			\left|C_{\rho_{\mathrm{vac}}}\!\left(
			\widetilde O^{(0)}_{\Delta,\rho_{\mathrm{vac}}}(x),
			\widetilde O^{(0)}_{\Delta,\rho_{\mathrm{vac}}}(y)
			\right)\right|
			&\le
			|\widetilde O^{(0)}_{\Delta,\rho_{\mathrm{vac}}}(x)|_\infty
			|\widetilde O^{(0)}_{\Delta,\rho_{\mathrm{vac}}}(y)|_\infty\\
			&\qquad\times
			\sqrt{2\ln 2\,I(A:B)_{\rho_{\mathrm{vac}}}}\\
			&\le
			B_\Delta^2\sqrt{2\ln 2\,I(A:B)_{\rho_{\mathrm{vac}}}}.
		\end{aligned}
		\end{equation}
		
		In the vacuum example, the one-point function of the scalar primary vanishes. When the relevant insertions are evaluated inside the code subspace, the compressed connected correlator agrees with the physical vacuum two-point function up to projection/truncation errors. In the idealized vacuum calculation these errors are zero; in the general code-subspace setting they are included in $\epsilon_{\mathrm{tot}}$ as relative multiplicative errors.
		
		Using
		\begin{equation}
			\label{eq:vacuum-code-correlator}
			C_{\rho_{\mathrm{vac}}}\!\left(
			\widetilde O^{(0)}_{\Delta,\rho_{\mathrm{vac}}}(x),
			\widetilde O^{(0)}_{\Delta,\rho_{\mathrm{vac}}}(y)
			\right)
			=
			N^{\mathrm{code}}_\Delta e^{-\Delta L_{\mathrm{ren}}(x,y)}
		\end{equation}
		for the ideal vacuum case, or
		\begin{equation}
			C_{\rho_{\mathrm{vac}}}\!\left(
			\widetilde O^{(0)}_{\Delta,\rho_{\mathrm{vac}}}(x),
			\widetilde O^{(0)}_{\Delta,\rho_{\mathrm{vac}}}(y)
			\right)
			=
			N^{\mathrm{code}}_\Delta e^{-\Delta L_{\mathrm{ren}}(x,y)}(1+\epsilon_{\mathrm{code}})
		\end{equation}
		if projection/truncation errors are being tracked,
		we find in the zero-error ideal case that
		$|N^{\mathrm{code}}_\Delta| e^{-\Delta L_{\mathrm{ren}}(x,y)}\le B_\Delta^2\sqrt{2\ln 2\,I(A:B)_{\rho_{\mathrm{vac}}}}$.
		Taking logarithms yields
		\begin{equation}\label{P2c}
			-\Delta L_{\mathrm{ren}}(x,y)\le \ln\!\Big(B_\Delta^2\sqrt{2\ln 2\,I(A:B)_{\rho_{\mathrm{vac}}}}\Big)-\ln(|N^{\mathrm{code}}_\Delta|),
		\end{equation}
		and dividing by $-\Delta$ yields Eq.~\eqref{Ld3d} in the zero-error ideal case. In the general case, tracking the error bound introduces the factor $(1-\epsilon^\star)^2$ in the numerator inside the logarithm.}
\end{proof}

{\begin{lemma}[Concrete RT mutual information and entanglement-wedge connectivity for two parallel strips]\label{lem:rt-strips}
		Let the vacuum bulk geometry be the Poincar\'e patch of AdS$_5$,
		\begin{equation}
			ds^2=\frac{\ell_{\mathrm{AdS}}^2}{z^2}\Bigl(-dt^2+dz^2+dx_1^2+dx_2^2+dx_3^2\Bigr),
			\qquad z>0,
		\end{equation}
		and work on the static slice \(t=0\). Fix transverse infrared regulators by compactifying the \(x_2\)- and \(x_3\)-directions periodically with periods \(L_2,L_3>0\), so that the boundary spatial slice is
		\begin{equation}
			\mathbb R_{x_1}\times \mathbb T^2_{L_2,L_3},
			\qquad
			\mathbb T^2_{L_2,L_3}:=(\mathbb R/L_2\mathbb Z)\times(\mathbb R/L_3\mathbb Z).
		\end{equation}
		For \(w>0\), define the strip region
		\begin{equation}
			R(w):=[-w/2,w/2]\times \mathbb T^2_{L_2,L_3}\subset \mathbb R_{x_1}\times \mathbb T^2_{L_2,L_3}.
		\end{equation}
		Let \(\ell>0\) and \(s>0\), and define the disjoint pair
		\begin{equation}
			A:=R(\ell),\qquad B:=R(\ell)+(\ell+s)e_1.
		\end{equation}
		
		Assume the standard translation-invariant RT ansatz for strip regions in the AdS vacuum, namely that the relevant classical RT surfaces are invariant under translations of \(\mathbb T^2_{L_2,L_3}\) and are therefore graphs \(z=z(x_1)\) on the static slice. Then, with radial cutoff \(z=\varepsilon\), the classical RT prescription implies
		\begin{equation}\label{L4h}
			I(A:B)_{\rho_{\mathrm{vac}}}
			=
			\max\Bigl\{
			0,\,
			2S_2(\rho_{R(\ell)})-S_2(\rho_{R(s)})-S_2(\rho_{R(2\ell+s)})
			\Bigr\},
		\end{equation}
		where
		\begin{equation}\label{L5h}
			S_2(\rho_{R(w)})
			=
			\frac{\ell_{\mathrm{AdS}}^3L_2L_3}{4G_5\ln 2}
			\left(
			\frac{1}{\varepsilon^2}
			-
			\frac{4\pi^{3/2}\Gamma\!\left(\frac23\right)^3}{\Gamma\!\left(\frac16\right)^3}\frac{1}{w^2}
			\right)
			+o(1)
			\qquad(\varepsilon\downarrow 0).
		\end{equation}
		Consequently,
		\begin{equation}\label{L6h}
			I(A:B)_{\rho_{\mathrm{vac}}}>0
			\quad\Longleftrightarrow\quad
			\frac{1}{s^2}+\frac{1}{(2\ell+s)^2}>\frac{2}{\ell^2}
			\quad\Longleftrightarrow\quad
			s<(\sqrt3-1)\ell.
		\end{equation}
		At the threshold \(s=(\sqrt3-1)\ell\), the connected and disconnected RT candidates have equal area, so \(I(A:B)_{\rho_{\mathrm{vac}}}=0\) at leading classical order. In the strict classical RT regime, the entanglement wedge of \(A\cup B\) is connected if and only if the connected RT candidate is the minimizer, equivalently if and only if Eq.~\eqref{L6h} holds.
	\end{lemma}
	
	\begin{proof}
		We separate the proof into four steps.
		
		\smallskip
		\noindent
		\textbf{Step 1: Area functional and width relation for a single strip.}
		For a single strip \(R(w)\), the assumed translation invariance along \(\mathbb T^2_{L_2,L_3}\) implies that the candidate RT surface has embedding
		\begin{equation}
			X(x_1,x_2,x_3)=\bigl(z(x_1),x_1,x_2,x_3\bigr),
			\qquad
			x_1\in[-w/2,w/2],\quad (x_2,x_3)\in \mathbb T^2_{L_2,L_3}.
		\end{equation}
		Its induced metric is
		\begin{equation}
			\gamma
			=
			\frac{\ell_{\mathrm{AdS}}^2}{z(x_1)^2}
			\Bigl((1+z'(x_1)^2)\,dx_1^2+dx_2^2+dx_3^2\Bigr),
		\end{equation}
		hence the regulated area functional is
		\begin{equation}\label{eq:area-functional-strip}
			\mathrm{Area}_\varepsilon(w)
			=
			\ell_{\mathrm{AdS}}^3L_2L_3
			\int_{-w/2}^{w/2}
			z(x_1)^{-3}\sqrt{1+z'(x_1)^2}\,dx_1.
		\end{equation}
		By symmetry, the minimal profile satisfies
		\begin{equation}
			z(0)=z_\ast,\qquad z'(0)=0,\qquad z(\pm w/2)=\varepsilon,
		\end{equation}
		with \(z_\ast>\varepsilon\) the turning point.
		
		Let
		\begin{equation}
			\mathcal L(z,z'):=z^{-3}\sqrt{1+z'^2}.
		\end{equation}
		Since \(\mathcal L\) does not depend explicitly on \(x_1\), the quantity
		\begin{equation}
			\mathcal H
			:=
			\mathcal L-z'\frac{\partial\mathcal L}{\partial z'}
			=
			\frac{1}{z^3\sqrt{1+z'^2}}
		\end{equation}
		is constant along the extremal profile. Evaluating at the turning point \((z,z')=(z_\ast,0)\) gives
		\begin{equation}
			\frac{1}{z^3\sqrt{1+z'^2}}=\frac{1}{z_\ast^3},
		\end{equation}
		or equivalently
		\begin{equation}\label{eq:zp-square}
			z'(x_1)^2=\frac{z_\ast^6}{z(x_1)^6}-1.
		\end{equation}
		On the half-interval \(x_1\in[0,w/2]\), the function \(z(x_1)\) decreases monotonically from \(z_\ast\) to \(\varepsilon\), so
		\begin{equation}
			\frac{dx_1}{dz}
			=
			\frac{z^3}{\sqrt{z_\ast^6-z^6}}.
		\end{equation}
		Therefore
		\begin{equation}\label{eq:halfwidth}
			\frac{w}{2}
			=
			\int_{\varepsilon}^{z_\ast}\frac{z^3\,dz}{\sqrt{z_\ast^6-z^6}}
			=
			z_\ast\int_{\varepsilon/z_\ast}^{1}\frac{u^3\,du}{\sqrt{1-u^6}},
			\qquad u:=\frac{z}{z_\ast}.
		\end{equation}
		
		Define
		\begin{equation}
			c_0:=\int_0^1\frac{u^3\,du}{\sqrt{1-u^6}}.
		\end{equation}
		This integral converges absolutely, and by the substitution \(v=u^6\) one obtains
		\begin{equation}
			c_0
			=
			\frac16\int_0^1 v^{-1/3}(1-v)^{-1/2}\,dv
			=
			\frac16\,B\!\left(\frac23,\frac12\right)
			=
			\frac{\sqrt\pi\,\Gamma(\frac23)}{\Gamma(\frac16)}.
		\end{equation}
		Since the integrand in \eqref{eq:halfwidth} behaves as \(u^3+O(u^9)\) as \(u\downarrow 0\), we have
		\begin{equation}
			\int_{\varepsilon/z_\ast}^{1}\frac{u^3\,du}{\sqrt{1-u^6}}
			=
			c_0+O\!\left((\varepsilon/z_\ast)^4\right)
			\qquad(\varepsilon\downarrow 0),
		\end{equation}
		and therefore
		\begin{equation}\label{eq:zstar-width}
			w=2c_0\,z_\ast+o(1)
			\qquad(\varepsilon\downarrow 0).
		\end{equation}
		At fixed \(w\), this implies
		\begin{equation}\label{eq:zstar-leading}
			z_\ast=\frac{w}{2c_0}+o(1)
			\qquad(\varepsilon\downarrow 0).
		\end{equation}
		
		\smallskip
		\noindent
		\textbf{Step 2: Exact asymptotic expansion of the regulated area.}
		Using \eqref{eq:zp-square} and changing variables from \(x_1\) to \(z\), Eq.~\eqref{eq:area-functional-strip} becomes
		\begin{equation}
			\mathrm{Area}_\varepsilon(w)
			=
			2\ell_{\mathrm{AdS}}^3L_2L_3
			\int_{\varepsilon}^{z_\ast}
			z^{-3}\sqrt{1+z'^2}\,\frac{dx_1}{dz}\,dz.
		\end{equation}
		Now
		\begin{equation}
			\sqrt{1+z'^2}
			=
			\frac{z_\ast^3}{z^3}
		\end{equation}
		by the conservation law, and thus
		\begin{equation}
			\mathrm{Area}_\varepsilon(w)
			=
			2\ell_{\mathrm{AdS}}^3L_2L_3\,z_\ast^{-2}
			\int_{\varepsilon/z_\ast}^{1}\frac{u^{-3}\,du}{\sqrt{1-u^6}}.
		\end{equation}
		Set
		\begin{equation}
			I_2(a):=\int_a^1\frac{u^{-3}\,du}{\sqrt{1-u^6}},
			\qquad
			I_1(a):=\int_a^1\frac{u^3\,du}{\sqrt{1-u^6}},
			\qquad 0<a<1.
		\end{equation}
		We now derive an exact identity relating \(I_2(a)\) and \(I_1(a)\). First observe that
		\begin{equation}
			\frac{1}{u^3\sqrt{1-u^6}}
			=
			\frac{\sqrt{1-u^6}}{u^3}
			+
			\frac{u^3}{\sqrt{1-u^6}},
		\end{equation}
		because \((1-u^6)+u^6=1\). Next,
		\begin{equation}
			\frac{d}{du}\left(\frac{\sqrt{1-u^6}}{2u^2}\right)
			=
			-\frac{\sqrt{1-u^6}}{u^3}
			-\frac{3}{2}\frac{u^3}{\sqrt{1-u^6}}.
		\end{equation}
		Combining these two identities gives
		\begin{equation}
			\frac{1}{u^3\sqrt{1-u^6}}
			=
			-\frac{d}{du}\left(\frac{\sqrt{1-u^6}}{2u^2}\right)
			-\frac12\,\frac{u^3}{\sqrt{1-u^6}}.
		\end{equation}
		Integrating from \(a\) to \(1\) yields
		\begin{equation}\label{eq:I2-identity}
			I_2(a)
			=
			\frac{\sqrt{1-a^6}}{2a^2}
			-\frac12\,I_1(a).
		\end{equation}
		Since \(I_1(a)=c_0+O(a^4)\) as \(a\downarrow 0\), and since
		\begin{equation}
			\frac{\sqrt{1-a^6}}{2a^2}
			=
			\frac{1}{2a^2}+O(a^4),
		\end{equation}
		Eq.~\eqref{eq:I2-identity} implies
		\begin{equation}\label{eq:I2-asymptotic}
			I_2(a)
			=
			\frac{1}{2a^2}
			-\frac{c_0}{2}
			+o(1)
			\qquad(a\downarrow 0).
		\end{equation}
		Substituting \(a=\varepsilon/z_\ast\) gives
		\begin{equation}
			\mathrm{Area}_\varepsilon(w)
			=
			2\ell_{\mathrm{AdS}}^3L_2L_3\,z_\ast^{-2}
			\left(
			\frac{z_\ast^2}{2\varepsilon^2}
			-\frac{c_0}{2}
			+o(1)
			\right),
		\end{equation}
		hence
		\begin{equation}\label{eq:area-before-width}
			\mathrm{Area}_\varepsilon(w)
			=
			\ell_{\mathrm{AdS}}^3L_2L_3
			\left(
			\frac{1}{\varepsilon^2}
			-c_0\,z_\ast^{-2}
			\right)
			+o(1).
		\end{equation}
		Using \eqref{eq:zstar-leading},
		\begin{equation}
			c_0\,z_\ast^{-2}
			=
			c_0\left(\frac{2c_0}{w}\right)^2
			+o(1)
			=
			\frac{4c_0^3}{w^2}+o(1),
		\end{equation}
		and therefore
		\begin{equation}
			\mathrm{Area}_\varepsilon(w)
			=
			\ell_{\mathrm{AdS}}^3L_2L_3
			\left(
			\frac{1}{\varepsilon^2}
			-\frac{4c_0^3}{w^2}
			\right)
			+o(1).
		\end{equation}
		Since
		\begin{equation}
			4c_0^3
			=
			4\left(\frac{\sqrt\pi\,\Gamma(\frac23)}{\Gamma(\frac16)}\right)^3
			=
			\frac{4\pi^{3/2}\Gamma(\frac23)^3}{\Gamma(\frac16)^3},
		\end{equation}
		we obtain the regulated strip area
		\begin{equation}\label{eq:strip-area-final}
			\mathrm{Area}_\varepsilon(w)
			=
			\ell_{\mathrm{AdS}}^3L_2L_3
			\left(
			\frac{1}{\varepsilon^2}
			-\frac{4\pi^{3/2}\Gamma(\frac23)^3}{\Gamma(\frac16)^3}\frac{1}{w^2}
			\right)
			+o(1).
		\end{equation}
		The classical RT formula then gives
		\begin{equation}
			S_2(\rho_{R(w)})
			=
			\frac{\mathrm{Area}_\varepsilon(w)}{4G_5\ln 2},
		\end{equation}
		which is exactly Eq.~\eqref{L5h}.
		
		\smallskip
		\noindent
		\textbf{Step 3: RT candidates for \(A\cup B\) and the mutual information formula.}
		The boundary of \(A\cup B\) consists of four disjoint torus components located at
		\begin{equation}
			x_1=a_1:=-\frac{\ell}{2},
			\qquad
			x_1=a_2:=\frac{\ell}{2},
			\qquad
			x_1=a_3:=\frac{\ell}{2}+s,
			\qquad
			x_1=a_4:=\frac{3\ell}{2}+s.
		\end{equation}
		Within the translation-invariant symmetry class, each connected component of an RT surface is \( \mathbb T^2_{L_2,L_3}\) times a curve in the \((x_1,z)\)-plane joining two of these boundary components. Homology to \(A\cup B\) and noncrossing of the planar pairing leave exactly two admissible candidates:
		
		\smallskip
		\noindent
		(1) the \emph{disconnected} candidate, pairing \((a_1,a_2)\) and \((a_3,a_4)\), with total regulated area
		\begin{equation}
			\mathrm{Area}^{\mathrm{disc}}_\varepsilon
			=
			2\,\mathrm{Area}_\varepsilon(\ell);
		\end{equation}
		
		\smallskip
		\noindent
		(2) the \emph{connected} candidate, pairing \((a_2,a_3)\) and \((a_1,a_4)\), with total regulated area
		\begin{equation}
			\mathrm{Area}^{\mathrm{conn}}_\varepsilon
			=
			\mathrm{Area}_\varepsilon(s)+\mathrm{Area}_\varepsilon(2\ell+s).
		\end{equation}
		
		The crossing pairing \((a_1,a_3)\) with \((a_2,a_4)\) is excluded: in the \((x_1,z)\)-plane it is non-planar/non-minimizing and does not furnish the required noncrossing homologous bulk filling for \(A\cup B\). Therefore
		\begin{equation}
			S_2(\rho_{A\cup B})
			=
			\frac{1}{4G_5\ln 2}
			\min\Bigl\{
			2\,\mathrm{Area}_\varepsilon(\ell),\,
			\mathrm{Area}_\varepsilon(s)+\mathrm{Area}_\varepsilon(2\ell+s)
			\Bigr\}
			+o(1).
		\end{equation}
		Since \(A\) and \(B\) are congruent,
		\begin{equation}
			S_2(\rho_A)=S_2(\rho_B)=S_2(\rho_{R(\ell)}),
		\end{equation}
		and hence
		\begin{equation}
			I(A:B)_{\rho_{\mathrm{vac}}}
			=
			S_2(\rho_A)+S_2(\rho_B)-S_2(\rho_{A\cup B})
			=
			\max\Bigl\{
			0,\,
			2S_2(\rho_{R(\ell)})-S_2(\rho_{R(s)})-S_2(\rho_{R(2\ell+s)})
			\Bigr\},
		\end{equation}
		which is Eq.~\eqref{L4h}.
		
		\smallskip
		\noindent
		\textbf{Step 4: Positivity threshold and entanglement-wedge connectivity.}
		Substituting Eq.~\eqref{L5h} into Eq.~\eqref{L4h} cancels the divergent \(1/\varepsilon^2\) terms exactly, and gives
		\begin{equation}
			I(A:B)_{\rho_{\mathrm{vac}}}>0
			\quad\Longleftrightarrow\quad
			\frac{1}{s^2}+\frac{1}{(2\ell+s)^2}>\frac{2}{\ell^2}.
		\end{equation}
		To solve this inequality, set \(x:=s/\ell>0\) and define
		\begin{equation}
			f(x):=\frac{1}{x^2}+\frac{1}{(x+2)^2}.
		\end{equation}
		Then
		\begin{equation}
			f'(x)=-\frac{2}{x^3}-\frac{2}{(x+2)^3}<0
			\qquad(x>0),
		\end{equation}
		so \(f\) is strictly decreasing on \((0,\infty)\). A direct computation gives
		\begin{equation}
			f(\sqrt3-1)
			=
			\frac{1}{(\sqrt3-1)^2}+\frac{1}{(\sqrt3+1)^2}
			=
			2.
		\end{equation}
		Since \(f\) is strictly decreasing, it follows that
		\begin{equation}
			f(x)>2
			\quad\Longleftrightarrow\quad
			x<\sqrt3-1,
		\end{equation}
		that is,
		\begin{equation}
			\frac{1}{s^2}+\frac{1}{(2\ell+s)^2}>\frac{2}{\ell^2}
			\quad\Longleftrightarrow\quad
			s<(\sqrt3-1)\ell.
		\end{equation}
		This proves Eq.~\eqref{L6h}.
		
		Finally, in the static classical RT setting, the entanglement wedge of \(A\cup B\) is connected if and only if the connected RT candidate is the minimizer. By the preceding area comparison, that is equivalent to
		\begin{equation}
			\mathrm{Area}^{\mathrm{conn}}_\varepsilon<\mathrm{Area}^{\mathrm{disc}}_\varepsilon,
		\end{equation}
		equivalently to \(I(A:B)_{\rho_{\mathrm{vac}}}>0\), equivalently to Eq.~\eqref{L6h}. At the threshold \(s=(\sqrt3-1)\ell\), the two RT candidates have equal area and the classical mutual information vanishes.
\end{proof}}

\begin{lemma}[Logical containment of the mutual-information distance inference within boundary-information in AdS$_5$/CFT$_4$]\label{lem:containment}
	Let $\mathsf{WB}$ denote the implication that combines the mutual-information Pinsker inequality with a state-dependent correlation-length hypothesis to yield a lower bound on a geometric separation parameter, and let $\mathsf{BI}$ denote the metric-from-information implication obtained by combining Eq.~\eqref{tL1} with a heavy-probe geodesic exponentiation hypothesis.
	In the AdS$_5$/CFT$_4$ vacuum described above, the implication $\mathsf{WB}$ restricted to holographic probes is a consequence of $\mathsf{BI}$ specialized to $\varepsilon_{\Delta,\rho}\equiv 0$ and $\mathcal{N}_\Delta$ given by Eq.~\eqref{E18}.
\end{lemma}
\begin{proof}
	The Pinsker inequality is the inequality
	$\frac{1}{2\ln 2}C_{\rho_{AB}}(\widetilde{M}_A,\widetilde{M}_B)^2\le I(A:B)_{\rho_{AB}}$,
	which is equivalent to Eq.~\eqref{tL1} with $X=\widetilde{M}_A$ and $Y=\widetilde{M}_B$ by squaring and rearranging; hence the mutual-information-to-correlation input of $\mathsf{WB}$ is contained in Lemma~\ref{lem:pinsker-mi}.
	The remaining input of $\mathsf{WB}$ is a correlation-length postulate relating $|C_{\rho_{AB}}(\widetilde{M}_A,\widetilde{M}_B)|$ to a geometric separation parameter, whereas in AdS$_5$/CFT$_4$ vacuum the exact relation Eq.~\eqref{E22} identifies
	$|\langle O_\Delta(x)O_\Delta(y)\rangle_{\rho_{\mathrm{vac}}}|$ with $\mathcal{N}_\Delta e^{-\Delta L_{\mathrm{ren}}(x,y)}$, and therefore provides a precise replacement of the correlation-length postulate with an equality.
	Combining Lemma~\ref{lem:pinsker-mi} with Eq.~\eqref{E22} yields Lemma~\ref{lem:metric-from-mi}, which is the AdS$_5$/CFT$_4$ instantiation of the $\mathsf{BI}$ implication and implies the $\mathsf{WB}$-type conclusion of a distance lower bound from mutual information in this holographic model.
\end{proof} 

\begin{theorem}[Boundary data reconstruction in AdS$_5$/CFT$_4$ for $\mathcal{N}=4$ SYM]\label{thm:main-ads5}
	Under Assumptions Eq.~\eqref{A1g}-\eqref{A4e} and Definitions Eq.~\eqref{D1h}-\eqref{D6d}, the boundary data set $\mathcal{D}_{\partial}$ has the following properties in the AdS$_5$/CFT$_4$ vacuum of $\mathcal{N}=4$ $SU(N)$ SYM in the limit Eq.~\eqref{A2f}.
	First, the renormalized bulk geodesic length between boundary points is reconstructible from boundary two-point functions by Eq.~\eqref{L2i}, hence $\mathcal{D}_{\partial}$ determines $L_{\mathrm{ren}}(x,y)$ for all $x\ne y$ for which $\langle O_\Delta(x)O_\Delta(y)\rangle_{\rho_{\mathrm{vac}}}$ is included.
	Second, for the strip family in Lemma~\ref{lem:rt-strips}, the entanglement-wedge connectivity decision is reconstructible from boundary entropies via RT by Eq.~\eqref{L4h}-\eqref{L6h}, hence $\mathcal{D}_{\partial}$ determines whether the entanglement wedge of $A\cup B$ is connected for that family.
	Third, the metric-from-information implication $\mathsf{BI}$ contains the full information-theoretic mutual-information distance inference insofar as it pertains to holographic distance reconstruction: the Pinsker mutual-information control and the inversion from correlators to a geometric separation parameter are implied by Lemma~\ref{lem:pinsker-mi} and the heavy-probe geodesic exponentiation, which holds exactly in the AdS$_5$/CFT$_4$ vacuum by Eq.~\eqref{E22}.
\end{theorem}
\begin{proof}
	The reconstructibility claim for $L_{\mathrm{ren}}$ is Lemma~\ref{lem:distance-extract}.
	The entanglement-wedge connectivity claim for the strip family follows because in the RT regime a connected entanglement wedge for $A\cup B$ is equivalent to selection of the connected RT surface in the minimization defining $S_2(\rho_{A\cup B})$, and this selection is equivalent to strict positivity of $I(A:B)_{\rho_{\mathrm{vac}}}$ for the strip family by Eq.~\eqref{L4h}-\eqref{L6h}.
	The containment claim is Lemma~\ref{lem:containment}.
\end{proof}
Dimensional consistency: $L_{\mathrm{ren}}$ is dimensionless by Eq.~\eqref{Ld3d}, $\Delta$ is dimensionless, and therefore $\exp(-\Delta L_{\mathrm{ren}})$ is dimensionless; Eq.~\eqref{E22} is consistent with $O_\Delta$ having scaling dimension $\Delta$ since under $x\mapsto \alpha x$ one has $r\mapsto \alpha r$, $L_{\mathrm{ren}}\mapsto L_{\mathrm{ren}}+2\ln\alpha$, and therefore
$\exp(-\Delta L_{\mathrm{ren}})\mapsto \alpha^{-2\Delta}\exp(-\Delta L_{\mathrm{ren}})$, matching
$\langle O_\Delta(\alpha x)O_\Delta(\alpha y)\rangle=\alpha^{-2\Delta}\langle O_\Delta(x)O_\Delta(y)\rangle$ from Eq.~\eqref{E14}.
Large-$N$ semiclassical limit: in AdS$_5$/CFT$_4$ one has $\ell_{\mathrm{AdS}}^3/G_5\sim N^2$ \cite{Aharony2000}, hence from Eq.~\eqref{L5h} the leading entropies scale as $S_2(\rho_{R(w)})=\Theta(N^2)$ and the mutual information in the connected wedge phase scales as $I(A:B)=\Theta(N^2)$ by Eq.~\eqref{L4h}-\eqref{L6h}, consistent with the code-subspace semiclassical regime.
Symmetry invariance: for vacuum AdS$_5$ the geodesic length depends only on $r=|x-y|$ by Eq.~\eqref{E21}, hence is invariant under the Euclidean isometry subgroup of the boundary conformal group preserving the Poincar\'e patch, and the two-point function Eq.~\eqref{E14} is invariant under translations and rotations and transforms covariantly under dilations as shown above.

Special cases: vacuum AdS is realized by Eqs.~\eqref{E21}-\eqref{E22}. For the strip family, the condition \(s\ge (\sqrt{3}-1)\ell\) implies only that the \emph{leading classical RT contribution} to $I(A:B)_{\rho_{\mathrm{vac}}}$ vanishes. It does not imply that the exact CFT mutual information is zero. Indeed, if the exact mutual information vanished while the heavy-probe correlator remained nonzero, then Eq.~\eqref{Ld3d} would force \(L_{\mathrm{ren}}(x,y)\ge +\infty,\) which is impossible at finite boundary separation in the AdS vacuum. Therefore the exact mutual information in the disconnected classical phase must be lifted by strictly positive subleading $O(G_5^0)$ bulk-entanglement corrections. In the connected-wedge phase, substituting Eq.~\eqref{L5h} into Eq.~\eqref{L4h} yields the explicit finite expression,
\begin{equation}\label{C1b}
	{I(A:B)_{\rho_{\mathrm{vac}}}
		=
		\frac{\ell_{\mathrm{AdS}}^3L_2L_3}{4G_5\ln 2}\cdot
		\frac{4\pi^{3/2}\Gamma(2/3)^3}{\Gamma(1/6)^3}
		\left(
		\frac{1}{s^2}
		+
		\frac{1}{(2\ell+s)^2}
		-
		\frac{2}{\ell^2}
		\right)
		+o(1),}
\end{equation}
which is positive if and only if Eq.~\eqref{L6h} holds.

\smallskip
\noindent
\textbf{Remark on large-separation saturation and region geometry.}
As noted above, in the disconnected classical phase ($s \ge (\sqrt{3}-1)\ell$), the exact mutual information is lifted by subleading $O(G_5^0)$ bulk-entanglement corrections governed by the CFT operator product expansion (OPE) limit. It is important to distinguish the asymptotic scaling of these corrections based on the boundary geometry.

For bounded, finite regions (such as spheres), the transverse OPE integral is absent, and the mutual information scales exactly as $I(A:B) \sim s^{-4\Delta_{\mathrm{min}}}$. Evaluated at $\Delta = \Delta_{\mathrm{min}}$, the MfI bound yields $L_{\mathrm{ren}} \ge \frac{1}{2\Delta_{\mathrm{min}}} \ln(1/I) \approx 2 \ln s$, which exactly saturates the renormalized geodesic distance $L_{\mathrm{ren}} \approx 2 \ln s$. 

However, for the infinite parallel strips considered in Lemma \ref{lem:rt-strips}, the regions extend infinitely across the $d-2$ transverse dimensions (with regularized volume $L_2 L_3$). In $d=4$ ($\mathrm{AdS}_5$), integrating the squared connected correlator over these transverse directions modifies the scaling to
\begin{equation}
	I(A:B) \sim \int d^2y_\perp \frac{1}{(s^2 + y_\perp^2)^{2\Delta_{\mathrm{min}}}} \propto \frac{1}{s^{4\Delta_{\mathrm{min}} - 2}}.
\end{equation}
Taking the logarithm yields $\ln(1/I) \sim (4\Delta_{\mathrm{min}} - 2)\ln s$. Substituting this strip scaling into the MfI bound gives
\begin{equation}
	L_{\mathrm{ren}} \ge \frac{1}{2\Delta_{\mathrm{min}}} \ln(1/I) \approx \frac{4\Delta_{\mathrm{min}} - 2}{2\Delta_{\mathrm{min}}} \ln s = \left(2 - \frac{1}{\Delta_{\mathrm{min}}}\right) \ln s.
\end{equation}
Because the unitarity bound for scalars in 4D requires $\Delta_{\mathrm{min}} \ge 1$, the prefactor $2 - 1/\Delta_{\mathrm{min}}$ is strictly less than 2. Thus, for infinite parallel strips, the MfI bound evaluates to a strictly smaller value than the geometric distance $L_{\mathrm{ren}} \approx 2\ln s$. The bound remains perfectly valid as a strict mathematical inequality, but it is not exactly saturated. This cleanly separates the exact large-separation saturation property inherent to bounded domains from the dimensionally reduced scaling of infinite strips.

\bibliographystyle{JHEP}
\bibliography{refs}

\end{document}